\documentclass{article}
\usepackage{authblk}
\usepackage{hyperref}
\usepackage[left=3cm,top=3cm,right=3cm,bottom=3cm]{geometry}
\usepackage{amsmath}
\usepackage{amsfonts}
\usepackage{amsthm}
\usepackage{graphicx}
\usepackage{caption}
\usepackage{float}
\restylefloat{table}
\usepackage{bbm}
\usepackage{comment}
\usepackage{tikz-cd}
\usepackage{MnSymbol}
\usepackage{complexity}
\usepackage{physics}
\usepackage{mwe}
\usepackage{subfig}
\usepackage{enumerate}
\usepackage{cleveref}

\newcounter{mycount}

\usepackage{array}
\newcolumntype{C}[1]{>{\centering\arraybackslash}m{#1}}

\newtheorem{theorem}{Theorem}[section]
\newtheorem{definition}{Definition}[section]
\newtheorem{lemma}{Lemma}[section]
\newtheorem{corollary}{Corollary}[section]
\newtheorem{proposition}{Proposition}[section]
\newtheorem{claim}{Claim}[section]
\newtheorem{fact}{Fact}[section]
\newtheorem{problem}{Problem}[section]

\DeclareMathOperator{\spec}{spec}
\DeclareMathOperator{\ply}{poly}
\DeclareMathOperator{\im}{Im}
\let\ker\relax
\DeclareMathOperator{\ker}{Ker}
\DeclareMathOperator{\Kcyc}{\mathcal{K}}

\DeclareMathOperator{\J}{\mathcal{J}}
\DeclareMathOperator{\minusST}{|-\rangle}
\DeclareMathOperator{\z}{| 0 \rangle}
\DeclareMathOperator{\1}{|1\rangle}

\DeclareMathOperator{\phiST}{|\phi\rangle}
\DeclareMathOperator{\identity}{\mathbb{I}}
\DeclareMathOperator{\Cl}{Cl}

\DeclareMathOperator{\Hbravyi}{\mathit{H}_{\textrm{Bravyi}}}
\DeclareMathOperator{\Hin}{\mathit{H}_{\textrm{in}}}
\DeclareMathOperator{\Hout}{\mathit{H}_{\textrm{out}}}
\DeclareMathOperator{\Hprop}{\mathit{H}_{\textrm{prop}}}
\DeclareMathOperator{\Hclock}{\mathit{H}_{\textrm{clock}}}
\DeclareMathOperator{\Hhs}{\mathit{H}_{\textrm{hs}}}

\newcommand{\qmsat}{{\textup{\sc Quantum $m$-$\SAT$}}\xspace}
\newcommand{\qfsat}{{\textup{\sc Quantum $4$-$\SAT$}}\xspace}
\newcommand{\qtsat}{{\textup{\sc Quantum $2$-$\SAT$}}\xspace}
\newcommand{\YES}{{\textup{YES}}\xspace}
\newcommand{\NO}{{\textup{NO}}\xspace}

\newcommand{\Hpropt}{H_\textrm{prop,t}}

\definecolor{tk}{RGB}{246,76,246}

\newcommand{\CNOT}{\text{CNOT}}

\begin{document}

\title{Gapped Clique Homology on weighted graphs is $\QMA_1$-hard and contained in $\QMA$}
\author[1]{Robbie King}
\author[2,3]{Tamara Kohler}
\affil[1]{\small Department of Computing and Mathematical Sciences, Caltech, Pasadena, USA}
\affil[2]{\small Instituto de Ciencias Matemáticas, Madrid, Spain}
\affil[3]{\small Department of Computer Science, Stanford University}
\date{}
\maketitle

\begin{abstract}
We study the complexity of a classic problem in computational topology, the homology problem: given a description of some space $X$ and an integer $k$, decide if $X$ contains a $k$-dimensional hole.
The setting and statement of the homology problem are completely classical, yet we find that the complexity is characterized by quantum complexity classes.
Our result can be seen as an aspect of a connection between homology and supersymmetric quantum mechanics \cite{witten1982supersymmetry}.

We consider clique complexes, motivated by the practical application of topological data analysis (TDA).
The clique complex of a graph is the simplicial complex formed by declaring every $k+1$-clique in the graph to be a $k$-simplex.
Our main result is that deciding whether the clique complex of a weighted graph has a hole or not, given a suitable promise on the gap, is $\QMA_1$-hard and contained in $\QMA$.

Our main innovation is a technique to lower bound the eigenvalues of the combinatorial Laplacian operator. 
For this, we invoke a tool from algebraic topology known as \emph{spectral sequences}. In particular, we exploit a connection between spectral sequences and Hodge theory \cite{forman1994hodge}.
Spectral sequences will play a role analogous to perturbation theory for combinatorial Laplacians.
In addition, we develop the simplicial surgery technique used in prior work \cite{crichigno2022clique}.

Our result provides some suggestion that the quantum TDA algorithm \cite{lloyd2016quantum} \emph{cannot} be dequantized. More broadly, we hope that our results will open up new possibilities for quantum advantage in topological data analysis.
\end{abstract}

\pagebreak
\setcounter{tocdepth}{1}
\tableofcontents

\pagebreak
\section{Introduction}

In quantum complexity theory, the goal is to characterize the capacity of quantum computers for solving computational problems.
There is no requirement that the problems considered be quantum mechanical in nature, but most of the key results in the field do consider such problems. See \cite{osborne2012hamiltonian, gharibian2015quantum} for reviews.
The field has been important both for understanding the potential of quantum computers and for giving insight into fundamental physics.
However, the capacity of quantum computers to solve problems that do not appear inherently quantum mechanical is less well understood.

In this work, we examine a classic problem in computational topology and find that it can be characterized by quantum complexity classes.
Topology studies properties of spaces that only depend on the continuity and connectivity between points and do not depend on distances between points.
In particular, we are concerned with \emph{homology}, a branch of topology that describes $k$-dimensional holes in a topological space.
For example, a circle constitutes a 1-dimensional hole, a hollow sphere a 2-dimensional hole, and so on.
The types of spaces we consider are motivated by the practical application of topological data analysis (TDA).
Our result thus has implications for quantum advantage in TDA and the quantum TDA algorithm introduced in \cite{lloyd2016quantum}.
In TDA, one applies techniques from topology to extract global information from data in a way that is resistant to local noise.
For background on TDA and its applications, see \cite{wasserman2018topological,petri2014homological,giusti2016two,reimann:hal-01706964}.

\medskip
{\noindent \bf Results.}
\smallskip

The computational problem we consider is perhaps the simplest question in homology -- Given some space, does it have a $k$-dimensional hole or not?
Our main result is that, for certain types of input spaces and a suitable promise on the gap, this problem is $\QMA_1$-hard and contained in $\QMA$. The complexity class $\QMA$ is the quantum analogue of the class $\MA$, and $\QMA_1$ is a one-sided error version of $\QMA$. (See \Cref{sec:q complexity} for full definitions.)

A problem of interest in TDA is to compute the number of holes in the \emph{clique complex} of a graph -- the clique complex of a graph is the simplicial complex formed by mapping every $k+1$-clique in the graph to a $k$-simplex. (See \Cref{sec:clique} for formal definition.)
We consider the problem of deciding whether a weighted graph's clique complex has a $k$-dimensional hole or not.
The decision version of the clique homology problem was shown to be $\QMA_1$-hard in \cite{crichigno2022clique}.
However, the problem is only believed to be inside $\QMA$ when the combinatorial Laplacian has a promise gap.
Our key contribution is showing that the clique homology problem remains $\QMA_1$-hard when the gap is imposed. Thus we can provide upper and lower bounds on the complexity of this classical problem in terms of quantum complexity classes.

\begin{problem} \label{informal_prom_prob}
\emph{(Gapped clique homology)}
Fix functions $k : \mathbb{N} \rightarrow \mathbb{N}$ and $g : \mathbb{N} \rightarrow [0,\infty)$, with $g(n) \geq 1 / \ply{n}$, $k(n) \leq n$. The input to the problem is a vertex-weighted graph $\mathcal{G}$ on $n$ vertices. The task is to decide whether:
\begin{itemize}
    \item {\bf \YES} \ the $k(n)^{\text{th}}$ homology group of $\Cl(\mathcal{G})$ is non-trivial $H_k(\mathcal{G}) \neq 0$
    \item {\bf \NO} \ the $k(n)^{\text{th}}$ homology group of $\Cl(\mathcal{G})$ is trivial $H_k(\mathcal{G}) = 0$ and the weighted combinatorial Laplacian $\Delta_k$ has minimum eigenvalue $\lambda_{\min}(\Delta_k) \geq g(n)$.
\end{itemize}
\end{problem}

The condition on the minimum eigenvalue $\lambda_{\min}(\Delta_k)$ can be interpreted as a promise that the graph is far from having a hole in the NO case; see below.

\begin{theorem} \label{QMA1_cor} 
\Cref{informal_prom_prob} is $\QMA_1$-hard and contained in $\QMA$.
\end{theorem}

At this point, we would like to make two clarifying remarks on our main result. First of all, one may wonder why we cannot get tight bounds on the complexity of \Cref{informal_prom_prob}; for example, could it be $\QMA$-hard, or contained in $\QMA_1$? We suspect that the true complexity of \Cref{informal_prom_prob} is $\QMA_1$. However, this is a somewhat fragile argument to make since containment in $\QMA_1$ depends on the specific choice of universal gate set for the verifier quantum circuit. Secondly, we would like to clarify the comparison to the previous results of \cite{crichigno2022clique}. By analogy to the local Hamiltonian problem without a promise gap, we can expect the decision version of the clique homology problem, as considered in \cite{crichigno2022clique}, to be in fact PSPACE-hard. Since PSPACE contains $\QMA$ and $\QMA_1$, the statement that the decision version is $\QMA_1$-hard does not provide compelling evidence that the clique homology problem has a quantum structure. On the other hand, our results study a related problem for which the complexity can be upper and lower bounded by quantum complexity classes.

\medskip
{\noindent \bf Techniques.}
\smallskip

The main technical contribution is the development of tools which allow us to lower bound the eigenvalues of the Laplacian operator from Hodge theory. With this purpose, we introduce a new technique to the field -- a powerful tool from algebraic topology known as \emph{spectral sequences}. In particular, we exploit a connection between spectral sequences and Hodge theory presented in \cite{forman1994hodge}. In homology, spectral sequences can play a powerful role analogous to perturbation theory in the analysis of perturbative gadgets \cite{kempe2006complexity}. The use of spectral sequences to perform a kind of perturbation theory on the Laplacians of simplicial complexes is novel to quantum information theory. Further, the ability to lower bound the Laplacian eigenvalues may be of independent interest. In addition, our work extends the simplicial surgery technique used in prior quantum complexity work \cite{crichigno2022clique}.

Our hardness proof will proceed by reducing from a particular local Hamiltonian problem. We would like to encode the Hamiltonian problem into the gapped clique homology problem. For this purpose, our main focus will be to establish the following theorem. (See Theorem \ref{main_thm} for formal version.)

\begin{theorem}\label{informal_main_thm}
\emph{(Main theorem, informal)}
Given a local Hamiltonian* $H$ on $n$ qubits, we can construct a vertex-weighted graph $\mathcal{G}$ on $\ply{n}$ vertices and a $k$ such that the combinatorial Laplacian $\Delta_k(\Cl(\mathcal{G}))$ satisfies
\begin{align*}
\lambda_{\min}(H) = 0 \ &\implies \ \lambda_{\min}(\Delta_k(\Cl(\mathcal{G}))) = 0 \\
\lambda_{\min}(H) \geq \frac{1}{\ply{n}} \ &\implies \ \lambda_{\min}(\Delta_k(\Cl(\mathcal{G}))) \geq \frac{1}{\ply{n}}
\end{align*}
* There are some conditions on the form of $H$, but the class is sufficiently expressive to be $\QMA_1$-hard.
\end{theorem}

In the language of the theorem, the bulk of the work is to establish $\lambda_{\min}(\Delta_k(\Cl(\mathcal{G}))) \geq 1/\ply{n}$ in the case $\lambda_{\min}(H) \geq 1/\ply{n}$. Here, spectral sequences will be essential.

\medskip
{\noindent \bf Spectral gap of Laplacian.}
\smallskip

How can we interpret the gap at the bottom of the spectrum of the combinatorial Laplacian $\lambda_{\min}(\Delta_k) \geq 1/\ply{n}$? At $k=0$, the combinatorial Laplacian $\Delta^0$ is equal to the usual graph Laplacian $L$ plus the projector onto the constant vector:
\begin{equation}
\Delta^0 = L +
\begin{pmatrix}
    1 & \dots & 1 \\
    \vdots & & \vdots \\
    1 & \dots & 1
\end{pmatrix}
\end{equation}
Thus the smallest eigenvalue of $\Delta^0$ corresponds to the first non-zero eigenvalue of the graph Laplacian $L$. This is precisely the eigenvalue which controls graph expansion; it appears in the well-known Cheeger inequality which relates the Laplacian spectrum to geometric connectivity of the graph \cite{levin2017markov}. This provides a geometric interpretation of the minimum eigenvalue of $\Delta^0$ -- it measures how far the graph is from being disconnected. A similar geometric interpretation holds for higher dimensional combinatorial Laplacians $\Delta_k$. Indeed, higher-dimensional Cheeger inequalities have been studied \cite{gundert2014higher,steenbergen2014cheeger,parzanchevski2016isoperimetric}, with connections to the field of \emph{high-dimensional expanders} \cite{lubotzky2018high}. A large minimum eigenvalue $\lambda_{\min}(\Delta_k)$ means that the graph is far from having a $k$-dimensional hole.

This leads us to an exciting future direction: Can we use graph operations and gadgets to perform gap amplification on the combinatorial Laplacian? In light of our $\QMA_1$-hardness result, this may have connections to the yet elusive \emph{quantum PCP conjecture} \cite{aharonov2013guest}.

\medskip
{\noindent \bf Implications.}
\smallskip

Related to deciding the existence of a hole is the problem of computing the normalized number of holes. For this problem there is an efficient quantum algorithm, known as the quantum TDA algorithm. (For a discussion of how this algorithm works, and what precisely we mean by `normalized number of holes' we refer readers to \Cref{sec:prior work}.)
A significant motivation for our work was understanding the complexity of the problem solved by this quantum TDA algorithm and its speedup over classical algorithms \cite{lloyd2016quantum, gunn2019review, gyurik2022towards, ubaru2021quantum, hayakawa2022quantum, mcardle2022streamlined, berry2022quantifying, akhalwaya2022towards, schmidhuber2022complexity, apers2022simple}.
Unlike many other quantum computing applications in machine learning \cite{tang2019quantum, gilyen2018quantum, chia2022sampling}, the quantum TDA algorithm has resisted \emph{dequantization}, and researchers still debate the presence of a speedup over the best possible classical algorithm.

Our result can be seen as providing some suggestion that the quantum TDA \emph{cannot} be dequantized.
We have shown that deciding whether clique complexes have holes is just as hard as deciding if a generic local Hamiltonian is frustration-free.
We can likewise translate the problem solved by the quantum TDA algorithm into a problem phrased in local Hamiltonians.
It is known that for generic local Hamiltonians this problem is very unlikely to be tractable on a classical computer. More precisely, it was shown in \cite{gyurik2022towards} that this problem is $\textrm{DQC}1$-hard (see \Cref{sec:prior work}).
However, the $\textrm{DQC}1$-hardness for generic Hamiltonians is inconclusive since some classical algorithms may exist that could exploit some unique structure in clique complexes to outperform algorithms for generic quantum Hamiltonians.
Our work suggests that problems involving clique complexes do \emph{not} possess exploitable structure, and are just as hard as the corresponding problem on a general Hamiltonian.

The results in our paper do not answer these questions conclusively -- our reduction does not immediately imply that \emph{all} problems regarding clique complexes are just as hard as the corresponding problem translated into the language of quantum Hamiltonians.
Nevertheless, we have developed techniques which are able to lower bound the eigenvalues of the Laplacian operator when reducing from a quantum Hamiltonian. 
We anticipate this could open up new possibilities in searching for quantum advantage in topological data analysis.
We elaborate on this and other open questions in \Cref{sec:open}.

\medskip
{\noindent \bf Discussion.}
\smallskip

The problem of deciding whether or not a space contains a hole makes no explicit reference to quantum mechanics, so it is surprising that its complexity turns out to be characterised by quantum complexity classes. Other examples of this kind are rare.
One story to compare to is that of the \emph{Jones polynomial}. Estimating the Jones polynomial, an invariant from knot theory, was shown to be BQP-complete in \cite{aharonov2006polynomial}.
That result is an aspect of a deep connection between topological quantum field theory and knot theory. 
Similarly, our result is an aspect of a deep connection between \emph{supersymmetry} and homology, which was previously explored in \cite{cade2021complexity,Crichigno:2020vue,crichigno2022clique}.
(We refer readers to \Cref{sec:SUSY} for a detailed discussion of this connection.)
Incidentally, both of these connections were explored in the 1980s by Witten \cite{witten1982supersymmetry, witten1989quantum}.

These results in quantum complexity suggest that a fruitful avenue for studying the possibility of quantum advantage in seemingly classical problems is to look for `hidden quantumness' -- mathematical problems that, at first glance, do not appear quantum; but can be mapped to specific families of quantum systems. 
One of the critical areas where quantum computers will offer practical advantanges over their classical counterparts is in studying quantum systems. By looking for such examples of `hidden quantumness', we may be able to extend the utility of quantum computation into more fields.

\medskip
{\noindent \bf Paper outline.}
\smallskip

In \Cref{sec:background} we provide the necessary background for our setting, before giving an overview of the proof of our main results in \Cref{sec:proof_overview}.
We provide an overview of related works in \Cref{sec:prior work}, and discuss open questions raised by our work in \Cref{sec:open}.
\Cref{sec:SUSY} is dedicated to explaining the link with supersymmetric quantum mechanics, and outlines the intuition behind the proof in this picture.
This closes the first part of the paper.

The second part of the paper develops techniques to rigorously prove our main results.
In \Cref{sec:prelims} we cover the necessary technical preliminaries.
In \Cref{sec:construction} we construct the gadgets needed for our reduction.
\Cref{spec_seq_sec} takes the gadgets and analyzes properties of their spectrum using spectral sequences.
In \Cref{sec:combining gadgets full} we describe how to combine many gadgets together, and analyze the spectrum of the combined complex, which constitutes our main technical theorem.
Finally, in \Cref{sec:final} we bring everything together to show our main result.

\section{Background} \label{sec:background}

In this section we outline essential background from simplicial homology and quantum complexity in order to understand the results.
More technical preliminaries required for our proof are contained in \Cref{sec:prelims}.

\subsection{Simplicial homology} \label{sec:homology_informal}

For a rigorous introduction to homology we refer readers to \Cref{sec:homology_formal} in the technical part of the paper.
Here we aim to introduce the concepts in enough detail to understand the statement of the main result, and the intuition behind the proof.

The building blocks of simplicial complexes are \emph{simplices}. 
Simplices can be thought of as the generalisation of triangles and tetrahedron to arbitrary dimensions.
A 0-simplex is a point, a 1-simplex is a line, 2- and 3-simplices are triangles and tetrahedra respectively.
In higher dimensions a $k$-simplex is defined as a $k$-dimensional polytope which is the convex hull of $k+1$-vertices.
We can denote a $k$-simplex by its vertices:
\begin{equation*}
\sigma = [x_0\ldots x_k]
\end{equation*}
where the $x_i$ are the vertices of the simplex. 
Simplices are oriented, with the orientation induced by the ordering of the vertices.
So permuting the vertices in a simplex leads to an equivalent simplex, possibly up to an overall sign:
\begin{equation*}
[x_{\pi(0)}\ldots x_{\pi(k)}] = \textrm{sgn}(\pi)[x_0\ldots x_k]
\end{equation*}

A \emph{simplicial complex} $\mathcal{K}$ is a collection of simplices satisfying two requirements: (A) if a simplex is in $\mathcal{K}$ then all its faces are also in $\mathcal{K}$ and (B) the intersection of any two simplices in $\mathcal{K}$ is a face of both the simplices. Intuitively we can think of constructing a simplicial complex by gluing simplices together along faces.

In general, we can consider linear combinations of $k$-simplices, known as $k$-\emph{chains}. The space of $k$-chains forms a vector space. To define the notion of a hole in a simplicial complex we need to introduce the boundary operator $\partial^k$.
The boundary operator acts on $k$-simplices as:
\begin{equation*}
\partial^k [x_0 \ldots x_k] = \sum_{j=0}^k (-1)^j [x_0 \ldots \hat{x}_j \ldots x_k]
\end{equation*}
where the notation $[x_0 \cdots \hat{x}_j \cdots x_k]$ means that the $j^{\textrm{th}}$ vertex is deleted.
In \Cref{fig:boundary} we demonstrate the action of the boundary map on a 2-simplex.
As the name suggests, the boundary map acting on a simplex gives the boundary of that simplex.
The action of $\partial^k$ can be extended by linearity to collections of simplices.

\begin{figure}
\centering
\includegraphics[scale=0.5]{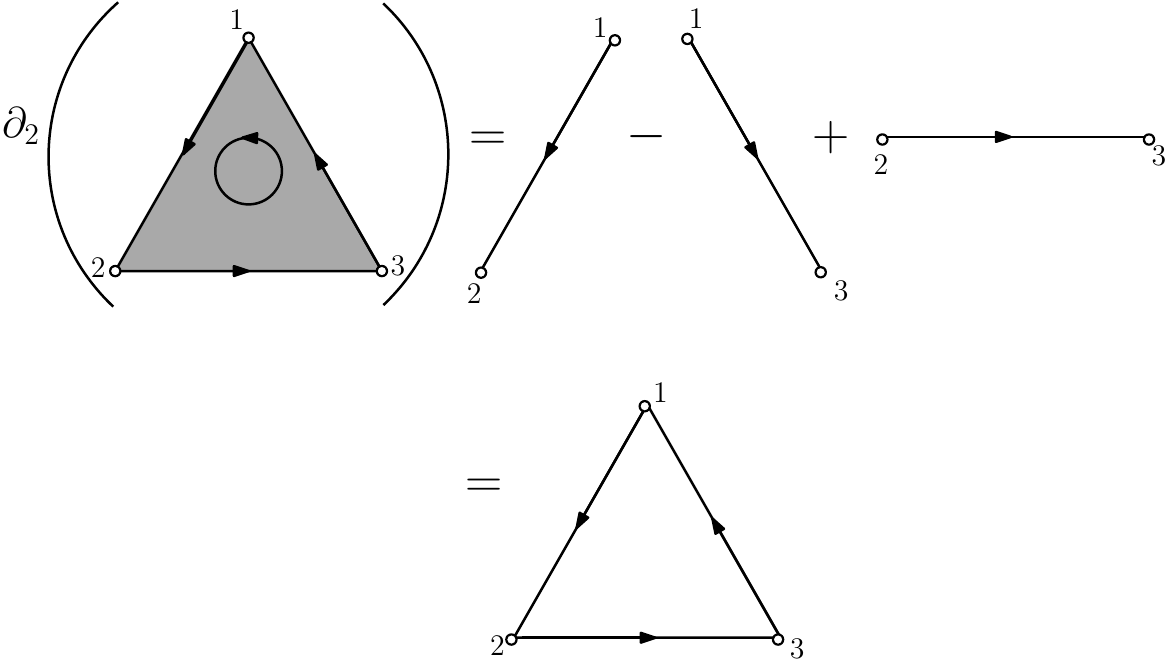}
\caption{The boundary operators action on a 2-simplex (i.e. a triangle).} \label{fig:boundary}
\end{figure}

We define any object that does not have a boundary as a cycle.
So a cycle $c$ satisfies $c \in \ker{\partial^k}$.
All boundaries are cycles, because boundaries don't themselves have a boundary.
In other words, the boundary operator is nilpotent:
\begin{equation} \label{boundary_nilpotent}
\partial^k \circ \partial^{k+1} =0
\end{equation}

How should we define what is a hole in a simplicial complex?
Intuitively a hole is a cycle which is not the boundary of anything. 
So a hole $h$ satisfies $h \in \ker{\partial^k}$, but there does not exist any $v$ such that $h = \partial^{k+1} v$.
Formally this means that holes are elements of the \emph{homology group}:
\begin{equation*}
H_k = \frac{\ker{\partial^k}}{\im{\partial^{k+1}}}
\end{equation*}
Note that the homology group is a quotient group, meaning that its elements are equivalence classes.
We can think of these equivalence classes as being sets of cycles that can be continuously deformed into each other.
Cycles which are boundaries can be continuously deformed to a single point, so these are trivial elements in homology.
If two non-trivial cycles cannot be continuously deformed into one another then they are the boundaries of different holes, so are different elements of homology.

It is possible to define a coboundary operator (see \Cref{sec:hodge} for details):
\begin{equation*}
d^k = (\partial^{k+1})^\dagger
\end{equation*}
Which can in turn be used to define the \emph{Laplacian}:
\begin{equation} \label{Laplacian_definition_informal}
\Delta_k = d^{k-1}\partial^k + \partial^{k+1}d^k
\end{equation}
which is a positive semi-definite operator. 
It can be shown (see \Cref{sec:hodge} for the proof) that $\ker(\Delta_k)$ is isomorphic to the homology group $H_k$.
In other words, a simplicial complex has a $k$-dimensional hole iff the Laplacian has a zero eigenvalue.

\subsection{Clique complexes}\label{sec:clique}

The computational complexity of determining whether or not a simplicial complex has a trivial homology group depends, of course, on how the simplicial complex is provided as input.
If we are given the simplicial complex as a list of simplices, the problem of deciding the homology is in $\P$. 
This is because we are doing linear algebra over a space whose dimension is equal to the number of $k$-simplices \cite{donald:1991}. 
To make the question more interesting, we would like a succinct description of the simplicial complex. 
This is the purpose of this section. Motivated from the practical task of topological data analysis, we will study clique complexes -- a class of simplicial complexes that can be represented by a graph $\mathcal{G}$.\footnote{We note there do exist succinct descriptions of general simplicial complexes e.g. by providing a list of vertices and maximal faces, or a list of vertices and minimal non-faces. We do not consider these input representations in this work.}

\begin{definition}
The \emph{clique complex} of a graph $\mathcal{G}$, denoted $\Cl(\mathcal{G})$, is the simplicial complex consisting of the \emph{cliques} of $\mathcal{G}$. A $k+1$-clique becomes a $k$-simplex. Throughout this paper, when there is no chance of confusion, we will sometimes abuse notation and write $\mathcal{G}$ when we are really referring to its clique complex $\Cl(\mathcal{G})$.
\end{definition}

Now the input size is the $n \times n$ adjacency matrix of $\mathcal{G}$, where $n$ is the number of vertices. Yet there could be up to $n \choose k+1$ $k$-simplices. If $k$ is growing with $n$, the number of $k$-simplices and hence the dimension of $\mathcal{C}^k(\mathcal{G})$ could be exponential in $n$. Combined with the relevance of the self-adjoint operator $\Delta_k$, we can start to see the emergence of quantum mechanical concepts in these homological objects. Hiding in this succinct graph is an exponential-dimensional Hilbert space $\mathcal{C}^k(\mathcal{G})$ with a Hamiltonian $\Delta_k$!

It should be noted that not all simplicial complexes arise as clique complexes. 
For example, the `hollow triangle' $\{[x_0x_1], [x_1x_2], [x_2x_0]\}$ cannot be the clique complex of any graph.

\begin{definition}
We say a simplicial complex $\mathcal{K}$ is \emph{2-determined} if it is the clique complex of its 1-skeleton graph $(\mathcal{K}^0,\mathcal{K}^1)$. A simplicial complex $\mathcal{K}$ is a clique complex if and only if it is 2-determined.
\end{definition}

\subsection{Quantum complexity theory}\label{sec:q complexity}

There are two complexity classes we will be interested in throughout this paper. The first, $\QMA$ is often referred to as the quantum analogue of the classical complexity class, $\NP$.
It is the set of problems where a proof (in the form of a quantum state) can be checked efficiently by a quantum computer. The second complexity class we deal with is a slight modification of $\QMA$ known as $\QMA_1$.
This is the `perfect completeness' version of $\QMA$. 
This means that in \YES cases we require the verifier to accept on valid witnesses with probability 1, while we still allow some probability of error in \NO cases. We postpone formal definitions to Section \ref{sec:qma1}.

The question of whether $\QMA_1$ is strictly contained within $\QMA$ is open.
The complexity classes are known to be distinct relative to a particular quantum oracle \cite{aaronson2008perfect}. 
However, the versions of $\QMA$ and $\QMA_1$ where the proofs are restricted to be classical bit strings \emph{are} equal, for certain choices of universal gate sets in the perfect completeness case \cite{jordan2012achieving}.
Equivalence of the classes for certain choices of universal gate set is also known to be true in the setting where there is exponentially small completeness-soundness gap \cite{fefferman2016complete}.

\section{Proof overview} \label{sec:proof_overview}

We stated \Cref{informal_prom_prob} without any reference to quantum mechanics, and \Cref{QMA1_cor} may appear surprising.
However, this classical problem does exhibit some characteristic properties of quantum mechanics.
Note first that the basis of the vector space $C_k$ are the $k$-simplices $\mathcal{G}^k$.
There could be as many as $n \choose k+1$ of these.
So if $k$ is growing linearly with $n$, the chainspace could have dimension exponential in $n$.
The emergence of an exponential-dimensional Hilbert space from a small object is a characteristic property of quantum mechanics.
Moreover, we can see the combinatorial Laplacian as playing the role of a quantum Hamiltonian, strengthening the link to quantum mechanics.
Viewing the setup in this manner, it is possible to demonstrate containment in $\QMA$ \cite{cade2021complexity}. 
Containment tells us we can frame the topological problem in quantum mechanical terms.
Demonstrating hardness is more involved, and is the main contribution of this work. 
Hardness gives us a converse: we can convert a quantum Hamiltonian to a topological object whose topology reflects the minimum eigenvalue of the Hamiltonian.

Let's first go over the argument for containment in $\QMA$.
The Laplacian $\Delta_k$ of $\Cl(\mathcal{G})$ is a sparse Hermitian operator.
If $\Cl(\mathcal{G})$ has a hole then $\Delta_k$ has a zero eigenvalue, and if it has no hole then by the gap every eigenvalue of $\Delta_k$ is bounded away from zero.
So our $\QMA$ verification protocol is simply to run quantum phase estimation \cite{kitaev1995quantum} on the witness state.
We accept if the measured energy is smaller than $g/2$, and reject otherwise.
In \YES cases a valid witness will be an eigenstate of $\Delta_k$ with eigenvalue zero.
In \NO cases all possible witnesses will fail with high probability.

The more challenging aspect of the work is to demonstrate hardness.
We will do this by reducing from $\qfsat$.
In \cite{bravyi2011efficient} the authors show $\QMA_1$-hardness of $\qfsat$ by constructing a family of local Hamiltonians $H$ which encode the computational histories of a $\QMA_1$-verification circuits, such that: 
\begin{enumerate}[a.]
\item If there is an accepting witness to the $\QMA_1$-verification circuit then $H$ has a zero energy eigenstate.
\item If there is no accepting witness then the minimum eigenvalue of $H$ is bounded away from zero.
\end{enumerate}
We will encode $H$ into some $\mathcal{G}$ via Theorem \ref{informal_main_thm}. $\mathcal{G}$ then has the properties:
\begin{enumerate}[i.]
\item \label{point:yes} If $H$ has a zero energy groundstate there is a hole in $\Cl(\mathcal{G})$.
\item \label{point:no} If the minimum eigenvalue of $H$ is bounded away from zero then the minimum eigenvalue of the Laplacian is bounded away from zero.
\end{enumerate}
This completes the hardness argument. Thus it remains to demonstrate Theorem \ref{informal_main_thm}.

\medskip
{\noindent \bf Proof overview of Theorem \ref{informal_main_thm}.}
\smallskip

The first step is to construct a graph $\mathcal{G}_1$ such that $\Cl(\mathcal{G}_1)$ has two holes.\footnote{This was also the first step in \cite{crichigno2022clique}; however, the graph we choose here is a different one.
In order to show hardness of the gapped problem, we need that the natural inner product on simplices respects that encoded computational basis states should be orthogonal. This feature was not present in the encoding of \cite{crichigno2022clique}.}
This means that the homology of $\Cl(\mathcal{G}_1)$ can encode the Hilbert space of one qubit. 
We then form our base qubit graph $\mathcal{G}_n$ by taking the $n$-fold join of $\mathcal{G}_1$. That is, we take $n$ copies of $\mathcal{G}_1$ where vertices are connected all-to-all between the different copies. 
Constructed in this way, $\mathcal{G}_n$ has $2^n$ $2n-1$-dimensional holes, one for each computational basis state.
Moreover, $\mathcal{G}_n$ has a tensor-product-like structure -- each copy of $\mathcal{G}_1$ in the join can be identified with a qubit.
At this stage, the kernel of the Laplacian is isomorphic to the entire encoded Hilbert space of $n$ qubits -- $\mathcal{G}_n$ by itself corresponds to the empty zero Hamiltonian. 

The next step is to design \emph{gadgets} which implement terms in the $H$.
First we decompose $H$ into a sum of local rank-1 projectors. 
Then for each $m$-local term $\ketbra{\phi}{\phi}$ we take the graph $\mathcal{G}_m$ (the $m$-fold join of the single qubit graph) and design a gadget which `fills in the hole' in $\Cl(\mathcal{G}_m)$ corresponding to the state $\ket{\phi}$. 
Constructively, this involves adding extra \emph{gadget vertices} to the graph, and adding edges between these new vertices and the original vertices from $\mathcal{G}_m$.
This serves to lift the cycle corresponding to $\ket{\phi}$ out of the homology by rendering it a boundary.
The clique complex of the resulting $m$-qubit graph has $2^m-1$ holes of dimension $2m-1$, encoding a $m$-local projector acting on a system of $m$ qubits. This is the content of \Cref{sec:construction}.

To construct a graph $\hat{\mathcal{G}}_n$ which implements the Hamiltonian $H = \sum_i \ketbra{\phi_i}{\phi_i}$ the procedure is as follows:
\begin{enumerate}
\item Start with the graph $\mathcal{G}_n$.
\item For each term $\ketbra{\phi_i}{\phi_i}$ in $H$, insert the gadget implementing that term onto the copies of $\mathcal{G}_1$ corresponding to qubits in the support of $\ketbra{\phi_i}{\phi_i}$.
\item For each term $\ketbra{\phi_i}{\phi_i}$, connect its gadget vertices \emph{all to all} with the vertices of $\mathcal{G}_n$ corresponding to qubits outside the support of $\ketbra{\phi_i}{\phi_i}$.
\item Do \emph{not} connect any gadget vertices coming from different Hamiltonian terms.
\end{enumerate}

With this candidate reduction in hand we then need to show that it satisfies the necessary properties.
Demonstrating that the resulting graph satisfies \Cref{point:yes} is straightforward -- our construction of the gadgets fills in the holes in $\Cl(\hat{\mathcal{G}}_n)$ corresponding to states that are lifted in $H$.
If the Hamiltonian $H$ has a zero energy groundstate then there is a state $\ket{\psi}$ in the Hilbert space of $n$-qubits that satisfies every projector in $H$.
This state corresponds to a hole in $\Cl({\mathcal{G}}_n)$ that has not been filled in by any gadget, thus the hole remains in $\Cl(\hat{\mathcal{G}}_n)$ has non-trivial homology.
If, on the other hand, $H$ is not satisfiable there is no state in the Hilbert space of $n$ qubits that satisfies every projector in $H$.
Therefore, the process of `filling in holes' via gadgets has removed $2^n$ holes from the homology of $\Cl(\mathcal{G}_n)$ to construct $\Cl(\hat{\mathcal{G}}_n)$. 
We demonstrate when constructing the gadgets that the method of constructing gadgets does not introduce any new `spurious' homology classes into the complex. 
Therefore, since all $2^n$ holes have been removed, and no no holes have been introduced, the resulting complex $\Cl(\hat{\mathcal{G}}_n)$ has trivial homology.\footnote{This same argument was used in \cite{crichigno2022clique} with a different graph construction to show that the decision version of the clique homology problem is $\QMA_1$-hard - see \Cref{sec:prior work} for details.} 

Demonstrating that the graph satisfies \Cref{point:no} is more challenging, and constitutes the main technical contribution of this work. 
At first glance there is no reason why \Cref{point:no} should hold.
We have encoded the ground space of $H$ into the homology of $\Cl(\hat{\mathcal{G}}_n)$.
But the Laplacian of $\Cl(\hat{\mathcal{G}}_n)$ has many more excited states than the spectrum of $H$, and it is plausible that the excited spectrum of the Laplacian includes very low energy states.

In order to get a handle on the excited spectrum of a single gadget, we need to ensure that the states we lift out of homology do not mix with the rest of the spectrum.
In order to do this we develop a way of \emph{weighting} the complex so that the gadgets can be viewed as perturbations of the original qubit complex.
In terms of the graph, this amounts to weighting the gadget vertices with a parameter $\lambda \ll 1$.
We then want to analyze the spectrum and eigenspaces, which sounds similar to the domain of perturbation theory and perturbative gadgets \cite{kempe2006complexity}. 
However, perturbation theory is not able to provide the level of generality and control we require for our purposes.\footnote{The difficulty is that we would like to go to arbitrarily high orders of perturbation theory. Using generic perturbation theory tools, this would be close to impossible.}

Here, the main innovation of our work enters. We use an advanced tool from algebraic topology known as \emph{spectral sequences} \cite{forman1994hodge, chow2006you, mccleary2001user} to analyze the spectrum of each gadget. This forms the content of \Cref{spec_seq_sec}, and we will now elaborate on this key technique.

Our complex has some vertices weighted by $\lambda$, which is a small perturbative parameter $\lambda \ll 1$. In order to understand the low-energy spectrum of the Laplacian, we would like to expand the kernel of the Laplacian perturbatively in $\lambda$. This gives us a sequence of vector spaces $E_0^k,E_1^k,E_2^k,\dots$ which provide increasingly close approximations to $\ker{\Delta_k}$.
\begin{equation*}
E_j^k \rightarrow \ker{\Delta_k} \ \text{as} \ j \rightarrow \infty
\end{equation*}
Taking all orders of $\lambda$ into account gives the true kernel of the Laplacian, which is isomorphic to the homology.
\begin{equation*}
\ker{\Delta_k} \cong H_k
\end{equation*}

From the weighting of the complex we can obtain a \emph{filtration} on the chain complex, and a filtered chain complex has an associated \emph{spectral sequence}, which consists of vector spaces $e_{j,l}^k$. At $j=0$, our spectral sequence $e_{0,l}^k$ consists of the $k$-simplices that are weighted by $\lambda^l$. For each `page' $j$ there are coboundary maps
\begin{equation*}
d_{j,l}^k : e_{j,l}^k \rightarrow e_{j,l+j}^{k+1}
\end{equation*}
and to obtain the page $e_{j+1, l}^k$ we take the coboundary of the previous page $e_{j,l}^k$. Now consider the direct sums
\begin{equation*}
e_j^k := \bigoplus_l e_{j,l}^k
\end{equation*}
The vector spaces $e_j^k$ provides a sequence of approximations to the true homology group $H_k$ as $j$ increases.
\begin{equation*}
e_j^k \rightarrow H_k \ \text{as} \ j \rightarrow \infty
\end{equation*}
Intuitively, subsequent pages of the spectral sequence take into account more terms in the filtration, which in our case corresponds to including terms of higher order in $\lambda$.

A result from algebraic topology tells us that these two sequences are in fact isomorphic \cite{forman1994hodge}.
\begin{equation*}
E_j^k \cong e_j^k
\end{equation*}
That is, the perturbative expansion of the Laplacian groundspace is isomorphic to the spectral sequence of the filtration. This provides a connection between Hodge theory and spectral sequences.
Our strategy from here is clear: We compute the spectral sequence $e_j^k$ algebraically and use the isomorphism of Ref.~\cite{forman1994hodge} to deduce the perturbative expansion $E_j^k$ of the Laplacian kernel and hence the low-energy spectrum of the Laplacian.
Our use of this isomorphism is reminiscent of many methods in algebraic topology, where one converts difficult questions in analysis and topology into the easier language of algebra.

Once we have completed the analysis of the spectrum and eigenspace of the Laplacian corresponding to a single gadget using spectral sequences, we must still analyze what happens when we put all the gadgets together.
The argument that the lowest eigenvalue of the Laplacian is bounded away from zero in \NO cases proceeds in two steps: 
\begin{enumerate}
\item In a \NO case, any state must have large overlap with the excited subspace of at least one of the gadgets.
\item States with low energy must have small overlap with the excited subspace of all the gadgets. Hence, by Step 1, such low energy states do not exist.
\end{enumerate}

The first point is straightforward - in \NO cases the Hamiltonian $H$ is not satisfiable, so it is not possible to construct a global state which is in the ground state of every projector.
Therefore the overlap of any global state with the zero energy groundstate of each gadget must be bounded away from one for at least one of the gadgets.
The second point is technically challenging -- it involves detailed understanding of the structure of the combined Laplacian and its eigenspaces, and constitutes the content of \Cref{sec:combining gadgets full}.

\section{Prior work}\label{sec:prior work}

Understanding the possiblity of quantum advantage in TDA has inspired a number of works in this area in recent years.

The first key result in the field was the quantum TDA algorithm of \cite{lloyd2016quantum}.
It gives an approximation to the \emph{normalized} $k^{\textrm{th}}$-Betti number of a simplicial complex.
The $k^{\textrm{th}}$-Betti number $\beta_k$ is the number of $k$ dimensional holes in the complex, and the normalized Betti number is given by $\beta_k/|S_k|$ where $S_k$ is the set of $k$-simplices.
We can understand the quantum TDA algorithm as running phase estimation on the Laplacian $\Delta_k$.
The input state for the phase estimation is the maximally mixed state over $k$-simplices:
\begin{equation*}
\frac{1}{|S_k|} \sum_{x \in S_k} |x\rangle\langle x|
\end{equation*}
The algorithm effectively samples the eigenvalues of the Laplacian $\Delta_k$, to any $1/\poly({n})$ precision. By counting the fraction of times a zero eigenvalue is observed, we get an estimate of $\beta_k/|S_k|$ to any $1/\poly(n)$ additive error.
In order for the input state to be prepared efficiently for a clique complex, the graph must be \emph{clique dense} (see \cite{gyurik2022towards}).

In \cite{gyurik2022towards}, the authors initiated the investigation into the complexity of the problem solved by the quantum TDA algorithm. They showed that if one applies the quantum TDA algorithm with a generic local Hamiltionian in the place of $\Delta_k$, then it is able to solve a $\textrm{DQC}1$-hard problem.
\footnote{$\textrm{DQC}1$ is the `one clean qubit' model of quantum computation where the initial state is limited to a single qubit in the state $|0\rangle$, along with a supply of maximally mixed qubits. (See \cite[Section 6.3]{brandao2008entanglement} for a formal definition.) It does not capture the full power of quantum computation, but is thought to be impossible to simulate efficiently with classical computation.}

Inspired by the connection between homology and supersymmetry, in \cite{cade2021complexity} it was shown that the problem of deciding whether a general chain complex (a generalisation of simplicial complexes) has an $k$-dimensional hole is $\QMA_1$-hard and contained in $\QMA$ (given a suitable promise gap). 
Moreover, it was shown in the same paper that estimating the normalized Betti numbers of a general chain complex is $\textrm{DQC}1$-hard and contained in BQP.

Other papers have considered the \emph{decision} version of \Cref{informal_prom_prob}:
\begin{problem} \label{dec_homology_prob}
\emph{(Decision clique homology)}
Let $\mathcal{G}$ be a graph on $n$ vertices, given by its adjacency matrix. We are also given an integer $k$. The task is to decide
\begin{itemize}
    \item {\bf YES} \ The $k^{\text{th}}$ homology group of $\Cl(\mathcal{G})$ is non-trivial $H_k(\mathcal{G}) \neq 0$.
    \item {\bf NO} \ The $k^{\text{th}}$ homology group of $\Cl(\mathcal{G})$ is trivial $H_k(\mathcal{G}) = 0$.
\end{itemize}
\end{problem}
In \cite{adamaszek2016complexity} it was shown that \Cref{dec_homology_prob} is $\NP$-hard, and in \cite{schmidhuber2022complexity} it was shown that \Cref{dec_homology_prob} remains $\NP$-hard when restricted to clique dense graphs.
These results culminated in \cite{crichigno2022clique} where it was shown that \Cref{dec_homology_prob} is $\QMA_1$-hard, including when restricted to clique-dense graphs.
However, \Cref{dec_homology_prob} is not believed to be inside $\QMA$. Moreover, the constructions used for the reductions in previous works do not satisfy the necessary gap to guarantee containment in $\QMA$.
It should be noted the previous result holds for both weighted and unweighted graphs, whereas our results hold for weighted graphs only.

On the more applied side of the field, a number of recent papers have looked more closely at the quantum TDA algorithm.
In \cite{ubaru2021quantum, hayakawa2022quantum, mcardle2022streamlined, berry2022quantifying, akhalwaya2022towards} improvements were made to the algorithm in \cite{lloyd2016quantum} which make it more practical to run.
\section{Future directions} \label{sec:open}

\paragraph{\emph{Does \Cref{QMA1_cor} hold for unweighted graphs?}}

An immediate question is whether the $\QMA_1$-hardness of the gapped problem still holds with an unweighted clique complex. For an unweighted clique complex, there is no perturbative regime. However, perhaps it is possible to simulate the weights with an unweighted complex. For example, if the weights were integers, we could consider replacing a vertex of weight $w$ with a clique of size $w$. Note that the weights in our construction are all between some $1/\poly(n)$ and $1$.

\paragraph{\emph{Does \Cref{informal_prom_prob} remain $\QMA_1$-hard when the input is restricted to graphs whose complements have bounded degree?}}

The supersymmetric Hamiltonian $H$ given in Section \ref{sec:SUSY} is precisely the Laplacian of the clique complex $\mathcal{G}$. However, the projector $P_i$ appearing in this Hamiltonian act on all the neighbours of $i$ in the \emph{complement} graph. Thus $H$ is only a \emph{local} Hamiltonian if the complement of $\mathcal{G}$ has degree bounded by some constant. Our reduction does \emph{not} have this property.
Does $\QMA_1$-hardness of the gapped problem still hold for graphs whose complements have bounded degree?
We note that the graphs coming from the reduction in \cite{crichigno2022clique} \emph{do} give rise to local supersymmetric Hamiltonians, but of course they reduce from the decision problem rather than the gapped problem.

\paragraph{\emph{Generalized gadgets}}

Can we generalize the gadget construction from \Cref{single_gadget_sec} to work on projectors onto (a) \emph{any} integer state (see  \Cref{integer_state_def}) and (b) arbitrary states. In \Cref{single_gadget_sec} we provide a construction that should generalize to any integer state, although we do not prove that it generalizes. Tackling arbitrary states is much more challenging, and it is not clear what the right approach is.

\paragraph{\emph{Is the following problem $\QMA$-complete}?}

\begin{problem}\label{Laplacian_min_evalue_prob}
\emph{(Laplacian minimum eigenvalue)}
Given a weighted graph $\mathcal{G}$, dimension $k$, and $0 \leq a < b$ with $b-a > 1/\poly(n)$, decide which of the following is true.
\begin{itemize}
    \item {\bf YES} \ $\lambda_{\min}(\Delta^k) \leq a$.
    \item {\bf NO} \ $\lambda_{\min}(\Delta^k) \geq b$.
\end{itemize}
where $\Delta^k$ is the $k$-Laplacian of $\mathcal{G}$.
\end{problem}

In this work, we have focused on a lower bound on the minimum eigenvalue of the Laplacian in the $\NO$ case. In the $\YES$ case, when the original Hamiltonian is frustration free, it was relatively easy to argue the Laplacian has a zero eigenvalue using \Cref{Hodge_prop}. Perhaps it is possible to get exact control on the Laplacian eigenvalues of our construction when the original Hamiltonian is frustrated. Rather than simply applying \Cref{Hodge_prop}, this would require a new upper bounding technique.

Succeeding with this for the minimum eigenvalue alone would be enough to establish a $\QMA$-completeness result. This is because precise control of the eigenvalue would allow us to reduce from the standard local Hamiltonian problem with a promise gap.

\paragraph{\emph{Are Laplacians of clique complexes \emph{universal quantum Hamiltonians}, in the sense of \cite{cubitt2018universal}?}}

Even further, perhaps controlling the eigenvalues of the Laplacian in our reduction is not only possible for the minimum eigenvalue, but also for excited eigenvalues. This opens up the exciting prospect that the Laplacian of clique complexes could be a \emph{universal quantum Hamiltonian}, in the sense of \cite{cubitt2018universal}. Here, the lowest $2^n$ eigenvalues of the Laplacian form a scaled version of the entire spectrum of the original Hamiltonian.

\paragraph{\emph{Is approximate Betti number estimation $\textrm{\emph{DQC}}1$-hard?}}

Can we show that the quantum TDA algorithm of \cite{lloyd2016quantum} cannot be dequantised? 
The precise problem solved by the quantum TDA algorithm was formalized in \cite{gyurik2022towards} and named \emph{approximate Betti number estimation} (ABNE):

\begin{problem} \label{ABNE_prob}
\emph{(Approximate Betti Number Estimation (ABNE))}
Let $\mathcal{G}$ be a graph on $n$ vertices, given by its adjacency matrix. We are also given an integer $k$, and a precision parameters $\varepsilon, \delta > 1/\poly(n)$. The task is to output an estimate $\chi$ which satisfies with high probability
\begin{equation*}
\frac{\beta_k}{{n \choose k+1}} - \varepsilon \ \leq \ \chi \ \leq \ \frac{|\{\lambda : \lambda \in \spec{\Delta^k} , \ 0\leq \lambda < \delta\}|}{{n \choose k+1}} + \varepsilon
\end{equation*}
where $\beta_k = \dim H^k(\mathcal{G})$ is the $k^{\textrm{th}}$ Betti number.
\end{problem}
\begin{fact}
ABNE can be solved with a quantum algorithm. The algorithm runs phase estimation on the operator $\widetilde{\Delta}^k$ from the proof of \Cref{main_cor}, with the input state which is maximally mixed over $(k+1)$-subsets, and counts the proportion of eigenvalues $\delta$.
\end{fact}

The same problem on generic sparse quantum Hamiltonians is called the Low Lying Spectral Density (LLSD) problem:

\begin{problem} \label{LLSD_prob}
We are given a $\poly(n)$-sparse Hamiltonian $H$ on $n$-qubits, a threshold $b$, precision parameters $\varepsilon, \delta > 1/\poly(n)$. The task is to output an estimate $\chi$ which satisfies with high probability
\begin{equation*}
\frac{|\{\lambda : \lambda \in \spec{H} , \ 0\leq \lambda < b\}|}{2^n} - \varepsilon \ \leq \ \chi \ \leq \ \frac{|\{\lambda : \lambda \in \spec{H} , \ 0\leq \lambda < b + \delta\}|}{2^n} + \varepsilon
\end{equation*}
\end{problem}

LLSD is known to be $\textrm{DQC}1$-hard \cite{gyurik2022towards}.
Showing that ABNE is $\textrm{DQC}1$-hard would provide complexity-theoretic proof that it cannot be dequantised.
For this, we would like to reduce from LLSD to ABNE.

Unfortunately, even the very strong universal quantum Hamiltonian property is not enough for the desired reduction. 
If we were to reduce from LLSD, the problem we would get would not quite be ABNE, since the denominator would still be $2^n$ rather than ${n \choose k+1}$. In our construction, ${n \choose k+1}$ is exponentially smaller than $2^n$, so amending the denominator would require us to modify the precision $\varepsilon$ to be exponentially small.

\paragraph{\emph{Applications to quantum PCP.}}

As discussed in the introduction, an exciting future direction is to develop technique to perform \emph{gap amplification} on the combinatorial Laplacian. The combinatorial Laplacian corresponds to a supersymmetric fermionic Hamiltonian, and we have shown that detecting a zero-eigenvalue is $\QMA_1$-hard. Thus this provides a tentative route towards a fermionic quantum PCP theorem \cite{aharonov2013guest}.

A related future direction is to better understand the geometric interpretation of the minimum eigenvalue of the combinatorial Laplacian via higher dimensional Cheeger inequalities \cite{gundert2014higher,steenbergen2014cheeger,parzanchevski2016isoperimetric}, and the connection to high-dimensional expansion \cite{lubotzky2018high}.

\section{Supersymmetry} \label{sec:SUSY}

In this section we will introduce \emph{supersymmetric quantum mechanics} (SUSY), and explain how it is connected to homology.
We follow the introductions to the subject in \cite{cade2021complexity,crichigno2022clique}.
Readers who are not interested in the connection between SUSY and our results can safely skip this section.

In Appendix \ref{app:susy_proof_sketch}, we will give a sketch of the main construction and proof in the SUSY quantum mechanics picture. This is not intended as a rigorous proof, but may help readers more familiar with physics to understand the intuition behind the formal proof.

\subsection{Supersymmetric quantum mechanics}

The Hilbert space of any quantum mechanical system can be decomposed into the space of fermionic states $\mathcal{H}_F$ and the space of bosonic states $\mathcal{H}_B$:
\begin{equation*}
\mathcal{H} = \mathcal{H}_B \oplus \mathcal{H}_F
\end{equation*}
where bosonic states are those that are symmetric under particle exchange, whereas fermionic states are anti-symmetric. 

Supersymmetric quantum systems exhibit a symmetry between fermionic and bosonic states.
More formally, in a $\mathcal{N}=2$ supersymmetric  system there is an operator, called the supercharge, which sends fermionic states to bosonic states, and vice versa:
\begin{equation*}
\mathcal{Q}:(\mathcal{H}_B,\mathcal{H}_F) = (\mathcal{H}_F,\mathcal{H}_B)
\end{equation*}
which is nilpotent:
\begin{equation*}
Q^2=0
\end{equation*}
The Hamiltonian of the sytem is then given by:
\begin{equation*}
H = \{Q,Q^\dagger \} = QQ^\dagger + Q^\dagger Q
\end{equation*}
Note that due to the nilpotency condition, $Q$ and $Q^\dagger$ commute with the Hamiltonian, so they are indeed symmetries of the system as claimed.

The equation $Q^2 = 0$ is reminiscient of \Cref{boundary_nilpotent} for the boundary operator. Indeed, $Q$ can be interpreted as a homological boundary operator in the abstract sense, and $\im{Q} \subseteq \ker{Q}$. The SUSY Hamiltonian $H$ then corresponds to the Laplacian operator from \Cref{Laplacian_definition_informal}.

A number of properties of the spectrum of SUSY systems follow immediately from the definition.
Firstly, we have:
\begin{equation*}
\langle\psi | H|\psi\rangle = \langle\psi |  (QQ^\dagger + Q^\dagger Q) |\psi\rangle = |Q|\psi\rangle|^2 +|Q^\dagger |\psi\rangle|^2 \geq 0
\end{equation*}
So the lowest possible energy in a SUSY system is $E=0$.

Consider first states with $E>0$.
It turns out that such states come in fermionic-bosonic pairs of the same energy.\footnote{See \Cref{pairing_sec} for a discussion of the same pairing phenomenon in simplicial complexes.}
Consider a bosonic energy eigenstate $|\psi_B\rangle$.
We have $H|\psi_B\rangle = E|\psi_B\rangle$, which implies that at least one of $Q|\psi_B\rangle$ or $Q^\dagger|\psi_B\rangle$ is non-zero, moreover it must be a fermionic state. 
Since $Q$ and $Q^\dagger$ are symmetries of the Hamiltonian the fermionic state will have the same energy as the original bosonic state. 
Moreover, the pairing up of the eigenstates, together with the nilpotency condidition, implies that exactly one of $Q|\psi_B\rangle$ or $Q^\dagger|\psi_B\rangle$ is non-zero (the same applies to fermionic energy eigenstates).

Now we turn to consider energy eigenstates with $E=0$.
These are known as SUSY eigenstates, and the question of whether or not a system has any SUSY eigenstates is a key question in the study of SUSY systems.
Note first that a state has energy $E=0$ iff:
\begin{equation*}
Q |\Phi\rangle = Q^\dagger |\Phi\rangle =0 
\end{equation*}

Therefore we have that $|\Phi\rangle \in \ker(Q)$ and likewise $|\Phi\rangle \in \ker(Q^\dagger)$.
We can also deduce that SUSY groundstates are not in $\im(Q)$ or $\im(Q^\dagger)$.
To see this assume that there exists a state $|\psi\rangle$ such that $|\Phi\rangle = Q^\dagger|\psi\rangle$.
Then since $Q^\dagger$ commutes with the Hamiltonian we must have that $|\psi\rangle$ is also a SUSY groundstate, which implies that $Q^\dagger|\psi\rangle=0$, a contradiction. 

Putting everything together we have that SUSY groundstates are in the homology of both $Q$ and $Q^\dagger$:
\begin{equation*}
|\Phi\rangle \in \frac{\ker(Q)}{\im(Q)} \textrm{\ \ \ and \ \ \ } |\Phi\rangle \in \frac{\ker(Q^\dagger)}{\im(Q^\dagger)} 
\end{equation*}

Elements which are in the homology of both an operator and its conjugate are known as \emph{harmonic} representatives of homology classes.
So, we have that SUSY groundstates are harmonic representatives of the homology classes of $Q$ and $Q^\dagger$.
This also implies that SUSY groundstates are in one-to-one correspondence with the homology classes themselves.
Therefore the question of deciding whether a homology group is non-trivial is equivalent to asking whether a certain SUSY system has a zero-energy eigenstate.

The \emph{Witten index} of a supersymmetric system is given by the difference between the number of bosonic and fermionic supersymmetric ground states.

\subsection{Fermion hard core model}

To connect with the clique homology problem we consider a particular SUSY system known as the \emph{fermion hard core model}, introduced in \cite{Fendley:2002sg}.
The fermion hard core model corresponds more naturally to the independence homology problem.\footnote{See \cite{Huijse:2010jia} for a nice overview of this model, its extensions, and the relation to independence homology. } The independence complex of a graph is the simplicial complex defined by declaring each independent set to be a simplex. The independence homology problem takes as input a graph $G$, and outputs \YES if the independence complex of the graph has a non-trivial homology group, and \NO otherwise.
The independence complex of a graph is the clique complex of the complement graph, so the two problems are computationally equivalent.
Throughout this section we will consider the independence homology problem. 

The fermion hard core model is defined on a graph $G = (V,E)$. 
The Hilbert space of the system is given by all configurations of fermions on $G$ subject to the condition that no two fermions occupy adjacent vertices (the hard-core condition).  
In other words, the Hilbert space corresponds to the independent sets of $G$.

The supercharges of the system are given by:
\begin{equation*}
Q = \sum_{i \in V}P_i a_i^\dagger \textrm{\ \ \ and \ \ \ } Q^\dagger = \sum_{i\in V} a_i P_i \textrm{\ \ \ where \ \ \ } P_i = \prod_{j|(i,j) \in E} (1-\hat{n}_j)
\end{equation*}
where $a_i, a^\dagger_i$ are the fermionic annihilation and creation operators respectively, and $\hat{n}_i = a_i^\dagger a_i$ is the fermionic number operator. 
Intuitively, $Q$ ($Q^\dagger$) adds (removes) a fermion to (from) the state, subject to the hard-core condition (which is enforced by the $P_i$). 
Straightforward manipulations then  give a Hamiltonian of the form:
\begin{equation*}
H = \sum_{(i,j) \in E} P_i a_i^\dagger a_j P_j + \sum_{i \in V} P_i
\end{equation*}

As outlined in the previous section, the zero-energy ground states of the fermion hard core model will be in one-to-one correspondence with elements of the homology of $Q^\dagger$.
It can be checked $Q^\dagger$ acts on the independence complex of $G$ as the boundary operator $\partial$. The space spanned by the $k$-simplices of the independence complex corresponds to the $k+1$-particle sector of Fock space.
Therefore the problem of whether the independence complex of a graph $I(G)$ has a non-trivial homology group is equivalent to the problem of whether or not the fermion hard core model defined on $G$ has a SUSY groundstate. Fixing a dimension $k$ at which we are studying the homology corresponds to restricting to the $k+1$-particle sector of Fock space.

In our proof we will weight each vertex in the graph by some complex number $\lambda$ in order to apply perturbative techniques to analyse the spectrum of our construction (see \cref{weighting_sec}).
In the supersymmetric picture this is an example of a technique known as staggering \cite{Huijse_2012,Bauer_2013}.\footnote{Note that the terminology `staggering' is the standard terminology for this technique in the fermion-hard core model. However, in graph theory what we are doing is weighting the vertices. In order to stay consistent with the terminology from the different fields we refer to the technique as staggering in this section, but as weighting in the rest of the paper.}
The idea is that we can modify the supercharges, as:
\begin{equation*}
Q = \sum_{i \in V}\lambda_i P_i a_i^\dagger \textrm{\ \ \ and \ \ \ } Q^\dagger = \sum_{i\in V} \lambda_i^* a_i P_i 
\end{equation*}
This gives a Hamiltonian:
\begin{equation} \label{fermion_Hamiltonian}
H = \sum_{(i,j) \in E} \lambda_i \lambda_j^* P_i a_i^\dagger a_j P_j + \sum_{i \in V} |\lambda_i|^2 P_i
\end{equation}
Applying staggering in this way does not alter the number of SUSY groundstates.
However, it can modify the excited part of the spectrum.
The system is still supersymmetric, so all fermionic eigenstates are still paired with a bosonic eigenstate of the same energy, and vice versa. 
But the eigenstates are generically shifted compared with the non-staggered case.

\pagebreak
\addcontentsline{toc}{section}{Proof of results} 
\noindent {\huge \huge \bf Proof of results}

\section{Preliminaries} \label{sec:prelims}

\subsection{$\QMA_1$ and $\qmsat$} \label{sec:qma1}

Let's begin by formally defining $\QMA$ and $\QMA_1$.
\begin{definition}[$\QMA$ \cite{kitaev2002classical}] \label{defQMA}
A problem $A = (A_{\text{yes}}, A_{\text{no}})$ is in $\QMA$ if there is a $\P$-uniform family of polynomial-time quantum circuits (the ``verifier'') $V_n$, one for each input size $n$, such that
\begin{itemize}
    \item If $x \in A_{\text{yes}}$, there exists a $\poly(n)$-qubit witness state $\ket{w}$ such that $\Pr[V_n(x,\ket{w}) = 1] \geq \frac{2}{3}$,
    \item If $x \in A_{\text{no}}$, then for any $\poly(n)$-qubit witness state $\ket{w}$, $\Pr[V_n(x,\ket{w}) = 1] \leq \frac{1}{3}$.
\end{itemize}
\end{definition}

The constants $\frac{1}{3}$, $\frac{2}{3}$ in the definition of $\QMA$ are conventional. The definition of $\QMA$ is equivalent as long as the acceptance and failure probabilities are separated by some inverse polynomial in the problem size \cite{kitaev2002classical}.

\begin{definition}[$\QMA_1$ \cite{bravyi2011efficient,gosset2016quantum}] \label{defQMA1}
A problem $A = (A_{\text{yes}}, A_{\text{no}})$ is in $\QMA_1$ if there is a $\P$-uniform family of polynomial-time quantum circuits (the ``verifier'') $V_n$, one for each input size $n$, such that
\begin{itemize}
    \item If $x \in A_{\text{yes}}$, there exists a $\poly(n)$-qubit witness state $\ket{w}$ such that $\Pr[V_n(x,\ket{w}) = 1] = 1$,
    \item If $x \in A_{\text{no}}$, then for any $\poly(n)$-qubit witness state $\ket{w}$, $\Pr[V_n(x,\ket{w}) = 1] \leq \frac{1}{3}$.
\end{itemize}
\end{definition}

The canonical $\QMA_1$-complete problem is $\qmsat$, in which you are asked to decide whether a local Hamiltonian composed of positive semi-definite local terms has an exactly zero eigenvalue, given a suitable promise gap. In this work we will reduce from \qfsat to the problem of deciding whether or not a particular homology group of a clique complex is non-trivial.

\begin{problem}\label{def:qmsat}\emph{(\qmsat\cite{bravyi2011efficient})}
Fix function $g : \mathbb{N} \rightarrow [0,\infty)$ with $g(n) \geq 1 / \ply{n}$. The input to the problem is a a list of $m$-local projectors $\Pi_j \in \mathcal{P}$. $\mathcal{P}$ are a set of projectors obeying certain constraints.\footnote{We have been deliberately vague in defining the constraints that the set of projectors must satisfy.
That is because the constraints depend on the gate set, and there is not a standard definition.
However, this is only an issue for showing {\it containment} in $\QMA_1$.
For showing $\QMA_1$-hardness there is no need to include any constraints on the form of the projectors in \Cref{def:qmsat}.
Since we are interested in showing $\QMA_1$-hardness and containment in $\QMA$ (where we only need to constrain the locality of the projectors), we will not impose any constraints on the allowable projectors (beyond their locality) throughout this work.}
Let
\begin{equation}
H = \sum_j \Pi_j
\end{equation}
The task is to decide whether:
\begin{itemize}
    \item {\bf \YES} \ $H$ is satisfiable i.e. there exists a state $\ket{\psi}$ with $H\ket{\psi}=0$.
    \item {\bf \NO} \ The minimum eigenvalue of $H$ is at least $g$.
\end{itemize}
\end{problem}

\qmsat is known to be $\QMA_1$-complete for $m\ge 3$ \cite{bravyi2011efficient,gosset2016quantum}.
\qtsat is known to be in $\P$ \cite{bravyi2011efficient}.

In \cite{bravyi2011efficient} a history state construction was used to reduce from a general problem in $\QMA_1$ to \qfsat.
Computational history states are of the form:
\begin{equation*}
\ket{\Phi}_{CQ} = \sum_{t=0}^T \ket{t}\ket{\psi_t}
\end{equation*}
where $\{\ket{t}\}$ is an orthonormal basis for $\mathcal{H}_C$, the clock register, and the $\ket{\psi_t} = \Pi_{i=0}^t U_i \ket{\psi_0}$ for some initial state $\ket{\psi_0}$ and some set of unitaries $\{U_i\}$. 
The first register of $\ket{\Phi}_{QC}$ encodes the time, while the second register (the `computational' register) encodes the state of the quantum circuit at time $t$. 

The idea in \cite{bravyi2011efficient} is to construct a local Hamiltonian (composed of projectors), whose zero energy ground states are history states, this Hamiltonian is given by:
\begin{equation*}
\Hhs =  \Hin +  \Hclock + \Hprop 
\end{equation*}
where $\Hin$ constrains the starting state $\ket{\psi_0}$ of the circuit, $\Hclock$ penalises any states in the clock register which don't encode valid times, and $\Hprop$ penalises any states where $\ket{\psi_t} \neq U_t\ket{\psi_{t-1}}$.
$\Hhs$ has a degenerate zero energy ground state, where all computations which start in a valid state $\ket{\psi_0}$ satisfy every constraint.
To break this degeneracy, and encode $\QMA_1$-verification circuits it is necessary to add one extra term to the Hamiltonian:
\begin{equation*}
\Hbravyi = \Hhs + \Hout
\end{equation*}
where $\Hout$ penalises any computation which outputs \NO, and gives zero energy to any computation which outputs \YES. 

A detailed overview of the reduction is given in \Cref{app:qma}.
In \Cref{table:states} we give an overview of the rank-1 projectors that we need to be able to implement in order to reduce from the construction in \cite{bravyi2011efficient}.
In \Cref{sec:construction} we will construct gadgets to implement each of these states. 

\vspace{\baselineskip} 

\noindent{\bf Note on gateset for $\QMA_1$}: It is important to note that the definition of $\QMA_1$ (see \Cref{defQMA1}) implicitly depends on a choice of universal gate set.
In standard $\QMA$ (see \Cref{defQMA}) the definition is independent of gate set, since all universal gate sets can approximate any unitary evolution.
However, the requirement of perfect completeness in $\QMA_1$ means that it may be necessary to implement a given unitary evolution exactly. 

When choosing a universal gate set for our construction we require that every gate in the  set should have only rational coefficients. This ensures that the only states we need to lift are integer states -- see \Cref{single_gadget_sec}.
Based on this requirement, we choose the universal gate set: $\mathcal{G} = \{\textrm{CNOT},U \}$ where $U$ is the `Pythagorean gate' \cite{crichigno2022clique}:
\begin{equation*}
U = \frac{1}{5} 
\begin{pmatrix}
3 & 4 \\
-4 & 3
    \end{pmatrix}
\end{equation*}
This choice of gate set is shown to be universal in \cite[Theorem 3.3]{Adleman:1997} and \cite[Theorem 1.2]{shi2003both}.

\begin{table} 
\begin{center}
\begin{tabular}{ |c|c|c| } 
 \hline
 \emph{Term in} $H_{\text{Bravyi}}$ & \emph{Penalizes state} $\ket{\psi_S}$ \\ 
 \hline
$\Hpropt' $& $ \frac{1}{\sqrt{2}} \left(\ket{1011}-\ket{1000} \right) $     \\ 
 \hline
  $\Hpropt(\text{CNOT})$ & $\frac{1}{\sqrt{2}} \left(\ket{0110}-\ket{0101} \right)$  \\
 \hline
   $\Hpropt(\text{CNOT})$ & $\frac{1}{\sqrt{2}}  \left(\ket{0010}-\ket{0001} \right)$  \\
  
\hline
$\Hpropt(U_{\mathit{Pyth.}})$ &$ \frac{1}{5\sqrt{2}} \left(-5\ket{011}+4\ket{100}+3\ket{101}\right)$   \\ 
 \hline
 $\Hpropt(U_{\mathit{Pyth.}})$ &$ \frac{1}{5\sqrt{2}}  \left(-5\ket{010}+3\ket{100}-4\ket{101}\right)$   \\ 
  \hline
   $\Hpropt(\text{CNOT})$ & $\frac{1}{\sqrt{2}}  \left(\ket{1101}-\ket{1010} \right)$ \\
 \hline
 $\Hpropt(\text{CNOT})$ & $\frac{1}{\sqrt{2}} \left(\ket{1011}-\ket{1100} \right) $  \\
 \hline
$\Hclock^{(1)}$ & $\ket{00}$    \\ 
 \hline
 $\Hclock^{(2)}$ & $\ket{11} $  \\ 
 \hline
 $\Hin$, $\Hout$ &$ \ket{011}  $  \\ 
 \hline
$ \Hclock^{(6)}$, $\Hclock^{(4)}$, $\Hclock^{(5)}$, $\Hclock^{(3)}$ &$ \ket{1100}$ \\
\hline
$ \Hclock^{(4)}$ &$ \ket{0111}$  \\
\hline
$ \Hclock^{(5)}$ &$ \ket{0001}$  \\
\hline
\end{tabular}
\caption{Projectors needed for quantum $4$-$\SAT$ with universal gate set $\mathcal{G}$. Note we collated projectors which are the same up to re-ordering the qubits involved.}
\label{table:states}
\end{center}
\end{table}

\subsection{Simplicial homology} \label{sec:homology_formal}

We will re-introduce the subject of homology more formally and from a `cohomology-first' perspective, since this will play nicely with the weighting introduced in \Cref{weighting_sec}, and later spectral sequences in \Cref{spec_seq_sec}.

\begin{definition}
A \emph{simplicial complex} $\mathcal{K}$ is a collection of subsets $\mathcal{K} = \mathcal{K}^0 \cup \mathcal{K}^1 \cup \dots$ such that
\begin{itemize}
    \item $\sigma \in \mathcal{K}^k$ have $|\sigma| = k+1$.
    \item If $\sigma \in \mathcal{K}$, then $\tau \in \mathcal{K}$ for all $\tau \subset \sigma$.
\end{itemize}
\end{definition}

Intuitively, a simplicial complex is a higher dimensional generalization of a graph. It has vertices $\mathcal{K}^0$, edges $\mathcal{K}^1$, triangles $\mathcal{K}^2$, tetrahedra $\mathcal{K}^3$, et cetera.

From a simplicial complex we can derive a \emph{chain complex}. Let $\mathcal{C}^k(\mathcal{K})$ be the complex vector space formally spanned by $\mathcal{K}^k$. This involves picking a conventional ordering for each simplex $\sigma = [v_0,\dots,v_k]$ and identifying $|[\pi(v_0),\dots,\pi(v_k)]\rangle = \textrm{sgn}(\pi) |\sigma\rangle$ for any permutation $\pi \in S_{k+1}$. Here $\textrm{sgn}(\pi)$ denotes the sign of the permutation, and it is known as the \emph{orientation} of the simplex. $\mathcal{C}^k(\mathcal{K})$ is known as a \emph{chain space}.

For a $k$-simplex $\sigma \in \mathcal{K}^k$, let $\text{up}(\sigma) \subset \mathcal{K}^0$ be the subset of vertices $v$ such that $\sigma \cup \{v\} \in \mathcal{K}^{k+1}$ is a $k+1$-simplex.

Define the \emph{coboundary map} $d^k: \mathcal{C}^k(\mathcal{K}) \rightarrow \mathcal{C}^{k+1}(\mathcal{K})$ to act as
\begin{equation*}
d^k |\sigma\rangle = \sum_{v \in \text{up}(\sigma)} |\sigma \cup \{v\}\rangle
\end{equation*}
for a $k$-simplex $\sigma \in \mathcal{K}^k$ - see \Cref{fig:coboundary}

\begin{figure}
\centering
\includegraphics[scale=0.8]{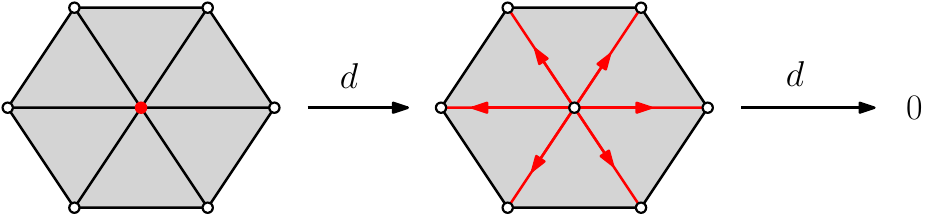}
\caption{The action of the coboundary map. In each step the coboundary map is acting on the parts of the complex shown in red. The third figure is empty because the lines that are acted on by the coboundary map in the second figure all have coboundaries composed of two triangles, and these coboundaries cancel out with those of the lines on either side as they have opposite orientation.}\label{fig:coboundary}
\end{figure}

This gives a chain of vector spaces with linear maps between them.
\begin{equation*}
\begin{tikzcd}
\mathcal{C}^{-1}(\mathcal{K}) \arrow[r,"d^{-1}"] & \mathcal{C}^{0}(\mathcal{K}) \arrow[r,"d^0"] & \mathcal{C}^{1}(\mathcal{K}) \arrow[r,"d^1"] & \mathcal{C}^{2}(\mathcal{K}) \arrow[r,"d^2"] & \dots
\end{tikzcd}
\end{equation*}
Here we define $\mathcal{C}^{-1}(\mathcal{K}) = \text{span}\{\emptyset\} \cong \mathbb{C}$ the 1-dimensional vector space, and $d^0$ maps the empty set $|\emptyset\rangle$ to the uniform superposition $\sum_{v \in \mathcal{K}^0} |v\rangle \in \mathcal{C}^0(\mathcal{K})$. With this convention, we are working with the \emph{reduced} cohomology.

It can be checked that $d^k \circ d^{k-1} = 0$ for each $k$. This gives a \emph{chain complex}.

\begin{definition}
A \emph{chain complex} is a chain of complex vector spaces $\mathcal{C}^k$ with linear maps $d^k : \mathcal{C}^k \rightarrow \mathcal{C}^{k+1}$ which satisfy $d^k \circ d^{k-1} = 0$ for all $k$.
\begin{equation*}
\begin{tikzcd}
\mathcal{C}^{-1} \arrow[r,"d^{-1}"] & \mathcal{C}^0 \arrow[r,"d^0"] & \mathcal{C}^1 \arrow[r,"d^1"] & \mathcal{C}^2 \arrow[r,"d^2"] & \dots
\end{tikzcd}
\end{equation*}
\end{definition}

$d^k \circ d^{k-1} = 0$ means that $\im{d^{k-1}} \subseteq \ker{d^k}$ for each $k$. This allows us to define the \emph{cohomology groups} as
\begin{equation*}
H^k = \frac{\ker{d^k}}{\im{d^{k-1}}}
\end{equation*}

Let $(\mathcal{C}^k)^*$ be the \emph{dual space} of $\mathcal{C}^k$. Formally, this is the space of linear functionals $f: \mathcal{C}^k \rightarrow \mathbb{C}$. Let $\partial^k = (d^{k-1})^* : (\mathcal{C}^k)^* \rightarrow (\mathcal{C}^{k-1})^*$ the dual map of $d^{k-1}$. $\partial^k$ are known as \emph{boundary maps}, as introduced in \Cref{sec:homology_informal}. 
As noted before, these act as:
\begin{equation*}
\partial^k |\sigma\rangle = \sum_{v \in \sigma} |\sigma \setminus \{v\}\rangle
\end{equation*}
for a $k$-simplex $\sigma \in \mathcal{K}^k$. (Here, by $|\sigma\rangle$ and $|\sigma \setminus \{v\}\rangle$, we technically mean the indicator functions of these simplices, which are members of the dual spaces $(\mathcal{C}^k)^*, (\mathcal{C}^{k-1})^*$.)
See \Cref{fig:boundary} for a diagrammatic representation.

We get the chain complex
\begin{equation*}
\begin{tikzcd}
(\mathcal{C}^{-1})^* & \arrow[l,"\partial^0"] (\mathcal{C}^0)^* & \arrow[l,"\partial^1"] (\mathcal{C}^1)^* & \arrow[l,"\partial^2"] (\mathcal{C}^2)^* & \arrow[l,"\partial^3"] \dots
\end{tikzcd}
\end{equation*}
From $d^k \circ d^{k-1} = 0 \ \forall k$ we get that $\partial^k \circ \partial^{k+1} = 0 \ \forall k$. This allows us to define the \emph{homology groups} as
\begin{equation*}
H^{k*} = \frac{\ker{\partial^k}}{\im{\partial^{k+1}}}
\end{equation*}
The homology groups $H^{k*}$ are the dual spaces of the cohomology groups $H^k$.

\subsection{Hodge theory} \label{sec:hodge}

Our next move will be to choose an inner product on $\mathcal{C}^k$, thus rendering it a \emph{Hilbert space}. This is equivalent to choosing an isomorphism between the space and its dual $\mathcal{C}^k \leftrightarrow \mathcal{C}^{k*}$. The most basic choice is the declare the simplices themselves to form an orthonormal basis. That is, for $\sigma, \tau \in \mathcal{K}^k$
\begin{equation*}
\langle\sigma|\tau\rangle = \begin{cases}
1 & \sigma = \tau \\
0 & \text{otherwise}
\end{cases}
\end{equation*}

We will modify this choice later, but for now let's consider this case. We can now drop the asterisks in the notation and identify $\mathcal{C}^{k*} = \mathcal{C}^k$. The $k$-boundary map is now the adjoint of the $k-1$-coboundary map
\begin{equation*}
\partial^k = (d^{k-1})^\dag
\end{equation*}

The inner product allows us to define the \emph{Laplacian} as
\begin{align*}
\Delta^k &: \mathcal{C}^k \rightarrow \mathcal{C}^k \\
\Delta^k &= d^{k-1} \partial^k + \partial^{k+1} d^k
\end{align*}
This is a positive semi-definite self-adjoint operator on the Hilbert space $\mathcal{C}^k$. In fact, we can split up the definition and write
\begin{align*}
\Delta^{\downarrow k} &= d^{k-1} \partial^k \\
\Delta^{\uparrow k} &= \partial^{k+1} d^k \\
\Delta^k &= \Delta^{\downarrow k} + \Delta^{\uparrow k}
\end{align*}
and now both $\Delta^{\downarrow k}$ and $\Delta^{\uparrow k}$ are individually positive semi-definite.

\begin{fact}\label{energy_formulas_lem}
\begin{align*}
\langle\psi|\Delta^{\downarrow k}|\psi\rangle &= ||\partial^k|\psi\rangle||^2 \\
\langle\psi|\Delta^{\uparrow k}|\psi\rangle &= ||d^k|\psi\rangle||^2 \\
\langle\psi|\Delta^k|\psi\rangle &= ||\partial^k|\psi\rangle||^2 + ||d^k|\psi\rangle||^2 \\
\end{align*}
\end{fact}

There is a close relationship between the Laplacian and the homology, which is described by \emph{Hodge theory}. The most basic proposition of Hodge theory is the following.
\begin{proposition} \label{Hodge_prop}
$\ker{\Delta^k}$ is canonically isomorphic to $H^k$.
\end{proposition}
\begin{proof}
Suppose $\Delta^k|\psi\rangle = 0$. This means that
\begin{align*}
0 &= \langle\psi|\Delta^k|\psi\rangle \\
&= || \partial^k |\psi\rangle ||^2 + || d^k |\psi\rangle ||^2
\end{align*}
where we used that $d^{k-1} = \partial^{k*}$ and $\partial^{k+1} = d^{k*}$. Thus
\begin{align*}
|| \partial^k |\psi\rangle ||^2 &= || d^k |\psi\rangle ||^2 = 0 \\
\implies \partial^k |\psi\rangle &= d^k |\psi\rangle = 0
\end{align*}
Thus
\begin{equation*}
|\psi\rangle \in \ker{d^k}
\end{equation*}
and
\begin{equation*}
|\psi\rangle \in \ker{\partial^k} = (\im{(\partial^k)^\dag})^\perp = (\im{d^{k-1}})^\perp
\end{equation*}
Each homology class $[|\psi\rangle] \in H^k = \ker{d^k} / \im{d^{k-1}}$ will have a unique representative orthogonal to $\im{d^{k-1}}$.
\end{proof}
The proposition tells us that each homology class has a unique \emph{harmonic} representative, where harmonic means that it is in the kernel of the Laplacian. The equation $\Delta^k |\psi\rangle = 0$ is a high-dimensional generalization of Laplace's equation, and $\ker{\Delta^k}$ is sometimes referred to as the \emph{harmonic subspace}.

\subsection{Pairing of Laplacian eigenstates} \label{pairing_sec}

From the chain complex
\begin{equation*}
\begin{tikzcd}
\mathcal{C}^{-1} & \arrow[l,"\partial^0"] \mathcal{C}^0 & \arrow[l,"\partial^1"] \mathcal{C}^1 & \arrow[l,"\partial^2"] \mathcal{C}^2 & \arrow[l,"\partial^3"] \dots
\end{tikzcd}
\end{equation*}
we can build the \emph{graded} vector space
\begin{equation*}
\mathcal{C} = \mathcal{C}^{-1} \oplus \mathcal{C}^0 \oplus \mathcal{C}^1 \oplus \dots
\end{equation*}
We will refer to this as the \emph{Fock space}, where this terminology comes from a connection to supersymmetric quantum systems. $\partial^k$ and $d^k$ give maps
\begin{equation*}
\partial, d : \mathcal{C} \rightarrow \mathcal{C}
\end{equation*}
by acting blockwise. The Laplacian becomes
\begin{align*}
\Delta &= d \partial + \partial d : \mathcal{C} \rightarrow \mathcal{C} \\
\Delta^\downarrow &= d \partial : \mathcal{C} \rightarrow \mathcal{C} \\
\Delta^\uparrow &= \partial d : \mathcal{C} \rightarrow \mathcal{C}
\end{align*}

$\{\partial,d\}$ generate a $C^\ast$-algebra representation on Fock space. Moreover, the Laplacian $\Delta$ commutes with $\partial$ and $d$, and hence commutes with this representation. Thus we can write each $\Delta$-eigenspace as a sum of irreducible $\{\partial,d\}$-subrepresentations.

\begin{proposition} \label{pairing_prop}
The kernel of the Laplacian $\ker{\Delta}$ consists of states that are annihilated by both $\partial$ and $d$. These are one-dimensional $\{\partial,d\}$-irreps, or \emph{singlets}. Let $E>0$ be an eigenvalue of the Laplacian $\Delta$. The $E$-eigenspace is a direct sum of 2-dimensional subspaces $\{|\psi^\uparrow\rangle, |\psi^\downarrow\rangle\}$ with $|\psi^\uparrow\rangle \in \mathcal{C}^k$ for some $k$ and $|\psi^\downarrow\rangle \in \mathcal{C}^{k+1}$ such that
\begin{align*}
&\partial|\psi^\uparrow\rangle = 0 \quad , \quad
d|\psi^\uparrow\rangle \propto |\psi^\downarrow\rangle \\
&\partial|\psi^\downarrow\rangle \propto |\psi^\uparrow\rangle \quad , \quad
d|\psi^\downarrow\rangle = 0
\end{align*}
which implies
\begin{align*}
&\Delta^\downarrow|\psi^\uparrow\rangle = 0 \quad , \quad
\Delta^\uparrow|\psi^\uparrow\rangle = E|\psi^\uparrow\rangle \\
&\Delta^\downarrow|\psi^\downarrow\rangle = E|\psi^\downarrow\rangle \quad , \quad
\Delta^\uparrow|\psi^\downarrow\rangle = 0
\end{align*}
These 2-dimensional subspaces are likewise $\{\partial,d\}$-irreps, or \emph{doublets}. We refer to the states $|\psi^\uparrow\rangle$ as `paired up', and the states $|\psi^\downarrow\rangle$ as `paired down'.
\end{proposition}
\begin{proof}
\emph{(Sketch.)}
$\langle\psi|\Delta|\psi\rangle = 0$ if and only if $\partial|\psi\rangle = d|\psi\rangle = 0$ by \Cref{energy_formulas_lem}, and the states in the kernel of the Laplacian are precisely the \emph{singlets} in this representation.

Now recall that $\partial : \mathcal{C}^{k+1} \rightarrow \mathcal{C}^k$, $d : \mathcal{C}^k \rightarrow \mathcal{C}^{k+1}$ and $\partial^2 = d^2 = 0$. These properties imply that all non-singlet irreducible $\{\partial,d\}$-subrepresentations are necessarily \emph{doublets} supported on neighboring blocks $\mathcal{C}^k \oplus \mathcal{C}^{k+1}$ for some $k$.
\end{proof}

\subsection{Joins}

The \emph{join} will be an important operation for us on simplicial complexes.

\begin{definition}
Given two simplicial complexes $\mathcal{K}$ and $\mathcal{L}$, define their \emph{join} $\mathcal{K} \ast \mathcal{L}$ to be the simplicial complex consisting of simplices $\sigma \otimes \tau := \sigma \cup \tau$ for all $\sigma \in \mathcal{K}$, $\tau \in \mathcal{L}$.
\end{definition}

The chain spaces and homology of the join are given by the Kunneth formula.

\begin{fact}\label{Kunneth_lem} \emph{(Kunneth formula)}
There are canonical isomorphisms
\begin{align*}
\mathcal{C}^k(\mathcal{K} \ast \mathcal{L}) &\cong \bigoplus_{i+j = k-1} \mathcal{C}^i(\mathcal{K}) \otimes \mathcal{C}^j(\mathcal{L}) \\
H^k(\mathcal{K} \ast \mathcal{L}) &\cong \bigoplus_{i+j = k-1} H^i(\mathcal{K}) \otimes H^j(\mathcal{L})
\end{align*}
\end{fact}

We would also like to relate the Laplacian of $\mathcal{K} \ast \mathcal{L}$ to the Laplacians of $\mathcal{K}$ and $\mathcal{L}$.

\begin{lemma}\label{Laplacian_join_lem}
If $|\psi\rangle \in \mathcal{C}^i(\mathcal{K})$ and $|\varphi\rangle \in \mathcal{C}^j(\mathcal{L})$ where $i+j=k-1$, then
\begin{equation*}
\Delta^k (|\psi\rangle \otimes |\varphi\rangle) = (\Delta^i |\psi\rangle) \otimes |\varphi\rangle + |\psi\rangle \otimes (\Delta^j |\varphi\rangle)
\end{equation*}
\end{lemma}
\begin{proof}
If $\sigma$ is an $i$-simplex of $\mathcal{K}$ and $\tau$ is a $j$-simplex of $\mathcal{L}$, and $i+j=k-1$, it can be checked that
\begin{equation*}
\partial^k(|\sigma\rangle \otimes |\tau\rangle) = (\partial^i |\sigma\rangle) \otimes |\tau\rangle + (-1)^{|\sigma|} |\sigma\rangle \otimes (\partial^j |\tau\rangle)
\end{equation*}
and
\begin{equation*}
d^k(|\sigma\rangle \otimes |\tau\rangle) = (d^i |\sigma\rangle) \otimes |\tau\rangle + (-1)^{|\sigma|} |\sigma\rangle \otimes (d^j |\tau\rangle)
\end{equation*}

Now
\begin{equation*}
\Delta^k = d^{k-1} \partial^k + \partial^{k+1} d^k
\end{equation*}
so
\begin{equation*}
\Delta^k(|\sigma\rangle \otimes |\tau\rangle) = (\Delta^i |\sigma\rangle) \otimes |\tau\rangle + |\sigma\rangle \otimes (\Delta^j |\tau\rangle)
\end{equation*}

By linearity, this extends to chains $|\psi\rangle \in \mathcal{C}^i(\mathcal{K})$, $|\varphi\rangle \in \mathcal{C}^j(\mathcal{L})$.
\end{proof}

Since we are building our simplicial complexes as the clique complexes of graphs, we must be able to implement the join at the level of the graphs. This is achieved by taking the two constituent graphs $\mathcal{G}$ and $\mathcal{G}'$ and including all edges between $\mathcal{G}$ and $\mathcal{G}'$.

\begin{definition}
The \emph{join} of two graphs $\mathcal{G} = (V,E)$ and $\mathcal{G}' = (V',E')$ is the graph $\mathcal{G} \ast \mathcal{G}'$ with vertices $V \cup V'$ and edges $E \cup E' \cup \{(u,v): u \in V, v \in V'\}$.
\end{definition}

\begin{fact} \label{fact:join}
The clique complex of the join of two graphs is the join of the clique complexes of the graphs.
\end{fact}

\subsection{Generalized octrahedra}\label{octahedra_sec}

Let $\mathfrak{g}_1$ denote the graph consisting of two disjoint points (with no edges), and let $\mathfrak{g}_n$ be the $n$-fold \emph{join} of $\mathfrak{g}_1$.\footnote{By \Cref{fact:join}, it is equivalent to think of the join as acting at the level of graphs or at the level of simplicial complexes.}
\begin{equation*}
\mathfrak{g}_n = \mathfrak{g}_1 \ast \dots \ast \mathfrak{g}_1 \quad \text{($n$ times)}
\end{equation*}
It will be useful to develop an intuitive interpretation for the complex $\mathfrak{g}_n$. The Kunneth formula tells us that it should have
\begin{equation*}
\dim{H^k} = \begin{cases}
1 & k = n-1 \\
0 & \text{otherwise}
\end{cases}
\end{equation*}
In fact, $\mathfrak{g}_n$ is topologically homeomorphic to the $(n-1)$-sphere $S^{n-1}$. But which triangulation of $S^{n-1}$ does it form in particular?

We can interpret $\mathfrak{g}_n$ as a ``generalized octahedron''. $\mathfrak{g}_2$ is the square loop, and $\mathfrak{g}_3$ is the standard octahedron. The $(n-1)$-simplices of $\mathfrak{g}_n$ are the $n$-cliques of the 1-skeleton. There are $2^n$ of these, corresponding to choosing one vertex from each copy of $\mathfrak{g}_1$. Notice how the number of simplices which make up the higher dimensional octahedron is exponential in the number of vertices.
These generalized octahedron are sometimes referred to as cross-polytopes.

\subsection{Weighting}\label{weighting_sec}

In our construction, we would like to consider \emph{weighted} simplicial complexes. The purpose of this section is to define a natural notion of weighting. We should clarify that the weighting will \emph{not} affect the homology of the simplicial complex; rather, it will only affect the Laplacian operator.

Recall that, in order to define the Laplacian, we had to choose an inner product on the chain spaces $\mathcal{C}^k$. We went with the most basic choice of declaring the simplices themselves to form an orthonormal basis. We will now relax this so that the simplices are orthogonal with weights. The more general inner product is
\begin{equation*}
\langle\sigma|\tau\rangle = \begin{cases}
w(\sigma)^2 & \sigma = \tau \\
0 & \text{otherwise}
\end{cases}
\end{equation*}
for $\sigma, \tau \in \mathcal{K}^k$.

This involves assigning a weight $w(\sigma) \geq 0$ to each simplex $\sigma$ in the simplicial complex $\mathcal{K}$. There are two issues associated with the generality of this definition. Firstly, recall that we introduced the clique complex to provide a more succinct description of a simplicial complex. We likewise need the inner product to be succinctly describable, so listing the weights of all simplices is not possible. The second is that we would like the inner product to respect the join operation. That is, after taking the join of two simplicial complexes $\mathcal{K} \ast \mathcal{L}$, we would like the inner product on $\bigoplus_{i+j = k-1} \mathcal{C}^i(\mathcal{K}) \otimes \mathcal{C}^j(\mathcal{L})$ to be induced from those on $\mathcal{C}^i(\mathcal{K})$ and $\mathcal{C}^j(\mathcal{L})$.

To solve these issues, we add more structure to the definition. Each \emph{vertex} $v$ in the simplicial complex is assigned a weight $w(v)$, and the weights of the higher simplices are induced via
\begin{equation*}
w(\sigma) = \prod_{v \in \sigma} w(v)
\end{equation*}

This allows a graph $\mathcal{G}$ with weighted \emph{vertices} to induce a weighted clique complex. Note that the edges of $\mathcal{G}$ are still binary (present or not present). This also ensures that, in the join construction, the weight of the tensor product of two simplices $\sigma, \tau$ is the product of the individual weights $w(\sigma \otimes \tau) = w(\sigma) w(\tau)$.

How do the coboundary and boundary operators now act on the weighted complex? Let's first consider the coboundary operator. We would like to transform to bases which are orthonormal in the new inner products. Our new basis for $\mathcal{C}^k(\mathcal{K})$ will be $\{|\sigma'\rangle : \sigma \in \mathcal{K}^k\}$, where
\begin{equation*}
|\sigma'\rangle = \frac{1}{w(\sigma)} |\sigma\rangle
\end{equation*}
is the unit vector of $|\sigma\rangle$. $d^k$ originally acted as
\begin{equation*}
d^k |\sigma\rangle = \sum_{v \in \text{up}(\sigma)} |\sigma \cup \{v\}\rangle
\end{equation*}
Written in the new orthonormal basis, this becomes
\begin{align}
d^k |\sigma'\rangle &= \sum_{v \in \text{up}(\sigma)} \frac{w(\sigma \cup \{v\})}{w(\sigma)} |(\sigma \cup \{v\})'\rangle \nonumber\\
&= \sum_{v \in \text{up}(\sigma)} w(v) |(\sigma \cup \{v\})'\rangle \label{coboundary_entries_eq}
\end{align}
$\partial^k$ is defined by
\begin{equation*}
\partial^k = (d^{k-1})^\dag
\end{equation*}
which thus acts by
\begin{equation} \label{boundary_entries_eq}
\partial^k |\sigma'\rangle = \sum_{v \in \sigma} w(v) |(\sigma \setminus \{v\})'\rangle
\end{equation}
From here onwards, we drop the primes on the standard orthonormal basis.

We anticipate it will be useful to the reader to explicitly describe the action of the Laplacian $\Delta^k$ of a weighted clique complex $\mathcal{G}$. Let $\sigma$, $\tau$ be two $k$-simplices. Let's say that $\sigma$ and $\tau$ have a \emph{similar common lower simplex} if we can remove a vertex $v_\sigma$ from $\sigma$ and a vertex $v_\tau$ from $\tau$ such that we get the same $k-1$-simplex, with the same orientation. Let's say that they have a \emph{dissimilar common lower simplex} if the same holds but with opposite orientation. Let's say that $\sigma$ and $\tau$ are \emph{upper adjacent} if they are lower adjacent and their union forms a $k+1$-simplex. $\sigma$ and $\tau$ being upper adjacent means that they are the faces of a common $k+1$-simplex.

\begin{fact} \label{fact_Laplacian_entries} (Similar to \cite[Theorem 3.3.4]{goldberg2002combinatorial})
\begin{equation*}
\langle\sigma|\Delta^k|\tau\rangle =
\begin{cases}
\left(\sum_{u \in \text{up}(\sigma)} w(u)^2\right) + \left(\sum_{v \in \sigma} w(v)^2\right) + 1 & \text{If $\sigma = \tau$.} \\
\\
w(v_\sigma) w(v_\tau) & \parbox{20em}{If $\sigma$ and $\tau$ have a similar common lower simplex and are \emph{not} upper adjacent. $v_\sigma$ and $v_\tau$ are the vertices removed from $\sigma$ and $\tau$ respectively to get the common lower simplex.} \\
\\
- w(v_\sigma) w(v_\tau) & \parbox{20em}{If $\sigma$ and $\tau$ have a dissimilar common lower simplex and are \emph{not} upper adjacent. $v_\sigma$ and $v_\tau$ are the vertices removed from $\sigma$ and $\tau$ respectively.} \\
\\
0 & \parbox{20em}{Otherwise. This includes the case that $\sigma$ and $\tau$ have no common lower simplex, and the case that they are upper adjacent.}
\end{cases}
\end{equation*}
\end{fact}

As a corollary, we can see that the Laplacian $\Delta^k$ is a $\poly(n)$-sparse matrix.

\subsection{Thickening} \label{sec:thickening}
Given a graph whose clique complex, $\mathcal{K}$, is a triangulation of $S^n$.
We would like to find a graph whose clique complex is topologically $\mathcal{K} \times I$ where $I = [0,1]$. Here, we describe a construction which we call \emph{thickening} which achieves this.

First, the following lemma tells us how to \emph{triangulate} $\mathcal{K} \times I$.
\begin{lemma}\label{thickening_lem_1}
Let $\mathcal{K}$ be a simplicial complex. Order the vertices $\mathcal{K}^0$. Let $\mathcal{L}$ be the simplicial complex with vertices $\mathcal{L}^0 = \mathcal{K}^0 \times \{0,1\}$ and simplices
\begin{equation*}
[(u_1,0)(u_2,0) \dots (u_a,0)]
\end{equation*}
whenever $[u_1u_2\dots u_a] \in \Kcyc$,
\begin{equation*}
[(u_1,1)(u_2,1)\dots (u_a,1)]
\end{equation*}
whenever
$[u_1u_2\dots u_a] \in \Kcyc$
and finally,
\begin{equation*}
[(u_1,0)(u_2,0) \dots (u_a,0) (v_1,1)\dots (v_b,1)]
\end{equation*}
whenever
\begin{itemize}
    \item $u_1 < \dots < u_a \leq v_1 < \dots < v_b$
    \item $[u_1\dots u_a] \in \Kcyc$
    \item $[v_1 \dots v_b] \in \Kcyc$
    \item if $u_a=v_1$ then $[u_1\dots u_a v_2 \dots v_b] \in \Kcyc$
    \item if $u_a\neq v_1$ then $[u_1\dots u_a v_1 \dots v_b] \in \Kcyc$
\end{itemize}
Then $\mathcal{L}$ is a triangulation of $\mathcal{K} \times I$.
\end{lemma}
\begin{proof}
To demonstrate that $\mathcal{L}$ is a triangulation of $\Kcyc \times I$ we must show that $\mathcal{L}$ subdivides $\Kcyc \times I$ -- i.e. that $\Kcyc \times I$ is filled by simplices, and no two simplices overlap. 

Consider the first point.
It is immediate from the definition of $\mathcal{L}$ that if the maximal simplices in $\Kcyc$ are $n$-simplices than the maximal simplices in $\mathcal{L}$ are $n+1$-simplices.
The maximal simplices in $\mathcal{L}$ are of the form:
\begin{equation*}
[(u_1,0)(u_2,0) \dots (u_a,0) (u_a,1)(v_2,1)\dots (v_b,1)]
\end{equation*}
where $a+b = n$
Note that every facet of a maximal simplex in $\mathcal{L}$ is either:
\begin{itemize}
\item of the form $[(u_1,0)(u_2,0) \dots (u_a,0)]$ -- i.e. lies on the boundary $\Kcyc \times \{0\}$ of the triangulation
\item  of the form $[(u_1,1)(u_2,1) \dots (u_a,1)]$ -- i.e. lies on the boundary $\Kcyc \times \{1\}$ of the triangulation
\item shared between at least two maximal simplices
\end{itemize}
The first two claims are trivial.
To see that the third claim is true note that $\Kcyc$ is a closed manifold.
Therefore consider a maximal simplex $s \in \mathcal{L}$. By the definition of $\mathcal{L}$, there exists a corresponding maximal simplex $s' \in \Kcyc$, such that  for any vertex $v$ we remove from $s$ to construct a facet there is a corresponding vertex $v'$ we can remove from $s'$.
Since $\Kcyc$ is closed we can always add a \emph{different} vertex $w'$ to $s' \setminus v'$ to give a different maximal simplex $t' \in \Kcyc$.
There will be a corresponding vertex $w$ we can add to the facet of $s$ to give a different maximal facet $t \in \mathcal{L}$, such that $t$ and $s$ share the facet obtained from $s$ by removing $v$. Therefore there are no `gaps' in $\mathcal{L}$, it completely fills $\Kcyc \times I$. 

Consider now the second point -- the simplices of $\mathcal{L}$ must not overlap.
Assume for contradiction that they do overlap. 
Then two simplices $s,t \in \mathcal{L}$ must share an intersection which is not a simplex.
But it is straightforward to check that any two maximal simplices of $\mathcal{L}$ either do not intersect, or intersect on a simplex.
Therefore the simplices of $\mathcal{L}$ do not overlap.
\end{proof}

Finally, we can implement this triangulation with a clique complex using the following lemma.
\begin{lemma}\label{thickening_lem_2}
Let $\mathcal{K}$ be a clique complex. Order the vertices $\mathcal{K}^0$. Let $\mathcal{G}$ be the graph with vertices $\mathcal{K}^0 \times \{0,1\}$ and edges
\begin{align*}
&\{((u,0),(v,0)):(u,v)\in\mathcal{K}^1\} \\
\cup \ &\{((u,1),(v,1)):(u,v)\in\mathcal{K}^1\} \\
\cup \ &\{((v,0),(v,1)):v\in\mathcal{K}^0\} \\
\cup \ &\{((u,0),(v,1)):(u,v)\in\mathcal{K}^1,u<v\}
\end{align*}
The clique complex of $\mathcal{G}$ is $\mathcal{L}$ from \Cref{thickening_lem_1}, and in particular has clique complex triangulating $\mathcal{K} \times I$.
\end{lemma}
\begin{proof}
It is straightforward to check that the cliques in $\mathcal{G}$ are precisely those listed in \Cref{thickening_lem_1}.
\end{proof}

\subsection{Perturbation of subspaces} \label{subspace_perturbation_sec}

In this section, we introduce a notion of a perturbation of a subspace. This will be useful in stating a central lemma in the argument, \Cref{spec_seq_lemma}, and investigating its consequences.

\begin{definition}
Consider a subspace $\mathcal{U} \subseteq \mathcal{V}$ of a complex vector space $\mathcal{V}$. Let $\mathcal{U}_\lambda \subseteq \mathcal{V}$ be a family of subspaces indexed by the continuous parameter $\lambda \in [0,1]$. We say that $\mathcal{U}_\lambda$ is a $\mathcal{O}(\lambda)$-perturbation of $\mathcal{U}$ if there exists orthonormal bases $\{|u\rangle\}_u$ for $\mathcal{U}$ and $\{|u,\lambda\rangle\}_u$ for each $\mathcal{U}_\lambda$ such that
\begin{equation} \label{perturbation_condition}
\left|\left| |u,\lambda\rangle - |u\rangle \right|\right| = \mathcal{O}(\lambda) \ \forall u
\end{equation}
\end{definition}

Using this definition, we prove a two-part lemma which will be useful later

\begin{lemma} \label{subspace_perturbation_lem}
Suppose $\mathcal{U}_\lambda$ is a $\mathcal{O}(\lambda)$-perturbation of subspace $\mathcal{U} \subseteq \mathcal{V}$. Let $\Pi$ be the orthogonal projection onto $\mathcal{U}$ and $\Pi_\lambda$ orthogonal projection onto $\mathcal{U}_\lambda$ for each $\lambda$. Then
\begin{enumerate}
    \item $||\Pi_\lambda - \Pi||_{\text{\emph{op}}} \leq \mathcal{O}(\lambda)$ in operator norm.
    \item If $|\psi_\lambda\rangle$ is a parametrized family of states such that $\langle\psi_\lambda|\Pi|\psi_\lambda\rangle = \mathcal{O}(\lambda^2)$, then $\langle\psi_\lambda|\Pi_\lambda|\psi_\lambda\rangle = \mathcal{O}(\lambda^2)$.
\end{enumerate}
\end{lemma}
\begin{proof}
Using the condition \Cref{perturbation_condition}, we can write
\begin{equation*}
|u,\lambda\rangle = |u\rangle + \mathcal{O}(\lambda) |\tilde{u}_\lambda\rangle
\end{equation*}
for some normalized vector $|\tilde{u}_\lambda\rangle$. Then
\begin{align*}
\Pi_\lambda &= \sum_u |u,\lambda\rangle \langle u,\lambda| \\
&= \sum_u \big(|u\rangle + \mathcal{O}(\lambda)|\tilde{u}_\lambda\rangle\big) \big(\langle u| + \mathcal{O}(\lambda)\langle \tilde{u}_\lambda|\big) \\
&= \Big(\sum_u |u\rangle \langle u|\Big) + \mathcal{O}(\lambda) \sum_u |u\rangle\langle \tilde{u}_\lambda| + \mathcal{O}(\lambda) \sum_u |\tilde{u}_\lambda\rangle\langle u| + \mathcal{O}(\lambda^2) \\
&= \Pi + \mathcal{O}(\lambda) \sum_u |u\rangle\langle \tilde{u}_\lambda| + \mathcal{O}(\lambda) \sum_u |\tilde{u}_\lambda\rangle\langle u| + \mathcal{O}(\lambda^2)
\end{align*}
We can immediately read off Part 1, that $||\Pi_\lambda - \Pi|| \leq \mathcal{O}(\lambda)$ in operator norm. As for Part 2,
\begin{align*}
\langle\psi_\lambda|\Pi_\lambda|\psi_\lambda\rangle &= \langle\psi_\lambda|\Pi|\psi_\lambda\rangle + \mathcal{O}(\lambda) \sum_u \langle\psi_\lambda|u\rangle\langle \tilde{u}_\lambda|\psi_\lambda\rangle + \mathcal{O}(\lambda) \sum_u \langle\psi_\lambda|\tilde{u}_\lambda\rangle\langle u|\psi_\lambda\rangle + \mathcal{O}(\lambda^2) \\
&= \mathcal{O}(\lambda^2) + \mathcal{O}(\lambda) \sum_u \langle\psi_\lambda|u\rangle\langle \tilde{u}_\lambda|\psi_\lambda\rangle + \mathcal{O}(\lambda) \sum_u \langle\psi_\lambda|\tilde{u}_\lambda\rangle\langle u|\psi_\lambda\rangle \\
&\leq \mathcal{O}(\lambda^2) + \mathcal{O}(\lambda) \Big(\sum_u |\langle u|\psi_\lambda\rangle|^2\Big)^{\frac{1}{2}} \Big(\sum_u |\langle \tilde{u}_\lambda|\psi_\lambda\rangle|^2\Big)^{\frac{1}{2}} \\
&\leq \mathcal{O}(\lambda^2) + \mathcal{O}(\lambda) \Big(\langle\psi_\lambda|\Pi|\psi_\lambda\rangle\Big)^{\frac{1}{2}} \\
&\leq \mathcal{O}(\lambda^2)
\end{align*}
We used the assumption $\langle\psi_\lambda|\Pi|\psi_\lambda\rangle = \mathcal{O}(\lambda^2)$, and then Cauchy-Schwarz, and finally the assumption again.
\end{proof}

\bigskip
\section{Hamiltonian to homology gadgets}\label{sec:construction}

Given an instance of quantum $m$-SAT $H$ on $n$ qubits, we are aiming to find a weighted graph $\mathcal{G}$ and some $k$ such that the ground energy of the $k^{\text{th}}$-order Laplacian $\Delta^k$ is related to the ground energy of $H$. We will do this by first constructing a \emph{qubit graph} $\mathcal{G}_n$ whose harmonic states can be identified with qubit states $|\psi\rangle \in (\mathbb{C}^2)^{\otimes n}$.\footnote{Here we do not use the same graph as in \cite{crichigno2022clique} because in order to maintaing a link between the spectrum of the Laplacian and the spectrum of the $\qfsat$ Hamiltonian we are reducing from we need our basis states to correspond to orthogonal holes in the clique complex, not just distinct homology classes.}
This means the $k$-homology of $\mathcal{G}_n$ will be $2^n$-dimensional. The clique complex of this graph will have no $k+1$-simplices, so $\Delta^{\uparrow k} = 0$.

The $m$-local Hamiltonian $H$ is a sum of rank-1 projectors
\begin{equation*}
H = \sum_{i=1}^t \phi_i
\end{equation*}
where each $\phi_i = |\phi_i\rangle\langle\phi_i|$ for a state $|\phi_i\rangle$ on at most $m$ qubits.
\begin{definition} \label{integer_state_def}
$|\phi\rangle$ is an \emph{integer state} if it can be written as
\begin{equation*}
|\phi\rangle = \frac{1}{\mathcal{Z}} \sum_{z \in \{0,1\}^m} a_z |z\rangle
\end{equation*}
where $a_z \in \mathbb{Z}$ are integers, and $\mathcal{Z} = (\sum_z |a_z|^2)^{\frac{1}{2}}$ is a normalization factor.
\end{definition}
We will assume $\{|\phi_i\rangle\}$ are integer states. For each term $\phi_i$ we will add a \emph{gadget} $\mathcal{T}_i$ to the graph which aims to implement the effect of this term on the groundspace of harmonic states. $\mathcal{T}_i$ will consist of additional vertices and edges which add $k+1$-simplices to the clique complex so that $\Delta^{\uparrow k}$ implements $\phi_i$. Topologically, this is achieved by designing $\mathcal{T}_i$ to be a triangulation of a $k+1$-manifold whose \emph{boundary} is the cycle corresponding to $|\phi_i\rangle$.

In our construction, the clique complex will be \emph{weighted} as described in \Cref{weighting_sec}. Recall that this is done by assigning weights to the vertices of the underlying graph. The vertices of the original qubit graph $\mathcal{G}_n$ will all have weight $1$. The vertices of the gadgets $\mathcal{T}_i$ will have weight $\frac{1}{\poly{n}} \leq \lambda \ll 1$. The weights are much smaller than 1, although only polynomially so.

\subsection{Qubit graph}

Let $\mathcal{G}_1$ be the \emph{bowtie graph} (see \Cref{fig:qubit_complex}). 
The clique complex has no 2-simplices, and the 1-homology is isomorphic to $\mathbb{C}^2$, spanned by equivalence classes of the left and right loops. Further, the loops themselves form an orthonormal basis of the harmonic states $\ker{\Delta^1} \subseteq \mathcal{C}^1$. Identify $|0\rangle$ with the left loop, (i.e. the cycle $[xa_3] + [a_3a_2] + [a_2a_4] + [a_4x]$) and $|1\rangle$ with the right loop (i.e. the cycle $[xb_3] + [b_3b_2] + [b_2b_4] + [b_4x]$). Note that $\mathcal{C}^1 \cong \mathbb{C}^8$ since there are 8 1-simplices (edges). The other 6 states have some higher energies on $\Delta^1$.

To construct $\mathcal{G}_n$, we employ the \emph{join} operation.
\begin{equation*}
\mathcal{G}_n = \mathcal{G}_1 \ast \dots \ast \mathcal{G}_1 \quad \text{($n$ times)}
\end{equation*}
Using the Kunneth formula \Cref{Kunneth_lem}, we see that $H^{2n-1} \cong (\mathbb{C}^2)^{\otimes n}$. Moreover, the computational qubit states $|z_1\rangle \otimes \dots \otimes |z_n\rangle$ for $z_i \in \{0,1\}$ form a natural orthonormal basis for the groundspace $\ker{\Delta^{2n-1}} \subseteq \mathcal{C}^{2n-1}$. Note this means that we take $k = 2n-1$. Denote
\begin{equation*}
\mathcal{H}_n := \ker{\Delta^{2n-1}} \cong (\mathbb{C}^2)^{\otimes n}
\end{equation*}

What is the cycle corresponding to a computational basis state $|z\rangle , z \in \{0,1\}^n$? $|z\rangle$ is the join of $n$ copies of the square loop. The square loop is a two-fold join $\mathfrak{g}_2 = \mathfrak{g}_1 \ast \mathfrak{g}_1$ from \Cref{octahedra_sec}. Thus $|z\rangle$ is a copy of the $(2n-1)$-dimensional octahedron $\mathfrak{g}_{2n}$. $|z\rangle$ is topologically homeomoprhic to $S^{2n-1}$.
\begin{figure}[H]
\begin{center}
\includegraphics[scale=0.6]{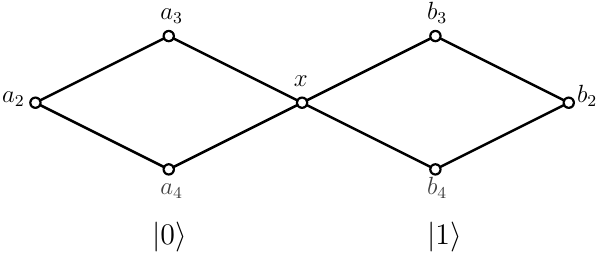}
\caption{The clique-complex that encodes a single-qubit.}
\label{fig:qubit_complex}
\end{center}
\end{figure}

\subsection{Single gadget} \label{single_gadget_sec}

First we focus on implementing a single term $\phi$ where $\phi=\ketbra{\phi}$. For this section we drop the subscript $i$ labelling the terms.

Let's first work on implementing the Hamiltonian term $\phi$ on the harmonic subspace $\mathcal{H}_m \subseteq \mathcal{C}^{2m-1}(\mathcal{G}_m)$, $\mathcal{H}_m \cong (\mathbb{C}^2)^{\otimes m}$ of the $m$-qubit graph. This will involve adding vertices and edges to $\mathcal{G}_m$, and these additional vertices and edges make up the gadget. Let the final graph be denoted $\hat{\mathcal{G}}_m$. Eventually we can then add the remaining $n-m$ qubits and find the relevant gadget for $\mathcal{G}_n$ by joining $\mathcal{G}_{n-m}$
\begin{equation*}
\hat{\mathcal{G}}_n = \hat{\mathcal{G}}_m \ast \mathcal{G}_{n-m}
\end{equation*}

Motivated by the equation $\langle\psi|\Delta^{\uparrow k}|\psi\rangle = ||d^k|\psi\rangle||^2$, we would like the gadget to be a \emph{triangulation} of a \emph{manifold} whose \emph{boundary} is precisely the cycle $|\phi\rangle$. This can be viewed topologically as \emph{filling in} the cycle $|\phi\rangle$. 
Throughout this section $\J$ refers to the cycle we want to fill in with the gadget (i.e. if we are constructing a gadget for the projector $\ketbra{\phi}{\phi}$ then $\J$ is the cycle that corresponds to the state $\ket{\phi}$), and $\mathcal{J}^0$ refers to the vertices in $\mathcal{J}$.
The procedure is as follows.
\begin{enumerate}
    \item \label{it:identify} Construct a clique complex $\mathcal{K}$ which is a triangulation of $S^{2m-1}$ and a relation:
    \begin{equation*}
        R \subseteq \{(z_i,z_j)| z_i,z_j \in \mathcal{K}^0\}
    \end{equation*}
    such that the map
    \begin{equation} \label{eq:f}
    \begin{split}
        &f:\Kcyc^0 \rightarrow \J^0 \\
       &f(z_i) = z_j \textrm{\ \ where \ \ } (z_i,z_j) \in R
        \end{split}
    \end{equation}
        is a surjective function, and the simplicial complex given by:
        \begin{equation}\label{eq:f complex}
        \{f(\sigma) \cap \J^0 | \sigma \in \Kcyc \}
        \end{equation}
        is a copy $\mathcal{J}$ of the cycle $|\phi\rangle$.\footnote{See the end of this section for a an example of how the function $f$ is applied and conditions on $R$ for $f$ to be a surjective function.} 
    \item Let $\mathcal{L}$ be the thickening of $\mathcal{K}$ described by \Cref{thickening_lem_1}. The vertices of $\mathcal{L}$ are $\mathcal{L}^0 = \mathcal{K}^0 \times \{0,1\}$. By \Cref{thickening_lem_2}, $\mathcal{L}$ is the clique complex of some graph. This creates a thickened $S^{2m-1}$ shell with outside layer $\mathcal{K}^0 \times \{0\}$ and inside layer $\mathcal{K}^0 \times \{1\}$.
    \item Add a central vertex $v_0$ which connects to all vertices of the inside layer $\mathcal{K}^0 \times \{1\}$. By the definition of clique complexes (\Cref{sec:clique}), this automatically introduces into the complex all the simplices which contain these new edges. We will refer to this process as \emph{coning off the cycle}.\footnote{The terminology arises from the concept of constructing a \emph{mapping cone} -- the cycle $\mathcal{L}$ can be represented as a map of a sphere onto the complex the process of adding this central vertex and the associated simplices can be seen as constructing the mapping cone.} Denote the simplicial complex at this step by $\hat\Kcyc$.
    \item Let $\mathcal{V} = \{\J^0 \cup \Kcyc^0 \times \{1\} \cup v_0\}$. Apply $f(\cdot)$ to $\mathcal{K}^0 \times \{0\}$ and construct the cycle:
    \begin{equation}\label{eq:f complex 2}
    \tilde \Kcyc = \{f(\sigma) \cap \mathcal{V} | \sigma \in \hat \Kcyc  \}
    \end{equation}
    as set out in \cref{eq:f complex} to get $\mathcal{J}$, a copy of the cycle $|\phi\rangle$.
    \item The weights of the vertices in the outer layer $\mathcal{J}^0$, which will be identified with vertices in the qubit graph $\mathcal{G}_m$, remain set to $1$.
    \item The weights of the inside layer $\mathcal{K}^0 \times \{1\}$ and the central vertex $v_0$ are set to $\lambda$ where $\frac{1}{\poly{n}} \leq \lambda \ll 1$.
    \item\label{final_step} The above forms the gadget. It remains to simply glue $\mathcal{J}$ onto the cycle $|\phi\rangle \in \Cl(\mathcal{G}_m)$ by identifying the vertices in $\J^0$ with the equivalent vertices in $\Cl(\mathcal{G}_m)$. 
    In \cref{sec:gadget proofs} we prove that the resulting complex is the clique complex of some new graph.
    Let the new graph with the gadget glued in be denoted $\hat{\mathcal{G}}_m$.
\end{enumerate}

The above procedure provides a method to fill-in any cycle, provided it is possible to complete \Cref{it:identify}.
For basis states \Cref{it:identify} is trivial -- the cycles themselves are already triangulations of $S^{2m-1}$ and clique complexes.
For arbitrary states we do not have a general method for carrying out the procedure.
However, for integer states (see \Cref{integer_state_def}) we can construct a general method for \Cref{it:identify}.
The method relies on the simplicial surgery techniques, first introduced in \cite{crichigno2022clique}.

\textbf{Note on the function $f(\cdot)$:} The function $f(\cdot)$ acts on the \emph{vertices} of a simplicial complex.
We then construct a new simplicial complex out of the new set of vertices, according to the method set out in \cref{eq:f complex}.
To see how this works explicitly we consider a simple example.
Take a simplicial complex, specified by its maximal faces: 
\begin{equation}
K = \{[x_0x_1x_2x_3],[x_0'x_1'x_2'x_3],[x_2x_3x_4x_5],[x_0x'_1x'_2,x_3],[x_0'x_1'x_1x_2'x_2x_3x_4x_5]\}
\end{equation}
Define a relation:
\begin{equation}
R = \{(x_i',x_i)| i \in [0,2] \} \cup \{(x_i,x_i)| i \in [0,5] \}
\end{equation}
Acting with $f(\cdot)$ on $K$ then gives a simplicial complex defined by maximal faces:
\begin{equation}
f(K) = \{[x_0x_1x_2x_3],[x_2x_3x_4x_5],[x_0x_1x_2,x_3],[x_0x_1x_2x_3x_4x_5]\}
\end{equation}
Note, that any simplices $\sigma \in K$ which do not contain any pairs from the relation $R$ map to simplices of the same dimension under $f(\cdot)$ (the vertices may have changed, but the number of vertices in the simplex is unchanged).
Simplices $\sigma \in K$ which contain one or more pairs of vertices from $R$ will map under $f(\cdot)$ to lower dimensional simplices, because simplices cannot contain two of the same vertex, so the number of vertices remaining in the simplex has decreased.

\textbf{Note on the relation R}: In order for $f(\cdot)$ to be a surjective function $f:\mathcal{K}^0 \rightarrow J^0$ as required the relation $R$ must satisfy:
\begin{itemize}
\item $R$ must be functional, i.e. for all $x,y,z \in \mathcal{K}^0$, $(x,y) \in R$ and $(x,z) \in R$ implies $y=z$
\item for all $y \in \mathcal{J}^0$ there must exist $x \in \mathcal{K}^0$ such that $(x,y) \in R$
\end{itemize}

\subsection{Constructing the gadgets for integer states} \label{constructing_cycles_sec}

Before describing the gadgets for the states in \Cref{table:states} we will outline some examples covering the cases up to $m=4$.
The only gadgets we actually need are those for $m=3$ and $m=4$. 
But we include the cases $m=1$ and $m=2$ for the reader to gain intuition about the construction before we tackle the more complicated situations.
The proof that the constructions we have outlined can be applied to all the states in \Cref{table:states} and fill in the desired states is given in \Cref{sec:gadget proofs}.

In \Cref{app:gadgets} we outline a general method to construct gadgets for arbitrary integer states. 
Throughout this section, and in \Cref{app:gadgets} we will use the notation from \Cref{fig:qubit_complex} to label the vertices of the copy of $\mathcal{G}_1$ corresponding to the first qubit.
For states of more than one qubit we will denote the vertices of the second copy of $\mathcal{G}_1$ (corresponding to the second qubit) by $x',a_1',a_2' \dots b_4'$, the vertices of the third copy of $\mathcal{G}_1$ by $x'',a_1'',a_2'' \dots b_4''$, etc.

\subsubsection{Single-qubit basis states}

As already stated, constructing the cycles for computational basis states is trivial.
We simply take $\mathcal{K} = \mathcal{J}$ where for single-qubits $\mathcal{J}$ is either the left or right cycle of the bow-tie.
To lift the state we could then add a single mediator - as in \Cref{fig:qubit_lift_zero_simple}.
While this method works perfectly for lifting basis states (in any dimension) it does not generalise to lifting more complicated states.
This is because when filling in more complicated cycles it is possible to fill in additional unwanted cycles using this simple method.
To avoid this we apply the \emph{thickening} method before adding the central vertex, as outlined in \Cref{single_gadget_sec}.

\begin{figure}[H]
\begin{center}
\includegraphics[scale=0.6]{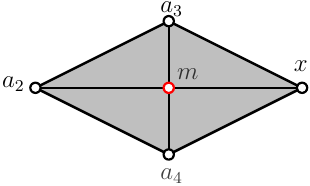}
\caption{The simplest method for `lifting' the $|0\rangle$ state.}
\label{fig:qubit_lift_zero_simple}
\end{center}
\end{figure}

In the case of the $\ketbra{0}{0}$ state the full method gives the clique complex shown in \Cref{fig:qubit_lift_zero_reall}.
The graph which results in this clique complex is simply the 1-skeleton of the complex shown in \Cref{fig:qubit_lift_zero_reall}.

\begin{figure}[H]
\begin{center}
\includegraphics[scale=0.4]{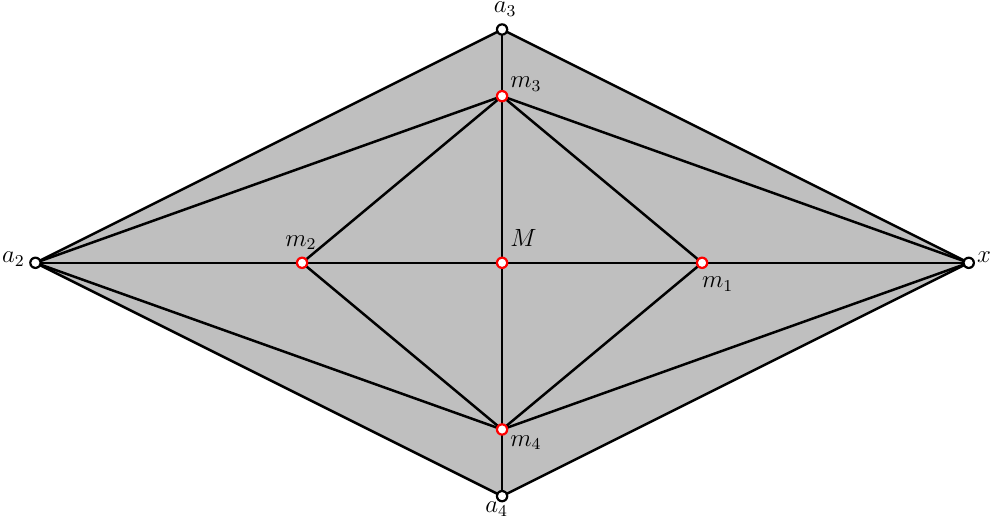}
\caption{The method we actually use for `lifting' the $|0\rangle$ state. This figure corresponds to applying the `thickening' procedure (outlined in \Cref{sec:thickening}) before adding the central vertex, as detailed in \Cref{single_gadget_sec}. This standard procedure works for all cycles (whereas the method used in \Cref{fig:qubit_lift_zero_simple} only works for basis states). Using a single method that works for arbitrary cycles simplifies the spectral sequence analysis in the proof of \Cref{spec_seq_lemma}.}
\label{fig:qubit_lift_zero_reall}
\end{center}
\end{figure}

\subsubsection{Superposition of two single-qubit basis states}

Now we consider slightly more complicated single-qubit states which are a superposition of two single-qubit basis states.
We will take $|-\rangle = |0\rangle - |1\rangle$ as an example (ignoring normalization).
We need to construct a triangulation of $S^1$, $\mathcal{K}$, where we can identify certain vertices to give two copies of $S^1$, concatenated with opposite orientation.
The two copies of $S^1$ share a single vertex, $x$ (\emph{all} single-qubit basis states contain the vertex $x$, so this cutting procedure works for \emph{any} single-qubit basis state).
So to construct the cycle $\Kcyc$ we split open the $|0\rangle$ and $
|1\rangle$ cycles by breaking the $x$ vertex into two -- one labeled $x$ and one labeled $x_1$.
Since we want to add the cycles with opposite orientations we choose opposite labelling for the $x$ and $x_1$ vertices in the two cycles (see \Cref{fig:filling_in_minus_copy_2}).

\begin{figure}[H]
\begin{center}
\includegraphics[scale=0.4]{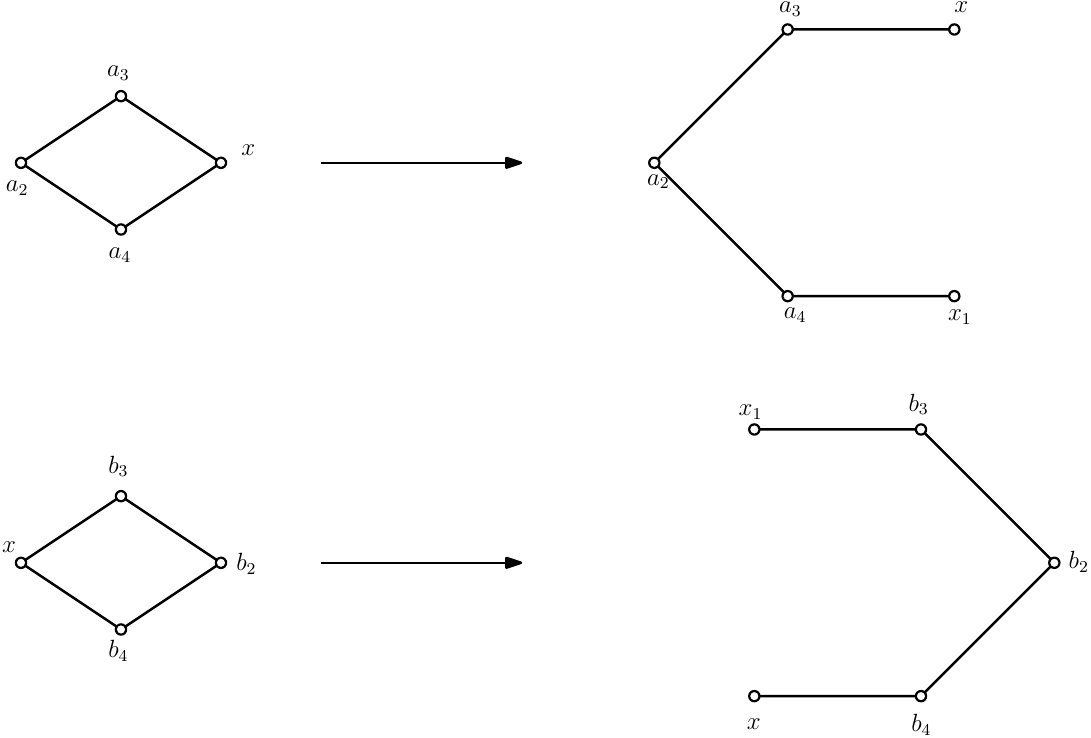}
\caption{To construct $\Kcyc$ for the $\minusST$}
\label{fig:filling_in_minus_copy_2}
\end{center}
\end{figure}

We then glue together the two cycles by identifying the $x$ and $x_1$ vertices from the different cycles with each other.
This gives the cycle $\Kcyc$ shown in \Cref{fig:qubit_lift_minus 1}.
Crucially the cycle $\Kcyc$ is a copy of $S^1$, so we can apply the method from \Cref{single_gadget_sec} to fill it in.

\begin{figure}[H]
\begin{center}
\includegraphics[scale=0.4]{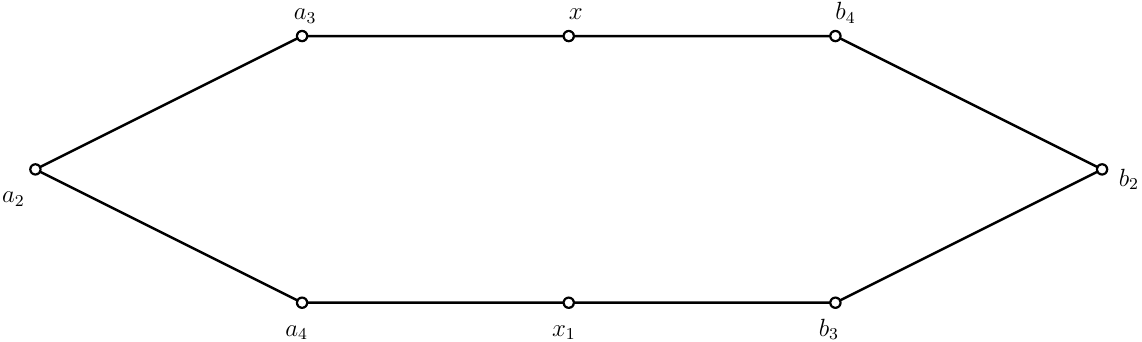}
\caption{This complex $\Kcyc$ is a triangulation of $S^1$ that gives the $|-\rangle$ cycle when the function $f(\cdot)$ is applied to it with the relation from \Cref{eq:R}.}
\label{fig:qubit_lift_minus 1}
\end{center}
\end{figure}

Applying the thickening and coning off procedure from \Cref{single_gadget_sec} to $\Kcyc$ gives the simplicial complex shown in \Cref{fig:qubit_lift_minus}.
To complete the gadget we have to construct a relation $R$ such that applying the function $f(\cdot)$ from \Cref{eq:f} to the outer vertices gives a copy of the cycle corresponding to the $\ket{-}$ state.
The relation is given by:
\begin{equation}\label{eq:R}
R = \{(x_1,x)\} \cup \{(v,v)|v\in \J^0\}
\end{equation}
It is straightforward to check that with this $R$, applying the function from \Cref{eq:f} to $\mathcal{K}^0$ gives the required cycle.

\begin{figure}[H]
\begin{center}
\includegraphics[scale=0.4]{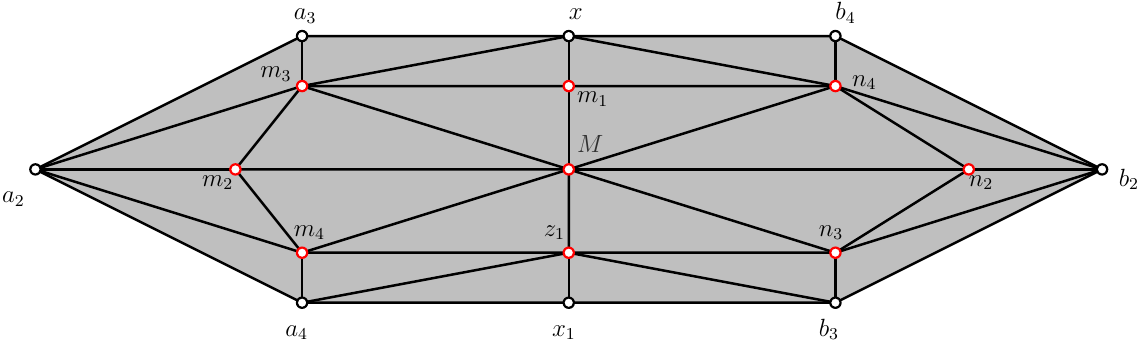}
\caption{The complex that fills in the $|- \rangle$ state \emph{before} we apply the function $f(\cdot)$ to the outer layer.}
\label{fig:qubit_lift_minus}
\end{center}
\end{figure}

\subsubsection{Superposition of more than two single-qubit basis states}

The final single-qubit example we will consider involves adding three computational basis states.
For concreteness we will consider the state $|0\rangle + 2|1\rangle$, but the method applies equally well to any other state involving three or more computational basis states.
To construct $\Kcyc$ for the state $|0\rangle + 2|1\rangle$ we now need to cut open two copies of the $\1$ cycle.
We label the $b_i$ vertices of the second copy of the $\1$ cycle by $b_{i,1}$.
When we cut open the $\z$ cycle and $\1$ cycles now we label the $x$ vertices of the $\z$ cycle by $x$ and $x_2$, and the $x$ vertices of the two $\1$ cycles are labelled by $x$, $x_1$ and $x_1$, $x_2$ respectively.
To glue the cycles together we identify the $x$, $x_1$ and $x_2$ vertices from the different cyclces, as shown in \Cref{fig:qubit_lift_minus minus}.
We can then apply the thickening procedure to $\Kcyc$, before adding a central vertex (as outlined in \Cref{single_gadget_sec}) to fill in $\Kcyc$. 
The final step in constructing the gadget is then to construct the relation $R$ and apply the function $f(\cdot)$ from \Cref{eq:f}.
For this example $R$ is given by:
\begin{equation}
R = \{(x_i,x)|i \in [1,2])\} \cup \{(b_{i,1},b_i)| i \in [2,4]\} \cup \{(v,v)| v\in \J^0\}
\end{equation}

\begin{figure}[H]
\begin{center}
\includegraphics[scale=0.4]{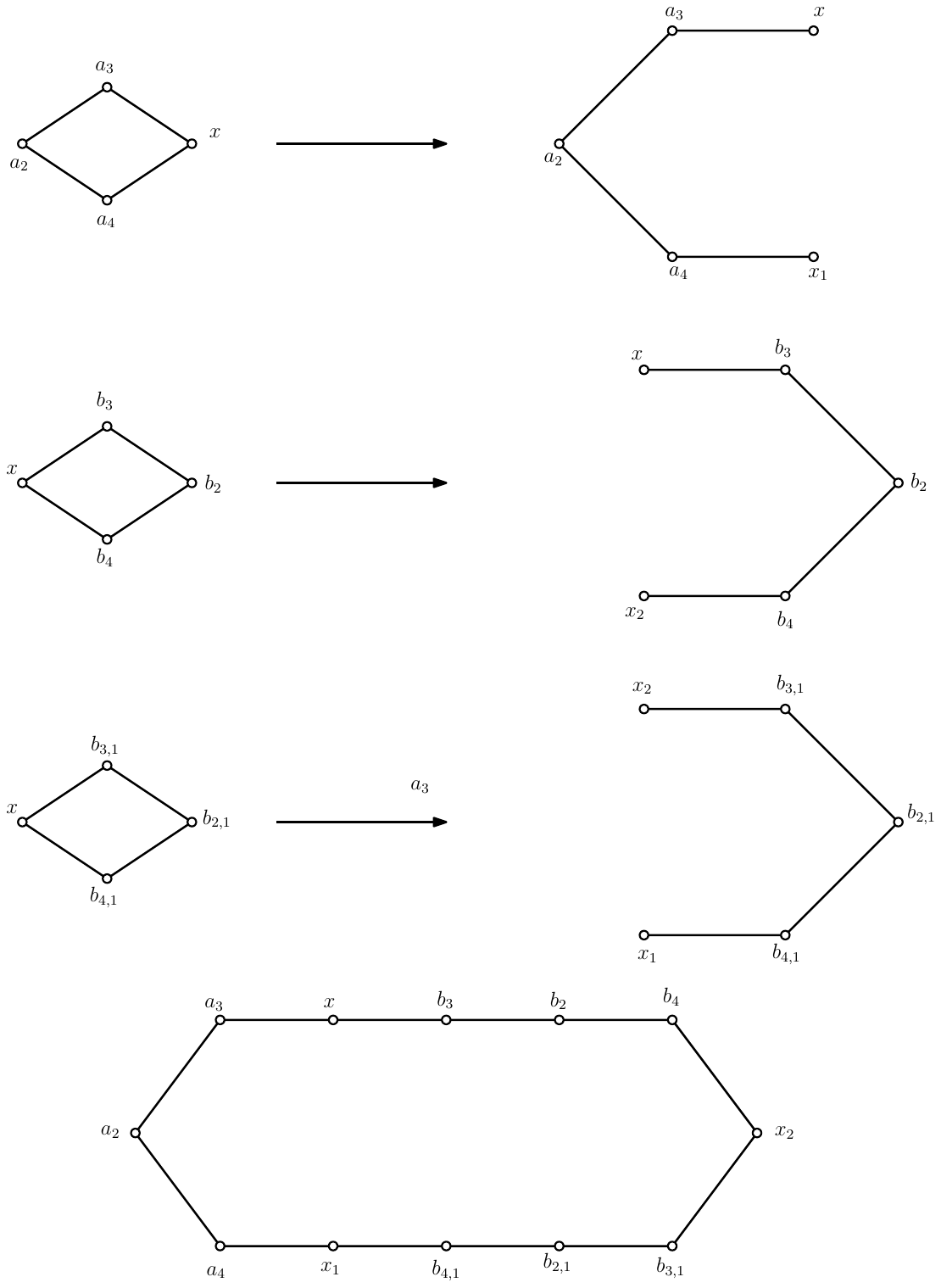}
\caption{Constructing a triangulation of $S^1$ that gives a copy of the cycle $|0\rangle + 2|1\rangle$ when the function $f(\cdot)$ is applied.}
\label{fig:qubit_lift_minus minus}
\end{center}
\end{figure}

\subsubsection{Two-qubit basis states}

The two-qubit basis states are triangulations of $S^3$, \Cref{fig:00cycle}.
As in the single-qubit case, the two-qubit basis states can be filled in by applying the thickening and coning off procedure from \Cref{single_gadget_sec} directly to the state itself.

\begin{figure}[H]
\begin{center}
\includegraphics[scale=0.4]{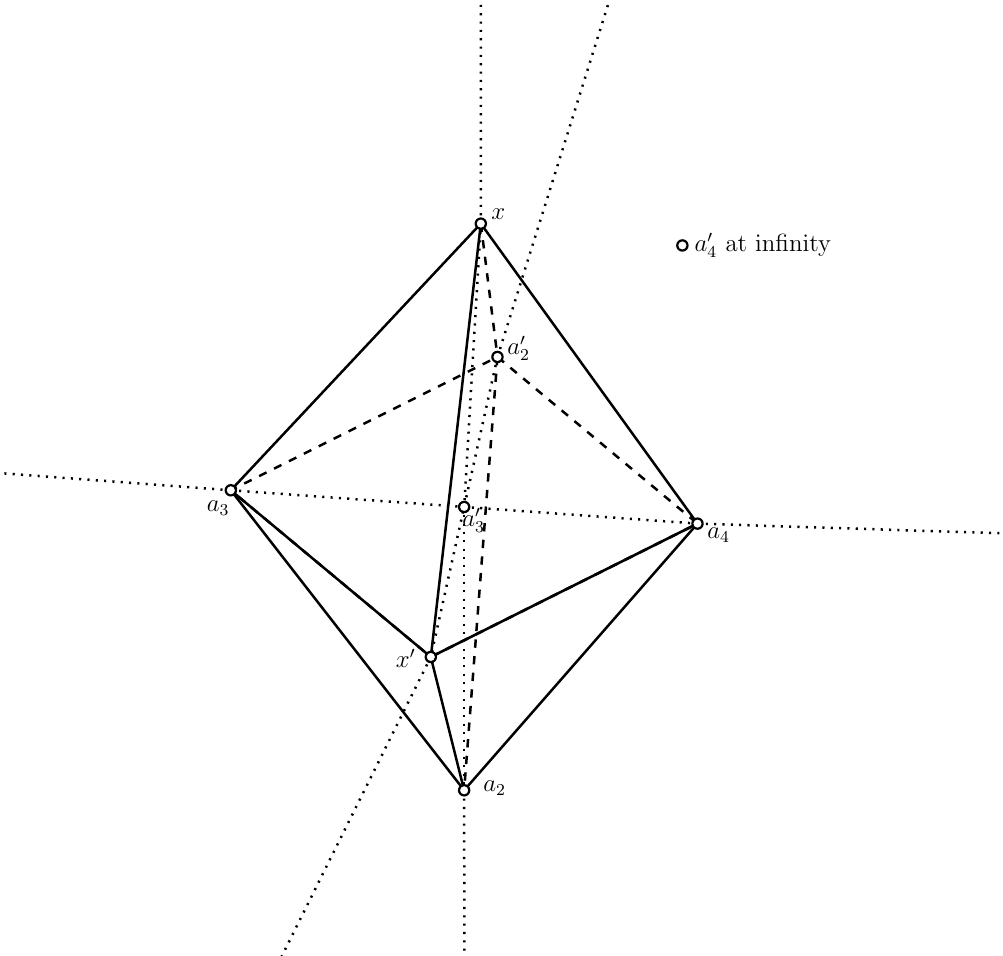}
\caption{The two-qubit basis states are 3-cycles. More specifically, they are the boundaries of 4-dimensional cross-polytopes (also known as 16-cells). These are 4-dimensional generalisations of octahedron, with eight vertices - equivalently they are bipiramids with an octahedron base. Their connectivity can be visualised by taking an octahedron, and adding one of the additional vertices in the centre, and the other additional vertex at infinity. The two additional vertices are connected to all the original octahedron vertices, but not to each other.}
\label{fig:00cycle}
\end{center}
\end{figure}

\subsubsection{Superposition of two two-qubit basis states}
\label{sec:2 qubit cutting}
Moving on to more complicated two-qubit states.
For concreteness we will consider the state $\ket{00}-\ket{11}$, but the same cutting and gluing procedure applies to any state comprised of two two-qubit computational basis states.
We now need to construct a triangulation of $S^3$ where certain vertices can be identified to give two copies of $S^3$.

The two copies of $S^3$ share the edge $[xx']$. In fact, \emph{all} two-qubit basis states contain the edge $[xx']$, so this cutting procedure works for \emph{any} two-qubit basis state.
In order to glue together the two copies of $S^3$ we will cut open this edge, and glue the two copies along the cut.
The portion of the  $|00\rangle$ cycle that includes the edge $[xx']$ is shown in \Cref{fig:singlet_one}.

\begin{figure}[H]
\begin{center}
\includegraphics[scale=0.4]{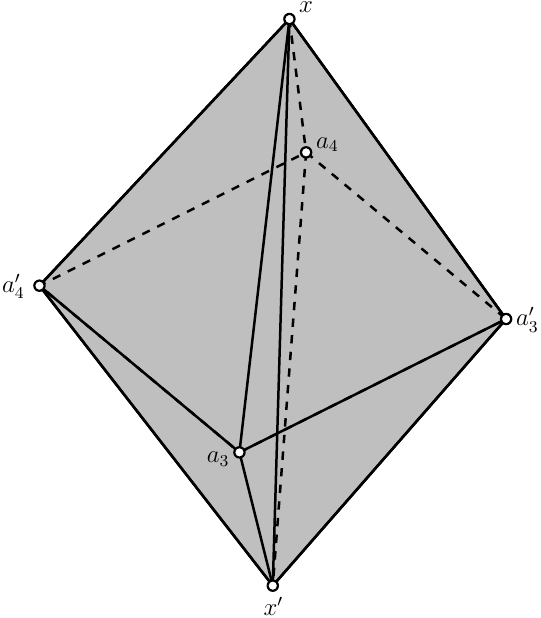}
\caption{All two-qubit cycles share the 1-simplex $[xx']$. We will `glue together' two-qubit cycles to form entangled states by cutting open this edge, and gluing along the cut. In this figure we visualise the portion of the $|00\rangle$ cycle that includes the 1-simplex $[xx']$.}
\label{fig:singlet_one}
\end{center}
\end{figure}

In order to glue together two 3-cycles along the $[xx']$ edge we need to cut the edge into an open 2-cycle.
We do that by introducing 4 new dummy vertices, $x_i$ for $i \in [1,4]$ as shown in \Cref{fig:singlet_two}.
We have to take care to connect the dummy vertices up to the $a_j$ or $b_j$ vertices surrounding the $[xx']$ edge correctly in order to make sure that the two cycles are glued together correctly, and no additional holes are introduced.
To see how we do that we will take the $
|00\rangle$ state as an example.
As can be seen in \Cref{fig:singlet_two} there is a square loop $[a_3a_4']+[a_4'a_4]+[a_4a_3']+[a_3'a_3]$ (a $\mathfrak{g}_2$) that surrounds $[xx']$ where each vertex of the vertices that make up the $\mathfrak{g}_2$ are connected to both $x$ and $x'$.
To complete our cutting open of the $|00\rangle$ cycle we associate each of the $x_i$ with one of the edges of the $\mathfrak{g}_2$, and connect the $x_i$ up to the two vertices that make up that edge (see \Cref{fig:singlet_three} and \Cref{fig:singlet_final} for an illustration).

So far we have been using the $|00\rangle$ cycle as an illustration.
We now apply exactly the same procedure to the $|11\rangle$ cycle - see \Cref{fig:singlet_final really}.\footnote{In this example we are adding the two cycles with opposite signs, so to add them with the same sign we reverse the orientation of the $[x_1x_2]+[x_2x_3]+[x_3x_4]+[x_4x_1]$ $\mathfrak{g}_2$ for one of the cycles}.
We then glue the two cycles together by identifying the $x_i$ from the two cycles.
This gives us our $\mathcal{K}$ cycle.
We then apply the thickening and coning off procedure from \Cref{single_gadget_sec}.
The final step in constructing the gadget is to construct the relation $R$ and apply the function $f(\cdot)$ from \Cref{eq:f} to the outer vertices to give the $\J$ cycle, which can be glued onto the original graph.
The relation for this example is given by:
\begin{equation}\label{eq:R1}
R = \{(x_i,x)|i \in [1,4]\} \cup \{(v,v)| v\in \J^0\}
\end{equation}

\begin{figure}[H]
\begin{center}
\includegraphics[scale=0.4]{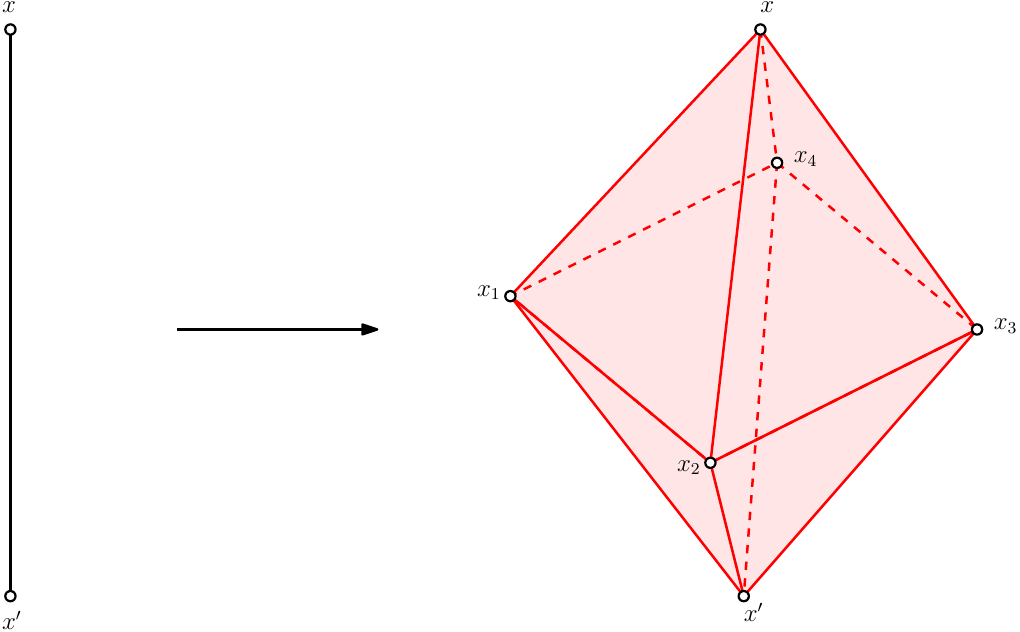}
\caption{In order to glue two 3-cycles along the 1-simplex $[xx']$ we need to cut $[xx']$ open to form an open 2 cycle. We do that by introducing four new vertices - $x_i$ for $i \in [1,4]$, and forming an octahedron by taking a square formed of the vertices $x_i$ and taking this as the base of a bipiramid with $x$, $x'$ being the final two vertices. We will later identify all the dummy vertices $x_i$ with the original vertex $x$ by applying the function $f(\cdot)$ with the relation from \Cref{eq:R1}.}
\label{fig:singlet_two}
\end{center}
\end{figure}

\begin{figure}[H]
\begin{center}
\includegraphics[scale=0.4]{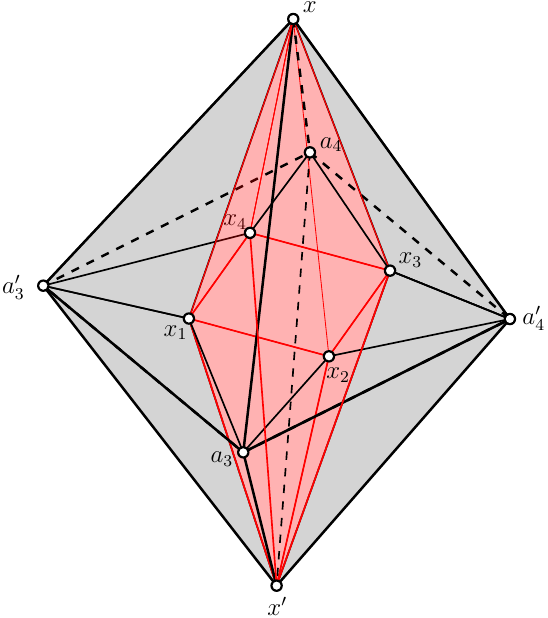}
\caption{When we cut open the 1-simplex $[xx']$ using the dummy vertices $x_i$ for $i \in [1,4]$ we have to take care to connect up the dummy vertices to real vertices so that we do not introduce extra holes. Each dummy vertex is associated to one of the four 3-simplexes that included the 1-simplex $[xx']$ in the original complex. We then connect the dummy vertex to all vertices from the associated 3-simplex. }
\label{fig:singlet_three}
\end{center}
\end{figure}

\begin{figure}[H]
\begin{center}
\includegraphics[scale=0.4]{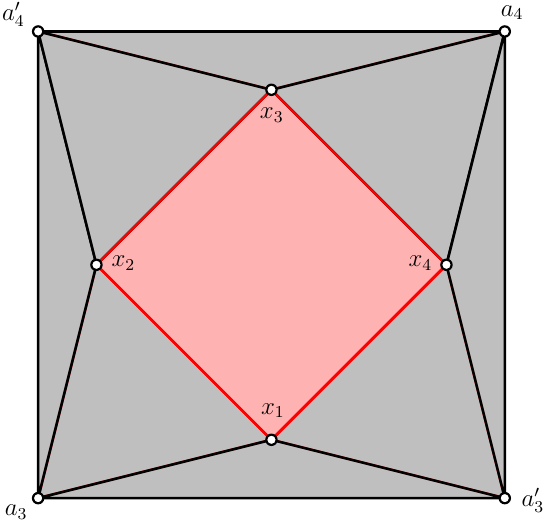}
\caption{A cross section of \Cref{fig:singlet_three} showing more clearly how the `dummy' vertices $x_i$ are connected to the original qubit vertices. Note - we use the $|00\rangle$ state as an example, but the procedure would be identical for cutting open any other basis state (simply replace $a$ ($a'$) vertices with $b$ ($b'$) vertices to change the state of the first (second) qubit to $|1\rangle$).}
\label{fig:singlet_final}
\end{center}
\end{figure}

\begin{figure}[H]
\begin{center}
\includegraphics[scale=0.4]{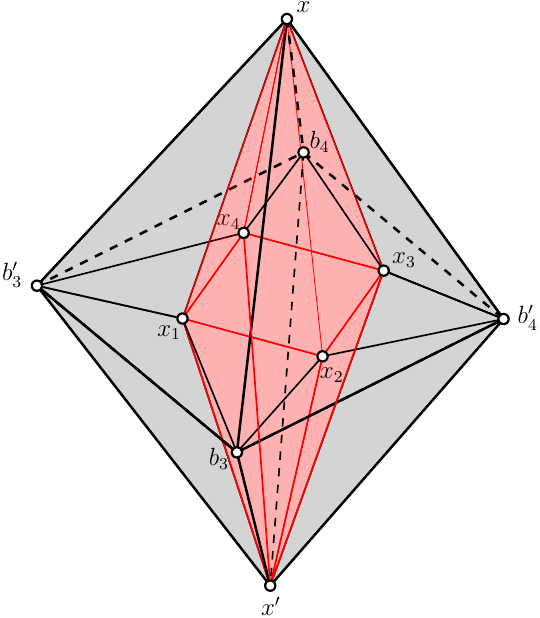}
\caption{To glue together the $|00\rangle$ and $|11\rangle$ cycles with opposite orientations we cut open the $|11\rangle$ cycle using the same method, with the same dummy vertices in the same orientation (if we wanted to construct the projector onto the state $|00\rangle + |11\rangle$ we would reverse the orientation of the dummy vertices for one of the cycles by swapping the positions of the $x_2$ and $x_4$ vertices).}
\label{fig:singlet_final really}
\end{center}
\end{figure}

\subsubsection{Superposition of three two-qubit basis states}

To construct the $\Kcyc$ for the state $|00\rangle + 2|11\rangle$ we now need to take two copies of the $|11\rangle$ cycle.
We label the $b_i$ vertices in the second cycle by $b_{i,1}$.
All the cycles are again cut open by introducing 4 dummy vertices to open the $[xx']$ edge into an open 2-cycle.
The dummy vertices in the $|00\rangle$ cycle are labelled by $x_1,x_2,x_3,x_4$.
The dummy vertices in the first $|11\rangle$ cycle are labelled by $x_2,x_3,x_5,x_6$. 
The dummy vertices in the second $|11\rangle$ cycle are labelled by $x_5,x_6,x_1,x_4$.
As usual, we then glue the cycles together by identifying the vertices with the same labels from the different cycles \Cref{fig: 00+ 2 11 main}.

To visualise what is happening when we glue the cycles together it is instructive to consider just the $\mathfrak{g}_2$ composed of the dummy vertices in each cycle.
For illustrative purposes we can consider these as open 1-cycles on the surface of a 2-sphere.
By identifying the vertices in the way outlined in the previous section we would be opening up a `viewing platform' that connects the three 2-spheres, as shown in \Cref{fig:multiples 1}.
What we are in fact doing is opening up a `viewing platform' in one dimension higher -- with the extra dimension coming from the $x,x'$ vertices that are common to all the cycles.

As usual, we complete the gadget applying the thickening and coning off procedure from \Cref{single_gadget_sec}, before applying the function $f(\cdot)$ from \Cref{eq:f} with the relation:
\begin{equation}\label{eq:R2}
R = \{(x_i,x)|i \in [1,6] \} \cup \{(b_{i,1},b_i)|i \in [2,4] \}\cup \{(b'_{i,1},b'_i)|i \in [2,4] \} \cup \{(v,v)| v\in \J^0 \}
\end{equation}

\begin{figure}[H]

\begin{minipage}{.5\linewidth}
\centering
\subfloat[]{\label{main:a}\includegraphics[scale=0.4]{figures/two_qubit_splitting_even_more_detail.pdf}}
\end{minipage}%
\begin{minipage}{.5\linewidth}
\centering
\subfloat[]
{\label{main:b}\includegraphics[scale=0.4]{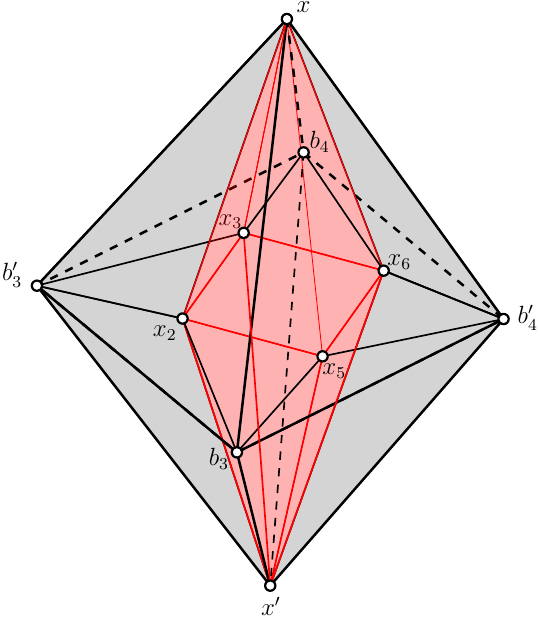}}
\end{minipage}\par\medskip
\centering
\subfloat[]{\label{main:c}\includegraphics[scale=0.4]{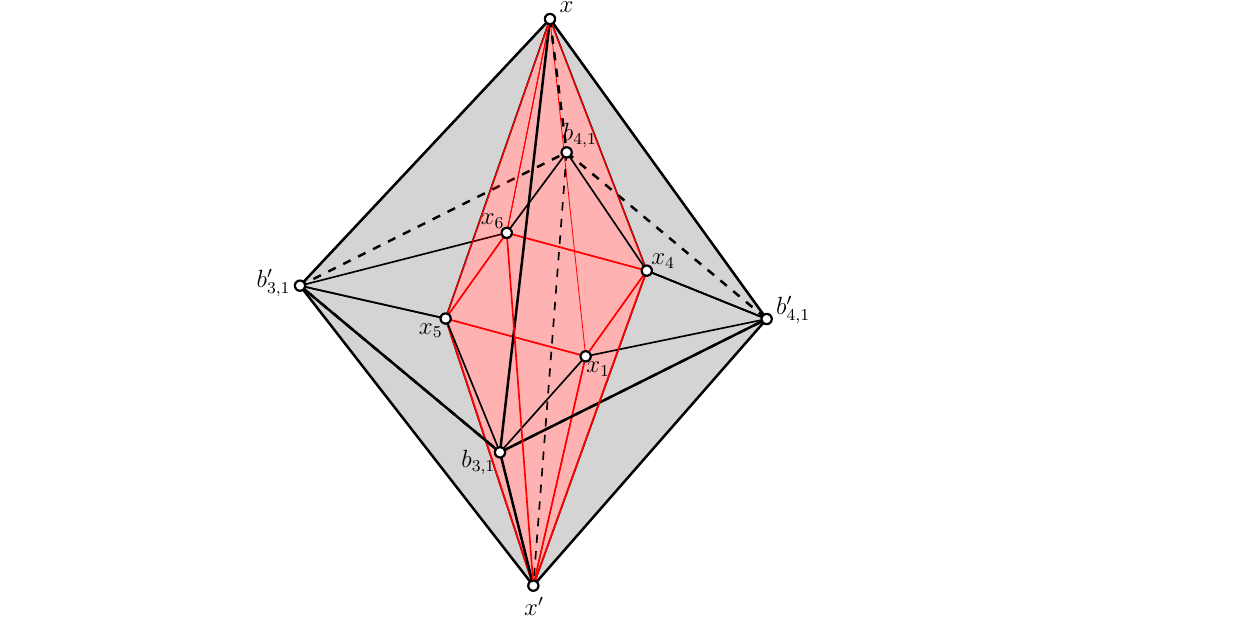}}

\caption{To construct the $\Kcyc$ for the $|00\rangle + 2|11\rangle$ cycle we now need to take two copies of the $|11\rangle$ cycle. We label the dummy vertices as outlined in the main text, and label the $b_i$ and $b_i'$ vertices of the second copy of the $|11\rangle$ cycle by $b_{i,1}$ and $b_{i,1}'$. To construct the $\J$ cycle from the $\Kcyc$ we apply the function $f(\cdot)$ with the relation in \Cref{eq:R2}.}
\label{fig: 00+ 2 11 main}
\end{figure}

\begin{figure}[H]
\begin{center}
\includegraphics[scale=0.6]{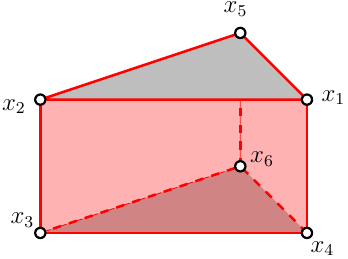}
\caption{To visualise what is happening when we glue together the three copies of $S^3$ it is useful to consider things in one dimension lower. We can do that by considering gluing together three copies of $S^2$ by opening up $\mathfrak{g}_2$ cycles in each of them. If we label the $\mathfrak{g}_2$ in the same way as we do the dummy vertices from \Cref{fig: 00+ 2 11 main} then we construct a wedge shaped `viewing platform' as shown here, which connects the three copies of $S^2$. In reality of course we are doing this in one dimension higher - and the extra dimension is coming from the $x,x'$ vertices which are common to all the cycles.}
\label{fig:multiples 1}
\end{center}
\end{figure}

\subsubsection{Superposition of more than three two-qubit basis states}

For the final example with two qubits we consider the state $|00\rangle + 3|11\rangle$.
The method is largely the same as in the previous section, where now we have three copies of the $|11\rangle$ cycle, and the labelling of the vertices is updated accordingly (see \Cref{fig: 00+ 3 11 main}).

However, there is a slight subtlety which appears when adding together four copies of $S^3$ (and remains for any larger number of copies).
Now when we construct the `viewing platform' to glue together the many copies, the roof and ceiling of the viewing platform are no longer triangles, they are now squares \Cref{fig:multiples 2}.
Therefore, we need to add two additional dummy vertices, and connect these dummy vertices up to the already existing dummy vertices as shown in \Cref{fig:multiples 2} to finish the gluing procedure.
We label these final dummy vertices by $x_9,x_{10}$.
As before, we apply the thickening and coning off procedure from \Cref{single_gadget_sec} before applying the function $f(\cdot)$ from \Cref{eq:f} with the relation:
\begin{equation}\label{eq:R3}
\begin{split}
R =& \{(x_i,x)|i \in [1,10] \} \cup \{(b_{i,j},b_i)|i \in [2,4], j \in [1,2] \}\\
&\cup \{(b'_{i,j},b'_i)|i \in [2,4], j \in [1,2] \} \cup \{(v,v)| v\in \J^0\}
\end{split}
\end{equation}

\begin{figure}[H]
\begin{minipage}{.5\linewidth}
\centering
\subfloat[]{\label{main2:a}\includegraphics[scale=0.4]{figures/two_qubit_splitting_even_more_detail.pdf}}
\end{minipage}%
\begin{minipage}{.5\linewidth}
\centering
\subfloat[]
{\label{main2:b}\includegraphics[scale=0.4]{figures/11_reverse.pdf}}
\end{minipage}\par\medskip
\begin{minipage}{.5\linewidth}
\centering
\subfloat[]{\label{main2:c}\includegraphics[scale=0.4]{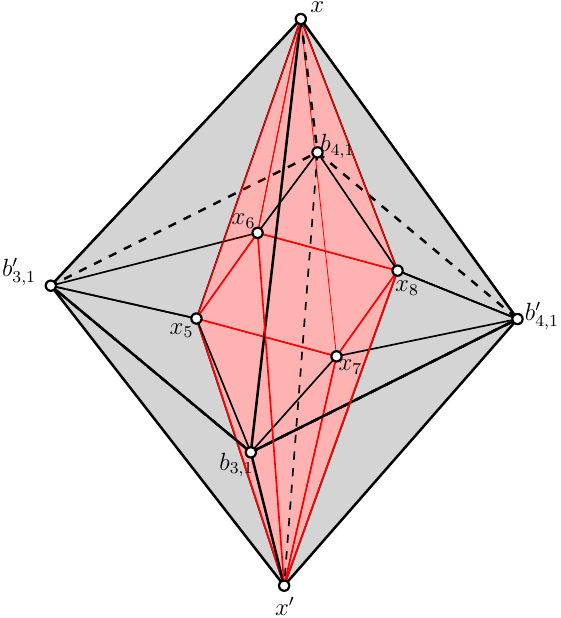}}
\end{minipage}%
\begin{minipage}{.5\linewidth}
\centering
\subfloat[]{\label{main2:d}\includegraphics[scale=0.4]{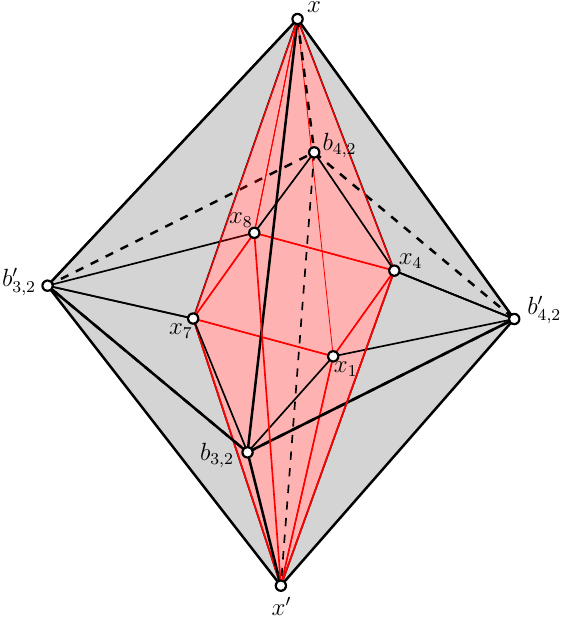}}
\end{minipage}
\caption{To construct the $\Kcyc$ for the $|00\rangle + 3|11\rangle$ cycle we now need to take three copies of the $|11\rangle$ cycle. 
The dummy vertices of the second $|11\rangle$ cycle are now labelled by $x_5,x_6,x_7,x_8$ while those of the third cycle are labelled by $x_7,x_8,x_1,x_4$.
We label the $b_i$ and $b_i'$ vertices of the third copy of the $|11\rangle$ cycle by $b_{i,2}$ and $b_{i,2}'$. To construct the $\J$ cycle from the $\Kcyc$ we apply the function $f(\cdot)$ with the relation in \Cref{eq:R3}.}
\label{fig: 00+ 3 11 main}
\end{figure}

\begin{figure}[H]
\begin{center}
\includegraphics[scale=0.6]{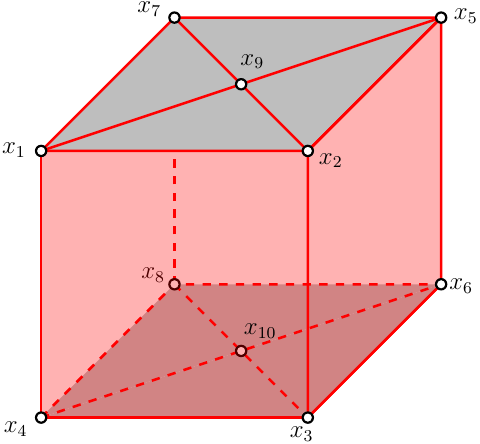}
\caption{As in the previous section, we can visualise what is occurring when we add the cycles together by considering adding together $S^2$ cycles which are being cut and glued along copies of $\mathfrak{g}_2$. When we have four or more cycles to glue together the `viewing platform' no longer has triangles for its floor and ceiling, instead it has squares. Therefore we need to introduce two final dummy vertices to close these squares and complete the gluing procedure.}
\label{fig:multiples 2}
\end{center}
\end{figure}

\subsubsection{Superposition of two three-qubit basis states}\label{sec:3 qubit cutting}

The three-qubit basis states are now triangulations of $S^5$.
The basis states themselves are, again, easy to handle.
To fill them in we simply apply the thickening and coning off procedure to the cycles themselves.
Tackling general integer states requires us to glue together triangulations of $S^5$ by cutting open $4$-cycles and gluing along them.

All the three-qubit basis states share the 2-simplex $[xx'x'']$ as shown in \Cref{fig:3 qubit initial}.
To cut this open we first cut open the $xx'$ edge by adding a dummy vertex $x_1$ in the middle as shown in \Cref{fig:3 qubit 2}.
This gives an open 1-cycle.
Before going further, we add edges between $x_1$ and every vertex that is connected to \emph{both} $x$ and $x'$.
For the $|000\rangle$ state this is the $a_2'',a_3'',a_4''$, the $a_3,a_4$ and the $a_3',a_4'$ vertices (shown in \Cref{fig:3 qubit 3}).

We will go from an open 1-cycle to an open 4-cycle by constructing an open 2-cycle.
We then use this 2-cycle as the base of a bi-pyramid with peaks $x$ and $x'$, giving an open 3-cycle.
Finally we use this bi-pyramid as the base of a second bi-pyramid with peaks $x''$ and $x_1$, giving an open 4-cycle - see \Cref{fig:3 qubit 4} 

\begin{figure}[H]
\begin{center}
\includegraphics[scale=0.4]{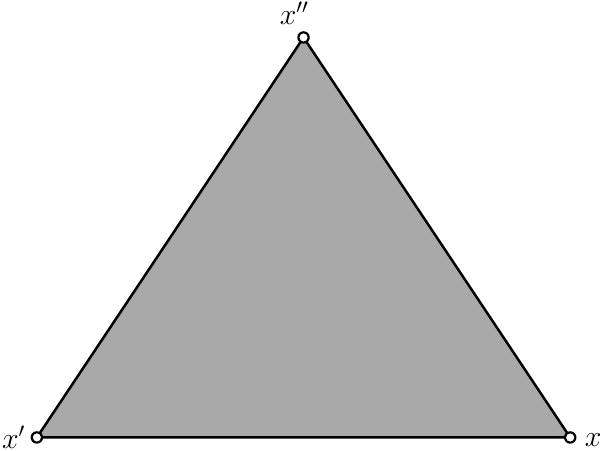}
\caption{All of the 5-cycles which encode the three-qubit basis states share this triangle.}
\label{fig:3 qubit initial}
\end{center}
\end{figure}

\begin{figure}[H]
\begin{center}
\includegraphics[scale=0.4]{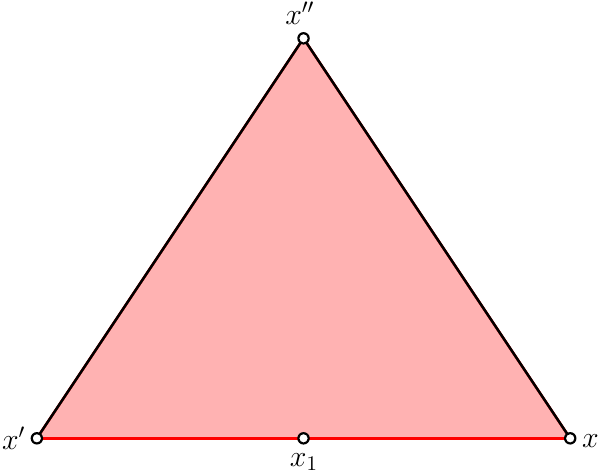}
\caption{To glue together two different three-qubit basis states we must cut open the common 2-simplex and create an open 4-cycle which we can glue the states along. The first step in this is to introduce a dummy vertex $x_1$ and use it to open the 2-simplex into a 1-cycle as shown.}
\label{fig:3 qubit 2}
\end{center}
\end{figure}

\begin{figure}[H]
\begin{center}
\includegraphics[scale=0.4]{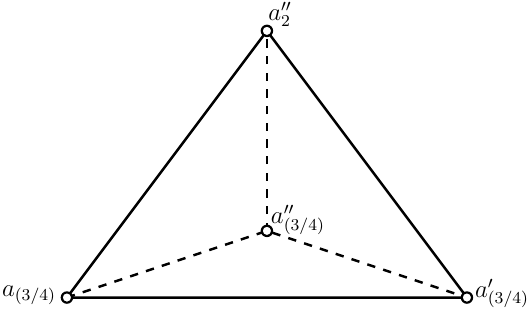}
\caption{The dummy vertex $x_1$ needs to be connected to every vertex that is connected to both $x$ and $x'$. The vertices that are connected to both form the 8 tetrahedron shown in this figure.}
\label{fig:3 qubit 3}
\end{center}
\end{figure}

\begin{figure}[H]
\begin{center}
\includegraphics[scale=0.4]{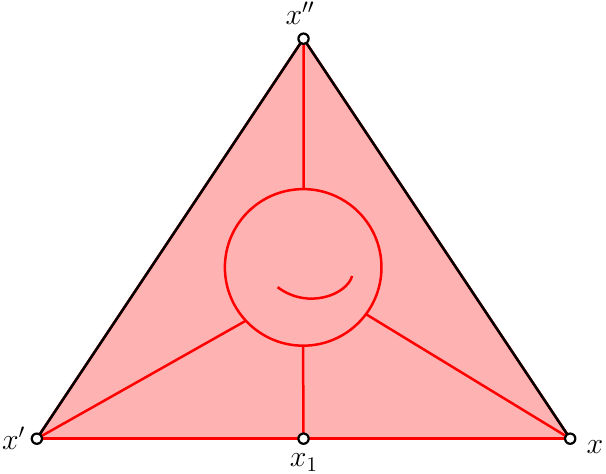}
\caption{To go from the open 1-cycle shown in \Cref{fig:3 qubit 2} to an open 4-cycle we will first construct an open 2-cycle. We will then use this 2-cycle as the base of a bi-pyramid with vertices $x$ and $x'$ as the peaks. This bi-pyramid is a 3-cycle. We will then use this 4-d bi-pyramid as the base of a second bi-pyramid that has vertices $x''$ and $x_1$ as the peaks. This gives an open 4-cycle as required.}
\label{fig:3 qubit 4}
\end{center}
\end{figure}

\begin{figure}[H]
\begin{center}
\includegraphics[scale=0.4]{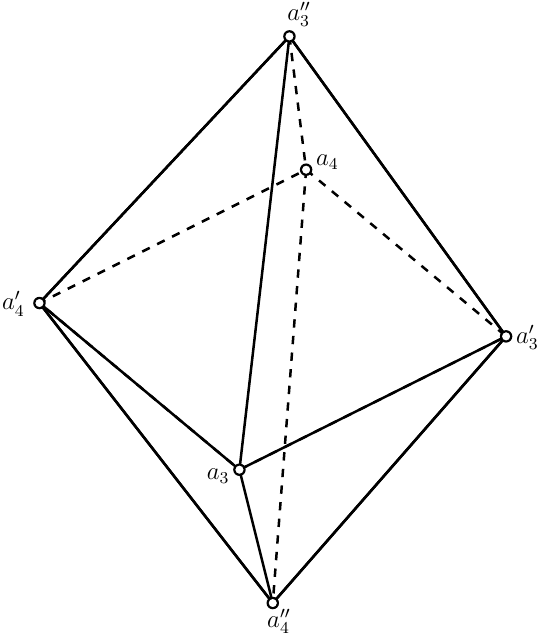}
\caption{To pick what 2-cycle should be used as the base of the first bi-pyramid we look at the vertices that are connected to all three of $x$, $x'$ and $x''$ in the original cycle. These vertices form an octahedron.}
\label{fig:3 qubit 5a}
\end{center}
\end{figure}

To pick what 2-cycle we should use as the base of the first bi-pyramid we look at the vertices that are connected to all three of $x$, $x'$ and $x''$ in the original cycle.
These will all need to be connected to the 2-cycle we use.
These vertices form an octahedron -- this octahedron is shown for the $|000\rangle$ state in \Cref{fig:3 qubit 5a}.

In analogy with the two-qubit case, we will pick a 2-cycle which is the dual of the octahedron shown in \Cref{fig:3 qubit 5a}. 
This gives a cube as shown in \Cref{fig:3 qubit 7}.
We can then connect up to the original vertices by assigning each vertex of the cube to one of the 2-simplices that triangulates the octahedron  see \Cref{fig:3 qubit 6}. 
The relationship between the cube and the octahedron can also be visualised in \Cref{fig:3 qubit 8}.

So far we have completed this process for the cycle corresponding to the $\ket{000}$ state.
If we wanted to construct the projector for the state $\ket{000} + \ket{010}$ we would then apply the same cutting procedure to the state $\ket{010}$.
Where we would label the vertices corresponding to the first qubit in the $\ket{010}$ state by $a_{i,1}$, the second qubit by $b'_i$ and the third qubit by $a'_{i,1}$.
The $x,x',x''$ and the dummy ($x_i$) vertices in the $\ket{010}$ cycle would have the same labels as those vertices in the $\ket{000}$ cycle.
We would then glue the two cycles together by identifying vertices with the same labels from the two cycles.
As usual, we complete the gadget by applying the thickening and coning off procedure, before applying the function $f(\cdot)$ from \Cref{eq:f} with the relation:
\begin{equation}\label{eq:R4}
\begin{split}
R = &  \{(x_i,x)|i \in [1,9] \} \cup \{(a_{i,1},a_i)|i \in [2,4] \} \\
&  \cup 
 \{(a''_{i,1},a''_i)|i \in [2,4] \}
 \cup \{(v,v)| v\in \J^0 \}
\end{split}
\end{equation}

\begin{figure}[H]
\begin{center}
\includegraphics[scale=0.4]{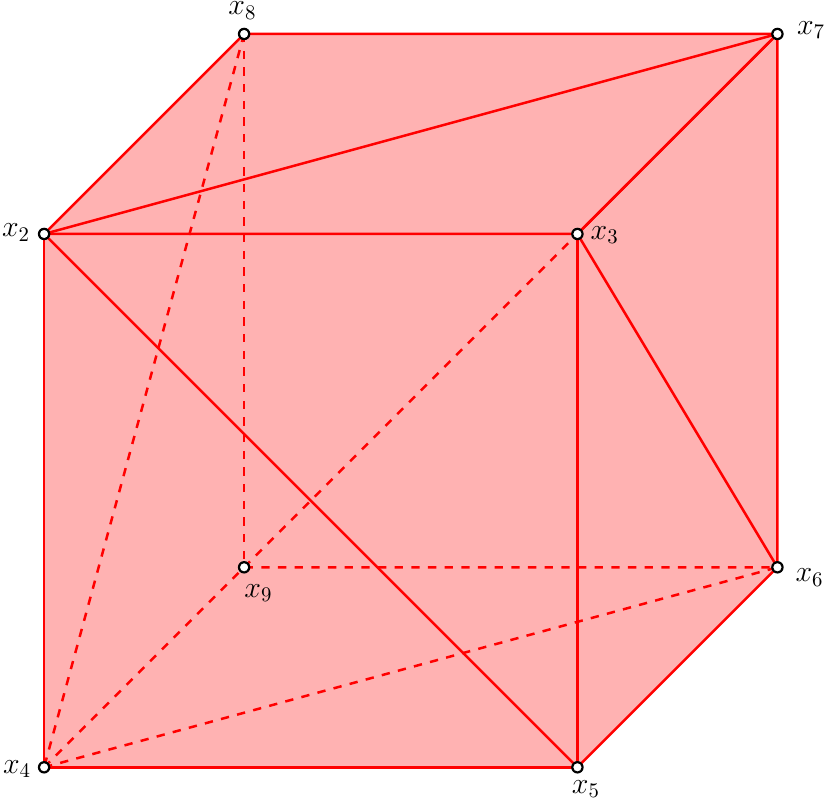}
\caption{We will introduce a 2-cycle which is the dual of the octahedron from \Cref{fig:3 qubit 5a}. This gives a cube as shown here. Using the dual polytope allows us to match up one dummy vertex with each 2-simplex that bounds the octaedron, in analogy with the two-qubit case.}
\label{fig:3 qubit 7}
\end{center}
\end{figure}

\begin{figure}[H]
\begin{minipage}{.5\linewidth}
\centering
\subfloat[]{\label{main8:a}\includegraphics[scale=0.4]{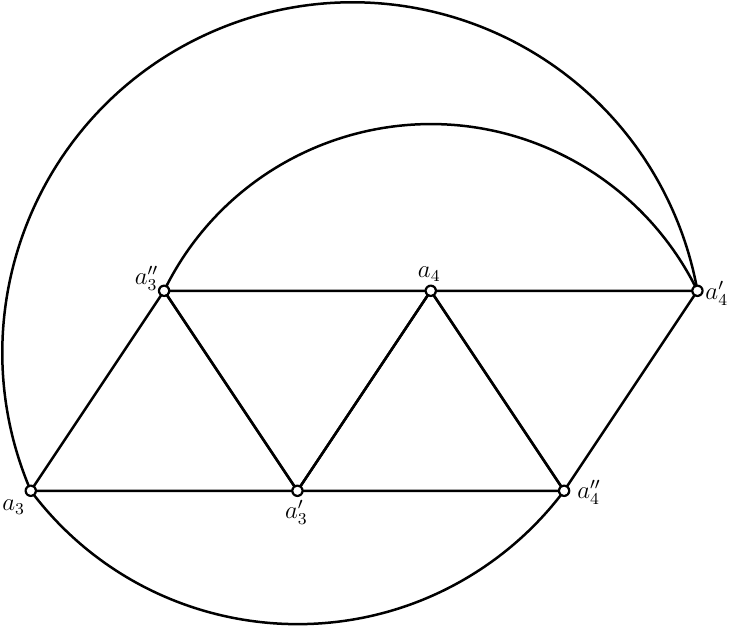}}
\end{minipage}%
\begin{minipage}{.5\linewidth}
\centering
\subfloat[]
{\label{main8:b}\includegraphics[scale=0.4]{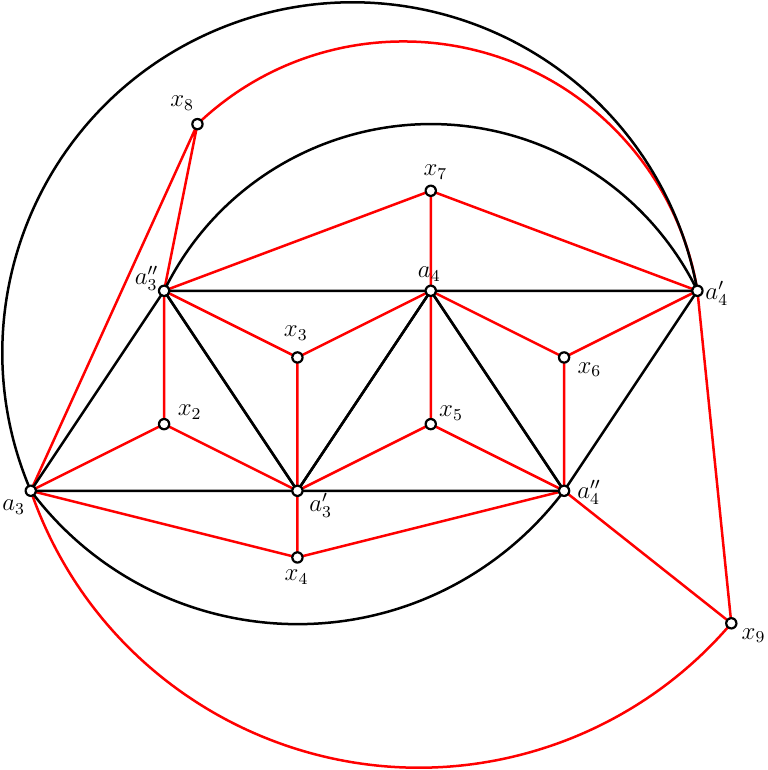}}
\end{minipage}\par
\caption{To see how the dummy vertices connect to the original vertices it is useful to look at a projection of the octahedron from \Cref{fig:3 qubit 5a}. Since we have chosen to use the dual of the octahedron as our open 2-cycle we can simply assign one vertex to each face of the octahedron.}
\label{fig:3 qubit 6}
\end{figure}

\begin{figure}[H]
\begin{center}
\includegraphics[scale=0.4]{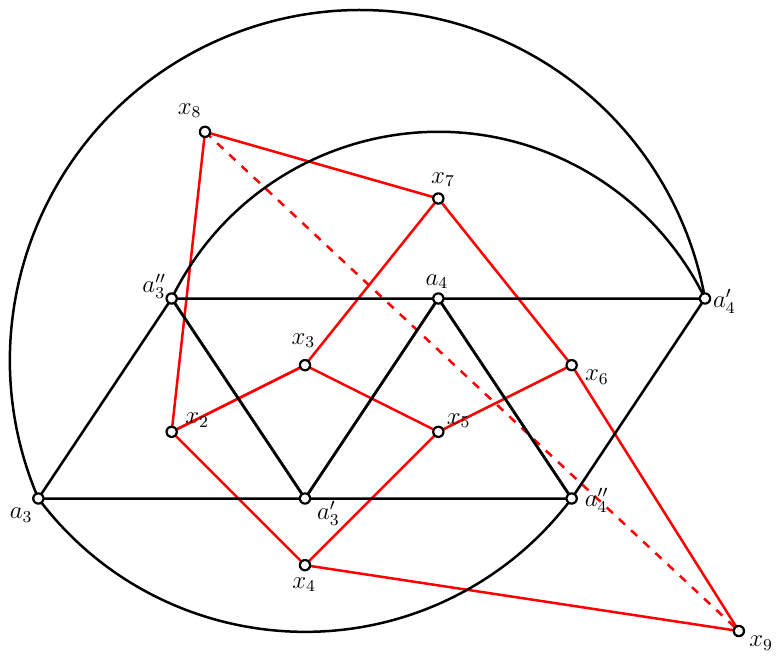}
\caption{This figure demonstrates how the cube of \Cref{fig:3 qubit 7} appears on the projection from \Cref{fig:3 qubit 6}. In the figure we have omitted the edges betwee dummy vertices and original vertices, and also omitted the diagonal lines from each face of the cube for clarity. }
\label{fig:3 qubit 8}
\end{center}
\end{figure}

\subsubsection{Superposition of three three-qubit basis states}

We will construct the projector onto three-qubit integer states with more than two cycles in analogy with the two-qubit case.
Recall that in the two-qubit case we glued together multiple cycles by cutting open the cycles, and sharing some of the vertices between different cycles.
We will do the same in the three-qubit case.
Here, the $x$, $x'$ and $x''$ vertices have to be the same for every cycle, since these are real vertices (not `dummy' vertices).
We will also keep the $x_1$ vertex common to every cycle.
So it is just the vertices that make up the cube that will only be partially shared between the different cycles.

In \Cref{fig:3 qubit 9} we demonstrate how the three cubes needed to cut open the three cycles for this integer state are formed. The first two cubes share half their vertices, and the final cube shares half its vertices with the first cube, and half its vertices with the second cube. 
As in the two-qubit case, when we glue these cycles together we have to take care to close any additional holes that are opened up.
In the three-qubit case there are four extra holes that are opened up.
In the case that there are three cycles to add together, we can close these extra holes with four vertices -- see \Cref{fig:3 qubit 10} for details.

To complete the construction of $\Kcyc$ we must attach the cubes to the other vertices from the computational basis states that we are adding together -- this is done in exactly the same manner as in the previous section for each cube / each computational basis state.
We then fill in $\Kcyc$ by applying the thickening and coning off procedure from \Cref{single_gadget_sec} to $\Kcyc$, before applying the function $f(\cdot)$ from \Cref{eq:f}.
The exact form of the relation will depend on which computational basis states are involved in the cycle $\J$, but as in previous examples all dummy vertices will be related to the vertices from $\J^0$ that they are copies of, while the non-dummy vertices are related to themselves.
For the states required in \Cref{table:states} the relations are given explicitly in \Cref{sec:gadget proofs}.

\begin{figure}[H]
\begin{center}
\includegraphics[scale=0.4]{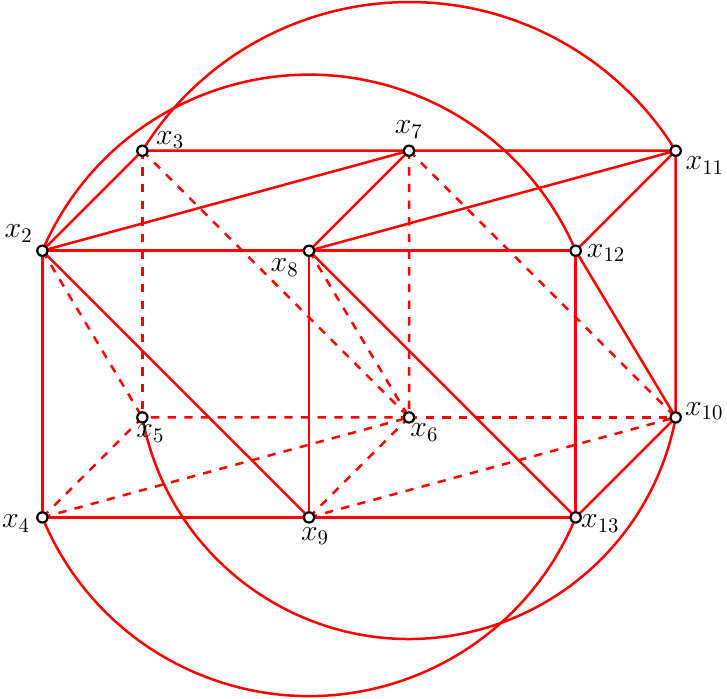}
\caption{Showing how the cubes for each cycle from \Cref{fig:3 qubit 7} are glued together. You may notice that in addition to the required open cycle bounded by the cubes, there are four additional open two cycles in this figure (above, below, and to either side of the cubes). See \Cref{fig:3 qubit 10} for an illustration of these holes, and details of how we close these them.}
\label{fig:3 qubit 9}
\end{center}
\end{figure}

\begin{figure}[H]
\begin{minipage}{.5\linewidth}
\centering
\subfloat[]{\label{main10:a}\includegraphics[scale=0.4]{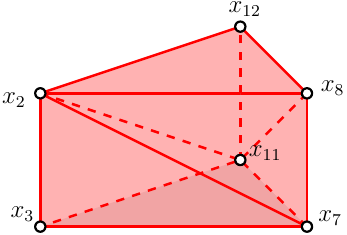}}
\end{minipage}%
\begin{minipage}{.5\linewidth}
\centering
\subfloat[]
{\label{main10:b}\includegraphics[scale=0.4]{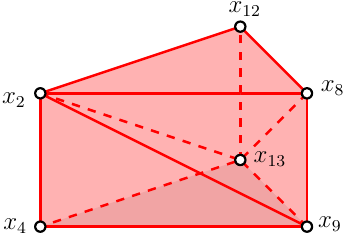}}
\end{minipage}\par\medskip
\begin{minipage}{.5\linewidth}
\centering
\subfloat[]{\label{main10:c}\includegraphics[scale=0.4]{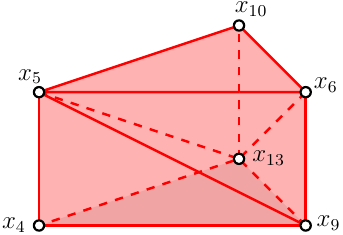}}
\end{minipage}%
\begin{minipage}{.5\linewidth}
\centering
\subfloat[]{\label{main10:d}\includegraphics[scale=0.4]{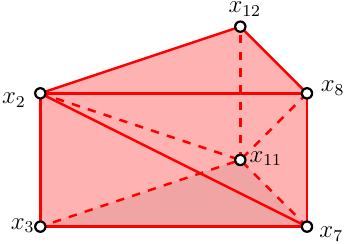}}
\end{minipage}
\caption{When we join together three of the three-qubit basis states as in \Cref{fig:3 qubit 9} we introduce four additional open 2-cycles, shown in \Cref{main10:a,main10:b,main10:c,main10:d} which must be closed. Each of the open cycles is closed by introducing an extra dummy vertex which is connected to all the vertices in that open cycle, as shown in \Cref{fig:3 qubit 10a}.}
\label{fig:3 qubit 10}
\end{figure}

\begin{figure}[H]
\begin{minipage}{.5\linewidth}
\centering
\subfloat[]{\label{main11:a}\includegraphics[scale=0.4]{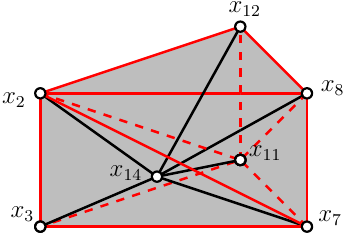}}
\end{minipage}%
\begin{minipage}{.5\linewidth}
\centering
\subfloat[]
{\label{main11:b}\includegraphics[scale=0.4]{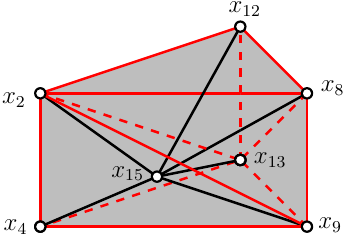}}
\end{minipage}\par\medskip
\begin{minipage}{.5\linewidth}
\centering
\subfloat[]{\label{main11:c}\includegraphics[scale=0.4]{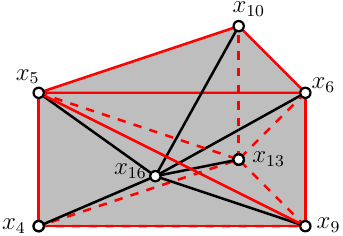}}
\end{minipage}%
\begin{minipage}{.5\linewidth}
\centering
\subfloat[]{\label{main11:d}\includegraphics[scale=0.4]{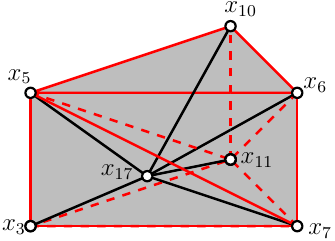}}
\end{minipage}
\caption{To fill in the additional open 2-cycles from \Cref{fig:3 qubit 9} we add one additional dummy vertex per cycle, and connect to all the vertices in the cycle.}
\label{fig:3 qubit 10a}
\end{figure}

\subsubsection{Superpositions of more than three three-qubit basis states} \label{sec:3 qubit construction}

The process for adding together more than three of the three-qubit basis states is very similar.
In \Cref{fig:3 qubit 11} we show how the four cubes that are used to cut open the four cycles are glued together.
As in the case where we only have three cycles, we glue the cubes together in a circular fashion, and each cube shares half its vertices with the cube preceding it, and half its vertices with the cube following it.
The only difference here is that we now have eight additional open 1-cycles which we need to fill in.
This occurs because the `floor' and the `ceiling' of the open 2-cycle we need to fill in are no longer 2-simplices, instead they are open 1-cycles -- see \Cref{fig:3 qubit 12}.
To fill these in we add one dummy vertex for each 1-cycle, and connect it to every vertex in the 1-cycle. 
We then add one vertex for every 2-cycle, and connect it to every vertex in the 2-cycle -- details are shown in \Cref{fig:3 qubit 13}.

Note that for the case of adding together four of the three-qubit basis states we could have closed the open 1-cycles by just adding one extra line per cycle, and no extra vertices.
However that method does not generalise to more than 4 basis states, whereas the current method generalises to arbitrary numbers of basis states.

To complete the construction of $\Kcyc$ we must attach the cubes to the other vertices from the computational basis states that we are adding together -- this is done in exactly the same manner as in \Cref{sec:3 qubit cutting} for each cube / each computational basis state.
This completes the construction of the $\Kcyc$.
We then fill in $\Kcyc$ by applying the thickening and coning off procedure from \Cref{single_gadget_sec} to $\Kcyc$, before applying the function $f(\cdot)$ from \Cref{eq:f}.
The exact form of the relation will depend on which computational basis states are involved in the cycle $\J$, but as in previous examples all dummy vertices will be related to the vertices from $\J^0$ that they are copies of, while the non-dummy vertices are related to themselves -- for the states required in \Cref{table:states} the relations are given explicitly in \Cref{sec:gadget proofs}.

\begin{figure}[H]
\begin{center}
\includegraphics[scale=0.4]{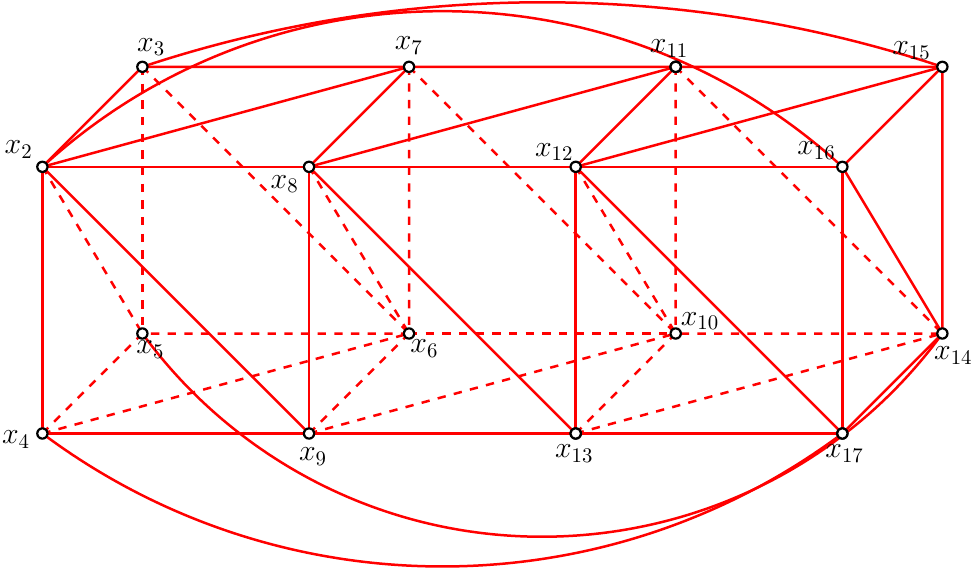}
\caption{Showing how the cubes for each cycle from \Cref{fig:3 qubit 7} are glued together when we have four cycles. In this figure there are additional one and two cycles that need to be filled in -- see \Cref{fig:3 qubit 12} for details.}
\label{fig:3 qubit 11}
\end{center}
\end{figure}

\begin{figure}[H]
\begin{minipage}{.5\linewidth}
\centering
\subfloat[]{\label{12:a}\includegraphics[scale=0.2]{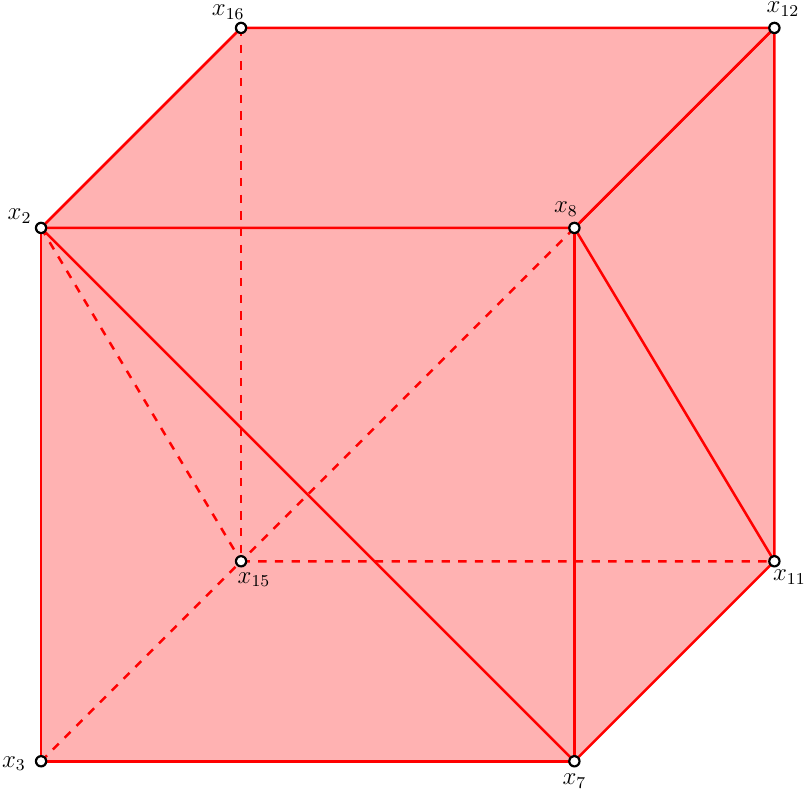}}
\end{minipage}%
\begin{minipage}{.5\linewidth}
\centering
\subfloat[]
{\label{12:b}\includegraphics[scale=0.2]{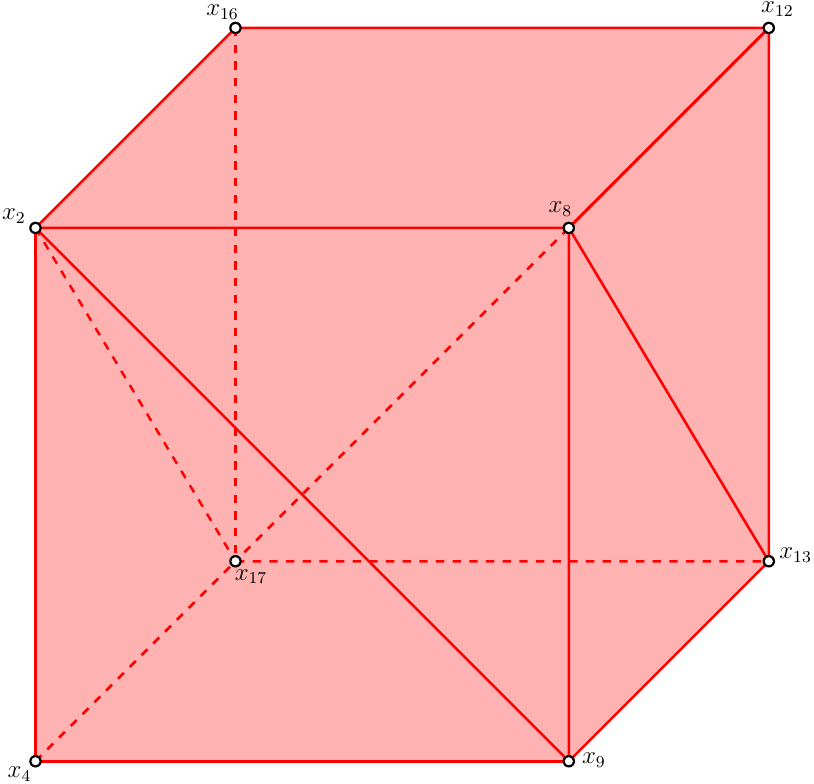}}
\end{minipage}\par\medskip
\begin{minipage}{.5\linewidth}
\centering
\subfloat[]{\label{12:c}\includegraphics[scale=0.2]{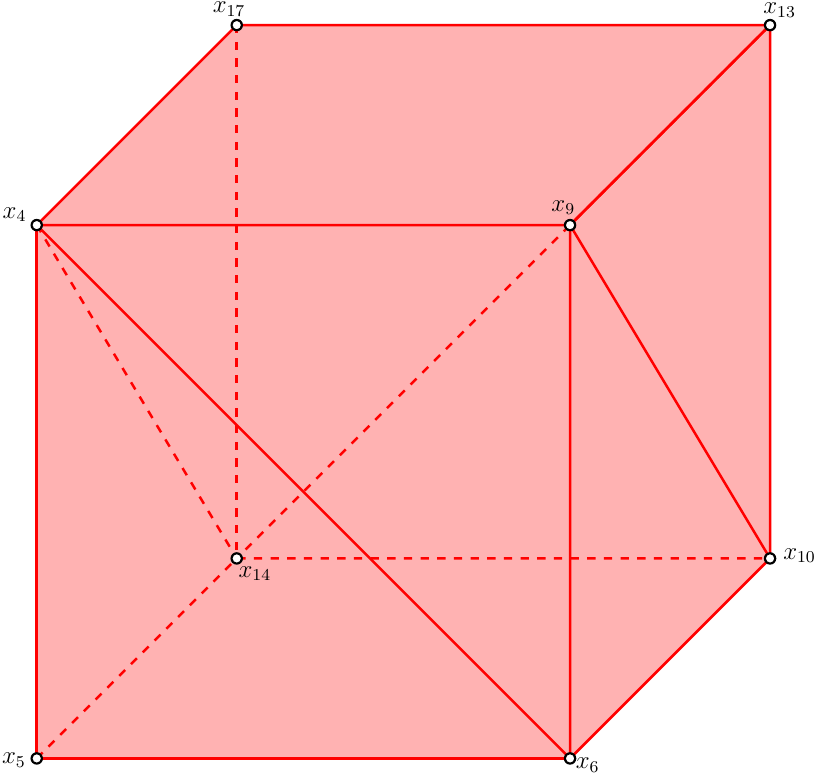}}
\end{minipage}%
\begin{minipage}{.5\linewidth}
\centering
\subfloat[]{\label{12:d}\includegraphics[scale=0.2]{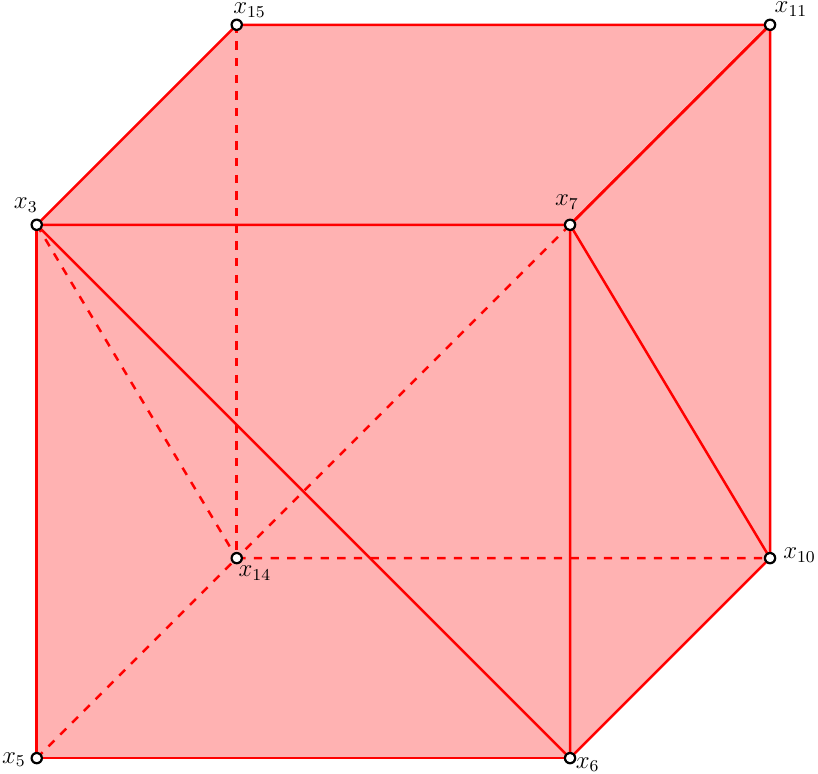}}
\end{minipage}
\caption{When joining together four (or more) three-qubit basis states, as well as the open 2-cycles that occurred in the previous case, we also find four open 1-cycles -- the `floors' and `ceilings' of the open 2-cycles shown here. For this case we use eight additional dummy vertex -- one to close each open 2-cycle, and one to close each open 1-cycle (see \Cref{fig:3 qubit 13}).}
\label{fig:3 qubit 12}
\end{figure}

\begin{figure}[H]
\begin{minipage}{.5\linewidth}
\centering
\subfloat[]{\label{13:a}\includegraphics[scale=0.2]{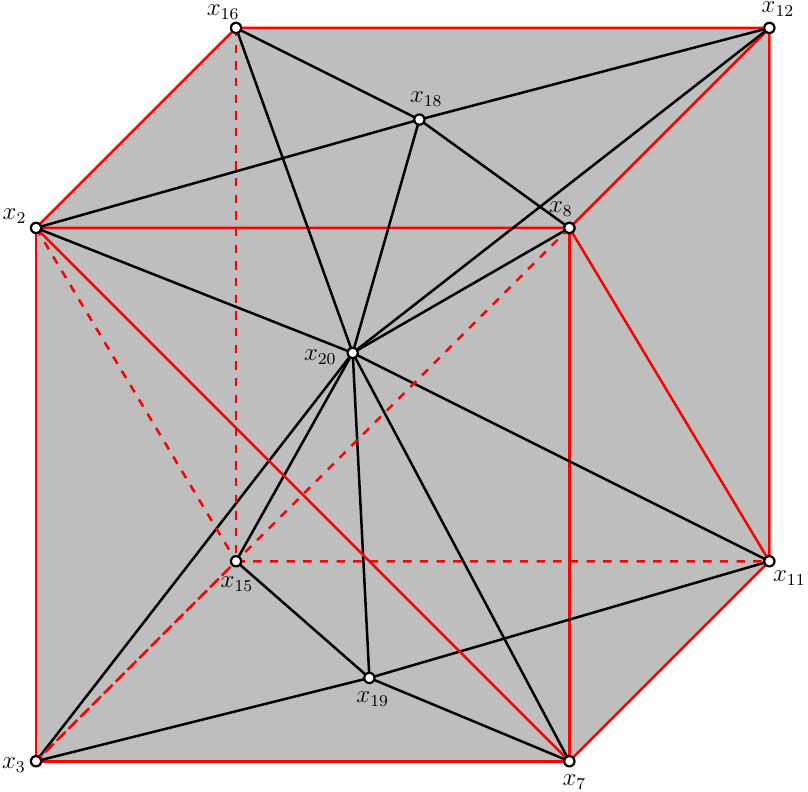}}
\end{minipage}%
\begin{minipage}{.5\linewidth}
\centering
\subfloat[]
{\label{13:b}\includegraphics[scale=0.2]{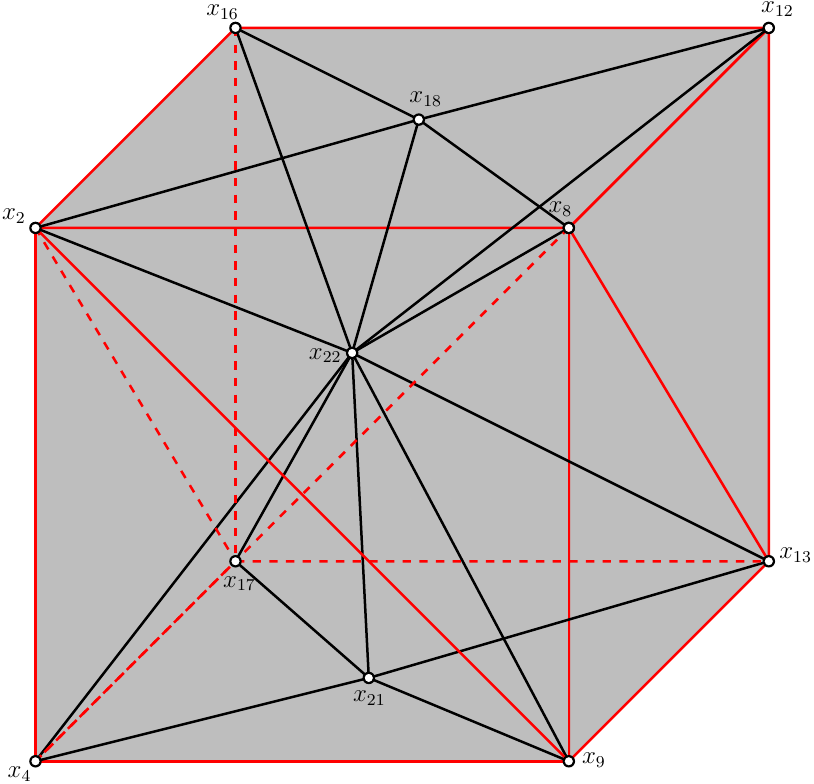}}
\end{minipage}\par\medskip
\begin{minipage}{.5\linewidth}
\centering
\subfloat[]{\label{13:c}\includegraphics[scale=0.2]{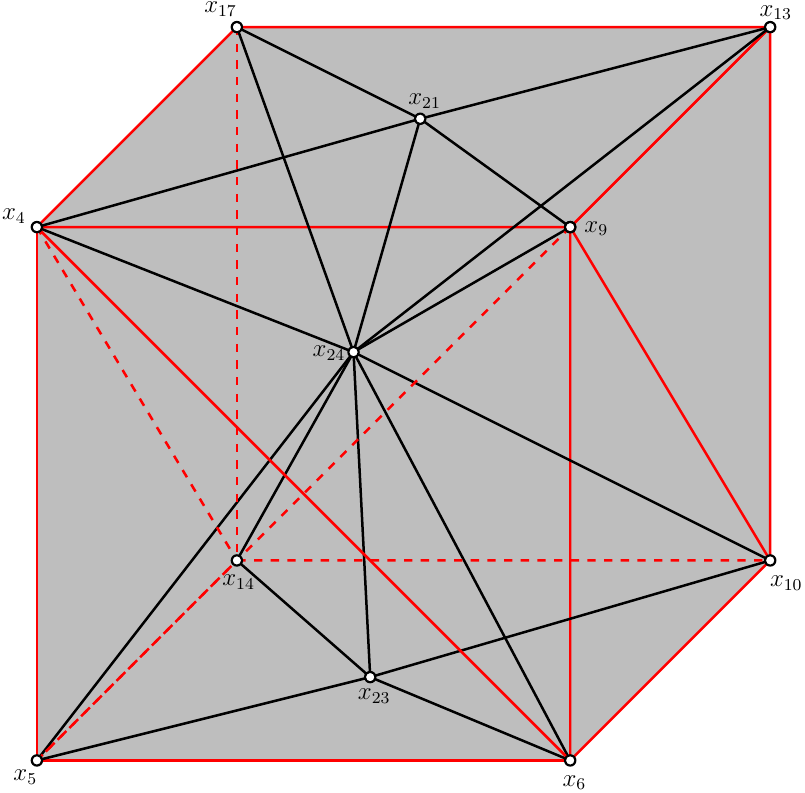}}
\end{minipage}%
\begin{minipage}{.5\linewidth}
\centering
\subfloat[]{\label{13:d}\includegraphics[scale=0.2]{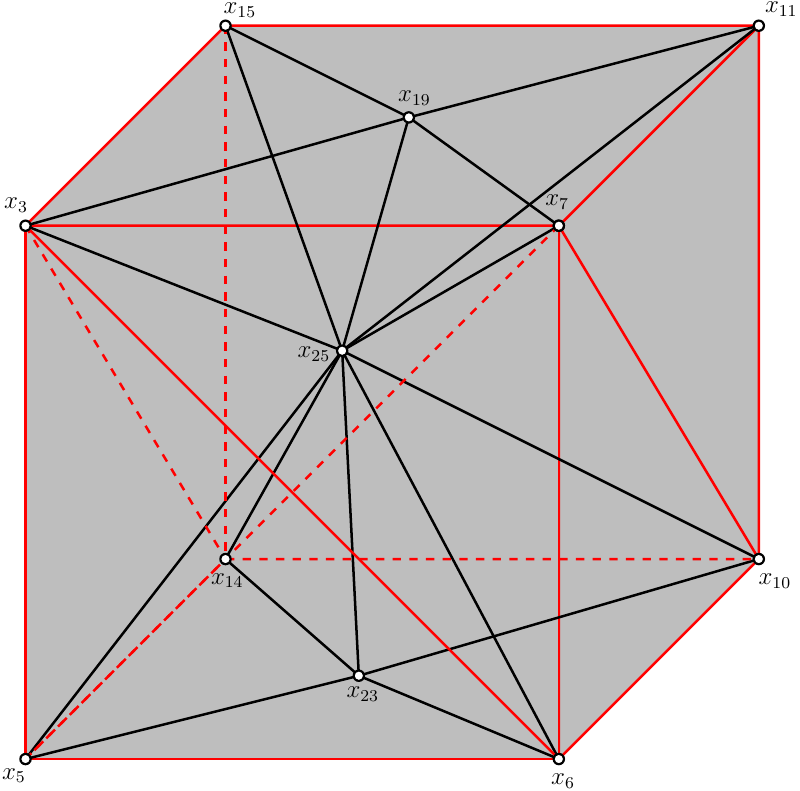}}
\end{minipage}
\caption{We close the open 1- and 2-cycles from \Cref{fig:3 qubit 12} using eight more dummy vertices. Note that this method for closing the cycles works for arbitrary many three-qubit basis states (if we had instead chosen to close the floor and ceilings by separating them into two triangles by adding one edge that would work for this example, but would not generalise to other three-qubit basis states).}
\label{fig:3 qubit 13}
\end{figure}

\subsubsection{Superpositions of two four-qubit basis states} \label{sec:4_qubit_construction}

The four-qubit basis states are triangulations of $S^7$.
The basis states themselves are, again, easy to handle.
To fill them in we simply apply the thickening and coning off procedure to the cycles themselves.
Tackling general integer states requires us to glue together triangulations of $S^7$ by cutting open $6$-cycles and gluing along them.
All of the four-qubit basis states required in \Cref{table:states} are either basis states, or superpositions of two computational basis states.
Therefore for the four-qubit case we will not consider states which are combinations of more than two computational basis states.

\begin{figure}[H]
\begin{center}
\includegraphics[scale=0.6]{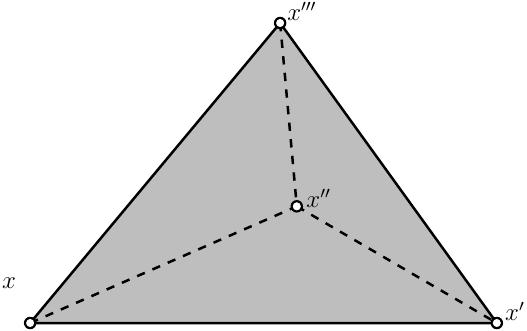}
\caption{Every four-qubit computational basis state shares the 4-simplex $[xx'x''x''']$ shown here.}
\label{fig:4_qubit_initial}
\end{center}
\end{figure}

All the four-qubit basis states share the 3-simplex $[xx'x''x''']$ as shown in \Cref{fig:4_qubit_initial}.
To cut this open we first form an open 3-cycle by:
\begin{enumerate}
\item cutting open the edge $[xx']$ by adding a vertex $x_1$ which is connected to $x$ and $x'$
\item adding an edge between $x_1$ and $x'''$
\item adding a vertex $x_2$ and connecting it to the vertices $x$, $x'$, $x''$ and $x_1$
\end{enumerate}
This gives the octahedron shown in \Cref{fig:4_qubit_2}.

Next we have to connect up the dummy vertices $x_1$ and $x_2$ to the original vertices outside of the simplex $[xx'x''x''']$ to ensure that we have not added any other extra holes.
To see how to do this we will consider each step in the construction of the octahedron from \Cref{fig:4_qubit_2}.
In the first step we have simply cut open the edge $[xx']$.
So to avoid creating any additional holes in the complex we need to connect $x_1$ up to every vertex which is connected to both $x$ and $x'$.
In the case of the $\ket{0000}$ cycle this is the vertices $\{a_3,a_4,a'_3,a'_4,a_2'',a''_3,a''_4,a_2''',a_3''',a_4'''\}$.\footnote{For other computational basis states if the state of a qubit changes from $\ket{0}$ to $\ket{1}$  we use the $b$ vertices from that qubit rather than the $a$ vertices. So, for example the $\ket{0101}$ cycle it would be the vertices $\{a_3,a_4,b'_3b'_4,a_2'',a''_3,a''_4,b_2''',b_3''',b_4'''\}$.}

The next step in the construction is to add an edge between $x_1$ and $x'''$.
We can see this step as using the vertex $x_1$ to cut open the 2-simplex $[xx'x''']$.
So we need to add edges between $x_1$ and every vertex that is connected to all three of $x,x',x'''$.
But this is a subset of the vertices which we have already connected $x_1$ to, so we do not have to do anything for this step.

Finally, we have added the $x_2$ vertex and connected it to $x,x',x''$.
We can see this step as cutting open the 2-simplex $[xx'x'']$.
So we have to add edges between $x_2$ and every vertex that is connected to all three of $x,x',x''$.
For the $\ket{0000}$ cycle this is the vertices $\{a_3,a_4,a'_3,a'_4,a''_3,a''_4,a_2''',a_3''',a_4'''\}$.\footnote{As before, for computational basis states if qubit $i$ is in the state $\ket{1}$ we use the $b$ vertices rather than the $a$ vertices for that qubit.}

It should be noted that which of $x_1$ and $x_2$ we consider adding first, and therefore which is connected to more vertices is arbitrary.
The vertices $x''$ and $x'''$ are equivalent in the original 4-simplex, but the action of breaking open the 4-simplex in this way breaks the symmetry.

\begin{figure}[H]
\begin{center}
\includegraphics[scale=0.6]{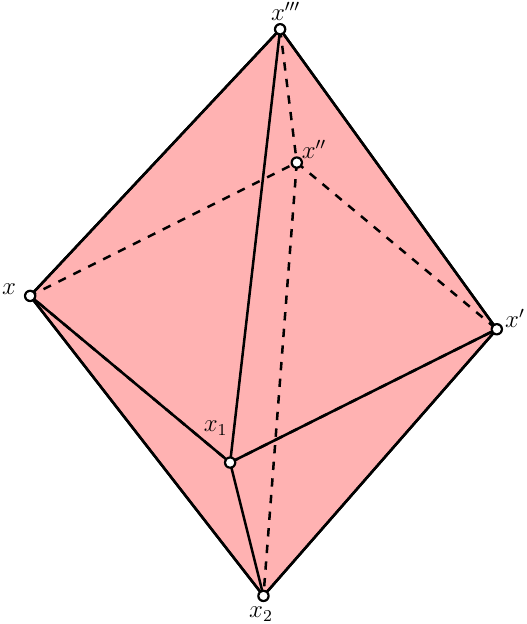}
\caption{The first step in combining two four-qubit basis states is to cut open the 4-simplex they share to form an octahedron, as shown here.}
\label{fig:4_qubit_2}
\end{center}
\end{figure}

We have now cut the closed 4-simplex shared by the two four-qubit basis states into an open 3-cycle.
But this is not a high enough dimensional hole to glue two copies of $S^7$ together through.
We need to extend this to an open 6-cycle.
We will do this by constructing a second open 3-cycle. 
We will then use this 3-cycle as the base of a bi-pyramid which has vertices $x,x'$ as its base.
This gives an open 4-cycle.
We then use this open 4-cycle as the base of a second bi-pyramid which has vertices $x'',x_1$ as its base.
This gives an open 5-cycle.
Finally we use this open 5-cycle as the base of a third bi-pyramid which has vertices $x''',x_2$ as its base.
This gives us our open 6-cycle as required. 

To pick which open 3-cycle we should construct to use as the base of the first bi-pyramid we lok at the vertices that are connected to all four of $x$, $x'$, $x''$, $x'''$ in the original 7-cycle.
These will all need to be connected to the 3-cycle we use in order to avoid creating extra holes in the complex.
These vertices form a 4-d cross-polytope, as shown in \Cref{fig:4_qubit_3}.
For our open 3-cycle we will construct another 4-d cross polytope, using eight dummy vertices $\{x_3,\cdots,x_{10}\}$, as shown in \Cref{fig:4_qubit_4}.
We will connect these eight dummy vertices to the original vertices from \Cref{fig:4_qubit_3} using the thickening procedure from \Cref{sec:thickening}.

Therefore, by introducing a total of 10 dummy vertices we can cut an open 6-cycle in the 7-cycles which form the four-qubit basis states.
Then, to glue two four-qubit basis states together we simply cut holes in both of them, and identify the dummy vertices with the same labels from the different cycles with each other.
This completes the construction of the $\Kcyc$.
To fill in $\Kcyc$ we apply the thickening and coning off procedure from \Cref{single_gadget_sec}.
Then we apply the function $f(\cdot)$ from \Cref{eq:f}.
The form of the relation depends on the exact projector we are constructing -- for the states required in \Cref{table:states} the relations are given explicitly in \Cref{sec:gadget proofs}.
 
\begin{figure}[H]
\begin{center}
\includegraphics[scale=0.6]{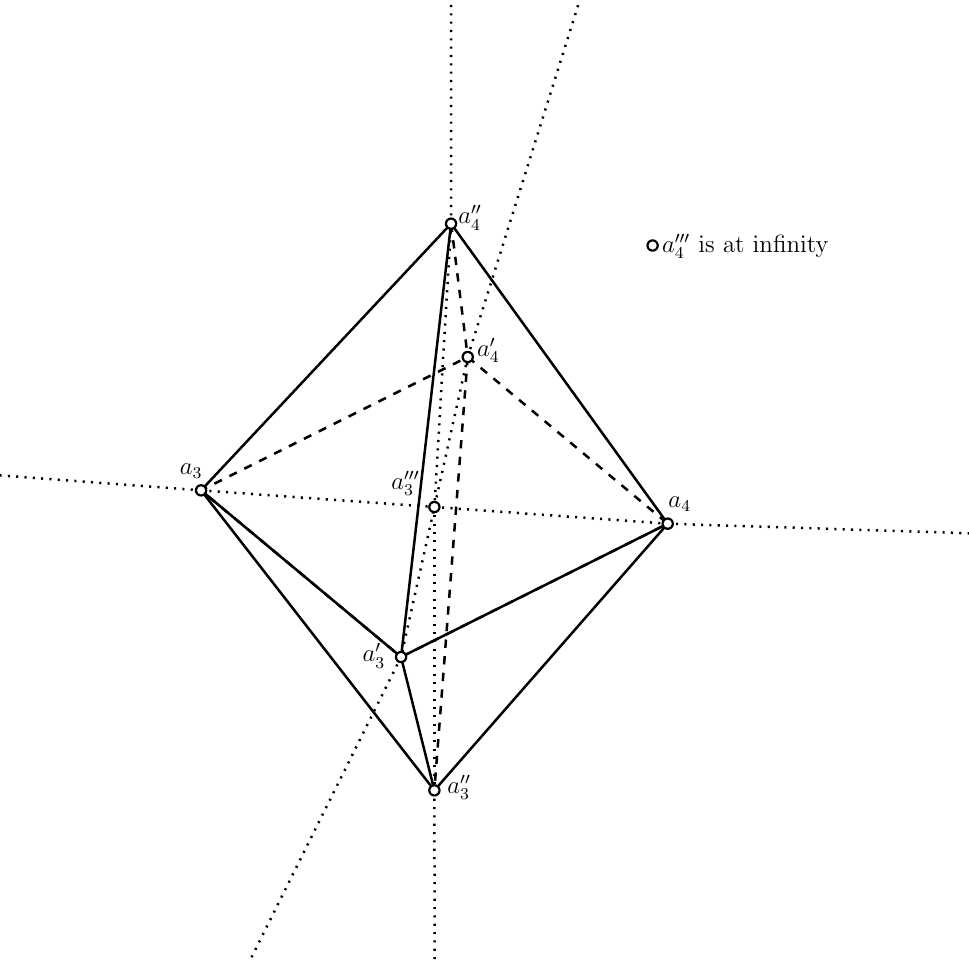}
\caption{The vertices that are connected to every vertex in the 4-simplex $[xx'x''x''']$ form a 4-d cross polytope, as shown here.}
\label{fig:4_qubit_3}
\end{center}
\end{figure}

\begin{figure}[H]
\begin{center}
\includegraphics[scale=0.6]{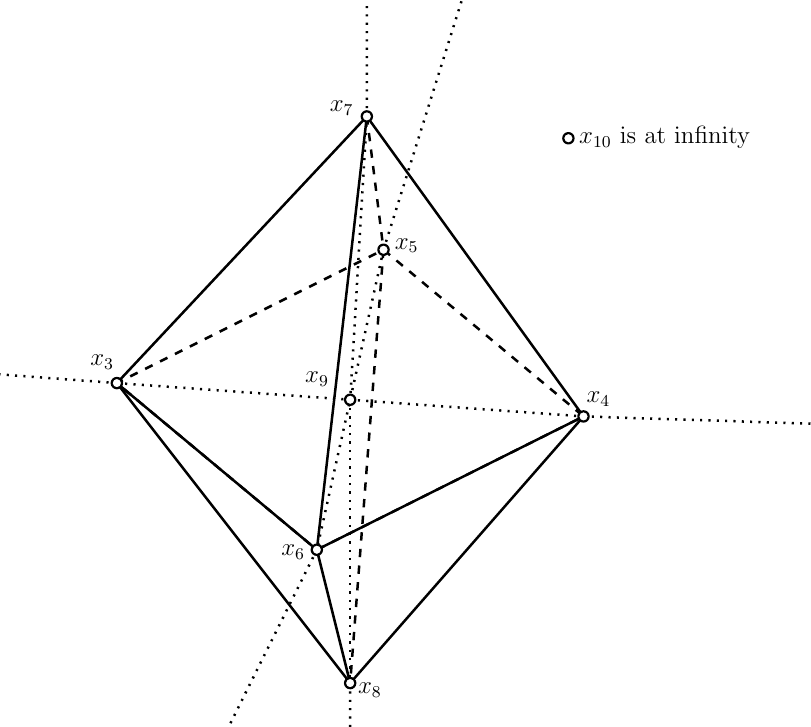}
\caption{We will introduce eight new dummy vertices to construct a 4-d cross-polytope, which will form the base for the first bi-pyramid needed to construct an open 6-cycle to glue the 7-cycles corresponding to two four-qubit computational basis states together.}
\label{fig:4_qubit_4}
\end{center}
\end{figure}

\subsection{Proof that the gadgets lift the desired states} \label{sec:gadget proofs}

From \Cref{table:states} are three types of gadget we have to consider:
\begin{enumerate}
    \item Gadgets for projectors onto computational basis states
    \item Gadgets for projectors onto the three-qubit entangled states arising from the Pythagorean gate
    \item Gadgets for projectors onto four-qubit entangled states
\end{enumerate}

The simplicial complexes involved are too big for us to compute the homology directly, instead we will prove that the gadgets have the desired behaviour. 
We will tackle each type of gadget in turn, but first we prove some general lemmas.
The first three geeral lemmas demonstrate that the thickening, coning off and applying $f(\cdot)$ procedure from \Cref{single_gadget_sec} has the required behaviour -- that is, it closes the hole in a $(2m-1)$-sphere, it does not introduce spurious homology classes, and it results in a clique complex (as opposed to a more general simplicial complex).
The final general lemma puts everything together to show that the method of constructing a triangulation $\Kcyc$ of $S^{2m-1}$ and filling it in before identifying vertices to give the desired cycle $\J$, does in fact remove the cycle $\J$ from homology.

Throughout this section we will use the following notation:
\begin{itemize}
\item $\J$ refers to the cycle we want to fill in with the gadget (i.e. if we are constructing a gadget for the projector $\ketbra{\phi}{\phi}$ then $\J$ is the cycle that corresponds to the state $\ket{\phi}$.
\item $\Kcyc$ refers to a simplicial complex which triangulates an $(2m-1)$-sphere.
\item $\hat \Kcyc$ refers to a simplicial complex obtained from $\Kcyc$ by applying the thickening and coning off procedure from \Cref{single_gadget_sec}. 
\item $\tilde \Kcyc$ refers to a simplicial complex obtained from $\hat\Kcyc$ by applying the function $f(\cdot)$ from \Cref{eq:f} with respect to some relation $R$ and constructing the complex as set out in \cref{eq:f complex 2}
\end{itemize}

\begin{lemma}\label{lem:thickening and coning}
$\hat \Kcyc$ is a triangulation of the closed $2m$-ball.
\end{lemma}
\begin{proof}
The thickening procedure from \Cref{sec:thickening} is designed to triangulate $\Kcyc \times I$ where $I = [0,1]$, the proof that the procedure achieves this is provided in \Cref{sec:thickening}.
The final step, of adding the central vertex and connecting it to every vertex in $\mathcal{K}\times \{1\}$ clearly closes the hole in $\Kcyc$ since it provides a state for which $\Kcyc \times 1$ is the boundary, and by the results of \Cref{sec:thickening}, $\Kcyc \times \{0\}$ is homologous to $\Kcyc \times \{1\}$. 
It is straightforward to see that it does so without adding any spurious homology classes.
More formally, since the cycle $\Kcyc$ can be represented as a map of a sphere into the complex, the process of adding the central vertex can be seen as constructing the mapping cone.
\end{proof}

\begin{lemma}\label{lem:identifying}
$\tilde \Kcyc$ has trivial homology.
\end{lemma}
\begin{proof}
Without loss of generality we demonstrate the result for the relation:
\begin{equation}\label{eq:R5}
R = \{(x_1,x)\} \cup \{(w,w)| w \in \hat\Kcyc^0, w\neq x_1\}
\end{equation}
where $x,x_1 \in \hat\Kcyc^0$ and $x \neq x_1$.
The result follows for more general relations since we can always choose to apply the function $f(\cdot)$ in stages, where at each stage we only identify two vertices.

First we show that every cycle $c \in \hat\Kcyc$ maps to a cycle $\tilde{c} \in \tilde{\Kcyc}$.
Let $c = \sum_{i=1}^{2m-1} \sigma_i$ where each $\sigma_i = [w_0 \dots w_k]$.
Then we have:
\begin{equation}
\partial \ket{c} = \sum_i \partial \ket{\sigma_i} = \sum_i\sum_j \ket{\sigma_i \setminus w_j}= 0
\end{equation}

Consider $\tilde{\sigma}_i = f(\sigma_i)$ for each $i$. 
If $\sigma_i$ only contains one of $x$ or $x_1$ then $\tilde{\sigma}_i$ is still a $k$-simplex.
However, any $\sigma_i$ which contains both $x$ and $x_1$ will map to a $k-1$-simplex. 
Consider one such simplex $\sigma_i = [xx_1w_2 \dots w_k]$.
We have that $\ket{\sigma_i \setminus x} \in \partial \sigma_i$ and $\ket{\sigma_i \setminus x_1} \in \partial \sigma_i$.
However, since $c$ is a cycle and satisfies $\partial c = 0$ there must exist other $\sigma_j, \sigma_k \in c$ where:
\begin{equation}
\sigma_j = [x w_2 \dots \tilde{w}_{k+1} \dots w_k] \textrm{\ \ \ and \ \ \ }\sigma_k = [x_1 w_2 \dots \tilde{w}_{k+2} \dots w_k]
\end{equation}
Under $f(\cdot)$ $\sigma_j$ and $\sigma_k$ map to two $k$-simplices which intersect on the $k-1$-simplex $\tilde{\sigma}_i= [xw_2 \dots w_k]$.
Crucially, the boundary elements of $\sigma_j$ and $\sigma_k$ which cancelled out with boundary elements of $\sigma_i$ in the sum $\partial c$ now map to boundary elements of $\tilde{\sigma}_j$ and $\tilde{\sigma}_k$ that cancel out with each other in the sum $\partial \tilde{c}$.
Therefore we have that:
\begin{equation}
\tilde{c} = f(c) = \sum_{i|\{x,x_1\} \not \subset \sigma_i } \tilde{\sigma}_i
\end{equation}
and:
\begin{equation}
\partial \tilde{c} = \sum_{i|\{x,x_1\} \not \subset \sigma_i} \partial \ket{\tilde{\sigma}_i} = \sum_{i|\{x,x_1\} \not\subset \sigma_i}\sum_j \ket{\tilde{\sigma}_i \setminus w_j} = 0
\end{equation}

Now consider a state $\ket{v} \in \hat\Kcyc$ such that $c = \partial\ket{v}$ (we know such a state exists because the $2m$-ball has trivial homology.)
We have:
\begin{equation}
\ket{v} = \sum_i s_i  \textrm{\ \ \ and \ \ \ } \partial\ket{v} = \sum_i \sum_j \ket{s_i \setminus w_j} = \ket{c}
\end{equation}
By the same logic outlined for the cycles, the state $\ket{v}$ maps under $f(\cdot)$ to a state $\tilde{v}$ satisfying:
\begin{equation}
\ket{\tilde{v}} = \sum_{i|\{x,x_1\}} \tilde{s}_i  \textrm{\ \ \ and \ \ \ } \partial\ket{\tilde{v}} = \sum_{i|\{x,x_1\}} \sum_j \ket{\tilde{s}_i \setminus w_j} = \tilde{\hat{c}}
\end{equation}

Therefore every cycle in $\hat\Kcyc$ maps to a cycle in $\tilde{\Kcyc}$ which is homologous to the trivial cycle.
However, we are not quite done yet.
We also need to consider cycles $\tilde{c} \in \tilde{\Kcyc}$ which are given by:
\begin{equation}
\tilde{c} = f(v)
\end{equation}
for some state $v \in \hat\Kcyc$ which is not a cycle.
This case is straightforward -- we can simply take a cycle $c \in \hat\Kcyc$ which contains $v$ (this is always possible to do, since every cycle in $\hat\Kcyc$ is homologous to the trivial cycle we can simply take a point on $v$ and continuously deform it until it is a cycle which contains all of $v$).
Where, when we say the cycle $c$ contains $v$, more formally we mean there exists a state $v'$ such that $c = \partial v'$ and $v \in v'$
We can then map this cycle to give $\tilde{c}' = f(c)$.
By the previous argument, $\tilde{c}'$ is a cycle which is homologous to the trivial cycle. 
But it is straightforward to see that $\tilde{c}$ is contained within $\tilde{c}'$ because $\tilde{c} = f(v) \in f(v')$, and by the earlier arguments we have $\partial{\tilde{c}'} = f(v')$.
Therefore $\tilde{c}$ must also be homologous to the trivial cycle.
\end{proof}

\begin{lemma}\label{lem:clique}
$\tilde \Kcyc$ is a clique complex.
\end{lemma}
\begin{proof}
Applying the thickening procedure from \Cref{sec:thickening} to a clique complex $\Kcyc$ manifestly results in a simplicial complex which is also a clique complex, since the corresponding graph is explicitly written down in \Cref{thickening_lem_2}.

By construction, applying the coning off procedure from \Cref{single_gadget_sec} to the clique complex after the thickening procedure has been applied also results in a clique complex.
This is clear because the coning off procedure requires adding edges between every vertex $w \in \Kcyc \times \{1\}$ and the central vertex $w_c$, \emph{and} adding every simplex that contains one of these edges. 
Therefore, the resulting complex will be `2-determined', and is therefore a clique complex (see \Cref{sec:clique}).

Applying $f(\cdot)$ to the complex does not remove any simplices (some simplices may be rendered equivalent to each other, and some may reduce in dimension, but no simplex is entirely removed).
Therefore applying $f(\cdot)$ does not affect the 2-determined property, so the final complex remains a clique complex.
\end{proof}

\begin{lemma}\label{lem:full gadget}
Let $\J \in \Cl(G)$ be a $(2m-1)$-hole in a clique complex which is not a triangulation of $S^{2m-1}$, and let $\Kcyc$ be a clique complex which is a triangulation of $S^{2m-1}$.
Assume there exists a relation $R$ such that applying the function $f(\cdot)$ from \Cref{eq:f} to the vertices $\Kcyc^0$ results in the cycle $\J$.
Then constructing $\tilde{\Kcyc}$ and gluing it to $\J \in \Cl(G)$ results in a new clique complex $\Cl(G')$ where $\J$ is a cycle, but not a hole. That is, $\J$ has been removed from the homology.
\end{lemma}

\begin{proof}
As demonstrated in \Cref{lem:thickening and coning}, taking $\Kcyc$ and applying the thickening and coning off procedure results in a clique complex $\hat{\Kcyc}$ which is a triangulation of the $2m$-ball.
Let $c \in \Kcyc$ be the cycle that is the linear combination of every $2m-1$-simplex in $\Kcyc$.
Then, by construction, $c \in \hat\Kcyc$, and moreover $c = \partial v$ where $v \in \hat\Kcyc$ is the linear combination of every $2m$-simplex in $\hat\Kcyc$.

Let $\tilde{\Kcyc}$ be the complex that results from applying the function $f(\cdot)$ with the relation $R$.
Then by \Cref{lem:identifying} $\tilde\Kcyc$ has trivial homology.
Moreover, by assumption the cycle $c$ has been mapped to the cycle $\J$.
Therefore the simplicial complex $\tilde\Kcyc$ contains the cycle $\J$, but $\J$ is trivial in the homology since $\J = \partial f(v)$.
Moreover, by \Cref{lem:clique}, $\tilde\Kcyc$ is the clique complex of some graph.

When we glue $\tilde\Kcyc$ to $\Cl(G)$ by identifying the vertices from $\J$ in the two complexes we do not induce any new edges between vertices in $\tilde\Kcyc$ and vertices in $\Cl(G)$.  
Therefore the new complex is 2-determined (since the two original complexes were 2-determined), so it is the clique complex of its 1-skeleton, which we will denote $G'$.
The cycle $\J \in \Cl(G')$ is homologous to the trivial cycle, since it is the boundary of $f(c) \in \Cl(G')$.
\end{proof}

\subsubsection{Proof that gadgets correctly lift the computational basis states}

\begin{theorem}
The gadgets for lifting the computational basis states fill in the cycles corresponding to those states, and do not fill in any other holes, or introduce any additional homology classes.
\end{theorem}
\begin{proof}
From \Cref{lem:thickening and coning} we have immediately that applying the thickening and coning off procedure to the projectors for the computational basis states fill in that cycle, without introducing spurious homology classes.
The only remaining thing to check is that when we glue this cycle onto the original graph we do not fill in any of the cycles that correspond to other computational basis states.
This can be verified by computing the Euler Characteristic of the resulting complex.
Recall that the Euler Characteristic of the complex gives the difference between the number of holes of even dimension and the number of holes of odd dimension(see \Cref{sec:SUSY}).
Since we have demonstrated that our process of constructing the graph does not induce any new holes in the clique complex, all the holes on an $m$-qubit graph must still be $(2m-1)$-holes.
We demonstrate in \cite{mathematica} that for each $m$-qubit projector in \Cref{table:states} the Witten index is $2^m-1$.
Since we have demonstrated that no spurious homology classes were created in the process, the $2^m-1$ homology classes that remain must correspond to the $2^m-1$ states that are in the kernel of the projector.
\end{proof}

\subsubsection{Proof that gadgets correctly lift the four-qubit entangled states}

\begin{theorem}
The gadgets for lifting the four-qubit entangled states fill in the cycles corresponding to those states, and do not fill in any other holes, or introduce any spurious homology classes.
\end{theorem}
\begin{proof}
For the gadgets for lifting entangled states we must show:
\begin{enumerate}
\item \label{item:4 1} The procedure outlined in \Cref{sec:4_qubit_construction} for cutting and gluing four-qubit cycles does indeed result in a clique complex $\Kcyc$ which triangulates $S^7$
\item \label{item:4 2} There exists a relation $R$ such that applying the function $f(\cdot)$ to the vertices of $\mathcal{K}$ gives a copy of the cycle we want to fill in, $\J$
\item \label{item:4 3} No other holes in $\Cl(\mathcal{G}_4)$ are filled in by the gadget
\end{enumerate}
Applying \Cref{lem:full gadget} then immediately gives the result.
For concreteness we will demonstrate \Cref{item:4 1} -- \Cref{item:4 3} for the state $\ket{1011}-\ket{1000}$, generalising the proof to the other 4-qubit entangled states in \Cref{table:states} is trivial.\footnote{The only difference between the states is which computational basis states are in the state. To change the computational basis states we have to change which cycles, we are gluing together -- this merely requires relabelling the $a$ and $b$ vertices appropriately.}

To demonstrate \Cref{item:4 1} we show that cutting open each 4-qubit cycle via the prescription in \Cref{sec:4_qubit_construction} results in a 7-sphere with a single seven-dimensional hole cut into it by the dummy vertices (i.e. it is a seven-dimensional hyperplane with the dummy vertices as its boundary). 
We can then argue that gluing two such `open' spheres along their cuts results in a single copy of $S^7$. 
Demonstrating this requires analysing the construction from \Cref{sec:4_qubit_construction} in detail,  the proof is provided in \Cref{app:gadget proofs}.

We can now turn to consider \Cref{item:4 2}.
We claim the relation:
\begin{equation}
\begin{split}
R =& \{(x_i,x)| i \in [1,10] \} \cup \{(b_{i,1},b_i)| i \in [2,4] \} \\
&\cup \{(a'_{i,1},a'_i)| i \in [2,4] \}
\cup \{(v,v)| v \in \J \}
\end{split}
\end{equation}
has the required behaviour.
We verify in \cite{mathematica} that applying $f(\cdot)$ to the vertices in $\mathcal{K}$ gives a simplicial complex with exactly those simplices present in $\mathcal{J}$.
Therefore all that is left to check is that the orientation of the simplices is correct.
We note that the orientation of the simplices within each cycle will be fixed relative to each other (since we start with the original cycles and cut them open, but do not change orientation within cycles), so it is just the orientation between the cycles which needs to be checked.
The orientation between the cycles is determined by the orientation of the shared $x_i$ vertices within the cycles.
It can be seen in \cite{mathematica} that we have fixed the $x_i$ vertices to be in the same orientation in the two cycles, which results in adding the cycles with opposite signs as required.

Finally, to verify \Cref{item:4 3} we have computed the Euler Characteristic of the complex that results from the gadget construction \cite{mathematica}.
It is 15, and since we have demonstrated that the gadget construction does not introduce new homology classes, the 15 remaining homology classes must all be $7$ 
-holes and they must correspond to the 15-dimensional kernel of the projector.
\end{proof}

\subsubsection{Proof that gadgets correctly lift the Pythagorean states}

\begin{theorem}
The gadgets for lifting the Pythagorean states fill in the cycles corresponding to those states, and do not fill in any other holes, or introduce any additional homology classes.
\end{theorem}
\begin{proof}
As in the previous section, for the Pythagorean gadgets we must show:
\begin{enumerate}
\item \label{item:3 1} The procedure outlined in \Cref{sec:3 qubit cutting} for cutting and gluing three-qubit cycles does indeed result in a clique complex $\Kcyc$ which triangulates $S^5$
\item \label{item:3 2} There exists a relation $R$ such that applying the function $f(\cdot)$ to the vertices of $\mathcal{K}$ gives a copy of the cycle we want to fill in, $\J$
\item \label{item:3 3} No other holes in $\Cl(\mathcal{G}_3)$ are filled in by the gadget
\end{enumerate}
Applying \Cref{lem:full gadget} then immediately gives the result.
For concreteness we will demonstrate \Cref{item:3 1} -- \Cref{item:3 3} for the state $-5\ket{011}+4\ket{100}+3\ket{101}$, but as in the previous section, generalising the proof to the other Pythagorean state from \Cref{table:states} is trivial.

The demonstration of \Cref{item:3 1} follows the same structure as in the previous section, although it is slightly more involved as there are more cycles to consider. 
The proof is provided in \Cref{app:gadget proofs}.

We can now turn to consider \Cref{item:3 2}.
We claim the relation:
\begin{equation}
\begin{split}
R =& \{(x_i,x)| i \in [1,45] \}  \\
&\cup \{(a_{i,j},a_i)| i \in [2,4], j \in [1,4] \}   \cup \{(b_{i,j},b_i)| i \in [2,4], j \in [1,6] \} \\
&\cup \{(a'_{i,j},a'_i)| i \in [2,4], j \in [1,6] \} \cup \{(b'_{i,j},b'_i)| i \in [2,4], j \in [1,4] \} \\
& \cup \{(a''_{i,j},a''_i)| i \in [2,4], j \in [1,3] \} \cup \{(b''_{i,j},b''_i)| i \in [2,4], j \in [1,7] \} \\
& \cup \{(v,v)| v \in \J \}
\end{split}
\end{equation}
has the required behaviour.
We verify in \cite{mathematica} that applying $f(\cdot)$ to the vertices in $\mathcal{K}$ gives a simplicial complex with exactly those simplices present in $\mathcal{J}$.
Therefore all that is left to check is that the orientation of the simplices is correct. 
As in the previous section, the orientation of the simplices within each cycle will be fixed relative to each other, so it is just the orientation between the cycles which needs to be considered.
It can be seen in \cite{mathematica} that the orientation of the $x_i$ vertices with respect to the vertices in the cycles has been picked so that the adjacent cycles that are added with the same sign (which is all but two pairs of adjacent cycles) are added together constructively, while the ordering of the two pairs of cycles where the cycles are added with opposite orientation has been reversed, to add these cycles with opposite signs.

Finally, to verify \Cref{item:3 3} we have computed the Euler Characteristic of the resulting complex \cite{mathematica}.
It is 7, and since we have demonstrated that the gadget construction does not introduce new homology classes, the 7 remaining homology classes must all be $5$-holes, and must correspond to the 7-dimensional kernel of the projector.
\end{proof}

\pagebreak
\section{Spectral sequences}\label{spec_seq_sec}

The purpose of this section is to prove the following lemma, which concerns the spectrum of the Laplacian of a single gadget $\hat{\mathcal{G}}_m$ from \Cref{single_gadget_sec}.

\begin{lemma}\label{spec_seq_lemma} \emph{(Single gadget lemma)}
Let $\hat{\mathcal{G}}_m$ be the weighted graph described in \Cref{single_gadget_sec}, implementing the projector onto the integer state $|\phi\rangle$ on $m$ qubits. Let $\hat{\Delta}^k$ be the Laplacian of this graph. Recall the definition of a subspace perturbation from \Cref{subspace_perturbation_sec}.
\begin{itemize}
    \item $\hat{\Delta}^{2m-1}$ has a $(2^m-1)$-dimensional kernel, which is a $\mathcal{O}(\lambda)$-perturbation of the subspace $\{|\psi\rangle \in \mathcal{H}_m : \langle\phi|\psi\rangle = 0\}$. Note $\mathcal{H}_m$ is embedded as $\mathcal{H}_m \subseteq \mathcal{C}^{2m-1}(\mathcal{G}_m) \subseteq \mathcal{C}^{2m-1}(\hat{\mathcal{G}}_m)$.
    \item The first excited state $|\hat{\phi}\rangle$ of $\hat{\Delta}^{2m-1}$ above the kernel is a $\mathcal{O}(\lambda)$-perturbation of $|\phi\rangle \in \mathcal{H}_m$, and it has energy $\Theta(\lambda^{4m+2})$.
    \item The next lowest eigenvectors have eigenvalues $\Theta(\lambda^2)$, and they are $\mathcal{O}(\lambda)$-perturbations of sums of $(2m-1)$-simplices touching the central vertex $v_0$.
    \item The rest of the eigenvalues are $\Theta(1)$.
\end{itemize}
\end{lemma}

Spectral sequences are a tool from algebraic topology which (among other things) analyze the homology of \emph{filtered} chain complexes. It turns out that weighting a subset of vertices by $\lambda \ll 1$, as we do in our construction, naturally gives rise to a certain filtration. In this setting, Ref.~\cite{forman1994hodge} showed a beautiful relationship between the spectral sequence and the perturbative eigenspaces of the Hodge theoretic Laplacian. It is this relationship which we exploit in this section to prove \Cref{spec_seq_lemma}. For a light introduction to spectral sequences, see \cite{chow2006you}; for a comprehensive textbook, see \cite{mccleary2001user}.

\subsection{Spectral sequence of a filtration}

Spectral sequences will work best for us in the \emph{cohomology} picture. First we define a \emph{filtration}.

\begin{definition}
Suppose we have a cohomological chain complex
\begin{equation*}
\begin{tikzcd}
\mathcal{C}^{-1} \arrow[r,"d^{-1}"] & \mathcal{C}^0 \arrow[r,"d^0"] & \mathcal{C}^1 \arrow[r,"d^1"] & \mathcal{C}^2 \arrow[r,"d^2"] & \dots
\end{tikzcd}
\end{equation*}
A \emph{filtration} on this chain complex is a nested sequence of subspaces
\begin{equation*}
\mathcal{C}^k = \mathcal{U}_0^k \supseteq \mathcal{U}_1^k \supseteq \mathcal{U}_2^k \supseteq \dots
\end{equation*}
for each $n$ such that
\begin{equation*}
d^k(\mathcal{U}_l^k) \subseteq \mathcal{U}_l^{k+1} \ \forall \ k,l
\end{equation*}
Our filtrations will be bounded, in the sense that $U_j^k = \{0\}$ for sufficiently large $j$, for each $k$.
\end{definition}

We can now develop the spectral sequence of such a filtration. The spectral sequences will consist of \emph{pages} indexed by $j$. Each page is an array of vector spaces $e_{j,l}^k$, one for each dimension $k$ and filtration level $l$.

The zeroth page is simply
\begin{equation*}
e_{0,l}^k = \mathcal{U}_l^k / \mathcal{U}_{l+1}^k
\end{equation*}
The chain complex coboundary map $d^k$ induces coboundary maps
\begin{equation*}
d_{0,l}^k : e_{0,l}^k \rightarrow e_{0,l}^{k+1}
\end{equation*}
since if two chains differ by an element of $\mathcal{U}_{l+1}^k$, then their coboundaries will differ by an element of $\mathcal{U}_{l+1}^{k+1}$.

Define the first page of the spectral sequence to be the cohomology of the zeroth page with respect to $d_{0,l}^k$, entrywise for each $k,l$.
\begin{equation*}
e_{1,l}^k = \ker{d_{0,l}^k} / \im{d_{0,l}^{k-1}}
\end{equation*}
Now the coboundary map $d^k$ induces coboundary maps 
\begin{equation*}
d_{1,l}^k : e_{1,l}^k \rightarrow e_{1,l+1}^{k+1}
\end{equation*}
This is because (a) the coboundary of any representative of an element in $\ker{d_{0,l}^k}$ is a cocycle in $\mathcal{U}_{l+1}^{k+1}$, and thus is the representative of some element of $\ker{d_{0,l+1}^{k+1}}$; (b) if we chose a different representative of the element in $\ker{d_{0,l}^k}$, the resulting element of $\ker{d_{0,l+1}^{k+1}}$ would differ by an element of $\im{d_{0,l+1}^k}$, so we end up with the same element of $e_{1,l+1}^{k+1}$; and (c) if our element of $\ker{d_{0,l+1}^{k+1}}$ differed by an element of $\im{d_{0,l}^{k-1}}$, we get the exact same element of $\ker{d_{0,l+1}^{k+1}}$.

In general, at page $j$ there are induced coboundary maps
\begin{equation*}
d_{j,l}^k : e_{j,l}^k \rightarrow e_{j,l+j}^{k+1}
\end{equation*}
and page $j+1$ is defined to be the cohomology of page $j$ entrywise with respect to these coboundary maps
\begin{equation*}
e_{j+1,l}^k = \ker{d_{j,l}^k} / \im{d_{j,l-j}^{k-1}}
\end{equation*}

One should have in mind the entries of a single page $j$ laid out in an array as follows. On this array, $d_{j,l}^k$ will map from a space to the one which is one step upwards and $j$ steps to the right.

\begin{center}
\begin{tabular}{ c | c c c c }
$k$ & $:$ & & & \\
\\
$1$ & $e_{j,0}^1$ & $e_{j,1}^1$ & $e_{j,2}^1$ & \\
\\
$0$ & $e_{j,0}^0$ & $e_{j,1}^0$ & $e_{j,2}^0$ & \\
\\
$-1$ & $e_{j,0}^{-1}$ & $e_{j,1}^{-1}$ & $e_{j,2}^{-1}$ & $\dots$ \\
\\
\hline
& $0$ & $1$ & $2$ & $l$
\end{tabular}
\end{center}

We can express the spaces $e_{j,l}^k$ more explicitly.

\begin{definition} \label{Z_B_def}
Let $\mathcal{C}^k = \mathcal{U}_0^k \supseteq \mathcal{U}_1^k \supseteq \mathcal{U}_2^k \supseteq \dots$ be a filtered chain complex with coboundary $d$. Define
\begin{align*}
Z_{j,l}^k &= \mathcal{U}_l^k \cap (d^k)^{-1}(\mathcal{U}_{l+j}^{k+1}) \\
B_{j,l}^k &= \mathcal{U}_l^k \cap d^{k-1}(\mathcal{U}_{l-j}^{k-1})
\end{align*}
\end{definition}

With these definitions in place, it turns out that the terms $e_{j,l}^k$ of the spectral sequence are equal to
\begin{equation*}
e_{j,l}^k = Z_{j,l}^k / (B_{j-1,l}^k + Z_{j-1,l+1}^k)
\end{equation*}

Let $e_j^k = \bigoplus_l e_{j,l}^k$. The point of spectral sequences is that, for sufficiently large $j$, $e_j^k$ is isomorphic to the cohomology of the complex
\begin{equation*}
e_{\infty}^k \cong H^k = \ker{d^k} / \im{d^{k-1}}
\end{equation*}
The $e_j^k$ spaces form an algebraic sequence of approximations to the true cohomology.

\subsection{Relationship to Hodge theory}

From our construction we have a clique complex $\mathcal{G}$ where a subset of the vertices are weighted by $\lambda \ll 1$, and the rest by $1$. The weight of a simplex is defined to be the product of the weights of the vertices involved in the simplex. In this context, there is a natural filtration on the chain complex $\mathcal{C}$. Let
\begin{equation*}
\mathcal{U}_l^k = \text{span}\{\sigma \in \mathcal{G}^k : w(\sigma) \in \{\lambda^l,\lambda^{l+1},\dots\}\}
\end{equation*}
This is the span of the $k$-simplices which involve at least $l$ `gadget' vertices. These spaces are nested
\begin{equation*}
\mathcal{C}^k = \mathcal{U}_0^k \supseteq \mathcal{U}_1^k \supseteq \mathcal{U}_2^k \supseteq \dots
\end{equation*}
and it can also be checked that
\begin{equation*}
d^k(\mathcal{U}_l^k) \subseteq \mathcal{U}_l^{k+1}
\end{equation*}
Thus we have a filtration.

We are interested in the low energy eigenstates of the Laplacian $\Delta^k$ of $\mathcal{G}$. In particular, we would like to examine the eigenvalues which are zero to first order in $\lambda$, and then second order, and so on. This is reminiscent of perturbation theory from quantum mechanics. Recall that
\begin{equation*}
\langle\psi|\Delta^k|\psi\rangle = ||\partial^k|\psi\rangle||^2 + ||d^k|\psi\rangle||^2
\end{equation*}

Motivated by this, define the isomorphism
\begin{align*}
\rho_\lambda^k : \mathcal{C}^k &\rightarrow \mathcal{C}^k \\
|\sigma\rangle &\rightarrow w(\sigma) |\sigma\rangle
\end{align*}
where $\sigma \in \mathcal{G}^k$ is a $k$-simplex which has weight $w(\sigma)$. Morally, $\rho_\lambda^k$ maps from the weighted chainspace to the unweighted chainspace. From this, construct the maps
\begin{align*}
\partial^k_\lambda &= \rho^k_\lambda \circ \partial^k \circ (\rho^k_\lambda)^{-1} \\
d^k_\lambda &= \rho^k_\lambda \circ d^k \circ (\rho^k_\lambda)^{-1}
\end{align*}

Define the spaces
\begin{align*}
E_j^k = \{|\psi\rangle \in \mathcal{C}^k \ &: \ \exists \ |\psi_\lambda\rangle = |\psi\rangle + \lambda|\psi_1\rangle + \lambda^2|\psi_2\rangle + \dots + \lambda^j|\psi_j\rangle \in \mathcal{C}^k[\lambda] \\
&\text{s.t.} \ \partial^k_\lambda|\psi_\lambda\rangle \in \lambda^j \mathcal{C}^{k-1}[\lambda] , \ d^k_\lambda|\psi_\lambda\rangle \in \lambda^j \mathcal{C}^{k+1}[\lambda]\}
\end{align*}
where $\mathcal{C}^k[\lambda]$ is the space of polynomials in $\lambda$ with coefficients in $\mathcal{C}^k$. Further define
\begin{equation*}
E_{j,l}^k = E_j^k \cap \mathcal{U}_l^k \cap (\mathcal{U}_{l+1}^k)^\perp
\end{equation*}
so that
\begin{equation*}
E_j^k = \oplus_l E_{j,l}^k
\end{equation*}
$E_j^k$ is the space of vectors which have perturbations in $\lambda$ which give energies of size $\mathcal{O}(\lambda^{2j})$ on the Laplacian $\Delta^k$. Here and throughout this section, $\mathcal{O}(\lambda^{l})$ is used as shorthand for polynomials in $\lambda$ which contain no terms of degree less than $l$. Taking $j\rightarrow\infty$ should give $E_\infty^k = H^k(\mathcal{C})$, and these spaces $E_j^k$ form a Hodge-theoretic sequence of approximations to the true homology.

\begin{proposition} \label{O_lambda_perturbation_clm}
The space
\begin{equation*}
\{|\psi\rangle \in \mathcal{C}^k \ : \ \text{$|\psi\rangle$ is an eigenvector of $\Delta^k$ with eigenvalue $\mathcal{O}(\lambda^{2j})$}\}
\end{equation*}
is a $\mathcal{O}(\lambda)$-perturbation of $E_j^k$, in the sense of \Cref{subspace_perturbation_sec}.
\end{proposition}
\begin{proof}
This follows from Theorem 2 of \cite{forman1994hodge}, combined with Rellich's theorem stated at the end of the introduction of \cite{forman1994hodge}.
\end{proof}

Theorem 7 from \cite{forman1994hodge} tells us remarkably that these $E_j^k$ spaces are \emph{isomorphic} to the $e_j^k$ spaces of the filtration. The $E_j^k$ spaces are our real objects of interest, but the $e_j^k$ spaces are tractable to calculate. This is reminiscent of a standard methodology in algebraic topology where we prove difficult topological and analytic statements by turning them into algebraic statements. This will become our strategy to prove \Cref{spec_seq_lemma}. The rest of this section is devoted to describing the isomorphism.

Recall \Cref{Z_B_def}. $\lambda^{-l}\rho^k_\lambda$ creates isomorphisms
\begin{align*}
Z_{j,l}^k \rightarrow \ &\widetilde{Z}_{j,l}^k \\
&:= \{|\psi_\lambda\rangle = |\psi_0\rangle + \lambda|\psi_1\rangle + \lambda^2|\psi_2\rangle + \dots \ : \ |\psi_i\rangle \in \mathcal{U}_{l+i}^k , \ d_\lambda|\psi_\lambda\rangle \in \lambda^j\mathcal{C}[\lambda]\} \\
Z_{j-1,l+1}^k \rightarrow \ &\lambda\widetilde{Z}_{j-1,l+1}^k \\
&:= \{|\psi_\lambda\rangle = \lambda|\psi_1\rangle + \lambda^2|\psi_2\rangle + \dots \ : \ |\psi_i\rangle \in \mathcal{U}_{l+i}^k , \ d_\lambda|\psi_\lambda\rangle \in \lambda^j\mathcal{C}[\lambda]\} \\
B_{j-1,l}^k \rightarrow \ &\widetilde{B}_{j-1,l}^k \\
&:= \lambda^{-j+1} d_\lambda \{|\psi_\lambda\rangle = |\psi_0\rangle + \lambda|\psi_1\rangle + \lambda^2|\psi_2\rangle + \dots \ : \ |\psi_i\rangle \in \mathcal{U}_{l-j+1+i}^{k-1} , \ d_\lambda|\psi_\lambda\rangle \in \lambda^{j-1}\mathcal{C}[\lambda]\} \\
&= \lambda^{-j+1} d_\lambda \widetilde{Z}_{j-1,l-j+1}^k
\end{align*}

Let $|\psi\rangle \in E_{j,l}^k$. Let
\begin{equation*}
|\psi_\lambda\rangle = |\psi\rangle + \lambda|\psi_1\rangle + \lambda^2|\psi_2\rangle + \dots
\end{equation*}
be the polynomial from the definition of $E_j^k$. Let $|\psi_{i,l}\rangle$ denote the projection of $|\psi_i\rangle$ onto $\mathcal{U}_l^k$, and write
\begin{align*}
|\psi_\lambda\rangle &= \bigg[|\psi\rangle + \sum_{i>0} \lambda^i|\psi_{i,l+i}\rangle\bigg] + \bigg[\sum_{i>0} \sum_{c\neq i} \lambda^i|\psi_{i,l+c}\rangle\bigg] \\
&= |\psi_\lambda^{(0)}\rangle + |\psi_\lambda^{(1)}\rangle
\end{align*}

\begin{proposition}\label{thm_7} \emph{(Theorem 7 from \cite{forman1994hodge})}

$|\psi_\lambda^{(0)}\rangle \in \widetilde{Z}_{j,l}^k$ and the map
\begin{align*}
E_{j,l}^k &\rightarrow \widetilde{Z}_{j,l}^k / (\lambda \widetilde{Z}_{j-1,l+1}^k + \widetilde{B}_{j-1,l}^k) \cong e_{j,l}^k \\
|\psi\rangle &\rightarrow [|\psi_\lambda^{(0)}\rangle]
\end{align*}
is an isomorphism. The corresponding map from $[|\psi\rangle] \in e_{j,l}^k$ to $E_{j,l}^k$ is to project the representative $|\psi\rangle$ onto $\mathcal{U}_l^k$.
\end{proposition}

\subsection{Example spectral sequence}

In this section, we get some practice with spectral sequences by calculating the spectral sequence of the weighted complex shown below. The vertices on the perimeter have weight 1, and the vertices in the interior have weight $\lambda \ll 1$. This is a filling in of a hexagon, and we are concerned with the $k=1$ homology. This complex does not correspond to any Hamiltonian on any number of qubits, but rather it is like a rank-1 projector on a 1-dimensional Hilbert space. Regardless, let's refer to the weight 1 vertices on the boundary as \emph{qubit} vertices, and the weight $\lambda$ vertices in the bulk as \emph{gadget} vertices. It will serve as a simple example which highlights some key features which will be present in the general gadget construction.

\begin{figure}[H]
\begin{center}
\includegraphics[scale=0.4]{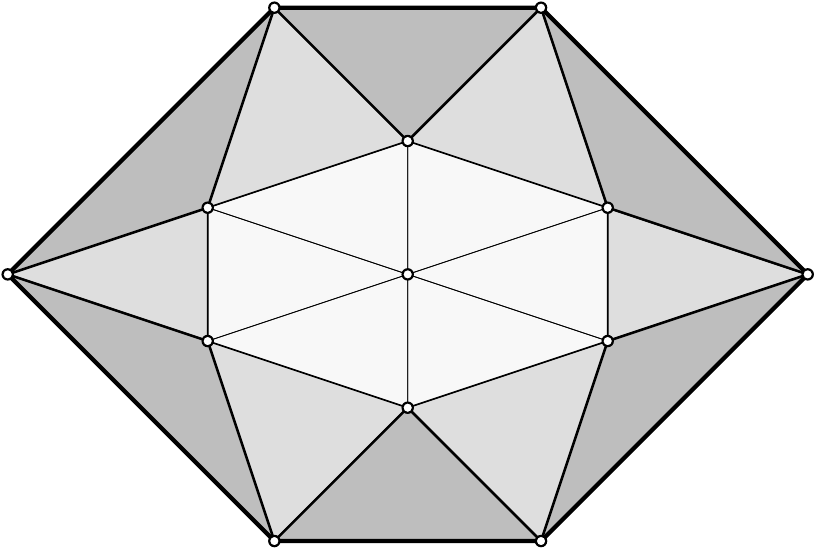}
\caption{The complex for the example spectral sequence. 2-simplices are shown darker and bolder if they are more heavily weighted.}
\label{fig:spec seq 1}
\end{center}
\end{figure}
The relevant chainspaces are
\begin{align*}
\mathcal{C}^2 &= \mathcal{U}_0^2 \supseteq \mathcal{U}_1^2 \supseteq \mathcal{U}_2^2 \supseteq \mathcal{U}_3^2 \\
\mathcal{C}^1 &= \mathcal{U}_0^1 \supseteq \mathcal{U}_1^1 \supseteq \mathcal{U}_2^1 \\
\mathcal{C}^0 &= \mathcal{U}_0^0 \supseteq \mathcal{U}_1^0
\end{align*}
The chain $\mathcal{C}^2$ is truncated at $\mathcal{U}_3^2$ since there are no triangles with more than 3 gadget vertices, and similarly for edges $\mathcal{C}^1$ and vertices $\mathcal{C}^0$.

The zeroth page of the spectral sequence is $e_{0,l}^k$ for $k=-1,0,1,2$, $l=0,\dots,k+1$. $e_{0,l}^k$ can be thought of as the space spanned by the $k$-simplices of weight $\lambda^l$. There are no triangles consisting only of qubit vertices, so $e_{0,0}^{2} = \{0\}$. Below we see pictorial representations of Page 0.

\begin{table}[H]
\begin{center}
\caption*{Page 0}
\begin{tabular}{ C{0.5cm} | C{3.5cm} C{3.5cm} C{3cm} C{3cm} C{2cm} }
$k$ & & & & & \\
\\
$2$ & & 
\includegraphics[scale=0.2]{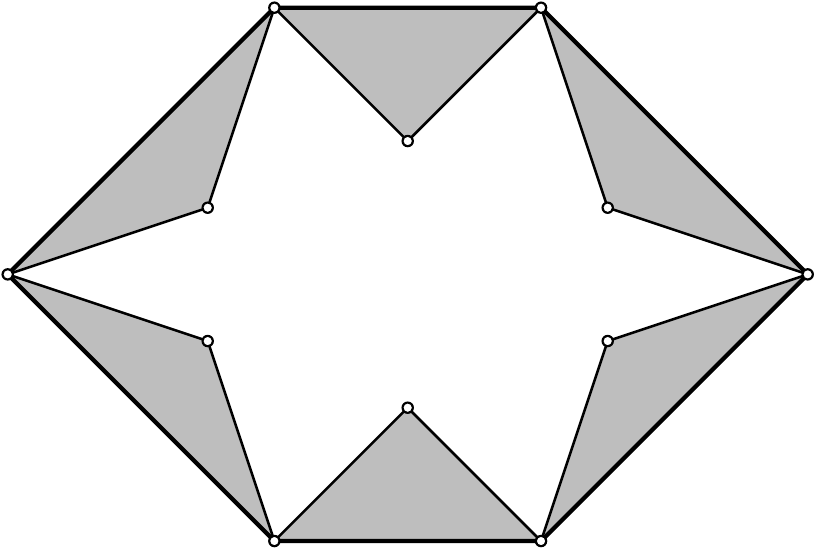}
& \includegraphics[scale=0.2]{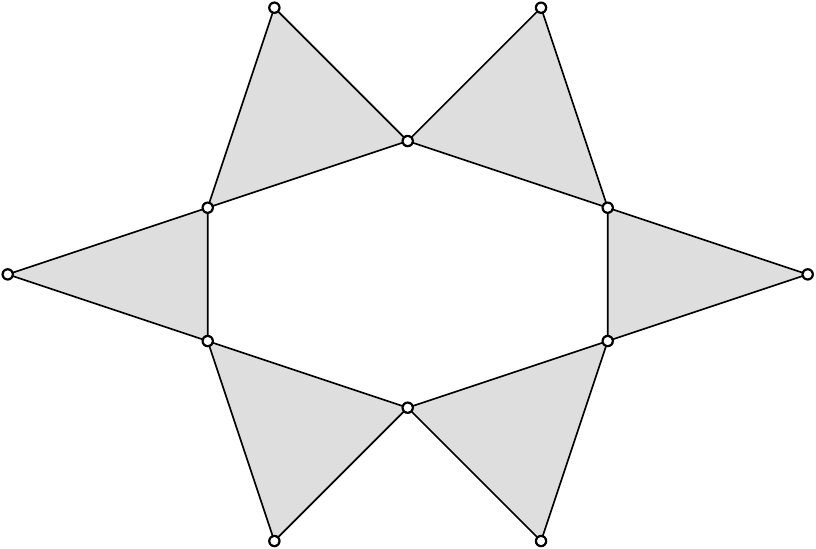} & \includegraphics[scale=0.2]{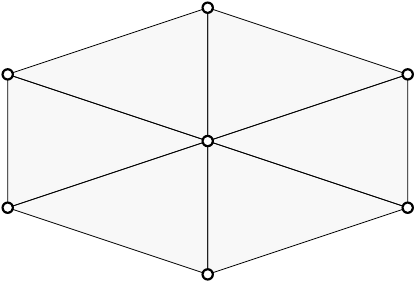} & \\
\\
$1$ & \includegraphics[scale=0.2]{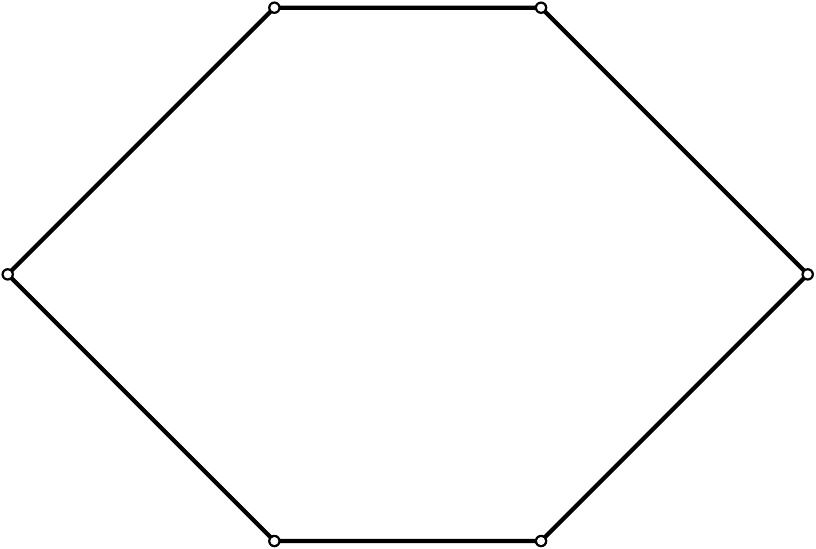} & \includegraphics[scale=0.2]{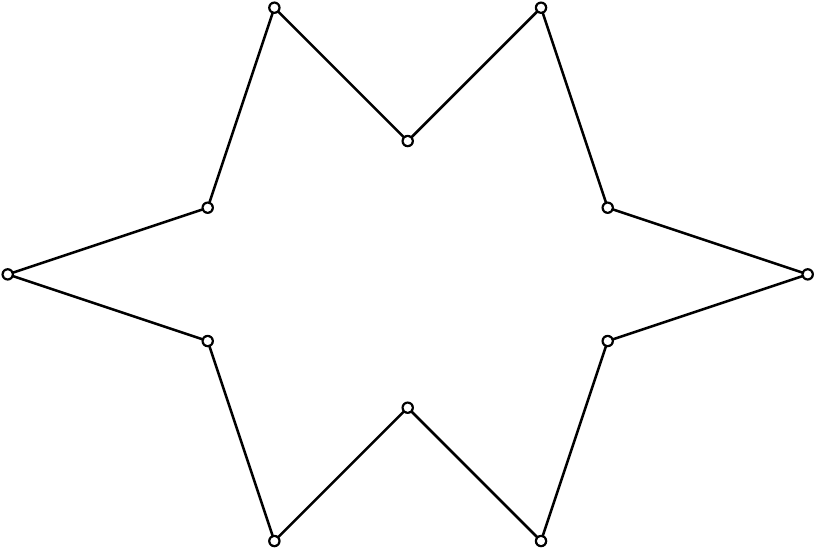} & \includegraphics[scale=0.2]{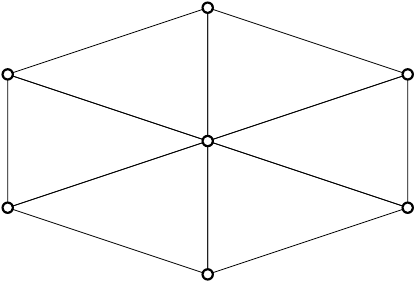} & & \\
\\
$0$ & \includegraphics[scale=0.2]{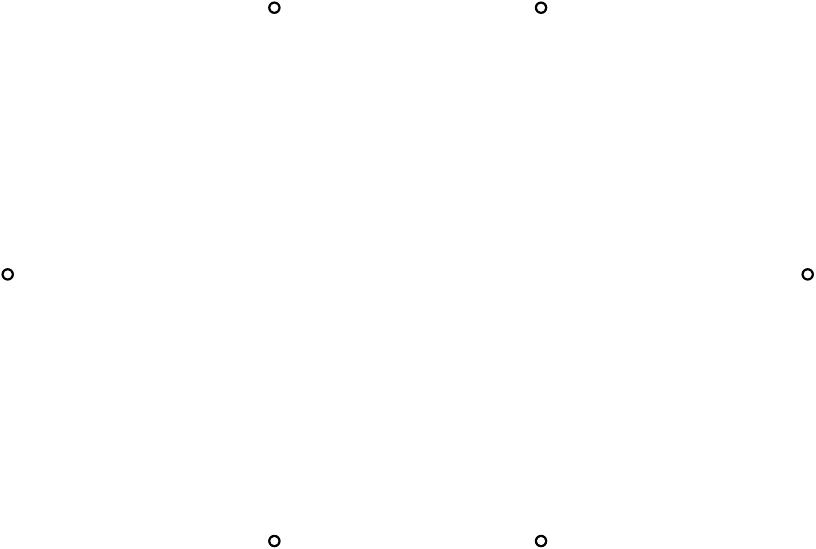} & \includegraphics[scale=0.2]{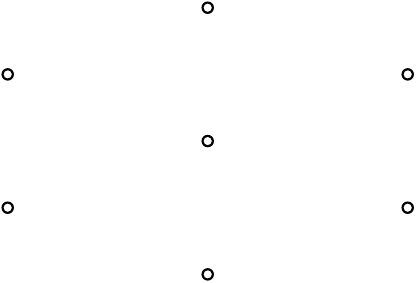} & & & \\
\\
$-1$ & \includegraphics[scale=0.5]{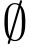}  & & & & \\
\\
\hline
& $0$ & $1$ & $2$ & $3$ & $l$
\end{tabular}
\end{center}
\end{table}

The coboundary map of the zeroth page maps
\begin{equation*}
d_{0,l}^k : e_{0,l}^k \rightarrow e_{0,l}^{k+1}
\end{equation*}
It maps upwards one step in the diagram, and it can be thought of as a qubit vertex coboundary map, which adds a qubit vertex to the simplex. If no qubit vertex can be added to the $k$-simplex $\sigma$, then $d_{0,l}^k |\sigma\rangle = 0$.

We are now in a position to compute the first page $e_{1,l}^k$ of the spectral sequence, which is defined as
\begin{equation*}
e_{1,l}^k = \ker{d_{0,l}^k} / \im{d_{0,l}^{k-1}}
\end{equation*}
It may be more intuitive to bear in mind that this is the same as
\begin{equation*}
e_{1,l}^k = \ker{\partial_{0,l}^k} / \im{\partial_{0,l}^{k+1}}
\end{equation*}
where
\begin{equation*}
\partial_{0,l}^k : e_{0,l}^k \rightarrow e_{0,l}^{k-1}
\end{equation*}
is the qubit boundary map, which removes qubit vertices. (Formally, we can define $\partial_{0,l}^k = (d_{0,l}^{k-1})^\dag$.)

The first column $l=0$, we simply get the homology of the qubit complex, which is topologically a single loop $S^1$. Thus $e_{1,0}^0 = e_{1,0}^2 = \{0\}$, but $e_{1,0}^1$ is the 1-dimensional space spanned by the loop of qubit vertices. For the purposes of this section, denote this state by $|\text{loop}\rangle$. The $e_{0,1}^1$ looks like it also has a loop; that is, a 1-chain without boundary. However, this 1-chain is in fact the boundary of the uniform superposition of triangles in $e_{0,1}^2$, so the homology at this position is zero $e_{1,1}^1 = \{0\}$. The triangle chainspaces $e_{0,1}^2$, $e_{0,2}^2$ have no cycles, so their homologies are zero $e_{1,1}^2 = e_{1,2}^2 = \{0\}$. In $e_{0,1}^0$, the only gadget vertex which is \emph{not} the boundary of an edge in $e_{0,1}^1$ is the central vertex, so $e_{1,1}^0$ is the 1-dimensional space spanned by this vertex. Similarly, the edges on the `outside' of $e_{0,2}^1$ are the boundaries of triangles in $e_{0,2}^2$, so the ones which survive in the homology $e_{1,2}^2$ are the spokes of the `star' touching the central vertex. For $e_{1,3}^2$, both boundary maps involved are zero, so $e_{1,3}^2 = e_{0,3}^2$ i.e. all the central triangles. See the below pictorial representations of the first page.

\begin{table}[H]
\begin{center}
\caption*{Page 1}
\begin{tabular}{ C{0.2cm} |C{0.1cm} C{2.5cm} C{0.1cm} C{1cm} C{0.1cm} C{2.5cm} C{0.2cm} C{2.5cm} C{0.1cm} C{2cm} }
$k$ & & & & & & & & & & \\
\\
$2$ & & & & & & & & \includegraphics[scale=0.2]{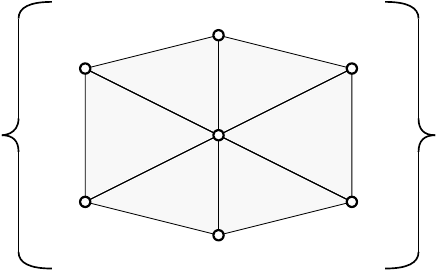} & &  \\
\\
 &  & &\includegraphics[scale=0.4]{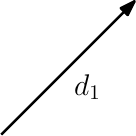} & & & & \includegraphics[scale=0.4]{figures/fig_B_copy_4.pdf} &  & &  \\
$1$ & & \includegraphics[scale=0.2]{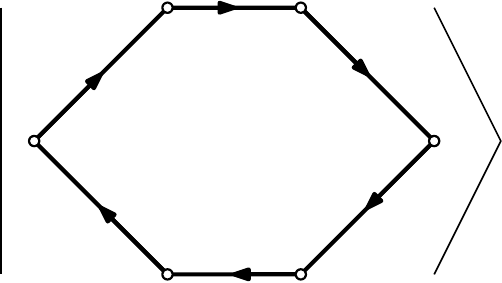} & & & & \includegraphics[scale=0.2]{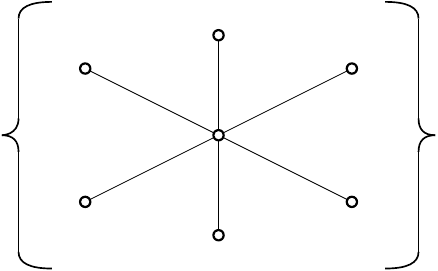}& & & & \\
\\
 & \includegraphics[scale=0.3]{figures/fig_B_copy_4.pdf} & & & &\includegraphics[scale=0.4]{figures/fig_B_copy_4.pdf} & &  &  & &  \\
$0$ & & & & \includegraphics[scale=0.2]{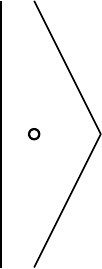} & & & & & & \\
 &  & & \includegraphics[scale=0.3]{figures/fig_B_copy_4.pdf} & & & &  &  & &  \\
\\
$-1$ & & & & & & & & & & \\
\\
\hline
& & $0$ & & $1$ & & $2$ & & $3$ & & $l$
\end{tabular}
\end{center}
\end{table}

The coboundary map of the first page maps
\begin{equation*}
d_{1,l}^k : e_{1,l}^k \rightarrow e_{1,l+1}^{k+1}
\end{equation*}
It maps `diagonally' one step up and one step to the right in the diagram. It can be thought of as a gadget vertex coboundary map, which adds a gadget vertex to the simplex. If no gadget vertex can be added to the $k$-simplex $\sigma$, then $d_{1,l}^k |\sigma\rangle = 0$. As before, it may be more intuitive to visualize the gadget vertex boundary map $\partial_{1,l}^k = (d_{1,l-1}^{k-1})^\dag$.

Using this, we aim to compute the second page $e_{2,l}^k$ of the spectral sequence, which is defined as
\begin{equation*}
e_{2,l}^k = \ker{d_{1,l}^k} / \im{d_{1,l-1}^{k-1}} = \ker{\partial_{1,l}^k} / \im{\partial_{1,l+1}^{k+1}}
\end{equation*}
Both boundary maps acting on $e_{1,0}^1$ are zero, so $e_{2,0}^1 = e_{1,0}^1$, which recall is the span of the state $|\text{loop}\rangle$. Apart from $e_{1,0}^1$, we have a 3-chain
\begin{equation*}
\begin{tikzcd}
\includegraphics[scale=0.2]{figures/fig_B_copy.pdf} \arrow[r,"d_{1,1}^0"] & \includegraphics[scale=0.2]{figures/fig_B_copy_2.pdf} \arrow[r,"d_{1,2}^1"] & \includegraphics[scale=0.2]{figures/fig_B_copy_3.pdf}
\end{tikzcd}
\end{equation*}
$e_{1,1}^0$ contains only the central vertex $v_0$, which is \emph{not} in the kernel of $d_{1,1}^0$, so $e_{2,1}^0 = \{0\}$. $e_{1,2}^1$ consists of the star of edges touching $v_0$. The coboundaries of these edges has support on the triangles wedged in between them. The only way for a superposition of these edges to have coboundaries cancelling on each of these triangles is to be proportional to the uniform superposition. But this state is precisely the coboundary of the central vertex $d_{1,1}^0 |v_0\rangle$. Thus the middle term has no homology and $e_{2,2}^1 = \{0\}$. It remains to calculate the space $e_{2,3}^2$. By counting dimensions, there should be a single state in the homology of $e_{1,3}^2$. This is because there are the same number of edges spanning $e_{1,2}^1$ as triangles spanning $e_{1,3}^2$, and precisely one state in $e_{1,2}^1$ was killed by $d_{1,2}^1$, namely the uniform superposition $d_{1,1}^0 |v_0\rangle$. Thus there is a unique state in $e_{2,3}^2$, which is the state in $e_{1,3}^2$ orthogonal to the image of $d_{1,2}^1$. But $(\im{d_{1,2}^1})^\perp = \ker{\partial_{1,3}^2}$, so we are looking for the unique state in the kernel of the gadget vertex boundary map $\partial_{1,3}^2$. We can now see that this state is the uniform superposition over the triangles, with matching orientations. Denote this state by $|\text{core}\rangle$.


\begin{table}[H]
\begin{center}
\caption*{Page 2}
\begin{tabular}{ C{0.2cm} | C{2cm} C{2cm} C{2cm} C{2cm} C{2cm} }
$k$ & & & & & \\
\\
$2$ & & \includegraphics[scale=0.4]{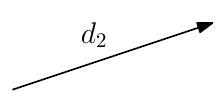} & & \includegraphics[scale=0.2]{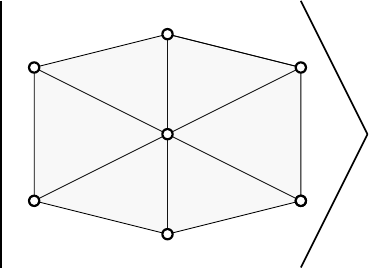} & \\
\\
$1$ & \includegraphics[scale=0.2]{figures/fig_B.pdf} & & \includegraphics[scale=0.4]{figures/arrow_d_2.pdf} & & \\
\\
$0$ & & & & & \\
\\
$-1$ & & & & & \\
\\
\hline
& $0$ & $1$ & $2$ & $3$ & $l$
\end{tabular}
\end{center}
\end{table}

We have seen that page 2 contains only two states: $|\text{loop}\rangle \in e_{2,0}^1$ and $|\text{core}\rangle \in e_{2,3}^2$. It turns out that Page 3 of the spectral sequence must be identical to page 2. To see why this is the case, consider the coboundary map of page 2 $d_{2,l}^k$. This maps via the `knight move' one step up and two steps to the right - see \Cref{fig:knight move}.
\begin{equation*}
d_{2,l}^k : e_{2,l}^k \rightarrow e_{2,l+2}^{k+1}
\end{equation*}
Thus $|\text{loop}\rangle$ is mapped into $e_{2,2}^2 = \{0\}$, and there is nothing in $e_{2,1}^1 = \{0\}$ to map to $|\text{core}\rangle$. Since all relevant coboundary maps are zero, taking the homology leaves the page unchanged.

\begin{figure}[H]
\centering
\includegraphics[scale=0.6]{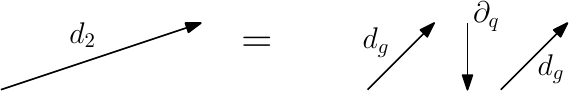}
\caption{The coboundary map of page 2, $d_2 = d_g\cdot \partial_q \cdot d_g$. This acts as a `knight move' as shown above.} \label{fig:knight move}
\end{figure}

\begin{table}[H]
\begin{center}
\caption*{Page 3}
\begin{tabular}{ C{0.2cm} | C{2cm} C{4cm} C{2cm} C{2cm} }
$k$ & & & & \\
\\
$2$ & & & \includegraphics[scale=0.2]{figures/fig_B_copy_5.pdf} & \\
& & \includegraphics[scale=0.4]{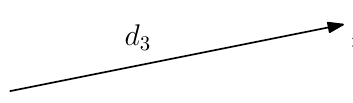} & & \\
$1$ & \includegraphics[scale=0.2]{figures/fig_B.pdf} & & & \\
\\
$0$ & & & & \\
\\
$-1$ & & & & \\
\\
\hline
& $0$ & $1 \qquad\qquad\qquad 2$ & $3$ & $l$
\end{tabular}
\end{center}
\end{table}

This argument no longer applies when we look at page 4. Now, the coboundary map of page 3 $d_{3,l}^k$ maps one step up and \emph{three} steps to the right. In particular, $d_{3,0}^1$ maps
\begin{equation*}
d_{3,0}^1 : e_{3,0}^1 \rightarrow e_{3,3}^2
\end{equation*}
We will argue that in fact (ignoring normalizations)
\begin{equation*}
d_{3,0}^1 \ |\text{loop}\rangle \ = \ |\text{core}\rangle
\end{equation*}
and these two states cancel each other out when we take the homology, leaving page 4 completely empty. To see this, we need to examine the coboundary map $d_{3,0}^1$ more closely. Returning to the zeroth page $e_{0,l}^k$, denote by $\partial_{\text{qubit},l}^k = \partial_{0,l}^k$ the `qubit boundary map'
\begin{equation*}
\partial_{\text{qubit},l}^k : e_{0,l}^k \rightarrow e_{0,l}^{k-1}
\end{equation*}
which acts by removing qubit vertices. Further, denote by $d_{\text{gadget},l}^k$ the `gadget coboundary map'
\begin{equation*}
d_{\text{gadget},l}^k : e_{0,l}^k \rightarrow e_{0,l+1}^{k+1}
\end{equation*}
which acts by adding gadget vertices. Now we can think of the map $d_{3,0}^1$ as acting by
\begin{equation*}
d_{3,0}^1 = d_{\text{gadget},2}^1 \circ \partial_{\text{qubit},2}^2 \circ d_{\text{gadget},1}^1 \circ \partial_{\text{qubit},1}^2 \circ d_{\text{gadget},0}^1
\end{equation*}
as shown in \Cref{fig:d3 fig}
With this new understanding, let's examine $d_{3,0}^1 |\text{loop}\rangle$. From \Cref{fig:sinking}, we can see that indeed $d_{3,0}^1 |\text{loop}\rangle = |\text{core}\rangle$.

\begin{figure}[H]
\centering
\includegraphics[scale=0.6]{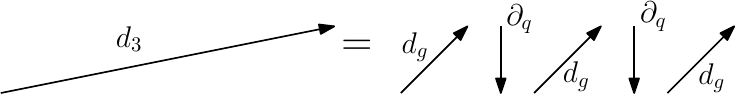}
\caption{The coboundary map of page 2, $d_3 = d_g\cdot \partial_q \cdot d_g \cdot \partial_q \cdot d_g$. This acts by moving one step up and three to the right as shown.} \label{fig:d3 fig}
\end{figure} 

\begin{figure}[H]
\centering
\includegraphics[scale=0.4]{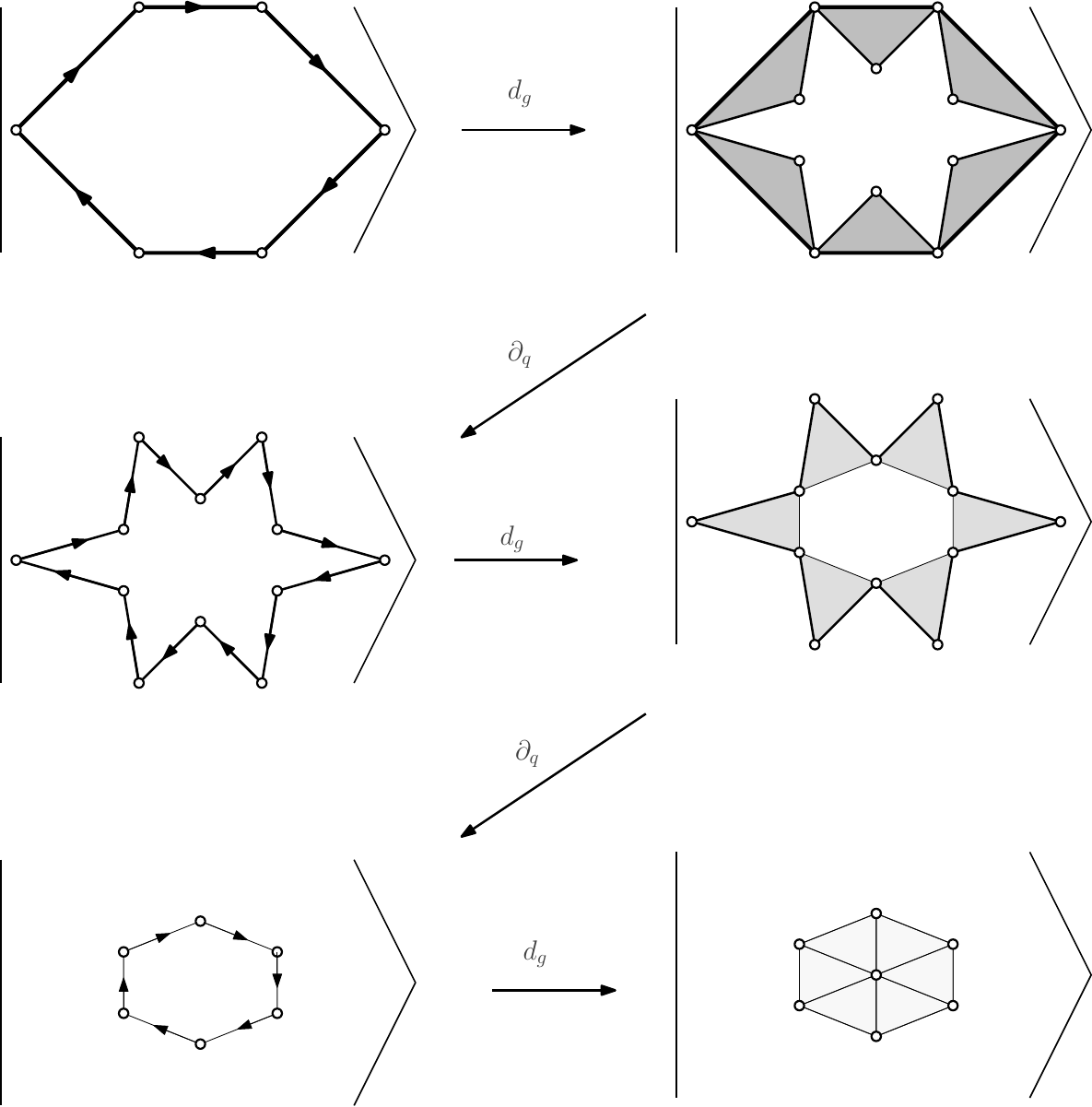}
\caption{Applying the $d_3$ map to the $\ket{\text{loop}}$ gives the $\ket{\text{core}}$ state as claimed.} \label{fig:sinking}
\end{figure}

\begin{table}[H]
\begin{center}
\caption*{Page 4}
\begin{tabular}{ c | c c c c c }
$k$ & & & & & \\
\\
$2$ & & & & & \\
\\
$1$ & & & & & \\
\\
$0$ & & & & & \\
\\
$-1$ & & & & & \\
\\
\hline
& $0$ & $1$ & $2$ & $3$ & $l$
\end{tabular}
\end{center}
\end{table}

\subsection{Proof of \Cref{spec_seq_lemma}}

We will apply the machinery of spectral sequences to the construction $\hat{\mathcal{G}}_m$ described in \Cref{single_gadget_sec}. Our real objects of interest for \Cref{spec_seq_lemma} are the spaces $E_j^k$. Our plan is to calculate the spaces $e_{j,l}^k$ and use the isomorphism in \Cref{thm_7} to analyze $E_j^k$. The isomorphism in \Cref{thm_7} allows us to learn about the analytical spaces $E_j^k$ using easier algebraic techniques.

In \Cref{single_gadget_sec}, we were concerned with the $(2m-1)$-homology, and we were filling in a $(2m-1)$-cycle $|\phi\rangle \in H^{2m-1}(\mathcal{G}_m)$ by adding a gadget of $2m$-simplices. Recall $\mathcal{G}_m^k \subseteq \hat{\mathcal{G}}_m^k$ is the original qubit complex before adding the gadget. Equivalently, $\mathcal{G}_m^k$ is the simplicial subcomplex consisting of simplices of weight 1. For this section, we will drop the $m$ subscripts $\mathcal{G}^k \subseteq \hat{\mathcal{G}}^k$. The original vertices $\mathcal{G}^0$ have weight 1, and will be referred to as \emph{qubit vertices}. The added vertices $\hat{\mathcal{G}}^0 \setminus \mathcal{G}^0$ have weight $\lambda$, and will be referred to as \emph{gadget vertices}. The weight of a simplex is defined to be the product of the weights of the vertices involved in the simplex. Recalling the construction of the gadgets from \Cref{single_gadget_sec}, introduce the following notations:
\begin{itemize}
    \item Let $[\text{bulk}]$ denote the set of simplices which involve the central vertex $v_0$, and let $[\text{non-bulk}]$ be the complement of this set. It can be checked that $[\text{non-bulk}] \subseteq \hat{\mathcal{G}}$ is in fact a simplicial subcomplex, but $[\text{bulk}]$ is \emph{not}.
    \item Let $\Omega_j^k \subseteq \hat{\mathcal{G}}^k$ be the $k$-simplices of weights $\{1,\lambda,\dots,\lambda^j\}$. It can be checked that $\Omega_j \subseteq \hat{\mathcal{G}}$ is a simplicial subcomplex for each $j$, and $\mathcal{C}^k(\Omega_j) = (\mathcal{U}_{j+1}^k)^\perp$. Tautologically, $\Omega_0^k = \mathcal{G}^k$ is the qubit complex.
\end{itemize}

\subsubsection{Spectral sequence of general gadget}

We now calculate the spectral sequence of the general gadget construction from \Cref{single_gadget_sec}. The following lemma, stated in generality, will prove useful.
\begin{lemma}\label{P_Q_lemma}
Let $\mathcal{P}$ be a simplicial complex and $\mathcal{Q} \subseteq \mathcal{P}$ a simplicial subcomplex. If $\mathcal{Q}$ has no $(k-1)$-cohomology, then
\begin{equation*}
\mathcal{C}^k(\mathcal{Q})^\perp \cap d^{k-1}(\mathcal{C}^{k-1}(\mathcal{Q})^\perp) = \mathcal{C}^k(\mathcal{Q})^\perp \cap \im{d^{k-1}}
\end{equation*}
\end{lemma}
\begin{proof}
It is clear that
\begin{equation*}
\mathcal{C}^k(\mathcal{Q})^\perp \cap d^{k-1}(\mathcal{C}^{k-1}(\mathcal{Q})^\perp) \subseteq \mathcal{C}^k(\mathcal{Q})^\perp \cap \im{d^{k-1}}
\end{equation*}
It remains to show the opposite inclusion.

Let $|\alpha\rangle \in \mathcal{C}^k(\mathcal{Q})^\perp \cap \im{d^{k-1}}$. We must show that $|\alpha\rangle \in \mathcal{C}^k(\mathcal{Q})^\perp \cap d^{k-1}(\mathcal{C}^{k-1}(\mathcal{Q})^\perp)$. We know there is a $|\beta\rangle \in \mathcal{C}^{k-1}(\mathcal{P})$ with $|\alpha\rangle = d^{k-1}|\beta\rangle$. Now
\begin{equation*}
d^{k-1}(|\beta\rangle_{\mathcal{C}^{k-1}(\mathcal{Q})})_{\mathcal{C}^k(\mathcal{Q})} = 0
\end{equation*}
where the subscript denotes \emph{restriction}. Here comes the key step: $\mathcal{Q}$ has no $(k-1)$-cohomology, which implies that there exists a $|\gamma\rangle \in \mathcal{C}^{k-2}(\mathcal{Q})$ with
\begin{equation*}
|\beta\rangle_{\mathcal{C}^{k-1}(\mathcal{Q})} = (d^{k-2} |\gamma\rangle)_{\mathcal{C}^{k-1}(\mathcal{Q})}
\end{equation*}
Now $(d^{k-1} \circ d^{k-2}) |\gamma\rangle = 0$, so
\begin{equation*}
d^{k-1}(d^{k-2}|\gamma\rangle)_{\mathcal{C}^{k-1}(\mathcal{Q})} + d^{k-1}(d^{k-2}|\gamma\rangle)_{\mathcal{C}^{k-1}(\mathcal{Q})^\perp} = 0
\end{equation*}
Consider
\begin{equation*}
|\beta\rangle_{\mathcal{C}^{k-1}(\mathcal{Q})^\perp} - (d^{k-2} |\gamma\rangle)_{\mathcal{C}^{k-1}(\mathcal{Q})^\perp} \in \mathcal{C}^{k-1}(\mathcal{Q})^\perp
\end{equation*}
This has
\begin{equation*}
d^{k-1}\Big(|\beta\rangle_{\mathcal{C}^{k-1}(\mathcal{Q})^\perp} - (d^{k-2} |\gamma\rangle)_{\mathcal{C}^{k-1}(\mathcal{Q})^\perp}\Big) = |\alpha\rangle
\end{equation*}
\end{proof}
\Cref{P_Q_lemma} is telling us that, if $\mathcal{Q} \subseteq \mathcal{P}$ is a simplicial subcomplex with no $(k-1)$-homology, then all $k$-coboundaries in $\mathcal{C}^k(\mathcal{Q})^\perp$ are in fact the coboundaries of chains in $\mathcal{C}^{k-1}(\mathcal{Q})^\perp$.

\bigskip
The zeroth page of the spectral sequence is simply $e_{0,l}^k = \mathcal{U}_l^k / \mathcal{U}_{l+1}^k$. By choosing the representative orthogonal to $\mathcal{U}_{l+1}^k$ in each equivalence class, we can think of $e_{0,l}^k$ as $\mathcal{U}_l^k \cap (\mathcal{U}_{l+1}^k)^\perp$, which are the simplices of weight $\lambda^l$. This will be a general trick in this section, to replace a quotient by an intersection with the orthogonal space. There are no simplices of dimension higher than $2m$, so there is nothing above the row $k=2m$; that is, $e_{0,l}^k = \{0\} \ \forall \ k>2m$. There are no $k$-simplices with more than $k+1$ gadget vertices, so there is nothing below the diagonal $l=k+1$; that is, $e_{0,l}^k = \{0\} \ \forall \ l>k+1$. There are no $2m$-simplices consisting only of qubit vertices, so $e_{0,0}^{2m} = \{0\}$. The first column $l=0$ is simply the chainspaces of the qubit complex $\mathcal{C}^k(\mathcal{G})$. The diagonal $l=k+1$ consists simply of $e_{0,k+1}^k = \mathcal{U}_{k+1}^k$, which are the simplices involving no qubit vertices. Note these are \emph{not} the same as the bulk chainspaces $\mathcal{C}^k([\text{bulk}])$. The spaces $e_{0,l}^k$ for $0<l<k+1$ consist of simplices which are a mixture of qubit and gadget vertices, which make up the \emph{thickening} of the gadget construction.

\begin{table}[H]
\begin{center}
\caption*{Page 0}
\begin{tabular}{ c | c c c c c c c }
$k$ & & & & & & & \\
\\
$2m$ & \{0\} & $\dots$ & $\dots$ & $e_{0,2m-1}^{2m}$ & $e_{0,2m}^{2m}$ & $\mathcal{U}_{2m+1}^{2m}$ & \\
\\
$2m-1$ & $\mathcal{C}^{2m-1}(\mathcal{G})$ & $\dots$ & $\dots$ & $e_{0,2m-1}^{2m-1}$ & $\mathcal{U}_{2m}^{2m-1}$ & & \\
\\
$2m-2$ & $\mathcal{C}^{2m-2}(\mathcal{G})$ & $\dots$ & $\dots$ & $\mathcal{U}_{2m-1}^{2m-2}$ & & & \\
\\
$\vdots$ & $\vdots$ & & $\udots$ & & & & \\
\\
$\vdots$ & $\vdots$ & $\udots$ & & & & & \\
\\
\hline
& $0$ & $\dots$ & $\dots$ & $2m-1$ & $2m$ & $2m+1$ & $l$
\end{tabular}
\end{center}
\end{table}

The coboundary map $d_{0,l}^k$ of the zeroth page maps upwards one step from $e_{0,l}^k$ to $e_{0,l}^{k+1}$. When acting on the representative in $\mathcal{U}_l^k \cap (\mathcal{U}_{l+1}^k)^\perp$, it can be thought of as a qubit vertex coboundary map, which adds a qubit vertex to the simplex.

We are now in a position to compute the first page $e_{1,l}^k$ of the spectral sequence, which is defined as $e_{1,l}^k = \ker{d_{0,l}^k} / \im{d_{0,l}^{k-1}}$. For the first column $l=0$, this simply gives us the cohomology of the qubit complex $e_{1,0}^k = H^k(\mathcal{G})$. These cohomology groups are zero except for $H^{2m-1}(\mathcal{G})$. Next let's look at the diagonal $l = k+1$. The claim is that these spaces are isomorphic to the bulk chainspaces $e_{1,k+1}^k = \mathcal{C}^k([\text{bulk}])$. The spaces $e_{1,k+1}^{k-1}$ are zero, so $\im{d_{0,k+1}^{k-1}} = \{0\}$, and $e_{1,k+1}^k$ are simply the cocycles $e_{1,k+1}^k = \ker{d_{0,k+1}^k}$. The simplices in $\mathcal{U}_{k+1}^k$ which are \emph{not} in $\mathcal{C}^k([\text{bulk}])$ do not vanish under $d_{0,k+1}^k$, and thus are not cocycles. It remains to look at $e_{1,l}^k$ for $1<l<k+1$. It turns out that these spaces are all zero.

\begin{claim} \label{clm_30}
$e_{1,l}^k = \{0\}$ for $1 \leq l \leq k$, $k = 1,\dots,2m$.
\end{claim}
\begin{proof}
We will split this up into two cases: $k \neq 2m-1$ and $k = 2m-1$. The first case $k \neq 2m-1$ will be easier, and the case $k = 2m-1$ will require an extra idea. Note the first case $k \neq 2m-1$ includes $k < 2m-1$ and $k = 2m$. The strategy in this proof is to apply \Cref{P_Q_lemma} to show that all states in the kernel of the relevant outgoing coboundary map are also in the image of the relevant incoming coboundary map.

Let's begin with the rows $k \neq 2m-1$. Recalling the interpretation $e_{0,l}^k = \mathcal{U}_l^k \cap (\mathcal{U}_{l+1}^k)^\perp$, let $|\alpha\rangle \in \mathcal{U}_l^k \cap (\mathcal{U}_{l+1}^k)^\perp$ with $|\alpha\rangle \in \ker{d_{0,l}^k}$. This tells us that $d^k |\alpha\rangle \in \mathcal{U}_{l+1}^{k+1}$. Restricted to the subcomplex $\Omega_l$, this says that $(d^k |\alpha\rangle)_{\mathcal{C}^{k+1}(\Omega_l)} = 0$. That is, $|\alpha\rangle$ is a \emph{cocycle} in $\Omega_l$. (Here and throughout, the subscript refers to restriction or orthogonal projection onto this subspace.) By considering the gadget construction from \Cref{single_gadget_sec}, $\Omega_l$ has no $k$-homology for $k \neq 2m-1$. ($\Omega_l$ \emph{does} in fact have $(2m-1)$-homology, so for this case we will need an extra trick. But for now, $k \neq 2m-1$ and $\Omega_l$ has no $k$-homology.) This means that $|\alpha\rangle$ is not only a \emph{cocycle} in $\Omega_l$ but necessarily also a \emph{coboundary}. Now apply \Cref{P_Q_lemma} with $\mathcal{P} = \Omega_l$ and $\mathcal{Q} = \Omega_{l-1}$. Note $|\alpha\rangle \in \mathcal{U}_l^k$ so indeed $|\alpha\rangle \perp \mathcal{C}^k(\Omega_{l-1})$. We get that $|\alpha\rangle = (d^{k-1} |\beta\rangle)_{\mathcal{C}^k(\Omega_l)}$ for some $|\beta\rangle \in \mathcal{C}^{k-1}(\Omega_l) \cap \mathcal{C}^{k-1}(\Omega_{l-1})^\perp$. But $\mathcal{C}^{k-1}(\Omega_l) \cap \mathcal{C}^{k-1}(\Omega_{l-1})^\perp = \mathcal{U}_l^{k-1} \cap (\mathcal{U}_{l-1}^{k-1})^\perp$, so $|\beta\rangle \in e_{0,l}^{k-1}$. Noting that $(d^{k-1} |\beta\rangle)_{\mathcal{C}^k(\Omega_l)}$ is precisely $(d^{k-1} |\beta\rangle)_{(\mathcal{U}_{l+1}^k)^\perp} = d_{0,l}^{k-1} |\beta\rangle$, we get that $|\alpha\rangle \in \im{d_{0,l}^{k-1}}$. The conclusion is that $\ker{d_{0,l}^k} = \im{d_{0,l}^k}$ and $e_{0,l}^k$ has no homology, so $e_{1,l}^k = \{0\}$.

Now we move onto the case $k = 2m-1$. As before, suppose $|\alpha\rangle \in \mathcal{U}_l^{2m-1} \cap (\mathcal{U}_{l+1}^{2m-1})^\perp$ with $|\alpha\rangle \in \ker{d_{0,l}^{2m-1}}$. Here, it is \emph{not} true that $\Omega_l$ has no $(2m-1)$-homology. However, since $|\alpha\rangle$ has no support on the qubit complex $\mathcal{G}$, the idea is that \emph{it cannot tell} that the complex has homology. We are able to imagine closing all the $(2m-1)$-cycles with fictitious auxiliary vertices, and \Cref{P_Q_lemma} will tell us that this move is not detected by $|\alpha\rangle$.

Let's formalize this. Recalling \Cref{single_gadget_sec}, $\mathcal{G} \subseteq \hat{\mathcal{G}}$ consists of $2^m$ copies of $S^{2m-1}$ corresponding to the $m$-bit strings. For each copy of $S^{2m-1}$ introduce an extra `auxiliary' vertex (of weight $1$) and connecting it with all vertices in this $S^{2m-1}$, thus `filling in the hole'. Denote the new objects after this operation with a star such as $\hat{\mathcal{G}}^\ast$. The effect of the auxiliary vertices is that $\hat{\mathcal{G}}^\ast$ has no $(2m-1)$-cohomology. This includes $\Omega_l^\ast$ for each $l$: these subcomplexes no longer have any $(2m-1)$-cohomology.

Continuing the argument, $|\alpha\rangle \in \mathcal{U}_l^{2m-1} \cap (\mathcal{U}_{l+1}^{2m-1})^\perp$ with $|\alpha\rangle \in \ker{d_{0,l}^{2m-1}}$, which tells us that $d^{2m-1} |\alpha\rangle \in \mathcal{U}_{l+1}^{2m}$ and, restricting to the subcomplex $\Omega_l$, $(d^{2m-1} |\alpha\rangle)_{\mathcal{C}^{2m}(\Omega_l)} = 0$. That is, $|\alpha\rangle$ is a \emph{cocycle} in $\Omega_l$. Now $|\alpha\rangle \in \mathcal{U}_l^{2m-1}$ so $|\alpha\rangle$ has no support on the qubit complex $\mathcal{G}$, $|\alpha\rangle_{\mathcal{G}} = 0$. Thus $|\alpha\rangle$ is likewise a \emph{cocycle} in $\Omega_l^\ast$. But $\Omega_l^\ast$ has no $(2m-1)$-homology, so $|\alpha\rangle$ is in fact a \emph{coboundary} in $\Omega_l^\ast$. Now apply a \Cref{P_Q_lemma} with $\mathcal{P} = \Omega_l^\ast$ and $\mathcal{Q} = \Omega_{l-1}^\ast$. We get that $|\alpha\rangle = (d^{2m-2} |\beta\rangle)_{\mathcal{C}^{2m-1}(\Omega_l^\ast)}$ for some $|\beta\rangle \in \mathcal{C}^{2m-2}(\Omega_l^\ast) \cap \mathcal{C}^{2m-2}(\Omega_{l-1}^\ast)^\perp$. But $\mathcal{C}^{2m-2}(\Omega_l^\ast) \cap \mathcal{C}^{2m-2}(\Omega_{l-1}^\ast)^\perp = \mathcal{U}_l^{2m-2} \cap (\mathcal{U}_{l-1}^{2m-2})^\perp$, so $|\beta\rangle \in e_{0,l}^{2m-2}$. Noting that $(d^{2m-2} |\beta\rangle)_{\mathcal{C}^k(\Omega_l^\ast)} = (d^{2m-2} |\beta\rangle)_{\mathcal{C}^k(\Omega_l)}$ is precisely $d_{0,l}^{2m-2} |\beta\rangle$, we get that $|\alpha\rangle \in \im{d_{0,l}^{2m-2}}$. The conclusion is that $\ker{d_{0,l}^{2m-1}} = \im{d_{0,l}^{2m-1}}$ and $e_{0,l}^{2m-1}$ has no homology, so $e_{1,l}^{2m-1} = \{0\}$.
\end{proof}

\begin{table}[H]
\begin{center}
\caption*{Page 1}
\begin{tabular}{ c | c c c c c c c }
$k$ & & & & & & & \\
\\
$2m$ & & & & & & $\mathcal{C}^{2m}([\text{bulk}])$ & \\
\\
$2m-1$ & $H^{2m-1}(\mathcal{G})$ & & & & $\mathcal{C}^{2m-1}([\text{bulk}])$ & & \\
\\
$2m-2$ & & & & $\mathcal{C}^{2m-2}([\text{bulk}])$ & & & \\
\\
$\vdots$ & & & $\udots$ & & & & \\
\\
$\vdots$ & & $\udots$ & & & & & \\
\\
\hline
& $0$ & $\dots$ & $\dots$ & $2m-1$ & $2m$ & $2m+1$ & $l$
\end{tabular}
\end{center}
\end{table}

The page 1 coboundary map $d_{1,l}^k$ maps diagonally upwards and to the right one step from $e_{0,l}^k$ to $e_{0,l}^{k+1}$. It can be thought of as a gadget vertex coboundary map, which adds a gadget vertex to the simplex. Acting on $\mathcal{C}^k([\text{bulk}])$, this has the same action as the regular coboundary map $\hat{d}^k$ of $\hat{\mathcal{G}}^k$.

Let's now calculate the second page of the spectral sequence, defined as $e_{2,l}^k = \ker{d_{1,l}^k} / \im{d_{1,l-1}^{k-1}}$. $H^{2m-1}(\mathcal{G})$ remains unchanged, since both the relevant coboundary maps are zero. Next we show that all terms on the diagonal $l=k+1$ vanish, except for $e_{2,2m+1}^{2m}$.

\begin{claim} \label{bulk_homology_clm}
$e_{2,k+1}^k = \{0\} \ \forall \ k = 0,\dots,2m-1$
\end{claim}
\begin{proof}
This proof will be very similar to that of \Cref{clm_30}, except that we will apply \Cref{P_Q_lemma} with $\mathcal{P}$ being the entire complex and $\mathcal{Q}$ being $[\text{non-bulk}]$. We will again split up into the two cases $k < 2m-1$ and $k = 2m-1$, with the case $k = 2m-1$ requiring the same extra idea. The relevant coboundary maps are now those of the first page $d_{1,k+1}^k$, and recall $e_{2,k+1}^k$ is defined to be the homology of the chain
\begin{equation*}
\begin{tikzcd}
\mathcal{C}^{k-1}([\text{bulk}]) \arrow[r,"d_{1,k}^{k-1}"] & \mathcal{C}^k([\text{bulk}]) \arrow[r,"d_{1,k+1}^k"] & \mathcal{C}^{k+1}([\text{bulk}])
\end{tikzcd}
\end{equation*}

Let's again begin with $k < 2m-1$. Let $|\alpha\rangle \in \mathcal{C}^k([\text{bulk}])$ with $|\alpha\rangle \in \ker{d_{1,k+1}^k}$. Acting on $\mathcal{C}^k([\text{bulk}])$, $d_{1,k+1}^k$ simply acts as the original coboundary map $d^k$, so in fact $|\alpha\rangle \in \ker{d^k}$ is a \emph{cocycle}. But the complex has no $k$-homology for $k \neq 2m-1$, so $|\alpha\rangle$ must also be a \emph{coboundary}. Now apply \Cref{P_Q_lemma} with $\mathcal{P} = \hat{\mathcal{G}}$ and $\mathcal{Q} = [\text{non-bulk}]$. We get that $|\alpha\rangle \in d^{k-1}(\mathcal{C}^{k-1}([\text{non-bulk}])^\perp) = d_{1,k}^{k-1}(\mathcal{C}^{k-1}([\text{bulk}]))$. We conclude that $\ker{d_{1,k+1}^k} = \im{d_{1,k}^{k-1}}$ and $e_{2,k+1}^k = \{0\}$.

Now we move onto the case $k = 2m-1$. Recalling \Cref{single_gadget_sec}, $\mathcal{G} \subseteq \hat{\mathcal{G}}$ consists of $2^m$ copies of $S^{2m-1}$ corresponding to the $m$-bit strings. For each copy of $S^{2m-1}$ introduce an extra auxiliary vertex (of weight $1$) and connecting it with all vertices in this $S^{2m-1}$, thus filling in the hole. Denote the new objects after this operation with a star such as $\hat{\mathcal{G}}^\ast$. The effect of the auxiliary vertices is that $\hat{\mathcal{G}}^\ast$ has no $(2m-1)$-homology.

As before, suppose $|\alpha\rangle \in \mathcal{C}^{2m-1}([\text{bulk}])$ with $|\alpha\rangle \in \ker{d_{1,2m}^{2m-1}} = \ker{d^{2m-1}}$. Now $|\alpha\rangle$ has no support on the qubit complex $\mathcal{G}$, thus $|\alpha\rangle$ is likewise a cocycle in $\hat{\mathcal{G}}^\ast$. But $\hat{\mathcal{G}}^\ast$ has no $(2m-1)$-homology, so $|\alpha\rangle$ is in fact a \emph{coboundary} in $\hat{\mathcal{G}}^\ast$. Now apply \Cref{P_Q_lemma} with $\mathcal{P} = \hat{\mathcal{G}}^\ast$ and $\mathcal{Q} = [\text{non-bulk}]^\ast$. We get that $|\alpha\rangle \in d^{2m-2}(\mathcal{C}^{2m-2}([\text{non-bulk}]^\ast)^\perp)$. But $\mathcal{C}^{2m-2}([\text{non-bulk}]^\ast)^\perp$ is simply $\mathcal{C}^{2m-2}([\text{bulk}])$, so $|\alpha\rangle \in d^{2m-2}(\mathcal{C}^{2m-2}([\text{bulk}])) = d_{1,2m-1}^{2m-2}(\mathcal{C}^{2m-2}([\text{bulk}]))$. We conclude that $\ker{d_{1,2m}^{2m-1}} = \im{d_{1,2m-1}^{2m-2}}$ and $e_{2,2m}^{2m-1} = \{0\}$.

The interpretation of this move is the same as in the proof of \Cref{clm_30}. Since $|\alpha\rangle$ has no support on the qubit complex $\mathcal{G}$, \emph{it cannot tell} that the complex has some $(2m-1)$-homology. We are able to imagine closing all the $(2m-1)$-cycles with fictitious auxiliary vertices, and \Cref{P_Q_lemma} will tell us that this is not detected by $|\alpha\rangle$.
\end{proof}

Finally let's investigate $e_{2,2m+1}^{2m}$. The entire space $\mathcal{C}^{2m}([\text{bulk}])$ is in $\ker d_{1,2m+1}^{2m}$ since $e_{1,2m+2}^{2m+1} = \{0\}$, so we have $e_{2,2m+1}^{2m} = \mathcal{C}^{2m}([\text{bulk}]) / \im{d_{1,2m}^{2m-1}}$. As usual, we can pick the representative orthogonal to the space we are quotienting and write $e_{2,2m+1}^{2m} = \mathcal{C}^{2m}([\text{bulk}]) \cap (\im{d_{1,2m}^{2m-1}})^\perp$. But $(\im{d_{1,2m}^{2m-1}})^\perp = \ker{\partial_{1,2m+1}^{2m}}$ where $\partial_{1,l}^k = (d_{1,l-1}^{k-1})^\dag$ can be interpreted as a gadget vertex boundary map which removes a gadget vertex from the simplex. So $e_{2,2m+1}^{2m} = \mathcal{C}^{2m}([\text{bulk}]) \cap \ker{\partial_{1,2m+1}^{2m}}$. Since $[\text{bulk}]$ is isomorphic to the interior of a $2m$-ball, the only state in $\mathcal{C}^{2m}([\text{bulk}])$ which has no $\partial_{1,2m+1}^{2m}$-boundary is the uniform superposition over all the $2m$ simplices in $[\text{bulk}]^{2m}$, with appropriate orientations. Let's denote this state by $|\text{core}\rangle \in \mathcal{C}^{2m}([\text{bulk}])$.

\begin{table}[H]
\begin{center}
\caption*{Page 2}
\begin{tabular}{ c | c c c c c c }
$k$ & & & & & & \\
\\
$2m$ & & & & & $\{|\text{core}\rangle\}$ & \\
\\
$2m-1$ & $H^{2m-1}(\mathcal{G})$ & & & & & \\
\\
$\vdots$ & & & & & & \\
\\
$\vdots$ & & & & & & \\
\\
\hline
& $0$ & $\dots$ & $\dots$ & $\dots$ & $2m+1$ & $l$
\end{tabular}
\end{center}
\end{table}

Pages $3$ to $2m+1$ will remain unchanged. This can be seen purely from the direction in which the coboundaries $d_{j,l}^k$ map. $d_{j,l}^k$ moves one step up and $j$ steps to the right, so these must necessarily be zero maps.
\begin{equation*}
d_{j,l}^k : e_{j,l}^k \rightarrow e_{j,l+j}^{k+1}
\end{equation*}
Since all relevant coboundary maps are zero, taking the homology leaves the page unchanged.

This argument no longer applies when we look at page $2m+2$. Now, the coboundary map of page $2m+1$ $d_{2m+1,l}^k$ maps one step up and $2m+1$ steps to the right. In particular, $d_{2m+1,0}^1$ maps
\begin{equation*}
d_{2m+1,0}^1 : H^{2m-1}(\mathcal{G}) \rightarrow \{|\text{core}\rangle\}
\end{equation*}
We will argue that in fact (ignoring normalizations)
\begin{equation*}
d_{2m+1,0}^1 \ |\phi\rangle \ = \ |\text{core}\rangle
\end{equation*}
where $|\phi\rangle \in H^{2m-1}(\mathcal{G})$ is the cycle being filled by the gadget $\hat{\mathcal{G}}$. Thus $|\phi\rangle$ is lifted out of the homology at this page. To see this, we need to examine the coboundary map $d_{2m+1,0}^1$ more closely. Returning to the zeroth page $e_{0,l}^k$, denote by $\partial_{\text{qubit},l}^k = \partial_{0,l}^k$ the `qubit boundary map'
\begin{equation*}
\partial_{\text{qubit},l}^k : e_{0,l}^k \rightarrow e_{0,l}^{k-1}
\end{equation*}
which acts by removing qubit vertices. Further, denote by $d_{\text{gadget},l}^k$ the `gadget coboundary map'
\begin{equation*}
d_{\text{gadget},l}^k : e_{0,l}^k \rightarrow e_{0,l+1}^{k+1}
\end{equation*}
which acts by adding gadget vertices. Now we can think of the map $d_{2m+1,0}^1$ as acting by
\begin{equation*}
d_{2m+1,0}^1 = d_{\text{gadget},2m}^{2m-1} \circ \partial_{\text{qubit},2m}^{2m} \circ \dots \circ d_{\text{gadget},1}^{1} \circ \partial_{\text{qubit},1}^2 \circ d_{\text{gadget},0}^1
\end{equation*}
With this new understanding, we can see that $d_{2m+1,0}^1 |\phi\rangle = |\text{core}\rangle$. Recalling \Cref{single_gadget_sec}, the cycle $|\phi\rangle$ on the outside layer gets transported through the thickening of the gadget to the inside layer, at which point it gets sent to $|\text{core}\rangle$ by the final gadget vertex coboundary map $d_{\text{gadget},2m}^{2m-1}$.

\begin{table}[H]
\begin{center}
\caption*{Page 2}
\begin{tabular}{ c | c c c c c c }
$k$ & & & & & & \\
\\
$2m$ & & & & & & \\
\\
$2m-1$ & $\{|\psi\rangle \in H^{2m-1}(\mathcal{G}) : \langle\phi|\psi\rangle = 0\}$ & & & & & \\
\\
$\vdots$ & & & & & & \\
$\vdots$ & & & & & & \\
\\
\hline
& $0$ & $\dots$ & $\dots$ & $\dots$ & $2m+1$ & $l$
\end{tabular}
\end{center}
\end{table}

\subsubsection{Finishing the proof}

By Theorem 4 of \cite{forman1994hodge}, the eigenvalues of the Laplacian decay like $\Theta(\lambda^{2j})$ for some $j$. Thus if we take the eigenspace of eigenvalues $\mathcal{O}(\lambda^{2j})$ and intersect it with the orthogonal complement of the eigenspace of eigenvalues $\mathcal{O}(\lambda^{2j+2})$, we get the eigenspace of eigenvalues that are exactly $\Theta(\lambda^{2j})$.

\Cref{O_lambda_perturbation_clm} tells us that the Laplacian eigenspace of eigenvalues $\mathcal{O}(\lambda^{2j})$ is a $\mathcal{O}(\lambda)$-perturbation of $E_j^{2m-1}$, in the sense of \Cref{subspace_perturbation_sec}. Thus the Laplacian eigenspace of eigenvalues $\Theta(\lambda^{2j})$ is a $\mathcal{O}(\lambda)$-perturbation of $E_j^{2m-1} \cap (E_{j+1}^{2m-1})^\perp$.

Thus to complete the proof of \Cref{spec_seq_lemma}, we use the isomorphism in \Cref{thm_7} to derive the spaces $E_j^k$.
Recall
\begin{equation*}
E_j^k = \bigoplus_l E_{j,l}^k
\end{equation*}
and by \Cref{thm_7} we can get $E_{j,l}^k$ by taking $e_{j,l}^k$ and projecting a representative from each equivalence class onto $\mathcal{U}_l^k$. This gives:
\begin{itemize}
    \item $E_0^{2m-1} = \mathcal{C}^{2m-1}(\hat{\mathcal{G}})$, with $E_{0,l}^{2m-1} = \mathcal{U}_l^{2m-1} \cap (\mathcal{U}_l^{2m-1})^\perp$ for each $l$. (These are tautological.)
    \item $E_1^{2m-1} = H^{2m-1}(\mathcal{G}) \oplus \mathcal{C}^{2m-1}([\text{bulk}])$, with $E_{1,0}^{2m-1} = H^{2m-1}(\mathcal{G})$ and $E_{1,2m}^{2m-1} = \mathcal{C}^{2m-1}([\text{bulk}])$.
    \item $E_j^{2m-1} = E_{j,0}^{2m-1} = H^{2m-1}(\mathcal{G})$ for all $j= 2,\dots,2m+1$.
    \item $E_j^{2m-1} = E_{j,0}^{2m-1} = \{|\psi\rangle \in H^{2m-1}(\mathcal{G}) : \langle\phi|\psi\rangle = 0\}$ for all $j\geq 2m+2$.
\end{itemize}

(Recall that $H^{2m-1}(\mathcal{G}) \cong \ker{\Delta^{2m-1}}$ where $\Delta^k$ is the Laplacian of the qubit complex $\mathcal{G}^k$. When the space $H^{2m-1}(\mathcal{G})$ appears above, it is the harmonic representative from $\ker{\Delta^k} \subseteq \mathcal{C}^k(\mathcal{G}) \subseteq \mathcal{C}^k(\hat{\mathcal{G}})$ which is present.)

\pagebreak
\section{Combining gadgets}\label{sec:combining gadgets full}

In this section, we describe how to combine many gadgets together to simulate a local Hamiltonian. We prove our main theorem, \Cref{main_thm}.

Recall from \Cref{single_gadget_sec} $\hat{\mathcal{G}}_m$ denotes the complex of a single gadget. We must be careful to keep track of what we think of as the original qubit complex and the additional gadget complex. The \emph{vertices} $\hat{\mathcal{G}}_m^0$ can be partitioned into the vertices of the original qubit graph $\mathcal{G}_m^0$, which have weight 1, and the added gadget vertices, which have weight $\lambda$. Let these added gadget vertices be denoted $\mathcal{T}^0$. The \emph{$k$-simplices} $\hat{\mathcal{G}}_m^k$ can be partitioned into the $k$-simplices of the original qubit complex $\mathcal{G}_m^k$, and the $k$-simplices which contain at least one vertex from $\mathcal{T}^0$. Let these extra $k$-simplices be denoted $\mathcal{T}^k$. Note, however, that $\mathcal{T}$ is \emph{not} a simplicial complex in its own right, since $\mathcal{T}^k$ involve vertices outside of $\mathcal{T}^0$. We can decompose the chain space of $\hat{\mathcal{G}}_m$ as a direct sum
\begin{equation*}
\mathcal{C}^k(\hat{\mathcal{G}}_m) = \mathcal{C}^k(\mathcal{G}_m) \oplus \mathcal{C}^k(\mathcal{T})
\end{equation*}

In \Cref{spec_seq_sec} we used the powerful tool of spectral sequences to understand the spectrum of the Laplacian of a single gadget $\hat{\mathcal{G}}_m$, captured by \Cref{spec_seq_lemma}. This lemma will be an essential ingredient later when we come to analyze the spectrum of many gadgets combined.

\subsection{Padding with identity} \label{padding_sec}

In order to add the remaining $n-m$ qubits, we join the graph $\mathcal{G}_{n-m}$ to get $\hat{\mathcal{G}}_n = \hat{\mathcal{G}}_m \ast \mathcal{G}_{n-m}$. This is analogous to tensoring a Hamiltonian term with identity on all the qubits outside of its support.

There is another way to look at the final complex $\hat{\mathcal{G}}_n$. To get $\hat{\mathcal{G}}_n$, we implement the gadget $\mathcal{T}$ described in Sections \ref{single_gadget_sec}, \ref{constructing_cycles_sec} on the copies of $\mathcal{G}_1$ corresponding to the qubits on which $\phi$ is supported, and then connecting the vertices of the gadget $\mathcal{T}^0$ \emph{all to all} with the qubit vertices in the copies of $\mathcal{G}_1$ corresponding to qubits outside the support of $\phi$.

\begin{lemma} \label{spec_seq_lemma_padded}
Let $\hat{\Delta}'^{k}$ be the Laplacian of $\hat{G}_n$.
\begin{itemize}
    \item $\hat{\Delta}'^{2n-1}$ has a $(2^m-1) \cdot 2^{n-m}$-dimensional kernel, which is a $\mathcal{O}(\lambda)$-perturbation of the subspace $\{|\psi\rangle \in \mathcal{H}_m : \langle\phi|\psi\rangle = 0\}$ tensored with $\mathcal{H}_{n-m}$.
    \item The first excited eigenspace of $\hat{\Delta}'^{2n-1}$ above the kernel is the $2^{n-m}$-dimensional space $|\hat{\phi}\rangle \otimes \mathcal{H}_{n-m}$, where $|\hat{\phi}\rangle$ is a $\mathcal{O}(\lambda)$-perturbation of $|\phi\rangle \in \mathcal{H}_m$, and it has energy $\Theta(\lambda^{4m+2})$.
    \item The next lowest eigenvectors have eigenvalues $\Theta(\lambda^2)$, and they are $\mathcal{O}(\lambda)$-perturbations of sums of $(2m-1)$-simplices touching the central vertex $v_0$, tensored with $\mathcal{H}_{n-m}$.
    \item The rest of the eigenvalues are $\Theta(1)$.
\end{itemize}
\end{lemma}
\begin{proof}
By \Cref{Kunneth_lem}, the new chainspace is
\begin{equation*}
\mathcal{C}^{2n-1}(\hat{\mathcal{G}}_n) = \mathcal{C}^{2m-1}(\hat{\mathcal{G}}_m) \otimes \mathcal{C}^{2(n-m)+1}(\mathcal{G}_{n-m})
\end{equation*}
By \Cref{Laplacian_join_lem}, the new Laplacian on $\hat{\mathcal{G}}_n$ is
\begin{equation*}
\hat{\Delta}'^{2n-1} = \hat{\Delta}^{2m-1} \otimes \mathbbm{1} + \mathbbm{1} \otimes \Delta^{2(n-m)+1}
\end{equation*}
where $\Delta^{2(n-m)+1}$ here is the Laplacian on $\mathcal{C}^{2(n-m)+1}(\mathcal{G}_{n-m})$. Notice that the two terms on the right-hand-side commute. $\Delta^{2(n-m)+1}$ has kernel $\mathcal{H}_{n-m}$ and its first excited eigenvalue is $\Theta(1)$. The conclusions follow from \Cref{spec_seq_lemma}.
\end{proof}

Define the new gadget simplices of $\hat{\mathcal{G}}_n$ to be $\mathcal{T}' = \mathcal{T} \ast \mathcal{G}_{n-m}$, and we have
\begin{equation*}
\mathcal{C}^{2n-1}(\mathcal{T}') = \mathcal{C}^{2m-1}(\mathcal{T}) \otimes \mathcal{C}^{2(n-m)-1}(\mathcal{G}_{n-m})
\end{equation*}

\subsection{Some facts about a single gadget} \label{facts_sec}

We now take the opportunity to prove some useful facts about the complex with a single gadget added. Suppose we implement the gadget corresponding to integer state $|\phi\rangle$, and let $|\hat{\phi}\rangle$ be as in \Cref{spec_seq_lemma_padded}. First we show that states in the subspace $|\hat{\phi}\rangle \otimes \mathcal{H}_{n-m}$ are \emph{cycles}. In other words, these states are \emph{paired up} in the language of \Cref{pairing_sec}.

\begin{claim} \label{paired_up_clm}
$\hat{\partial}^{2n-1} (|\hat{\phi}\rangle \otimes \mathcal{H}_{n-m}) = 0$
\end{claim}
\begin{proof}
$\mathcal{H}_{n-m}$ consists of cycles, so it is sufficient to show $\hat{\partial}^{2m-1} |\hat{\phi}\rangle = 0$. $|\hat{\phi}\rangle$ is an eigenstate of $\hat{\Delta}^{2m-1}$, thus by \Cref{pairing_prop} it must be paired up $|\hat{\phi}\rangle \in \im{\hat{\partial}^{2m}}$ or paired down $|\hat{\phi}\rangle \in \im{\hat{d}^{2m-2}}$. By \Cref{spec_seq_lemma}, $|\hat{\phi}\rangle$ is a $\mathcal{O}(\lambda)$ perturbation of $|\phi\rangle \in \mathcal{H}_m$. $|\phi\rangle$ is a cycle $\hat{\partial}^{2m-1} |\phi\rangle = 0$, so $|\phi\rangle$ is orthogonal to $\im{\hat{d}^{2m-2}}$. This guarantees that the $\mathcal{O}(\lambda)$-perturbation $|\hat{\phi}\rangle = |\phi\rangle + \mathcal{O}(\lambda)$ is \emph{not} contained in $\im{\hat{d}^{2m-2}}$. Thus $|\hat{\phi}\rangle$ is paired up and $\hat{\partial}^{2m-1} |\hat{\phi}\rangle = 0$.
\end{proof}

Next we will state some facts about the \emph{bulk} of the gadget. Let $[\text{bulk}]$ denote the simplices touching the central vertex $v_0$ of gadget $\mathcal{T}$ (see \Cref{single_gadget_sec}), with chainspaces $\mathcal{C}^k([\text{bulk}])$, and let $\Pi^{[k]}$ be the projection onto $\mathcal{C}^k([\text{bulk}])$.

\begin{claim} \label{lambda_clm}
All states $|\psi\rangle$ have $||\Pi^{[2n-2]} \hat{\partial}^{2n-1} |\psi\rangle|| = \mathcal{O}(\lambda) |||\psi\rangle||$ and $||\Pi^{[2n]} \hat{d}^{2n-1} |\psi\rangle|| = \mathcal{O}(\lambda) |||\psi\rangle||$.
\end{claim}
\begin{proof}
All vertices touching $[\text{bulk}]$ have weight $\lambda$. Equations \ref{coboundary_entries_eq}, \ref{boundary_entries_eq} then give the conclusion.
\end{proof}

The proof of the following claim relies on a fact from \Cref{spec_seq_sec}.

\begin{claim} \label{bulk_clm}
A normalized state $|\psi\rangle \in \mathcal{C}^{2n-1}([\text{\emph{bulk}}])$ has $||\Pi^{[2n-2]} \hat{\partial}^{2n-1} |\psi\rangle|| = \Omega(\lambda)$ or $||\Pi^{[2n]} \hat{d}^{2n-1} |\psi\rangle|| = \Omega(\lambda)$.
\end{claim}
\begin{proof}
Since all vertices relevant to these claims are weighted by $\lambda$, we can consider the \emph{unweighted} complex. Then it is sufficient to show that either $||\Pi^{[2n-2]} \hat{\partial}^{2n-1} |\psi\rangle|| \neq 0$ or $||\Pi^{[2n]} \hat{d}^{2n-1} |\psi\rangle|| \neq 0$. Returning to the weighted complex simply introduces a factor of $\lambda$.

Suppose $||\Pi^{[2n]} \hat{d}^{2n-1} |\psi\rangle|| = 0$. That is, $\Pi^{[2n]} \hat{d}^{2n-1} |\psi\rangle = 0$. This in fact tells us that $\hat{d}^{2n-1} |\psi\rangle = 0$, since $\hat{d}^{2n-1}$ maps $\mathcal{C}^{2n-1}([\text{bulk}])$ into $\mathcal{C}^{2n}([\text{bulk}])$.

\Cref{bulk_homology_clm} from \Cref{spec_seq_sec} is telling us that $\mathcal{C}^{2n-1}([\text{bulk}]) \cap \ker{\hat{d}^{2n-1}} = \hat{d}^{2n-2}(\mathcal{C}^{2n-2}([\text{bulk}]))$. Note the proof of \Cref{bulk_homology_clm} does \emph{not} rely on any spectral sequence machinery, but only \Cref{P_Q_lemma} and the argument in the proof of \Cref{bulk_homology_clm} where we close the homology with auxiliary vertices.

The result is that, from $||\Pi^{[2n]} \hat{d}^{2n-1} |\psi\rangle|| = 0$, we can deduce that $|\psi\rangle \in \hat{d}^{2n-2}(\mathcal{C}^{2n-2}([\text{bulk}]))$. Remembering $\hat{\partial}^{2n-1} = (\hat{d}^{2n-2})^\dag$, this guarantees that $\hat{\partial}^{2n-1} |\psi\rangle$ has some component in $\mathcal{C}^{2n-2}([\text{bulk}])$ and $||\Pi^{[2n-2]} \hat{\partial}^{2n-1} |\psi\rangle|| \neq 0$.
\end{proof}

\subsection{Combining gadgets} \label{combine_gadget_sec}

We have seen how to construct a gadget $\mathcal{T}$ to implement a single local rank-1 projector $\phi=|\phi\rangle\langle\phi|$ where $|\phi\rangle$ is an integer state. We would now like to implement a Hamiltonian
\begin{equation*}
H = \sum_{i=1}^t \phi_i
\end{equation*}
which is a sum of such terms. This will involve adding a gadget $\mathcal{T}_i$ for each term $\phi_i$. We would like to add these gadgets in an independent way. From \Cref{padding_sec}, we have a way of gluing in a gadget $\mathcal{T}'_i$ onto $\mathcal{G}_n$ which corresponds to $\phi_i$. To implement multiple terms $\phi_i$, we simply glue in the gadgets separately. One may wonder if we should include any edges between the gadget vertices if different gadgets; we do \emph{not} include any of these edges.

The full procedure is described as follows, which includes the padding-with-identity step from \Cref{padding_sec}:
\begin{enumerate}
    \item For each $i$, add gadget $\mathcal{T}_i$ to the copies of $\mathcal{G}_1$ corresponding to qubits in the support of $|\phi_i\rangle$.
    \item For each $i$, connect gadget vertices $\mathcal{T}_i^0$ \emph{all to all} with the qubit vertices in the copies of $\mathcal{G}_1$ corresponding to qubits outside the support of $|\phi_i\rangle$ to get $\mathcal{T}'_i$.
    \item Do \emph{not} connect any gadget vertices from different gadgets $\mathcal{T}_i^0 \leftrightarrow \mathcal{T}_j^0$.
\end{enumerate}

From now on, we drop the primes on $\mathcal{T}_i'$, $\hat{\Delta}_i'^k$. Let the final weighted graph after adding many gadgets in this way be denoted $\hat{\mathcal{G}}_n$, with Laplacian $\hat{\Delta}^k$. 

\begin{fact}
After adding many gadgets with this procedure, the total chain space can be decomposed as
\begin{equation*}
\mathcal{C}^{2n-1}(\hat{\mathcal{G}}_n) = \mathcal{C}^{2n-1}(\mathcal{G}_n) \oplus \mathcal{C}^{2n-1}(\mathcal{T}_1) \oplus \dots \oplus \mathcal{C}^{2n-1}(\mathcal{T}_t)
\end{equation*}
\end{fact}

The following claim says that the \emph{up} Laplacian respects this decomposition.

\begin{claim} \label{up_Laplacian_clm}
When acting on the entire chainspace $\mathcal{C}^{2n-1}(\hat{\mathcal{G}}_n)$, we can write the \emph{up} Laplacian $\hat{\Delta}^{\uparrow 2n-1}$ of the entire complex $\hat{\mathcal{G}}_n$ as the sum of the up Laplacians $\hat{\Delta}_i^{\uparrow 2n-1}$ of the individual gadgets $\mathcal{C}^{2n-1}(\mathcal{G}_n) \oplus \mathcal{C}^{2n-1}(\mathcal{T}_i)$.
\begin{equation*}
\hat{\Delta}^{\uparrow 2n-1} = \hat{\Delta}^{\uparrow 2n-1}_1 + \dots + \hat{\Delta}^{\uparrow 2n-1}_t
\end{equation*}
\end{claim}

\begin{proof}
It is sufficient to check
\begin{equation*}
\langle\psi|\hat{\Delta}^{\uparrow 2n-1}|\psi\rangle = \langle\psi_1|\hat{\Delta}^{\uparrow 2n-1}_1|\psi_1\rangle + \dots + \langle\psi_t|\hat{\Delta}^{\uparrow 2n-1}_t|\psi_t\rangle
\end{equation*}
for all states $|\psi\rangle \in \mathcal{C}^{2n-1}(\hat{\mathcal{G}}_n)$, where $|\psi_i\rangle$ is the component in $\mathcal{C}^{2n-1}(\mathcal{G}_n) \oplus \mathcal{C}^{2n-1}(\mathcal{T}_i)$. But
\begin{align*}
\langle\psi|\hat{\Delta}^{\uparrow 2n-1}|\psi\rangle &= ||\hat{d}^{2n-1}|\psi\rangle||^2 \\
&= ||\hat{d}_1^{2n-1}|\psi_1\rangle||^2 + \dots + ||\hat{d}_t^{2n-1}|\psi_t\rangle||^2 \\
&= \langle\psi_1|\hat{\Delta}^{\uparrow 2n-1}_1|\psi_1\rangle + \dots + \langle\psi_t|\hat{\Delta}^{\uparrow 2n-1}_t|\psi_t\rangle
\end{align*}
where in the second line, we used that the different gadgets do not share any $2n$-simplices. That is, there are no $2n$-simplices which contain vertices from more than one gadget.
\end{proof}

The hope is that, after this procedure to combine the gadgets, they will implement a version of the Hamiltonian $H$ on the simulated qubit subspace $\mathcal{H}_n$. This will be reflected in our main theorem \Cref{main_thm}, which states that a $1/\poly(n)$ lower bound on the spectrum of $H$ results in a $1/\poly(n)$ lower bound on the spectrum of the Laplacian $\hat{\Delta}^{2n-1}$.

The high level overview of the proof of \Cref{main_thm} is as follows. Assume that $\lambda_{\min}(H) \geq 1/\poly(n)$. We will argue by contradiction. Suppose there is a state $|\varphi\rangle$ of extremely low energy on $\hat{\Delta}^{2n-1}$. Consider a single gadget $\mathcal{T}_i$. $|\varphi\rangle$ is forced to have small overlap with the excited eigenspaces of the single-gadget Laplacian $\hat{\Delta}_i^{2n-1}$. Thus the restriction of $|\varphi\rangle$ to $\mathcal{C}^{2n-1}(\hat{\mathcal{G}}_n) \oplus \mathcal{C}^{2n-1}(\mathcal{T}_i)$ must lie close to $\ker{\hat{\Delta}_i^{2n-1}}$. \Cref{spec_seq_lemma_padded} tells us that $\ker{\hat{\Delta}_i^{2n-1}}$ is close to the space in the simulated qubit subspace $\mathcal{H}_n$ which is \emph{orthogonal to} the state $|\phi_i\rangle$ which is being filled in by gadget $\mathcal{T}_i$. Thus the restriction of $|\varphi\rangle$ to $\mathcal{C}^{2n-1}(\hat{\mathcal{G}}_n)$ is close to the subspace of $\mathcal{H}_n$ orthogonal to $|\phi_i\rangle$. This must hold for each gadget $\mathcal{T}_i$, so the restriction of $|\varphi\rangle$ to $\mathcal{C}^{2n-1}(\hat{\mathcal{G}}_n)$ must be in $\mathcal{H}_n$ and simultaneously orthogonal to all states $\{|\phi_i\rangle\}$. But the $1/\poly(n)$ spectrum lower bound on $H = \sum_i \phi_i$ forbids this.

\pagebreak
\begin{theorem} \label{main_thm} \emph{(Main theorem, formal)}
Suppose $H$ is a $m$-local Hamiltonian on $n$ qubits with $t$ terms, where each term is a rank-1 projector onto an integer state. Starting from $H$, let $\hat{\mathcal{G}}_n$ be the weighted graph described above, with Laplacian $\hat{\Delta}^k$. For any $g > 0$, there is a constant $c > 0$ sufficiently small such that setting
\begin{align*}
\lambda \ &= \ c t^{-1} g \\
E \ &= \ c \lambda^{4m+2} t^{-1} g
\end{align*}
gives
\begin{align*}
\lambda_{\min}(H) = 0 \ &\implies \ \lambda_{\min}(\hat{\Delta}^{2n-1}) = 0 \\
\lambda_{\min}(H) \geq g \ &\implies \ \lambda_{\min}(\hat{\Delta}^{2n-1}) \geq E
\end{align*}
\end{theorem}

\bigskip
\begin{proof}
First we show that $\lambda_{\min}(H) = 0 \ \implies \ \lambda_{\min}(\hat{\Delta}^{2n-1}) = 0$. We first argue that the complex $\hat{\mathcal{G}}_n$ has some $(2n-1)$-homology, and then we invoke \Cref{Hodge_prop}.

Recall that the clique complex of the graph $\mathcal{G}_n$ (which is the $n$-fold join of the initial qubit graph $\mathcal{G}_1$) has a homology group with rank $2^n$. 
This graph corresponds to the zero Hamiltonian -- every state in the Hilbert space of $n$ qubits is a zero-energy ground state.
When we add the gadget for the term $\ketbra{\phi_i}{\phi_i}$ to the graph $\mathcal{G}_n$, we fill in the cycles $\ket{\phi_i} \otimes \mathbb{C}^{\otimes (n-m)}$ by rendering them the boundaries of some $2n$-dimensional objects.
Crucially, we do not fill in any cycles other than $\ket{\phi_i} \otimes \mathbb{C}^{\otimes (n-m)}$ by adding this gadget.
Thus any state $\ket{\psi}$ which is orthogonal to $\text{span}\{|\phi_i\rangle\}$ will still give an element of homology. 
The state $|\psi\rangle$ satisfying $H|\psi\rangle = 0$ corresponds to a cycle $|\psi\rangle \in \mathcal{H}_n$ which is orthogonal to $\text{span}\{|\phi_i\rangle\}$. Thus $\hat{\mathcal{G}}_n$ has non-trivial $(2n-1)$-homology $H^{2n-1}(\hat{\mathcal{G}}_n) \neq \{0\}$. Invoking \Cref{Hodge_prop} tells us that $\ker{\hat{\Delta}^{2n-1}} \neq \{0\}$. That is, $\lambda_{\min}(\hat{\Delta}^{2n-1}) = 0$.

\bigskip
It remains to tackle the case $\lambda_{\min}(H) \geq g$. We begin with some notation.
\begin{definition} \label{proof_def}
\begin{itemize}
\item Let $\hat{\partial}^k$, $\hat{d}^k$, $\hat{\Delta}^k$ be the boundary, coboundary maps and Laplacian on the entire complex with all gadgets $\mathcal{C}^k(\hat{\mathcal{G}}_n)$, with $\partial^k$, $d^k$, $\Delta^k$ still reserved for the original qubit complex $\mathcal{C}^k(\mathcal{G}_n)$.

\item Let $\hat{\partial}_i^k$, $\hat{d}_i^k$, $\hat{\Delta}_i^k$ be the boundary, coboundary maps and Laplacian on the single gadget complex $\mathcal{C}^k(\mathcal{G}_n) \oplus \mathcal{C}^k(\mathcal{T}_i)$.

\item We can decompose the chainspaces of the entire complex as
\begin{equation*} 
\mathcal{C}^k(\hat{\mathcal{G}}_n) = \mathcal{C}^k(\mathcal{G}_n) \oplus \mathcal{C}^k(\mathcal{T}_1) \oplus \dots \oplus \mathcal{C}^k(\mathcal{T}_t)
\end{equation*}
Let $\Pi_0^k$ be projection onto $\mathcal{C}^k(\mathcal{G}_n)$ and $\Pi_i^k$ projection onto $\mathcal{C}^k(\mathcal{T}_i)$. These are a complete set of projectors on $\mathcal{C}^k(\hat{\mathcal{G}})$:
\begin{equation} \label{complete_projectors_eq}
\Pi_0^k + \sum_i \Pi_i^k = \text{id}
\end{equation}

\item Let $[\text{\emph{bulk}}]_i$ denote the simplices touching the central vertex $v_0$ of gadget $\mathcal{T}_i$ (see \Cref{single_gadget_sec}), with chainspaces $\mathcal{C}^k([\text{\emph{bulk}}])$. Let $[\text{\emph{bulk}}] = \sqcup_i [\text{\emph{bulk}}]_i$, so $\mathcal{C}^k([\text{\emph{bulk}}]) = \bigoplus_i \mathcal{C}^k([\text{\emph{bulk}}]_i)$. Let $\Pi^{[k]}_i$ be the projection onto $\mathcal{C}^k([\text{\emph{bulk}}]_i)$, and $\Pi^{[k]} = \bigoplus_i \Pi^{[k]}_i$ projection onto $\mathcal{C}^k([\text{\emph{bulk}}])$.

\item \Cref{spec_seq_lemma_padded} gives us an orthogonal decomposition of $\mathcal{C}^{2n-1}(\mathcal{G}_n) \oplus \mathcal{C}^{2n-1}(\mathcal{T}_i)$ into four spaces, for each gadget $\mathcal{T}_i$.
\begin{enumerate}
    \item The space of eigenvectors of eigenvalues $\Theta(1)$ - call it $\mathcal{A}_i$. Let $\Pi^{(\mathcal{A})}_i$ be the projection onto $\mathcal{A}_i$.
    \item The space of eigenvectors of eigenvalues $\Theta(\lambda^2)$ - call it $\mathcal{B}_i$. Let $\Pi^{(\mathcal{B})}_i$ be the projection onto $\mathcal{B}_i$. $\mathcal{B}_i$ is a $\mathcal{O}(\lambda)$-perturbation of $\mathcal{C}^{2n-1}([\text{\emph{bulk}}]_i)$.
    \item The eigenspace with eigenvalue $\Theta(\lambda^{4m_i+2})$. Let $\hat{\Phi}_i$ project onto this space. $\im{\hat{\Phi}_i}$ is a $\mathcal{O}(\lambda)$-perturbation of $|\phi_i\rangle \otimes \mathcal{H}_{n-m} \subseteq \mathcal{H}_n$. Let $\Phi_i$ project onto $|\phi_i\rangle \otimes \mathcal{H}_{n-m}$.
    \item The kernel of $\hat{\Delta}^{2n-1}_i$. Let $\hat{\Phi}^\perp_i$ project onto $\ker{\hat{\Delta}^{2n-1}_i}$. $\im{\hat{\Phi}^\perp_i}$ is a $\mathcal{O}(\lambda)$-perturbation of $|\phi_i\rangle^\perp \otimes \mathcal{H}_{n-m} \subseteq \mathcal{H}_n$. Let $\Phi^\perp_i$ project onto $|\phi_i\rangle^\perp \otimes \mathcal{H}_{n-m}$.
\end{enumerate}
We have a complete set of projectors
\begin{equation} \label{sing_gadg_eigenspaces_eq}
\Pi^{(\mathcal{A})}_i + \Pi^{(\mathcal{B})}_i + \hat{\Phi}_i + \hat{\Phi}^\perp_i = \Pi^{2n-1}_0 + \Pi^{2n-1}_i
\end{equation}

\item Finally, let $\Pi^{(\mathcal{H})}_n$ be the projection onto $\mathcal{H}_n$, with image in $\mathcal{C}^{2n-1}(\mathcal{G}_n)$.
\end{itemize}
\end{definition}

\bigskip
Now for the proof. Assume $\lambda_{\min}(H) \geq g$, and assume for contradiction that the \emph{normalized} state $|\varphi\rangle \in \mathcal{C}^{2n-1}(\hat{\mathcal{G}}_n)$ has
\begin{equation*}
\langle\varphi|\hat{\Delta}^{2n-1}|\varphi\rangle < E
\end{equation*}
where $E$ is defined in the theorem statement. We will show a contradiction by deriving $\langle\varphi|\varphi\rangle < 1$.

Consider the following calculation.
\begin{align*}
\langle\varphi|\varphi\rangle &= \langle\varphi|\big(\Pi_0^{2n-1} + \sum_i \Pi_i^{2n-1}\big)|\varphi\rangle \\
&= \sum_i \langle\varphi|\big(\Pi_0^{2n-1} + \Pi_i^{2n-1}\big)|\varphi\rangle - (t-1) \langle\varphi|\Pi_0^{2n-1}|\varphi\rangle \\
&= \sum_i \langle\varphi|\big(\hat{\Phi}^\perp_i + \hat{\Phi}_i + \Pi^{(\mathcal{A})}_i + \Pi^{(\mathcal{B})}_i\big)|\varphi\rangle - (t-1) \langle\varphi|\Pi_0^{2n-1}|\varphi\rangle \\
&= \Bigg( \langle\varphi| \sum_i \hat{\Phi}^\perp_i|\varphi\rangle - (t-1) \langle\varphi|\Pi_0^{2n-1}|\varphi\rangle \Bigg) + \langle\varphi| \sum_i \hat{\Phi}_i|\varphi\rangle + \langle\varphi| \sum_i \Pi^{(\mathcal{A})}_i|\varphi\rangle + \langle\varphi| \sum_i \Pi^{(\mathcal{B})}_i|\varphi\rangle
\end{align*}
Here we used Equations \ref{sing_gadg_eigenspaces_eq} and \ref{complete_projectors_eq}.

Examining the term in brackets, we have
\begin{align*}
\langle\varphi| \sum_i \hat{\Phi}^\perp_i |\varphi\rangle &- (t-1) \langle\varphi|\Pi_0^{2n-1}|\varphi\rangle \\
&= \langle\varphi| \sum_i \big(\Phi^\perp_i + \mathcal{O}(\lambda)\big) |\varphi\rangle - (t-1) \langle\varphi|\Pi_0^{2n-1}|\varphi\rangle \\
&= \langle\varphi| \sum_i \Phi^\perp_i |\varphi\rangle - (t-1) \langle\varphi|\Pi_0^{2n-1}|\varphi\rangle + \mathcal{O}(\lambda t) \\
&= \langle\varphi| \sum_i \big(\Pi^{(\mathcal{H})}_n - \Phi_i\big) |\varphi\rangle - (t-1) \langle\varphi|\Pi_0^{2n-1}|\varphi\rangle + \mathcal{O}(\lambda t) \\
&= \langle\varphi|\Pi^{(\mathcal{H})}_n|\varphi\rangle - \langle\varphi| \sum_i \Phi_i |\varphi\rangle - (t-1) \langle\varphi|\big(\Pi_0^{2n-1} - \Pi^{(\mathcal{H})}_n\big)|\varphi\rangle + \mathcal{O}(\lambda t)
\end{align*}
In the second line we used a consequence of \Cref{spec_seq_lemma_padded}, combined with Part 1 of \Cref{subspace_perturbation_lem}: $\hat{\Phi}^\perp_i = \Phi^\perp_i + \mathcal{O}(\lambda)$.

Consider the term $\langle\varphi| \sum_i \Phi_i |\varphi\rangle$. $\sum_i \Phi_i$ is precisely the implementation of the Hamiltonian $H$ on the simulated $n$-qubit subspace $\mathcal{H}_n$. This is where we use the assumption on the minimum eigenvalue of $H$.
\begin{align*}
\langle\varphi| \sum_i \Phi_i |\varphi\rangle &= \big(\langle\varphi| \Pi^{(\mathcal{H})}_n\big) \sum_i \Phi_i \big(\Pi^{(\mathcal{H})}_n |\varphi\rangle\big) \\
&\geq g \ ||\Pi^{(\mathcal{H})}_n |\varphi\rangle||^2 \\
&= g \ \langle\varphi| \Pi^{(\mathcal{H})}_n |\varphi\rangle
\end{align*}

Putting this all together, we have
\begin{align*}
\langle\varphi|\varphi\rangle &\leq \langle\varphi|\Pi^{(\mathcal{H})}_n|\varphi\rangle \cdot (1-g) - (t-1) \langle\varphi|\big(\Pi_0^{2n-1} - \Pi^{(\mathcal{H})}_n\big)|\varphi\rangle + \mathcal{O}(\lambda t) \\
& \qquad + \langle\varphi| \sum_i \hat{\Phi}_i|\varphi\rangle + \langle\varphi| \sum_i \Pi^{(\mathcal{A})}_i|\varphi\rangle + \langle\varphi| \sum_i \Pi^{(\mathcal{B})}_i|\varphi\rangle
\end{align*}

The first term is strictly less than 1, and the second term is non-positive since $\Pi_0^{2n-1} - \Pi^{(\mathcal{H})}_n \succeq 0$. If we can somehow show that the final three terms are small, then this will give the contradiction $\langle\varphi|\varphi\rangle < 1$. And indeed, it is true that the low-energy assumption $\langle\varphi|\hat{\Delta}^{2n-1}|\varphi\rangle < E$ forces the terms involving $\{\hat{\Phi}^i\}$, $\{\Pi_i^{(\mathcal{A})}\}$, $\{\Pi_i^{(\mathcal{B})}\}$ to be small. This forms the content of Lemmas \ref{phi_terms_lem}, \ref{A_terms_lem}, \ref{B_terms_lem}, whose proofs are postponed to \Cref{postponed_proofs_sec}.

\begin{lemma} \label{phi_terms_lem}
For $m$-local Hamiltonian $H$ with $t$ terms, recall $\hat{\Delta}^k$ is the Laplacian of the corresponding graph $\hat{\mathcal{G}}_n$. Let $\{\hat{\Phi}_i\}$ be as defined in \Cref{proof_def}. If $|\varphi\rangle \in \mathcal{C}^{2n-1}(\hat{\mathcal{G}}_n)$ is a state with $\langle\varphi|\hat{\Delta}^{2n-1}|\varphi\rangle < E$, then
\begin{equation*}
\langle\varphi| \sum_i \hat{\Phi}_i |\varphi\rangle = \mathcal{O}(\lambda^{-(4m+2)} E t)
\end{equation*}
\end{lemma}

\begin{lemma} \label{A_terms_lem}
For $m$-local Hamiltonian $H$ with $t$ terms, recall $\hat{\Delta}^k$ is the Laplacian of the corresponding graph $\hat{\mathcal{G}}_n$. Let $\{\Pi^{(\mathcal{A})}_i\}$ be as defined in \Cref{proof_def}. If $|\varphi\rangle \in \mathcal{C}^{2n-1}(\hat{\mathcal{G}}_n)$ is a state with $\langle\varphi|\hat{\Delta}^{2n-1}|\varphi\rangle < \mathcal{O}(\lambda^2)$, then
\begin{equation*}
\langle\varphi| \sum_i \Pi^{(\mathcal{A})}_i |\varphi\rangle = \mathcal{O}(\lambda^2 t)
\end{equation*}
\end{lemma}

\begin{lemma} \label{B_terms_lem}
For $m$-local Hamiltonian $H$ with $t$ terms, recall $\hat{\Delta}^k$ is the Laplacian of the corresponding graph $\hat{\mathcal{G}}_n$. Let $\{\Pi^{(\mathcal{B})}_i\}$ be as defined in \Cref{proof_def}. If $|\varphi\rangle \in \mathcal{C}^{2n-1}(\hat{\mathcal{G}}_n)$ is a state with $\langle\varphi|\hat{\Delta}^{2n-1}|\varphi\rangle < \mathcal{O}(\lambda^4)$, then
\begin{equation*}
\langle\varphi| \sum_i \Pi^{(\mathcal{B})}_i |\varphi\rangle = \mathcal{O}(\lambda^2 t)
\end{equation*}
\end{lemma}

Note that we set $E = o(\lambda^{4m+2})$ in the statement of \Cref{main_thm}, so $\langle\varphi|\hat{\Delta}^{2n-1}|\varphi\rangle < E$ implies the conditions of Lemmas \ref{A_terms_lem}, \ref{B_terms_lem}. Returning to the calculation armed with Lemmas \ref{phi_terms_lem}, \ref{A_terms_lem}, \ref{B_terms_lem}, we get
\begin{align*}
\langle\varphi|\varphi\rangle &= \langle\varphi|\Pi^{(\mathcal{H})}_n|\varphi\rangle \cdot (1-g) - (t-1) \langle\varphi|\big(\Pi_0^{2n-1} - \Pi^{(\mathcal{H})}_n\big)|\varphi\rangle + \mathcal{O}(\lambda t) \\
& \quad + \mathcal{O}(\lambda^{-(4m+2)} E t) + \mathcal{O}(\lambda^2 t) + \mathcal{O}(\lambda^2 t)
\end{align*}
Choosing
\begin{align*}
\lambda \ &= \ c t^{-1} g \\
E \ &= \ c \lambda^{4m+2} t^{-1} g
\end{align*}
for a constant $c$ sufficiently small, this becomes
\begin{align*}
1 &= \langle\varphi|\varphi\rangle \\
&\leq \langle\varphi|\Pi^{(\mathcal{H})}_n|\varphi\rangle \cdot (1-g) - (t-1) \langle\varphi|\big(\Pi_0^{2n-1} - \Pi^{(\mathcal{H})}_n\big)|\varphi\rangle + \frac{1}{10} g \\
&\leq 1 - g + \frac{1}{10} g \\
&< 1
\end{align*}
a contradiction. This concludes the proof of \Cref{main_thm}.
\end{proof}

\pagebreak
\section{Gapped Clique Homology is $\QMA_1$-hard and contained in $\QMA$}\label{sec:final}

Here we restate and prove the main result.

\begin{problem} \label{gapped_homology_prob}
\emph{(Problem \ref{informal_prom_prob} restated.)}
Fix functions $k : \mathbb{N} \rightarrow \mathbb{N}$ and $g : \mathbb{N} \rightarrow [0,\infty)$, with $E(n) \geq 1 / \poly(n)$. The input to the problem is a weighted graph $\mathcal{G}$ on $n$ vertices. The task is to decide whether:
\begin{itemize}
    \item {\bf \YES} \ the $k(n)^{\text{th}}$ homology group of $\Cl(\mathcal{G})$ is non-trivial $H^k(\mathcal{G}) \neq 0$
    \item {\bf \NO} \ the $k(n)^{\text{th}}$ homology group of $\Cl(\mathcal{G})$ is trivial $H^k(\mathcal{G}) = 0$ and the Laplacian $\Delta^k$ has minimum eigenvalue $\lambda_{\min}(\Delta^k) \geq E$.
\end{itemize}
\end{problem}

\begin{theorem} \label{main_cor} 
\emph{(Theorem \ref{QMA1_cor} restated.)}
\Cref{gapped_homology_prob} is $\QMA_1$-hard and  contained in $\QMA$.
\end{theorem}

\begin{proof}
\Cref{main_thm} and $\QMA_1$-hardness of quantum $m$-$\SAT$ give $\QMA_1$-hardness of \Cref{gapped_homology_prob}. It remains to argue that this problem is also inside $\QMA$.

We want to show that we can verify the minimum eigenvalue of $\Delta^k$ is zero. $\Delta^k$ acts on $\mathcal{C}^k(\mathcal{G})$ for a clique complex $\Cl(\mathcal{G})$. We must first find some way of embedding this Hilbert space into the qubits which our quantum verifier will have at its disposal. Say that $\mathcal{G}$ contains $n$ vertices. Create a qubit for each vertex $v \in \mathcal{G}^0$ in the clique complex. $\mathcal{C}^k(\mathcal{G})$ is spanned by the $(k+1)$-cliques of the graph $\mathcal{G}$. Define the embedding
\begin{align*}
\iota : \mathcal{C}^k &\hookrightarrow (\mathbb{C}^2)^{\otimes n} \\
|\sigma\rangle &\hookrightarrow |z_\sigma\rangle \qquad \forall \sigma \in \mathcal{G}^k
\end{align*}
where $|z_\sigma\rangle$ is the computational basis state of the indicator bitstring $z_\sigma \in \{0,1\}^n$ of $\sigma \in \mathcal{G}^k$.

The witness to the $\QMA$ verification protocol will ideally be $\iota|\psi\rangle$ where $|\psi\rangle \in \mathcal{C}^k(\mathcal{G})$ is a state of energy $\langle\psi|\Delta^k|\psi\rangle = 0$. Consider the operator on $(\mathbb{C}^2)^{\otimes n}$ given as follows, for computational basis states $|x\rangle,|y\rangle$
\begin{equation*}
\langle x|\widetilde{\Delta}^k|y\rangle = \begin{cases}
\langle x|\Delta^k|y\rangle & \text{$x,y$ are $(k+1)$-cliques} \\
A & x=y, \ \text{$x$ is \emph{not} a $(k+1)$-clique} \\
0 & x\neq y, \ \text{$x$ \emph{or} $y$ is \emph{not} a $(k+1)$-clique}
\end{cases}
\end{equation*}
$\widetilde{\Delta}^k$ acts as $\Delta^k$ on $\im{\iota}$, and on $(\im{\iota})^\perp$ it acts as $A \cdot \mathbbm{1}$ for some sufficiently large energy penalty $A \geq E$.

\begin{fact}\label{sparse_fct}
$\widetilde{\Delta}^k$ is sparse. Moreover, the graph $\mathcal{G}$ gives us efficient sparse access to $\widetilde{\Delta}^k$.
\end{fact}

This follows from \Cref{fact_Laplacian_entries}.

\begin{fact}\label{Ham_sim_fct}
Given efficient sparse access to a self-adjoint operator $H$ on $n$ qubits, we can efficiently implement controlled-$e^{iHt}$ and hence we can efficiently run phase estimation on $H$.
\end{fact}

This follows from \cite{berry2007efficient, childs2011simulating, berry2014exponential, low2017optimal}.

Facts \ref{sparse_fct}, \ref{Ham_sim_fct} give us the quantum verifier's protocol: given witness state $|\widetilde{\psi}\rangle$, use phase estimation to measure the energy $\langle\widetilde{\psi}|\widetilde{\Delta}^k|\widetilde{\psi}\rangle$. It remains to show completeness and soundness for this protocol.

{\bf Completeness.} In the YES case the prover sends $|\widetilde{\psi}\rangle = \iota|\psi\rangle$, which has $\langle\widetilde{\psi}|\widetilde{\Delta}^k|\widetilde{\psi}\rangle = 0$.

{\bf Soundness.} In the NO case, we know that all states $|\psi\rangle \in \mathcal{C}^k(\mathcal{G})$ have $\langle\psi|\Delta^k|\psi\rangle \geq E$. Decompose a state $|\widetilde{\psi}\rangle \in (\mathbbm{C}^2)^{\otimes n}$ as $|\widetilde{\psi}\rangle = |\psi\rangle + |\psi'\rangle$ with $|\psi\rangle \in \im{\iota}$ and $|\psi'\rangle \in (\im{\iota})^\perp$. Then
\begin{align*}
\langle\widetilde{\psi}|\widetilde{\Delta}^k|\widetilde{\psi}\rangle &= \langle\psi|\Delta^k|\psi\rangle + A \langle\psi'|\psi'\rangle \\
&\geq E \langle\psi|\psi\rangle + A \langle\psi'|\psi'\rangle \\
&\geq E
\end{align*}
\end{proof}


\pagebreak
\appendix

\section{Intuition of proof in supersymmetric picture} \label{app:susy_proof_sketch}

Phrased in the language of SUSY quantum mechanics, \Cref{QMA1_cor} says:

\begin{corollary}
Deciding whether the fermion hard core model defined on a graph $G$ has a SUSY groundstate with fermion number $k$ is $\QMA_1$-hard and contained in $\QMA$.
\end{corollary}

The formal proof of this statement in the paper comes from the rigorous proof of \Cref{QMA1_cor}, and applying the connection between SUSY groundstates and homology.
In this section we give an outline of the proof expressed in the language of SUSY quantum mechanics. 
While this is not a rigorous proof, we hope it provides an accessible route into the full proof for readers who are more familiar with physics than with homology. 

\vspace{\baselineskip}
\noindent \emph{Proof overview in SUSY framework:} 

To demonstrate $\QMA_1$-hardness we reduce from $\qmsat$.
We do this by constructing a graph $G_n$ such that the fermion hard core model defined on the graph has a $2^n$-dimensional zero energy subspace. 
We can encode the Hilbert space of $n$ qubits into this subspace. 
We then construct gadgets which we can `glue' to the graph to simulate terms in the history state Hamiltonian $\Hbravyi$ used to show $\QMA_1$-completeness of $\qmsat$. 
We design these gadgets in such a way that the fermion hard core model defined on the final graph $\hat{\boldsymbol{G}}_n$ has a SUSY ground state iff the history state Hamiltonian $\Hbravyi$ has a zero energy eigenstate. 
Moreover, we demonstrate that the resulting fermion hard core model Hamiltonian inherits the gap on the smallest eigenvalue in $\NO$ cases from the $\qmsat$ Hamiltonian.
This gap on the lowest eigenvalue in $\NO$ cases is sufficient to demonstrate containment in $\QMA$. 

The first step in the proof is to construct a graph $G_1$ such that the fermion hard core model defined on $G_1$ has two orthonormal SUSY groundstates which are both fermionic integer states.
Where we say that a state is a fermionic integer state if it can be written, up to overall normalisation, as:
\begin{equation*}
\ket{\psi} = \sum_{i} n_i \ket{f_i}
\end{equation*}
where $n_i\in \mathbb{Z}$ and $\ket{f_i} = \prod_{j}a_j^\dagger \ket{\textrm{vac}}$.

\begin{figure}
\centering
\includegraphics[scale=0.6]{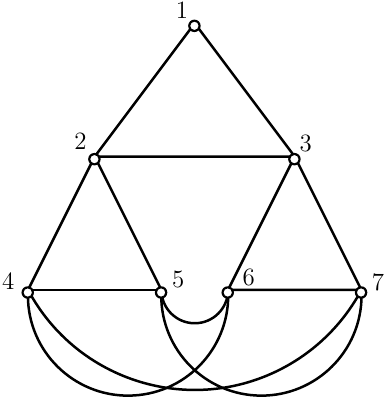}
\caption{The fermion hard core model defined on this graph encodes the Hilbert space of a single qubit in its supersymmetric groundstate subspace.}\label{fig:fermion_pic}
\end{figure}

The graph we choose for our construction is shown in \Cref{fig:fermion_pic}.
While this graph may appear complicated, its independence complex, shown in \Cref{fig:ind_pic} is much easier to visualise.
The two SUSY groundstates which we will identify with the computational basis states of a single qubit are given by the loops in \Cref{fig:ind_pic}.
In terms of fermionic operators they are given by:
\begin{equation*}
\ket{0} = (a^\dagger_1 - a^\dagger_{2} )(a^\dagger_{6} - a^\dagger_{7} )\ket{\textrm{vac}} \textrm{\ \ \ and \ \ \ } \ket{1} = (a^\dagger_1 - a^\dagger_{3} )(a^\dagger_{4} - a^\dagger_{5})\ket{\textrm{vac}}
\end{equation*}
So, the fermion hard core model defined on ${G}_1$ has two two-fermion SUSY groundstates, and no other SUSY ground states.
The graph ${G}_n$ given by $n$ disconnected copies of ${G}_1$ is then our intial graph, where the Hilbert space of $n$ qubits is encoded into the zero energy subspace with $2n$ fermions.
There are no other SUSY ground states.

\begin{figure}
\centering
\includegraphics[scale=0.6]{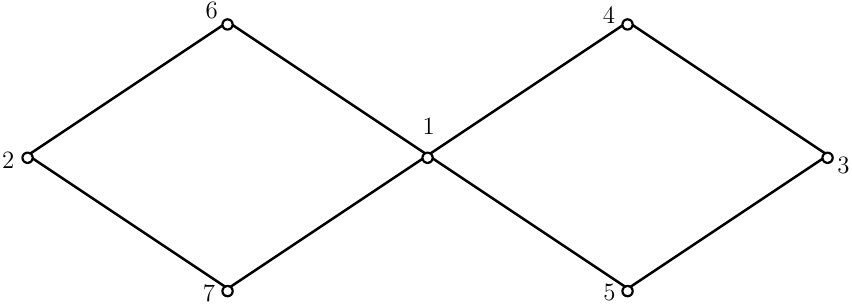}
\caption{This is the complement of the graph in \Cref{fig:fermion_pic}. The SUSY groundstates of the fermion hard-core model defined on the graph shown in \Cref{fig:fermion_pic} are given by the two 1-cycles shown in this figure.}\label{fig:ind_pic}
\end{figure}

The next step in the proof is to design the gadgets.
We do this for each local term needed to construct $\Hbravyi$.
The construction of the gadgets is technical, and a discussion of the details is deferred to \Cref{sec:construction}.
The idea is that for each $m$ local projector $\Pi_m = \ket{\phi}\bra{\phi}$ in $\Hbravyi$ we take the graph ${G}_m$ given by $n$ disconnected copies of ${G}_1$.
We then design a graph $\hat{{G}}_m$ by adding `mediator vertices' which are staggered by $\lambda \ll 1$ such that the spectrum of the fermion hard core model defined on the resulting graph has the following properties:
\begin{enumerate}[i.]
\item \label{spec:i} The SUSY groundspace of the fermion hard core model defined on $\hat{G}_m$ has dimension $2^m-1$, and the fermionic states encoding the states in the groundspace of $\Pi_m$ are $\lambda$-perturbations of states in the SUSY groundspace of $\hat{G}_m$
\item \label{spec:ii} The first excited state, $\ket{\hat{\phi}}$, has energy $\mathcal{O}(\lambda^{4m+2})$ and is a $\lambda$-perturbation of the fermionic state that encodes $\ket{\phi}$. 
\item The rest of the spectrum (which does not correspond to the qubit Hilbert space) has higher energy, moreover we are able to derive the form of these states for arbitrary gadgets
\end{enumerate}
The graph $\hat{{G}}_m$ therefore encodes an $m$-local projector acting on a system of $m$ qubits.
To extend $\hat{{G}}_m$ to a graph that encodes an $m$-local projector acting on a system of $n$ qubits we add $n-m$ copies of ${G}_1$ to our system, disconnected from the graph $\hat{{G}}_m$.
We call this new graph $\hat{\boldsymbol{G}}_m$.

To construct the graph $\hat{\boldsymbol{G}}_n$ corresponding to $\Hbravyi$ the procedure is as follows:
\begin{enumerate}
\item Construct the graph ${G}_n$
\item For each projector in $\Hbravyi$ add the mediators from $\hat{{G}}_m$ to the graph, connected to the relevant vertices from the original graph
\item For each projector in $\Hbravyi$, add edges between all the mediators from that gadget and all the mediators from every other gadget in the graph
\end{enumerate}

Designed in this way it is simple to show that the fermion hard core model defined on $\hat{\boldsymbol{G}}_n$ has a SUSY groundstate iff $\Hbravyi$ is satisfiable.\footnote{This is immediate because SUSY ground states are in one-to-one correspondence with holes in the independence complex of $G$. Our construction of the gadgets fills in the holes in $I(\boldsymbol{G}_n)$ corresponding to states that are lifted in $\Hbravyi$. So there is a hole in $I(\boldsymbol{G}_n)$, and hence a SUSY groundstate, iff there is a state which satisfies every projector in $\Hbravyi$.}
Showing that the lowest eigenvalue of the fermion hard core model in cases where there is no SUSY groundstate is bounded away from zero is more challenging. 

There are two steps to bounding the energy away from zero in \NO cases:
\begin{enumerate}
\item Show that in \NO cases every state in the Hilbert space of fermions on $\hat{\boldsymbol{G}}_n$ must have large overlap with the excited subspace of at least one of the gadgets
\item Demonstrate that this implies that every state in the Hilbert space in \NO cases has high energy
\end{enumerate}

The first point is straightforward - in \NO cases the Hamiltonian $\Hbravyi$ is not satisfiable, so it is not possible to construct a global state which is in the ground state of every projector.
Therefore the overlap of any global state with zero energy subspace of each gadget must be bounded away from one for at least one of the gadgets.

To understand the intuition behind the second point note that the hard core condition, together with the fact that mediator vertices are connected all-to-all between the different gadgets ensures that it is not possible for mediator vertices from two gadgets to be occupied simultaneously.
In other words, we can write any state as:
\begin{equation*}
\ket{\Psi} = \sum_i k_i \ket{\psi_i}
\end{equation*}
where each $\ket{\psi_i}$ can be defined purely on a graph $\hat{\boldsymbol{G}}_m$ for some projector in $\Hbravyi$.
It is straightforward to see that in \NO cases $\bra{\psi_i}H_{\hat{\boldsymbol{G}}_n}\ket{\psi_i}$ is bounded away from zero for every $\ket{\psi_i}$ in the above decomposition.
This is because, as stated above, the $\ket{\psi_i}$ cannot be in the ground state subspace of every gadget, and since the $\ket{\psi_i}$ are defined on the graphs associated to individual gadgets we do not need to consider interference between gadgets to calculate the energy of individual terms in the decomposition.
The final step is to show that interference between different terms in the decomposition cannot reduce the overall energy.
That is, we need to bound the effect of $\bra{\psi_j}H_{\hat{\boldsymbol{G}}_n}\ket{\psi_i}$ terms on the overall energy. 
This can be done by decomposing the Hilbert space into the eigenstates of each gadget, and using the information derived about the form of each eigenstate to argue that this negative interference cannot occur.

\section{Postponed proofs} \label{postponed_proofs_sec}

In this section, we prove Lemmas \ref{phi_terms_lem}, \ref{A_terms_lem}, \ref{B_terms_lem} one at a time. The proof of Lemma \ref{phi_terms_lem} is comparatively simple, since we need only use that the energy of $|\varphi\rangle$ on the \emph{up} Laplacian is small $\langle\varphi|\hat{\Delta}^{\uparrow 2n-1}|\varphi\rangle < E$. Claim \ref{up_Laplacian_clm} tells us that the up Laplacians play nicely with the decomposition into separate gadgets, and there is no interference between gadgets. Unfortunately, we start running into more trouble when we must consider also down Laplacians, which \emph{can} interfere between different gadgets. For this reason, the proofs of Lemmas \ref{A_terms_lem}, \ref{B_terms_lem} are more involved.

Here is some notation which will be useful in the proof of Lemmas \ref{A_terms_lem}, \ref{B_terms_lem}.
\begin{definition}
Recall the block decomposition
\begin{equation*}
\mathcal{C}^{2n-1}(\hat{\mathcal{G}}_n) = \mathcal{C}^{2n-1}(\mathcal{G}_n) \oplus \mathcal{C}^{2n-1}(\mathcal{T}_1) \oplus \dots \oplus \mathcal{C}^{2n-1}(\mathcal{T}_t)
\end{equation*}
Write
\begin{equation*}
|\varphi\rangle = |\omega_0\rangle + |\omega_1\rangle + \dots + |\omega_t\rangle
\end{equation*}
where $|\omega_0\rangle \in \mathcal{C}^{2n-1}(\mathcal{G}_n)$, $|\omega_i\rangle \in \mathcal{C}^{2n-1}(\mathcal{T}_i)$. Further, introduce the notation
\begin{equation*}
|\varphi_i\rangle = (\Pi_0^{2n-1} + \Pi_i^{2n-1}) |\varphi\rangle = |\omega_0\rangle + |\omega_i\rangle
\end{equation*}
\end{definition}

\subsection{Proof of Lemma \ref{phi_terms_lem}}

\begin{proof}
Recall the definition of $|\hat{\phi}_i\rangle$ from Lemma \ref{spec_seq_lemma_padded}. Lemma \ref{spec_seq_lemma_padded} tells us that $|\hat{\phi}_i\rangle \otimes \mathcal{H}_{n-m}$ is a $\hat{\Delta}^{2n-1}_i$-eigenspace with eigenvalue $\lambda^{4m_i+2}$. Combined with Claim \ref{paired_up_clm}, we can see that these states are in fact $\hat{\Delta}^{\uparrow 2n-1}_i$-eigenstates with eigenvalue $\lambda^{4m_i+2}$. Thus $\hat{\Delta}^{\uparrow 2n-1}_i \succeq \lambda^{4m_i+2} \cdot \hat{\Phi}_i$; states in $\im{\hat{\Phi}_i}$ have energy at least $\lambda^{4m_i+2}$ on $\hat{\Delta}^{\uparrow 2n-1}_i$, and states orthogonal to $\im{\hat{\Phi}_i}$ have energy at least zero. Now $\hat{\Delta}^{2n-1} \succeq \hat{\Delta}^{\uparrow 2n-1} = \hat{\Delta}^{\uparrow 2n-1}_1 + \dots + \hat{\Delta}^{\uparrow 2n-1}_t \succeq \hat{\Delta}^{\uparrow 2n-1}_i$ using Claim \ref{up_Laplacian_clm}, so in fact $\hat{\Delta}^{2n-1} \succeq \lambda^{4m_i+2} \cdot \hat{\Phi}_i$. But our state $|\varphi\rangle$ only has energy $\langle\varphi|\hat{\Delta}^{2n-1}|\varphi\rangle < E$ on $\hat{\Delta}^{2n-1}$. Thus
\begin{align*}
\langle\varphi| \hat{\Phi}_i |\varphi\rangle &\leq \lambda^{-(4m_i+2)} \cdot \langle\varphi| \hat{\Delta}^{2n-1} |\varphi\rangle \\
&= \mathcal{O}(\lambda^{-(4m+2)} E) \\
\implies \langle\varphi| \sum_i \hat{\Phi}_i |\varphi\rangle &= \mathcal{O}(\lambda^{-(4m+2)} E t)
\end{align*}
\end{proof}

\subsection{Proof of Lemma \ref{A_terms_lem}}

\begin{proof}
Our strategy is to first show that $|\varphi_i\rangle$ is low energy on $\hat{\Delta}^{2n-1}_i$.
\begin{equation*}
\langle\varphi|\hat{\Delta}^{2n-1}|\varphi\rangle < \mathcal{O}(\lambda^2) \ \implies \
\begin{cases}
\langle\varphi|\hat{\Delta}^{\downarrow 2n-1}|\varphi\rangle < \mathcal{O}(\lambda^2) & \implies \ ||\hat{\partial}^{2n-1}|\varphi\rangle||^2 < \mathcal{O}(\lambda^2) \\
\langle\varphi|\hat{\Delta}^{\uparrow 2n-1}|\varphi\rangle < \mathcal{O}(\lambda^2) & \implies \ ||\hat{d}^{2n-1}|\varphi\rangle||^2 < \mathcal{O}(\lambda^2)
\end{cases}
\end{equation*}
Now
\begin{align*}
&||\hat{d}^{2n-1}|\varphi\rangle||^2 < \mathcal{O}(\lambda^2) \\
\implies \ &||\hat{d}_1^{2n-1}|\varphi_1\rangle||^2 + \dots + ||\hat{d}_t^{2n-1}|\varphi_1\rangle||^2 < \mathcal{O}(\lambda^2) \\
\implies \ &\left|\left| \hat{d}_i^{2n-1} |\varphi_i\rangle \right|\right|^2 < \mathcal{O}(\lambda^2) \ \forall i
\end{align*}
using Claim \ref{up_Laplacian_clm}.
It is not so simple for the boundaries. Write
\begin{equation*}
\left|\left| \hat{\partial}_i^{2n-1} |\varphi_i\rangle \right|\right|^2
= \left|\left| \Pi_0^{2n-2} \hat{\partial}_i^{2n-1} |\varphi_i\rangle \right|\right|^2
+ \left|\left| \Pi_i^{2n-2} \hat{\partial}_i^{2n-1} |\varphi_i\rangle \right|\right|^2
\end{equation*}
For the restriction to the gadget chainspace $\mathcal{C}^{2n-2}(\mathcal{T}_i)$, we indeed have
\begin{equation*}
\left|\left| \Pi_i^{2n-2} \hat{\partial}_i^{2n-1} |\varphi_i\rangle \right|\right|^2 < \mathcal{O}(\lambda^2)
\end{equation*}
since the $(2n-2)$-simplices in gadget $\mathcal{T}_i$ are not shared with any other gadgets. (This is a similar reasoning to the proof of Claim \ref{up_Laplacian_clm}.) However, the $(2n-2)$-simplices in the qubit complex $\mathcal{G}_n^{2n-2}$ touches $(2n-1)$-simplices from many gadgets, and thus the boundaries of multiple $\{|\varphi_j\rangle\}_j$ could contribute to $\Pi_0^{2n-2} \hat{\partial}^{2n-1} |\varphi\rangle$ and possibly cancel out. Here we need a different argument. (To go from the first line to the second line in what follows, note that $\hat{\partial}^{2n-1}$ and $\hat{\partial}_i^{2n-1}$ have identical actions when acting on $\mathcal{C}^{2n-1}(\mathcal{G}) \oplus \mathcal{C}^{2n-1}(\mathcal{T}_i)$.)
\begin{align*}
\Pi_0^{2n-2} \hat{\partial}^{2n-1}|\varphi\rangle &= \Pi_0^{2n-2} \hat{\partial}^{2n-1} \big( |\omega_0\rangle + |\omega_i\rangle \big) + \Pi_0^{2n-2} \hat{\partial}^{2n-1} \sum_{j\neq i} |\omega_j\rangle \\
&= \Pi_0^{2n-2} \hat{\partial}_i^{2n-1} |\varphi_i\rangle + \Pi_0^{2n-2} \hat{\partial}^{2n-1} \sum_{j\neq i} |\omega_j\rangle \\
\implies \ \Pi_0^{2n-2} \hat{\partial}_i^{2n-1} |\varphi_i\rangle &= \Pi_0^{2n-2} \hat{\partial}^{2n-1}|\varphi\rangle - \Pi_0^{2n-2} \hat{\partial}^{2n-1} \sum_{j\neq i} |\omega_j\rangle \\
\implies \ \left|\left| \Pi_0^{2n-2} \hat{\partial}_i^{2n-1} |\varphi_i\rangle \right|\right| &\leq \left|\left| \Pi_0^{2n-2} \hat{\partial}^{2n-1}|\varphi\rangle \right|\right| + \left|\left| \Pi_0^{2n-2} \hat{\partial}^{2n-1} \sum_{j\neq i} |\omega_j\rangle \right|\right| \\
&< \mathcal{O}(\lambda) + \left|\left| \Pi_0^{2n-2} \hat{\partial}^{2n-1} \sum_{j\neq i} |\omega_j\rangle \right|\right|
\end{align*}
Here we used triangle inequality. But,
\begin{equation*}
\left|\left| \Pi_0^{2n-2} \hat{\partial}^{2n-1} \sum_{j\neq i} |\omega_j\rangle \right|\right| = \mathcal{O}(\lambda)
\end{equation*}
since the vertices $\mathcal{T}_j^0$ have weight $\lambda$, recalling Equation \ref{boundary_entries_eq}. Thus
\begin{align*}
\left|\left| \Pi_0^{2n-2} \hat{\partial}_i^{2n-1} |\varphi_i\rangle \right|\right| &< \mathcal{O}(\lambda) + \mathcal{O}(\lambda) = \mathcal{O}(\lambda) \\
\implies \ \langle\varphi_i| \hat{\Delta}_i^{2n-1} |\varphi_i\rangle &= \left|\left| \hat{d}_i^{2n-1} |\varphi_i\rangle \right|\right|^2 + \left|\left| \Pi_i^{2n-2} \hat{\partial}_i^{2n-1} |\varphi_i\rangle \right|\right|^2 + \left|\left| \Pi_0^{2n-2} \hat{\partial}_i^{2n-1} |\varphi_i\rangle \right|\right|^2 \\
&= \mathcal{O}(\lambda^2) + \mathcal{O}(\lambda^2) = \mathcal{O}(\lambda^2)
\end{align*}

$\mathcal{A}_i$ is a $\hat{\Delta}_i^{2n-1}$-eigenspace with eigenvalues $\Theta(1)$, so $\hat{\Delta}_i^{2n-1} \succeq \Theta(1) \cdot \Pi^{(\mathcal{A})}_i$. Thus
\begin{align*}
\langle\varphi_i| \Pi^{(\mathcal{A})}_i |\varphi_i\rangle &\leq \Theta(1) \cdot \langle\varphi_i| \hat{\Delta}_i^{2n-1} |\varphi_i\rangle = \mathcal{O}(\lambda^2) \\
\implies \langle\varphi| \sum_i \Pi^{(\mathcal{A})}_i |\varphi\rangle &= \mathcal{O}(\lambda^2 t)
\end{align*}
\end{proof}

\subsection{Proof of Lemma \ref{B_terms_lem}}

\begin{proof}
We will use what we are told by Lemma \ref{spec_seq_lemma_padded} about the form of $\mathcal{B}_i$. Namely, $\mathcal{B}_i$ is a $\mathcal{O}(\lambda)$-perturbation of the space $\mathcal{C}^{2n-1}([\text{bulk}]_i)$ of $(2n-1)$-simplices touching the central vertex of gadget $\mathcal{T}_i$.

We must rule out the possibility that, when we combine many gadgets, new states of very low energy emerge with high overlap on the $\{\mathcal{B}_i\}$ subspaces. What is the situation we are fighting against? We know states with high overlap on a single $\mathcal{B}_i$ subspace must have energy at least $\Omega(\lambda^2)$. But this state could be a cocycle, and perhaps the $\Omega(\lambda)$ boundary lives in the qubit complex $\mathcal{C}^{2n-2}(\mathcal{G})$. If this could happen, then perhaps many such states from different gadgets could be superposed in such a way that their boundaries destructively interfere on the qubit complex, leading to a state of very low energy. 

However, Claim \ref{bulk_clm} will let us show that not only must the states in an individual $\mathcal{B}_i$ have coboundaries or boundaries of size $\Omega(\lambda)$, but further these $\Omega(\lambda)$ coboundaries and boundaries must be supported on $[\text{bulk}]_i$.

We will show that, if $\langle\varphi|\hat{\Delta}^{2n-1}|\varphi\rangle < \mathcal{O}(\lambda^4)$, then $\langle\varphi_i| \Pi^{(\mathcal{B})}_i |\varphi_i\rangle = \mathcal{O}(\lambda^2)$ for each $i$. Lemma \ref{spec_seq_lemma_padded} and Part 2 of Lemma \ref{subspace_perturbation_lem} tell us that it is enough to show $\langle\varphi_i| \Pi^{[2n-1]}_i |\varphi_i\rangle = \mathcal{O}(\lambda^2)$ for each $i$. We will show the contrapositive. That is, we will assume
\begin{equation} \label{bulk_assumption_eq}
\langle\varphi_i| \Pi^{[2n-1]}_i |\varphi_i\rangle \notin \mathcal{O}(\lambda^2)
\end{equation}
and aim to derive a contradiction with $\langle\varphi|\hat{\Delta}^{2n-1}|\varphi\rangle \leq \mathcal{O}(\lambda^4)$. Our strategy to do this is to use Claim \ref{bulk_clm} to first derive that either $||\Pi^{[2n-2]}_i \hat{\partial}_i^{2n-1} |\varphi_i\rangle||$ or $||\Pi^{[2n]}_i \hat{d}_i^{2n-1} |\varphi_i\rangle||$ must be big. That is, the components of the boundary and coboundary on $[\text{bulk}]_i$ cannot both be small. Then since $\{[\text{bulk}_i]\}$ are separated from each other and `cannot interfere', this will necessitate that $\langle\varphi|\hat{\Delta}^{2n-1}|\varphi\rangle$ is big, providing the contradiction.

We will use the subspace decomposition
\begin{align}
\Pi^{2n-1}_0 + \Pi^{2n-1}_i &= \Pi^{[2n-1]}_i + \Pi^{[2n-1]\perp}_i \nonumber\\
&= \Pi^{[2n-1]}_i + \Pi^{[2n-1]\perp}_i \big( \Pi^{(\mathcal{A})}_i + \Pi^{(\mathcal{B})}_i + \hat{\Phi}_i + \hat{\Phi}^\perp_i \big) \nonumber\\
&= \Pi^{[2n-1]}_i + \Pi^{[2n-1]\perp}_i \Pi^{(\mathcal{A})}_i + \Pi^{[2n-1]\perp}_i \Pi^{(\mathcal{B})}_i + \big( \hat{\Phi}_i + \hat{\Phi}^\perp_i \big) - \Pi^{[2n-1]}_i \big( \hat{\Phi}_i + \hat{\Phi}^\perp_i \big) \label{bulk_decomposition_eq}
\end{align}
where $\Pi^{[2n-1]\perp}_i := \Pi^{2n-1}_0 + \Pi^{2n-1}_i - \Pi^{[2n-1]}_i$.

Let's apply Claim \ref{bulk_clm} to the state $\Pi^{[2n-1]}_i |\varphi_i\rangle \in \mathcal{C}^{2n-1}([\text{bulk}]_i)$. Equation \ref{bulk_assumption_eq} gives $|| \Pi^{[2n-1]}_i |\varphi_i\rangle || \notin \mathcal{O}(\lambda)$. From this, the two cases from Claim \ref{bulk_clm} are $||\Pi^{[2n-2]}_i \hat{\partial}_i^{2n-1} \Pi^{[2n-1]}_i |\varphi_i\rangle|| \notin \mathcal{O}(\lambda^2)$ and $||\Pi^{[2n]}_i \hat{d}_i^{2n-1} \Pi^{[2n-1]}_i |\varphi_i\rangle|| \notin \mathcal{O}(\lambda^2)$.

\bigskip
{\bf Case 1.} \ $||\Pi^{[2n-2]}_i \hat{\partial}_i^{2n-1} \Pi^{[2n-1]}_i |\varphi_i\rangle|| \notin \mathcal{O}(\lambda^2)$

In this case, we will show $||\Pi^{[2n-2]}_i \hat{\partial}_i^{2n-1} |\varphi_i\rangle||$ cannot be small. Using Equation \ref{bulk_decomposition_eq},
\begin{align*}
\Pi^{[2n-2]}_i \hat{\partial}_i^{2n-1} |\varphi_i\rangle &= \Pi^{[2n-2]}_i \hat{\partial}_i^{2n-1} \Pi^{[2n-1]}_i |\varphi_i\rangle \\
& \qquad + \Pi^{[2n-2]}_i \hat{\partial}_i^{2n-1} \Pi^{[2n-1]\perp}_i \Pi^{(\mathcal{A})}_i |\varphi_i\rangle + \Pi^{[2n-2]}_i \hat{\partial}_i^{2n-1} \Pi^{[2n-1]\perp}_i \Pi^{(\mathcal{B})}_i |\varphi_i\rangle \\
& \qquad + \Pi^{[2n-2]}_i \hat{\partial}_i^{2n-1} \big( \hat{\Phi}_i + \hat{\Phi}^\perp_i \big) |\varphi_i\rangle - \Pi^{[2n-2]}_i \hat{\partial}_i^{2n-1} \Pi^{[2n-1]}_i \big( \hat{\Phi}_i + \hat{\Phi}^\perp_i \big) |\varphi_i\rangle \\
\implies ||\Pi^{[2n-2]}_i \hat{\partial}_i^{2n-1} |\varphi_i\rangle|| &\geq ||\Pi^{[2n-2]}_i \hat{\partial}_i^{2n-1} \Pi^{[2n-1]}_i |\varphi_i\rangle|| \\
& \qquad - ||\Pi^{[2n-2]}_i \hat{\partial}_i^{2n-1} \Pi^{[2n-1]\perp}_i \Pi^{(\mathcal{A})}_i |\varphi_i\rangle|| - ||\Pi^{[2n-2]}_i \hat{\partial}_i^{2n-1} \Pi^{[2n-1]\perp}_i \Pi^{(\mathcal{B})}_i |\varphi_i\rangle|| \\
& \qquad - ||\Pi^{[2n-2]}_i \hat{\partial}_i^{2n-1} \big( \hat{\Phi}_i + \hat{\Phi}^\perp_i \big) |\varphi_i\rangle|| - ||\Pi^{[2n-2]}_i \hat{\partial}_i^{2n-1} \Pi^{[2n-1]}_i \big( \hat{\Phi}_i + \hat{\Phi}^\perp_i \big) |\varphi_i\rangle||
\end{align*}
by triangle inequality. We have $||\Pi^{[2n-2]}_i \hat{\partial}_i^{2n-1} \Pi^{[2n-1]}_i |\varphi_i\rangle|| \notin \mathcal{O}(\lambda^2)$ by assumption. To conclude that $||\Pi^{[2n-2]}_i \hat{\partial}_i^{2n-1} |\varphi_i\rangle||$ is big, we will argue that the remaining terms on the right hand side are small.

We know from before that $||\Pi^{(\mathcal{A})}_i |\varphi_i\rangle|| = \mathcal{O}(\lambda)$, thus by Claim \ref{lambda_clm} $||\Pi^{[2n-2]}_i \hat{\partial}_i^{2n-1} \Pi^{[2n-1]\perp}_i \Pi^{(\mathcal{A})}_i |\varphi_i\rangle|| = \mathcal{O}(\lambda^2)$. From Lemma \ref{spec_seq_lemma_padded}, $||\Pi^{[2n-1]\perp}_i \Pi^{(\mathcal{B})}_i|| = ||\Pi^{[2n-1]\perp}_i \big(\Pi^{[2n-1]}_i + \mathcal{O}(\lambda)\big)|| = \mathcal{O}(\lambda)$ and $||\Pi^{[2n-1]}_i \big(\hat{\Phi}_i + \hat{\Phi}^\perp_i\big)|| = ||\Pi^{[2n-1]}_i \big(\Phi_i + \Phi^\perp_i + \mathcal{O}(\lambda)\big)|| = \mathcal{O}(\lambda)$. Thus by Claim \ref{lambda_clm} $||\Pi^{[2n-2]}_i \hat{\partial}_i^{2n-1} \Pi^{[2n-1]\perp}_i \Pi^{(\mathcal{B})}_i |\varphi_i\rangle|| = \mathcal{O}(\lambda^2)$ and $||\Pi^{[2n-2]}_i \hat{\partial}_i^{2n-1} \Pi^{[2n-1]}_i \big( \hat{\Phi}_i + \hat{\Phi}^\perp_i \big) |\varphi_i\rangle|| = \mathcal{O}(\lambda^2)$. Finally, $||\hat{\partial}_i^{2n-1} \big( \hat{\Phi}_i + \hat{\Phi}^\perp_i \big) |\varphi_i\rangle|| = \mathcal{O}(\lambda^{2m_i+1}) = \mathcal{O}(\lambda^2)$ so $||\Pi^{[2n-2]}_i \hat{\partial}_i^{2n-1} \big( \hat{\Phi}_i + \hat{\Phi}^\perp_i \big) |\varphi_i\rangle|| = \mathcal{O}(\lambda^2)$. Altogether, we get
\begin{equation*}
||\Pi^{[2n-2]}_i \hat{\partial}_i^{2n-1} |\varphi_i\rangle|| \notin \mathcal{O}(\lambda^2)
\end{equation*}

\bigskip
{\bf Case 2.} \ $||\Pi^{[2n]}_i \hat{d}_i^{2n-1} \Pi^{[2n-1]}_i |\varphi_i\rangle|| \notin \mathcal{O}(\lambda^2)$

This will be similar to Case 1. Again using Equation \ref{bulk_decomposition_eq},
\begin{align*}
\Pi^{[2n]}_i \hat{d}_i^{2n-1} |\varphi_i\rangle &= \Pi^{[2n]}_i \hat{d}_i^{2n-1} \Pi^{[2n-1]}_i |\varphi_i\rangle \\
& \qquad + \Pi^{[2n]}_i \hat{d}_i^{2n-1} \Pi^{[2n-1]\perp}_i \Pi^{(\mathcal{A})}_i |\varphi_i\rangle + \Pi^{[2n]}_i \hat{d}_i^{2n-1} \Pi^{[2n-1]\perp}_i \Pi^{(\mathcal{B})}_i |\varphi_i\rangle \\
& \qquad + \Pi^{[2n]}_i \hat{d}_i^{2n-1} \big( \hat{\Phi}_i + \hat{\Phi}^\perp_i \big) |\varphi_i\rangle - \Pi^{[2n]}_i \hat{d}_i^{2n-1} \Pi^{[2n-1]}_i \big( \hat{\Phi}_i + \hat{\Phi}^\perp_i \big) |\varphi_i\rangle \\
\implies ||\Pi^{[2n]}_i \hat{d}_i^{2n-1} |\varphi_i\rangle|| &\geq ||\Pi^{[2n]}_i \hat{d}_i^{2n-1} \Pi^{[2n-1]}_i |\varphi_i\rangle|| \\
& \qquad - ||\Pi^{[2n]}_i \hat{d}_i^{2n-1} \Pi^{[2n-1]\perp}_i \Pi^{(\mathcal{A})}_i |\varphi_i\rangle|| - ||\Pi^{[2n]}_i \hat{d}_i^{2n-1} \Pi^{[2n-1]\perp}_i \Pi^{(\mathcal{B})}_i |\varphi_i\rangle|| \\
& \qquad - ||\Pi^{[2n]}_i \hat{d}_i^{2n-1} \big( \hat{\Phi}_i + \hat{\Phi}^\perp_i \big) |\varphi_i\rangle|| - ||\Pi^{[2n]}_i \hat{d}_i^{2n-1} \Pi^{[2n-1]}_i \big( \hat{\Phi}_i + \hat{\Phi}^\perp_i \big) |\varphi_i\rangle||
\end{align*}
By assumption, $||\Pi^{[2n]}_i \hat{d}_i^{2n-1} \Pi^{[2n-1]}_i |\varphi_i\rangle|| \notin \mathcal{O}(\lambda^2)$. We know from before that $||\Pi^{(\mathcal{A})}_i |\varphi_i\rangle|| = \mathcal{O}(\lambda)$, thus by Claim \ref{lambda_clm} $||\Pi^{[2n]}_i \hat{d}_i^{2n-1} \Pi^{[2n-1]\perp}_i \Pi^{(\mathcal{A})}_i |\varphi_i\rangle|| = \mathcal{O}(\lambda^2)$. From Lemma \ref{spec_seq_lemma_padded}, $||\Pi^{[2n-1]\perp}_i \Pi^{(\mathcal{B})}_i|| = ||\Pi^{[2n-1]\perp}_i \big(\Pi^{[2n-1]}_i + \mathcal{O}(\lambda)\big)|| = \mathcal{O}(\lambda)$ and $||\Pi^{[2n-1]}_i \big(\hat{\Phi}_i + \hat{\Phi}^\perp_i\big)|| = ||\Pi^{[2n-1]}_i \big(\Phi_i + \Phi^\perp_i + \mathcal{O}(\lambda)\big)|| = \mathcal{O}(\lambda)$. Thus by Claim \ref{lambda_clm} $||\Pi^{[2n]}_i \hat{d}_i^{2n-1} \Pi^{[2n-1]\perp}_i \Pi^{(\mathcal{B})}_i |\varphi_i\rangle|| = \mathcal{O}(\lambda^2)$ and $||\Pi^{[2n]}_i \hat{d}_i^{2n-1} \Pi^{[2n-1]}_i \big( \hat{\Phi}_i + \hat{\Phi}^\perp_i \big) |\varphi_i\rangle|| = \mathcal{O}(\lambda^2)$. Finally, $||\hat{d}_i^{2n-1} \big( \hat{\Phi}_i + \hat{\Phi}^\perp_i \big) |\varphi_i\rangle|| = \mathcal{O}(\lambda^{2m_i+1}) = \mathcal{O}(\lambda^2)$ so $||\Pi^{[2n]}_i \hat{d}_i^{2n-1} \big( \hat{\Phi}_i + \hat{\Phi}^\perp_i \big) |\varphi_i\rangle|| = \mathcal{O}(\lambda^2)$. Altogether, we get
\begin{equation*}
||\Pi^{[2n]}_i \hat{d}_i^{2n-1} |\varphi_i\rangle|| \notin \mathcal{O}(\lambda^2)
\end{equation*}

\bigskip
We have concluded that either $||\Pi^{[2n-2]}_i \hat{\partial}_i^{2n-1} |\varphi_i\rangle|| \notin \mathcal{O}(\lambda^2)$ or $||\Pi^{[2n]}_i \hat{d}_i^{2n-1} |\varphi_i\rangle|| \notin \mathcal{O}(\lambda^2)$, so
\begin{align*}
&||\Pi^{[2n-2]}_i \hat{\partial}_i^{2n-1} |\varphi_i\rangle||^2 + ||\Pi^{[2n]}_i \hat{d}_i^{2n-1} |\varphi_i\rangle||^2 \notin \mathcal{O}(\lambda^4) \\
\implies \ &\langle\varphi_i| \hat{d}_i^{2n-2} \Pi^{[2n-2]}_i \hat{\partial}_i^{2n-1} |\varphi_i\rangle + \langle\varphi_i| \hat{\partial}_i^{2n} \Pi^{[2n]}_i \hat{d}_i^{2n-1} |\varphi_i\rangle \notin \mathcal{O}(\lambda^4)
\end{align*}

Now consider the operators $\hat{d}_i^{2n-2} \Pi^{[2n-2]}_i \hat{\partial}_i^{2n-1}$, $\hat{\partial}_i^{2n} \Pi^{[2n]}_i \hat{d}_i^{2n-1}$. On $[\text{bulk}]_i$, $\hat{\partial}_i^k$ and $\hat{d}_i^k$ are the same as $\hat{\partial}^k$ and $\hat{d}^k$, thus
\begin{align*}
\hat{d}_i^{2n-2} \Pi^{[2n-2]}_i \hat{\partial}_i^{2n-1} &= \hat{d}^{2n-2} \Pi^{[2n-2]}_i \hat{\partial}^{2n-1} \\
\hat{\partial}_i^{2n} \Pi^{[2n]}_i \hat{d}_i^{2n-1} &= \hat{\partial}^{2n} \Pi^{[2n]}_i \hat{d}^{2n-1}
\end{align*}
So now we can deduce
\begin{align*}
\langle\varphi|\hat{\Delta}^{2n-1}|\varphi\rangle &= \langle\varphi|\hat{d}^{2n-2} \hat{\partial}^{2n-1}|\varphi\rangle + \langle\varphi|\hat{\partial}^{2n} \hat{d}^{2n-1}|\varphi\rangle \\
&\geq \langle\varphi|\hat{d}^{2n-2} \Pi^{[2n-2]}_i \hat{\partial}^{2n-1}|\varphi\rangle + \langle\varphi|\hat{\partial}^{2n} \Pi^{[2n]}_i \hat{d}^{2n-1}|\varphi\rangle \\
&\notin \mathcal{O}(\lambda^4)
\end{align*}
a contradiction to $\langle\varphi|\hat{\Delta}^{2n-1}|\varphi\rangle \leq \mathcal{O}(\lambda^4)$.

We have concluded that $\langle\varphi_i | \Pi^{(\mathcal{B})}_i|\varphi_i\rangle = \mathcal{O}(\lambda^2)$. This tells us
\begin{equation*}
\langle\varphi| \sum_i \Pi^{(\mathcal{B})}_i |\varphi\rangle = \mathcal{O}(\lambda^2 t)
\end{equation*}
\end{proof}

\section{Technical details of $\QMA_1$-hardness construction} \label{app:qma}

In \Cref{sec:qma1} we claim that in order to construct a reduction from $\qmsat$ to \Cref{informal_prom_prob} we only need to implement the projectors shown in \Cref{table:states}.
In this appendix we outline the construction from \cite{bravyi2011efficient} in more detail, to demonstrate how we obtain \Cref{table:states}.

The construction in \cite{bravyi2011efficient} maps a quantum circuit, $U=U_L\cdots U_2U_1$, $U_j \in \mathcal{G}$ operating on $N$ qubits into the ground state of a local Hamiltonian $H(U)$ acting on Hilbert space $\mathcal{H}_{\text{clock}} \otimes \mathcal{H}_{\text{comp}}$ where $\mathcal{H}_{\text{clock}}=(\mathbb{C}^4)^{\otimes L}$ and $\mathcal{H}_{\text{comp}}=(\mathbb{C}^2)^{\otimes N }$. 
The ground state of $H(U)$ is a `history state' \cite{kitaev2002classical} of the form:
\begin{equation*}
\ket{\Omega(\psi_{wit})} = \frac{1}{T}\sum_{t=1}^L \ket{t}\ket{\psi_t}
\end{equation*}
where $\ket{\psi_t} = \Pi_{i=1}^t U_i \ket{\psi_{wit}}$ and ${\ket{t}}$ is an orthonormal basis for $\mathcal{H}_{\text{clock}}$.

A single clock particle in \cite{bravyi2011efficient} has four possible basis states - $\ket{u} = \ket{00}$ - unborn, $\ket{a_1} = \ket{01}$ - active phase one, $\ket{a_2} = \ket{10}$ - active phase two and $\ket{d} = \ket{11}$ - dead.  
As time `passes' each clock particle evolves from the unborn state, through the two active stages, and eventually to the dead state. 
`Legal' clock states are defined as those obeying certain constraints \cite{bravyi2011efficient}:
\begin{enumerate}
    \item The first clock particle is either active or dead
    \item The last clock particle is either unborn or active.
    \item There is at most one active clock particle 
    \item If clock particle $j$ is dead then all particles $k$ satisfying $1\leq k \leq j$ are also dead
\end{enumerate}

In \cite{bravyi2011efficient} it is shown that the ground space of the Hamiltonian $\Hclock = \sum_{j=1}^6 \Hclock^{(j)}$ is spanned by legal clock states, where:
\begin{equation*}
\Hclock^{(1)} = \ket{u}\bra{u}_1
\end{equation*}
\begin{equation*}
\Hclock^{(2)} = \ket{d}\bra{d}_L
\end{equation*}
\begin{equation*}
\Hclock^3 = \sum_{1\leq j \leq k \leq L} \left(\ket{a_1}\bra{a_1} + \ket{a_2}\bra{a_2} \right)_j \otimes \left(\ket{a_1}\bra{a_1} + \ket{a_2}\bra{a_2} \right)_k
\end{equation*}
\begin{equation*}
\Hclock^{(4)} =\sum_{1\leq j \leq k \leq L} \left(\ket{a_1}\bra{a_1} + \ket{a_2}\bra{a_2}+ \ket{u}\bra{u} \right)_j \otimes \left(\ket{d}\bra{d}  \right)_k
\end{equation*}
\begin{equation*}
\Hclock^{(5)} = \sum_{1\leq j \leq k \leq L} \left(\ket{u}\bra{u} \right)_j \otimes \left( \ket{a_1}\bra{a_1} + \ket{a_2}\bra{a_2}+  \ket{d}\bra{d}  \right)_k
\end{equation*}
\begin{equation*}
\Hclock^{(6)} = \sum_{1 \leq j \leq L-1} \ket{d}\bra{d}_j \otimes \ket{u}\bra{u}_{j+1}
\end{equation*}

The computational qubits can be divided into input data and witness registers, $R_{in}$ and $R_{wit}$ such that $|R_{in}| + |R_{wit}| = N$. The input qubits are initialised into the all zero state at time zero via:
\begin{equation*}
\Hin = \ket{a_1}\bra{a_1}_1 \otimes \left( \sum_{b \in R_{in}} \ket{1}\bra{1}_b \right)
\end{equation*}
The `propagation' Hamiltonian is defined as:
\begin{equation*}
\Hprop = \sum_{t=1}^L  \left(\Hpropt + \Hpropt' \right)
\end{equation*}
where:

\begin{equation*}
\Hpropt = \frac{1}{2}\left[\left(\ket{a_1}\bra{a_1} + \ket{a_2}\bra{a_2} \right)\otimes \identity - \ket{a_2}\bra{a_1} \otimes U_t - \ket{a_1}\bra{a_2} \otimes U_t^\dagger\right]
\end{equation*}
\begin{equation*}
\Hpropt' = \frac{1}{2}\left(\ket{a_2,u}\bra{a_2,u} + \ket{d,a_1}\bra{d,a_1} -\ket{d,a_1}\bra{a_2,u} - \ket{a_2,u}\bra{d,a_1} \right)
\end{equation*}

It can easily be checked that the zero energy ground space of:
\begin{equation*}
H(U) =\Hin + \Hclock + \Hprop 
\end{equation*}
is spanned by computational history states of the form:
\begin{equation*}
\ket{\Omega(\psi_{wit})} = \sum_{t=1}^L\left(  \ket{d}^{\otimes t-1} \ket{a_1}_t \ket{u}^{\otimes L-t} \otimes \ket{Q_{t-1}} +\ket{d}^{\otimes t-1} \ket{a_2}_t \ket{u}^{\otimes L-t} \otimes \ket{Q_{t}}    \right)
\end{equation*}

where $\ket{Q_0} = \ket{0}^{\otimes R_{in}} \otimes \ket{\psi_{wit}}$, and $\ket{Q_t} = U_t \ket{Q_{t-1}}$, $t \in [1,L]$. 
Therefore the ground states of $H(U)$ encode the computational history of the circuit $U = U_L...U_1$.

The final step is to penalise computational histories where the circuit rejects the witness, this is achieved using the projector:
\begin{equation*}
\Hout = \ket{a_2}\bra{a_2}_L \otimes \left( \sum_{b \in R_{out}} \ket{1}\bra{1}_b \right)
\end{equation*}
This gives energy to computational history states $\ket{\Omega(\psi_{wit})}$ if the circuit $U$ rejects the witness $\ket{\psi_{wit}}$ with non-zero probability. 

So the final Hamiltonian:
\begin{equation*}
H = \Hin + \Hclock + \Hprop + \Hout
\end{equation*}
has a zero energy ground state iff there exists a witness $\ket{\psi_{wit}}$ such that circuit $U$ accepts with probability one. 

With a more detailed understanding of the construction in \cite{bravyi2011efficient} we can demonstrate that the projectors needed are given by those in \Cref{table:states}.

For each $\Hin$ only a single projector is needed.
\begin{equation*}
\Hin = \ket{011}\bra{011}
\end{equation*}
This is a three qubit projector of rank 1.
$\Hout$ is equivalent, up to permuting the qubits involved in the interaction:
\begin{equation*}
\Hout = \ket{101}\bra{101}
\end{equation*}

Implementing $\Hclock$ requires six projectors.
$\Hclock^{(1)}$ and $\Hclock^{(2)}$ are a two qubit projectors with rank 1:
\begin{equation*}
\Hclock^{(1)} = \ket{00}\bra{00}
\end{equation*}

\begin{equation*}
\Hclock^{(2)} = \ket{11}\bra{11}
\end{equation*}

The remaining terms in $\Hclock$ each act on 4 qubits. $\Hclock^{(3)}$ and $\Hclock^{(5)}$ each have rank 4, $\Hclock^{(4)}$ has rank 3, and $\Hclock^{(6)}$ has rank 1:
\begin{equation*}
\Hclock^{(3)} = \ket{0101}\bra{0101} + \ket{0110}\bra{0110} + \ket{1001}\bra{1001} +\ket{1010}\bra{1010}
\end{equation*}

\begin{equation*}
\Hclock^{(4)} = \ket{0111}\bra{0111} + \ket{1011}\bra{1011} + \ket{0011}\bra{0011}
\end{equation*}

\begin{equation*}
\Hclock^{(5)} = \ket{0001}\bra{0001} + \ket{0010}\bra{0010} + \ket{0011}\bra{0011}
\end{equation*}

\begin{equation*}
\Hclock^{(6)} = \ket{1100}\bra{1100}
\end{equation*}

For $\Hprop$ we have to consider two projectors for $\Hpropt$ and one for $\Hpropt'$.
The projector for propagation under the pythagorean gate is a three qubit projector of rank 2 given by:

\begin{equation*}
\begin{split}
\Hpropt(U_{\mathit Pyth.})=\frac{1}{2}[\left(\ket{01}\bra{01} + \ket{10}\bra{10} \right)\otimes \identity \\
- \frac{1}{5} \left( 3\ket{100}\bra{010}-4\ket{100}\bra{011}+4\ket{101}\bra{010}+3\ket{101}\bra{011}  \right) \\
- \frac{1}{5}\left( 3\ket{010}\bra{100} +4\ket{010}\bra{101} -4\ket{011}\bra{100} +3\ket{011}\bra{101} \right) ]
\end{split}
\end{equation*}

The projector for evolution under $\CNOT$ is a four qubit projector of rank 4:

\begin{equation*}
\begin{split}
\Hpropt(CNOT)= \frac{1}{2}[\left(\ket{01}\bra{01} + \ket{10}\bra{10} \right)\otimes \identity \\
- \left( \ket{1000}\bra{0100}+\ket{1001}\bra{0101} +\ket{1010}\bra{0111} +\ket{1011}\bra{0110}   \right)\\
- \left( \ket{0100}\bra{1000}+ \ket{0101}\bra{1001}+ \ket{0110}\bra{1011}+ \ket{0111}\bra{1010}\right)]
\end{split}
\end{equation*}

\section{Further details on gadgets for integer states}

\subsection{General method for constructing gadgets for integer states} \label{app:gadgets}

Before we discuss the general method for constructing integer states for arbitrary numbers of qubits, we will go through it for the case of 1-, 2-, and 3-qubit states.
This can be read in conjunction with the constructions in \Cref{sec:construction} to gain intuition for what is happening.
It should be noted that in for the 4-qubit states needed for the reduction we do not make use of this general method.
This is because the 4-qubit states from \Cref{table:states} all only involve two cycles.
This means that the simplest way of combining them is not to use this general method, but to use the method from \Cref{sec:4_qubit_construction}.
The method outlined here could be applied to the 4-qubit states, but it leads to a blow up in the number of vertices in the construction, and is not needed for the reasons already outlined.
 
\subsubsection{General procedure for 1-qubit integer states}

To construct $\Kcyc$ for a general one qubit integer state $\phiST = \sum_{i}n_i|z_i\rangle$ we apply the following procedure:
\begin{enumerate}
\item Take $n_0$ copies of the $\z$ cycle and $n_1$ copies of the $\1$ cycle, and label the $a_i$ vertices and $b_i$ vertices by $a_{i,j}$ $b_{i,k}$ where $i \in [2,4]$, $j \in [1,n_0]$ and $k \in [1,n_1]$
\item Cut open each cycle by duplicating the $x$ vertices
\item Label the $x$ vertex adjacent to $a_{3,i}$ by $x_i$, and the $x$ vertex adjacent to $a_{4,i}$ by $x_{i+1}$. 
\item If the $n_0 > 0$ and $n_1 > 0$ label the $x$ vertex adjacent to $b_{4,i}$ by $x_{n_0 +i}$ and the $x$ vertex adjacent to $b_{3,i}$ by $x_{n_0+i+1}$. Label the $x$ vertex adjacent to $b_{3,n_1}$ by $x_1$. (If either $n_0$ or $n_1$ is less than zero reverse the orderings of the $b$ cycles).
\item Glue the cycles together by identifying $x$ vertices with the same labels.
\end{enumerate}

We then apply the thickening and coning off procedure from \Cref{single_gadget_sec} to `fill in' the cycle $\Kcyc$.
To complete the gadget we apply the function $f(\cdot)$ from \Cref{eq:f} to the vertices $\mathcal{K}^0$ with the relation: 
\begin{equation*}
\begin{split}
R =& \{(x,x_i)|\forall i\} \cup \{(a_{j,i},a_j)| \forall i, j \in [2,4]\} \\ &\cup \{(b_{j,i},b_j)| \forall i, j \in [2,4]\}
\cup \{(v,v)|v\in \J^0\} 
\end{split}
\end{equation*}

\subsubsection{General procedure for 2 qubit integer states}

To construct $\Kcyc$ for a general two qubit integer state $\phiST = \sum_{i}n_i|z_i\rangle$ we apply the following procedure:
\begin{enumerate}
\item Take $n_i$ copies of the $|z_i\rangle$ cycle for $i \in [1,4]$.
Label the $a_i$ ($b_i$) vertices and $a_i'$ ($b_i'$) vertices by $a_{i,j}$ ($b_{i,j}$) and $a_{i,j}'$ ($b_{i,j}'$) where $i \in [2,4]$, $j \in [1,\sum_in_i]$
\item Cut open each cycle along the $1$-simplex $[xx']$ by introducing four dummy $x_j$ vertices, as outlined in \Cref{sec:2 qubit cutting}
\item Label the $x_j$ vertices from the $i^{\textrm{th}}$ cycle by $x_{2(i-1)+1},x_{2(i-1)+2},x_{2(i-1)+3},x_{2(i-1)+4}$, where the ordering of the vertices is dependent on whether the cycle is being added with positive or negative sign
\item Glue the cycles together by identifying vertices with the same labels.
\item If $\sum_i n_i \geq 4$ introduce two final dummy vertices, and connect each of the final dummy vertices to either the `ceiling' or the `floor' of the viewing platform as shown in \Cref{fig:multiples 2}
\end{enumerate}

We then apply the thickening and coning off procedure from \Cref{single_gadget_sec} to `fill in' the cycle $\Kcyc$.
To complete the gadget we apply the function $f(\cdot)$ from \Cref{eq:f} to the vertices $\mathcal{K}^0$ with the relation: 
\begin{equation*}
\begin{split}
R =& \{(x,x_i)|\forall i\} \cup \{(a_{j,i},a_j)| \forall i, j \in [2,4]\}\cup \{(a'_{j,i},a'_j)| \forall i, j \in [2,4]\} \\ 
& \cup \{(b_{j,i},b_j)| \forall i, j \in [2,4]\}\cup \{(b'_{j,i},b'_j)| \forall i, j \in [2,4]\}\\
&\cup \{(v,v)|v\in \J^0\} 
\end{split}
\end{equation*}

\subsubsection{General procedure for three qubit integer states}

To construct $\Kcyc$ for a general three qubit integer state $\phiST = \sum_{i}n_i|z_i\rangle$ we apply the following procedure:
\begin{enumerate}
\item Take $n_i$ copies of the $|z_i\rangle$ cycle for $i \in [1,8]$.
Label the $a_i$ ($b_i$) vertices, the $a_i'$ ($b_i'$) vertices and the $a_i''$ ($b_i''$) by $a_{i,j}$ ($b_{i,j}$), $a_{i,j}'$ ($b_{i,j}'$) and and $a_{i,j}''$ ($b_{i,j}''$) where $i \in [2,4]$, $j \in [1,\sum_in_i]$
\item Cut open each cycle by introducing nine dummy $x_j$ vertices in a bi-piramid, as outlined in \Cref{sec:3 qubit cutting}.
\item The $x_1$ vertex from all the open cycles will be shared. Label the $x_j$ for $j \in [2,8]$ vertices from each cycle so that each cycle shares half the vertices with the cycle preceding it, and half the vertices with the cycle following it (the order of the vertices determines whether the cycles are added with positive or negative orientation).
\item Glue the cycles together by identifying vertices with the same labels.
\item If $\sum_i n_i \geq 4$ introduce four more dummy vertices, and connect each of the final dummy vertices to one of the four open one cycles that have been introduced by this procedure, as shown in \Cref{fig:3 qubit 13}.
\item Add four final dummy vertices, and connect each of these to the vertices from one of the four open 2-cycles, as shown in \Cref{fig:3 qubit 13}.
\end{enumerate}

We then apply the thickening and coning off procedure from \Cref{single_gadget_sec} to `fill in' the cycle $\Kcyc$.
To complete the gadget we apply the function $f(\cdot)$ from \Cref{eq:f} to the vertices $\mathcal{K}^0$ with the relation: 
\begin{equation*}
\begin{split}
R =& \{(x,x_i)|\forall i\} \cup \{(a_{j,i},a_j)| \forall i, j \in [2,4]\}\cup \{(a'_{j,i},a'_j)| \forall i, j \in [2,4]\}\cup \{(a''_{j,i},a''_j)| \forall i, j \in [2,4]\} \\ 
& \cup \{(b_{j,i},b_j)| \forall i, j \in [2,4]\}\cup \{(b'_{j,i},b'_j)| \forall i, j \in [2,4]\}\cup \{(b''_{j,i},b''_j)| \forall i, j \in [2,4]\}\\
&\cup \{(v,v)|v\in \J^0\} 
\end{split}
\end{equation*}

\subsubsection{General procedure for constructing cycles for integer states}

To construct $\Kcyc$ for arbitrary $N$ qubit integer states we first need to understand how to `cut open' the $2N-1$-cycles that form the $N$ qubit basis states.
The $2N-1$ cycles corresponding to each basis state share the common $N-1$-simplex given by $[x^{(1)}x^{(2)},\cdots x^{(N)}]$.
We will cut open this $N-1$-simplex to give an open $2(N-1)$-cycle that we can glue cycles together along via the following prescription:
\begin{enumerate} 
\item Take the $N-1$-simplex and construct an $N-1$-dimensional cross-polytope $\mathfrak{g}_{N-1}$ by adding $N-2$ vertices $x_1,\cdots,x_{N-2}$. When constructing the cross-polytope every dummy vertex that is added must share an edge with $x^{(1)}$.\footnote{This rule makes the quotienting rules to go from $\Kcyc$ to $\J$ simpler. For the three qubit case there are two options of how to construct $\mathfrak{g}_2$ following this rule, in higher dimensions there will be more choices -- which option is chosen is arbitrary.}
\item Take each $x_i$ for $i \in [1,N-2]$. Let the vertices it shares an edge with in $\mathfrak{g}_{N-1}$ be denoted $V$. If $x_i$ `cuts open' an existing simplex add an edge between $x_i$ and every vertex that shares an edge with {\it every} vertex in $V$. If $x_i$ is adding new simplices to the cross polytope (i.e. not cutting open an existing simplex) add edges between $x_i$ and the vertices which form part of simplices that have been divided in the process.
\item The vertices that share edges with every vertex in the $[x^{(1)}x^{(2)}\cdots x^{(N)}]$ simplex will be $c_{3}^{(k)}$ and $c_{4}^{(k)}$ for $k \in [1,N]$ and where the $c_{i}^{(k)}$ for $i \in [2,4]$ are either $a$ or $b$ vertices depending on whether the $k^\mathrm{th}$ qubit is in the $\ket{0}$ or $\ket{1}$ state in the cycle being cut. These $2N$ vertices form an $N$-dimensional cross polytope $\mathfrak{g}_N$ (to see this note that the $c_{3}^{(k)}$ and $c_{4}^{(k')}$ share an edge iff $k \neq k'$).
\item Use $2^N$ new dummy vertices $x_{N-1} \cdots x_{N-1+2^N}$ to construct an open $N$-cube (which is the dual of $\mathfrak{g}_N$ and is an $N-1$-cycle). Triangulate each $N-1$-face of the $N$-cube via a standard triangulation.
\item Take each vertex of the $N$-cube and add edges between it and every vertex in the corresponding facet of $\mathfrak{g}_N$. 
\item Form a sequence of bi-pyramids -- the first has the $N$-cube as its base, and its peaks are any pair of non-adjacent vertices in $\mathfrak{g}_{N-1}$. The next bi-pyramid has its base as the previous bi-pyramid, and its peaks as any new pair of non-adjacent vertices in $\mathfrak{g}_{N-1}$. Continue in this manner, constructing a sequence of $N-1$-bi-pyramids.
\end{enumerate}

The final bi-pyramid will be an open $2(N-2)$ cycle as required.

Then to construct $\Kcyc$ for a general $N$ qubit integer state $\phiST = \sum_{i}n_i|z_i\rangle$ we apply the following procedure:
\begin{enumerate}
\item Take $n_i$ copies of the $|z_i\rangle$ cycle for $i \in [1,2^N]$.
Label the $a_i^{(k)}$ ($b_i^{(k)}$) vertices by $a_{i,j}^{(k)}$ ($b_{i,j}^{(k)}$), where $i \in [2,4]$, $k \in [1,N]$, $j \in [1,\sum_in_i]$
\item Cut open each of the $2N-1$-cycles that form the $N$-qubit basis states along their common $N-1$-simplex as outlined above.
\item The $x_1, \cdots x_{N-2}$ vertices from all the open cycles will be shared. Label the $x_j$ for $j \in [N-1,N-1+2^N]$ vertices from each cycle so that each cycle shares half the vertices with the cycle preceding it, and half the vertices with the cycle following it (the order of the vertices determines whether the cycles are added with positive or negative orientation).
\item Glue the cycles together by identifying vertices with the same labels.
\item This process will have left $2(n-1)$ facets of each $n$-cube unpaired. These unpaired facets will form open cycles with those from other $n$-cubes, which need to be filled in. For $\sum_i n_i \geq 3$ there will be $2(n-1)$ $d$-cycles for $d \in [3,n-1]$ which need to be filled in -- each can be filled in by adding a dummy vertex which is connected to every vertex in that cycle. For $\sum_i n_i \geq 4$ there will also be $2(n-1)$ 2-cycles which can be filled in in the same manner.
\end{enumerate}

We then apply the thickening and coning off procedure from \Cref{single_gadget_sec} to `fill in' the cycle $\Kcyc$.
To complete the gadget we apply the function $f(\cdot)$ from \Cref{eq:f} to the vertices $\mathcal{K}^0$ with the relation: 
\begin{equation*}
\begin{split}
R =& \{(x,x_i)|\forall i\} \cup \{(a^{(k)}_{j,i},a^{(k)}_j)| \forall i, \forall k, j \in [2,4]\} \\ & \cup \{(b^{(k)}_{j,i},b^{(k)}_j)| \forall i, \forall k, j \in [2,4]\}
\cup \{(v,v)|v\in \J^0\} 
\end{split}
\end{equation*}

\subsection{Postponed gadget proofs}\label{app:gadget proofs}

In this section we provide proofs of the final lemmas needed to prove that the gadget constructions for the states in \Cref{table:states} work as intended.

\begin{lemma}
The procedure outlined in \Cref{sec:4_qubit_construction} for cutting and gluing four-qubit cycles results in a clique complex $\Kcyc$ which triangulates $S^7$.
\end{lemma}

\begin{proof}
We prove this for the state $\ket{1011}-\ket{1000}$, generalising the proof to the other 4 qubit entangled states in \Cref{table:states} is trivial.
    
The cycle $c_0$ corresponding to $\ket{1011}$ is given by the join:
\begin{figure}[H]
\centering
\includegraphics[scale=0.8]{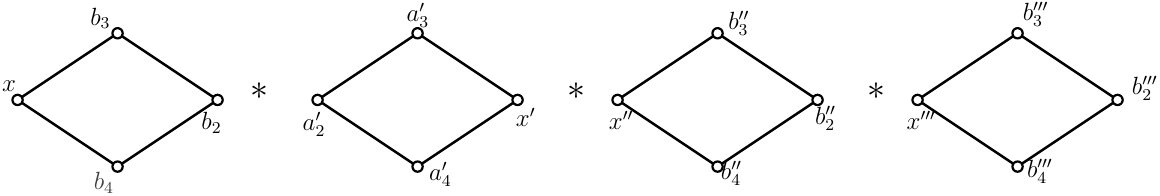}
\end{figure}
This forms an 8-dimensional cross-polytopes, $ \mathfrak{g}_8$.
Let $\mathcal{C}_0$ denote the complex which results from applying the cutting process from \Cref{sec:4_qubit_construction} to this copy of $ \mathfrak{g}_8$ -- i.e. $\mathcal{C}_0$ is constructed by taking the join above, and adding the dummy vertices $x_0 \dots x_{10}$ following the prescription in \Cref{sec:4_qubit_construction}.
We begin by arguing that $\mathcal{C}_0$ is a triangulation of a 7-sphere with a copy of $\mathfrak{g}_7$ cut into it by the dummy vertices.
Topologically this is equivalent to claiming that $\mathcal{C}_0$ is a triangulation of a 7-dimensional hyperplane with boundary given by the copy of $\mathfrak{g}_7$.
It is immediate from the construction in \Cref{sec:4_qubit_construction} that the dummy vertices together with the original vertices $x,x',x'',x'''$ form a copy of $\mathfrak{g}_7$.
So all that remains is to show that no additional holes were formed in the process of constructing this hole.
We achieve this by demonstrating that every facet of every maximal simplex (i.e. every 7-simplex) in $\mathcal{C}_0$ is either a face of the copy of $\mathfrak{g}_7$ that bounds the hole, or is shared with another simplex in $\mathcal{C}_0$.

The hole in $c_0$ is cut into the 3-simplex $\mathcal{S} = [xx'x''x''']$.
So any 7-simplices from $c_0$ which do not share any lower simplices with $\mathcal{S}$ are clearly unaffected by the cutting open procedure from \Cref{sec:3 qubit cutting}, and since in $c_0$ all their facets were shared with another simplex in $c_0$, the same is true in $\mathcal{C}_0$. 

Consider next the 7-simplices in $c_0$ which share exactly one 0-simplex with $\mathcal{S}$. There are 64 such simplices (each of $x,x',x'',x'''$ are involved in 16). 
These 7-simplices are also unaffected by the process of cutting open $\mathcal{S}$ (since only one vertex from the 3-simplex is involved in each 7-simplex no faces of the 7-simplices are removed).
So these 7-simplices are still present in $\mathcal{C}_0$, and all their facets are shared with other 7-simplices which are also present in $\mathcal{C}_0$.

Consider next the 7-simplices in $c_0$ which contain exactly two vertices from $\mathcal{S}$.
Here we have to consider two different types of simplex.
First consider 7-simplices which contain the 1-simplices $[xx'']$, $[xx''']$, $[x'x'']$, $[x'x''']$ or $[x''x''']$.
There are 80 such simplices (16 for each pair).
These 7-simplices are again unaffected by the process of cutting open $\mathcal{S}$ because the lower simplices they share with $\mathcal{S}$ are still intact after adding the dummy vertices (i.e. the edges they share with $\mathcal{S}$ are not cut open by the process of adding the $x_i$).

Next consider 7-simplices which share the 1-simplex $[xx']$ with $\mathcal{S}$.
There are 16 such 7-simplices:
\begin{equation}\label{eq:simplices 1}
\begin{split}
&[xb_3x'a_3'b_2''b_3''b_2''b_3''] \textrm{\ \ \  \ \ \ } [xb_3x'a_3'b_2''b_3''b_2''b_4''] \textrm{\ \ \  \ \ \ }
[xb_3x'a_3'b_2''b_4''b_2''b_3'']
\textrm{\ \ \  \ \ \ }
[xb_3x'a_3'b_2''b_4''b_2''b_4''] \\
&[xb_3x'a_4'b_2''b_3''b_2''b_3''] \textrm{\ \ \  \ \ \ } [xb_3x'a_4'b_2''b_3''b_2''b_4''] \textrm{\ \ \  \ \ \ }
[xb_3x'a_4'b_2''b_4''b_2''b_3'']
\textrm{\ \ \  \ \ \ }
[xb_3x'a_4'b_2''b_4''b_2''b_4''] \\
&[xb_4x'a_3'b_2''b_3''b_2''b_3''] \textrm{\ \ \  \ \ \ } [xb_4x'a_3'b_2''b_3''b_2''b_4''] \textrm{\ \ \  \ \ \ }
[xb_4x'a_3'b_2''b_4''b_2''b_3'']
\textrm{\ \ \  \ \ \ }
[xb_4x'a_3'b_2''b_4''b_2''b_4''] \\
&[xb_4x'a_4'b_2''b_3''b_2''b_3''] \textrm{\ \ \  \ \ \ } [xb_4x'a_4'b_2''b_3''b_2''b_4''] \textrm{\ \ \  \ \ \ }
[xb_4x'a_4'b_2''b_4''b_2''b_3'']
\textrm{\ \ \  \ \ \ }
[xb_4x'a_4'b_2''b_4''b_2''b_4'']
\end{split}
\end{equation}

We note that the dummy vertex $x_1$ shares an edge with both $x$ and $x'$, as well as with the vertices $\{b_3,b_4,a'_3a'_4,b_2'',b''_3,b''_4,b_2''',b_3''',b_4'''\}$ (see \Cref{sec:4_qubit_construction}).
Therefore, each of the 7-simplices in \Cref{eq:simplices 1} is removed from the complex by the process of adding the dummy vertices (since $[xx']$ is no longer in the complex).
But each simplex $s \in c_0$ is replaced by two 7-simplices -- one $s'$ containing the edge $[xx_1]$ and the other $s''$ containing the edge $[x'x_1]$.
The simplices $s'$ and $s''$ share the common 6-simplex given by removing $x$ and $x'$ from the vertex set and adding the vertex $x_1$.
In $c_0$ the other facets of each simplex $s$ in \Cref{eq:simplices 1} were all shared with either:
\begin{itemize}
\item another 7-simplex, $t$, that contains one of $x,x'$. Since $t$ only contains one of $x,x'$ we have that $t \in \mathcal{C}_0$, and moreover $t$ will now share a facet with either $s'$ or $s''$ (depending on whether it is $x$ or $x'$ that is contained in $t$)
\item another 7-simplex, $t$, that contains $x,x'$ and no other vertices from $\mathcal{S}$. Then $t$ is one of the other simplices from \Cref{eq:simplices 1} that has been replaced by two simplices $t',t'' \in \mathcal{C}_0$. We will have that $t'$ shares a facet with $s'$, and $t''$ shares a facet with $s''$.
\item another 7-simplex $t$ that contains $x$, $x'$ and one or more other vertices from $\mathcal{S}$. Then $s',s''$ will share a facet with one of the other new simplices in $\mathcal{C}_0$ that we will consider next
\end{itemize}
So, while the 7-simplices in \Cref{eq:simplices 1} are not present in $\mathcal{C}_0$, the simplices that replace them share every facet with a simplex that is in $\mathcal{C}_0$.

Next we turn to consider simplices which share exactly three vertices with $\mathcal{S}$. 
Again, here we must consider two types of simplices.
First we consider simplices which share the 2-simplices $[xx''x''']$ or $[x'x''x''']$ with $\mathcal{S}$.
There are 32 such 7-simplices (16 for each 2-simplex).
These 7-simplices are not cut open by the process of adding the dummy vertices, because the 2-simplices they share with $\mathcal{S}$ are left intact by the process of cutting $\mathcal{S}$ open.

The other type of simplex we must consider are those that share the 2-simplices $[xx'x'']$ or $[xx'x''']$ with $\mathcal{S}$.
Again, there are 32 such simplices.
These 7-simplices are removed from the complex by the process of adding the dummy vertices, because the 2-simplices they share with $\mathcal{S}$ are no longer present in the complex after $\mathcal{S}$ has been cut open.
We will consider first the 7-simplices which share the 2-simplex $[xx'x'']$ with $\mathcal{S}$:
\begin{equation}\label{eq:simplices 2}
\begin{split}
&[xb_3x'a_3'x''b_3''b_2'''b_3'''] \textrm{\ \ \  \ \ \ } [xb_3x'a_3'x''b_3''b_2'''b_4'''] \textrm{\ \ \  \ \ \ }
[xb_3x'a_3'x''b_4''b_2'''b_3''']
\textrm{\ \ \  \ \ \ }
[xb_3x'a_3'x''b_4''b_2'''b_4'''] \\
&[xb_3x'a_4'x''b_3''b_2'''b_3'''] \textrm{\ \ \  \ \ \ } [xb_3x'a_4'x''b_3''b_2'''b_4'''] \textrm{\ \ \  \ \ \ }
[xb_3x'a_4'x''b_4''b_2'''b_3''']
\textrm{\ \ \  \ \ \ }
[xb_3x'a_4'x''b_4''b_2'''b_4'''] \\
&[xb_4x'a_3'x''b_3''b_2'''b_3'''] \textrm{\ \ \  \ \ \ } [xb_4x'a_3'x''b_3''b_2'''b_4'''] \textrm{\ \ \  \ \ \ }
[xb_4x'a_3'x''b_4''b_2'''b_3''']
\textrm{\ \ \  \ \ \ }
[xb_4x'a_3'x''b_4''b_2'''b_4'''] \\
&[xb_4x'a_4'x''b_3''b_2'''b_3'''] \textrm{\ \ \  \ \ \ } [xb_4x'a_4'x''b_3''b_2'''b_4'''] \textrm{\ \ \  \ \ \ }
[xb_4x'a_4'x''b_4''b_2'''b_3''']
\textrm{\ \ \  \ \ \ }
[xb_4x'a_4'x''b_4''b_2'''b_4''']
\end{split}
\end{equation}
We note that $x_2$ shares an edge with $x$, $x'$ and every other vertex in each of the 7-simplices in \Cref{eq:simplices 2}.
Therefore, as before, each simplex $s \in c_0$ has been removed, but replaced by two new 7-simplices, one $s'$ containing the edge $[xx_2]$ and the other $s''$ containing the edge $[x'x_2]$, where $s',s''$ share the 6-simplex given by removing $x$ and $x'$ from the vertex set, and adding $x_2$.
Again we have to consider the other simplices that shared facets with the simplices in \Cref{eq:simplices 2} to demonstrate that there are no new holes in $\mathcal{C}_0$.
In $c_0$ the other facets of each simplex $s$ in \Cref{eq:simplices 2} were all shared with either:
\begin{itemize}
\item another 7-simplex, $t$, that contains one of $x,x'$. Since $t$ only contains one of $x,x'$ we have that $t \in \mathcal{C}_0$, and moreover $t$ will now share a facet with either $s'$ or $s''$ (depending on whether it is $x$ or $x'$ that is contained in $t$)
\item another 7-simplex, $t$, that contains $x,x'$ and no other vertices from $\mathcal{S}$. Then $t$ is one of the simplices from \Cref{eq:simplices 1} that has been replaced by two simplices $t',t'' \in \mathcal{C}_0$. We note that $t,t''$ do \emph{not} share a facet with $s'$ or $s''$, since these simplices all differ on two vertices, rather than one. 
However, $x_1$ and $x_2$ share an edge, and this leads to the the inclusion of a set of simplices in $\mathcal{C}_0$ which do not directly arise from simplices in $c_0$, but which fill in holes between simplices. These are:
\begin{equation}\label{eq:simplices ff}
\begin{split}
&[x_1x_2xb_3a_3'b_3''b_2'''b_3'''] \textrm{\ \ \  \ \ \ } [x_1x_2xb_3a_3'b_3''b_2'''b_4'''] \textrm{\ \ \  \ \ \ }
[x_1x_2xb_3a_3'b_4''b_2'''b_3''']
\textrm{\ \ \  \ \ \ }
[x_1x_2xb_3a_3'b_4''b_2'''b_4'''] \\
&[x_1x_2xb_3a_4'b_3''b_2'''b_3'''] \textrm{\ \ \  \ \ \ } [x_1x_2xb_3a_4'b_3''b_2'''b_4'''] \textrm{\ \ \  \ \ \ }
[x_1x_2xb_3a_4'b_4''b_2'''b_3''']
\textrm{\ \ \  \ \ \ }
[x_1x_2xb_3a_4'b_4''b_2'''b_4'''] \\
&[x_1x_2xb_4a_3'b_3''b_2'''b_3'''] \textrm{\ \ \  \ \  \ } [x_1x_2xb_4a_3'b_3''b_2'''b_4'''] \textrm{\ \ \  \ \ \ }
[x_1x_2xb_4a_3'b_4''b_2'''b_3''']
\textrm{\ \ \  \ \  \ }
[x_1x_2xb_4a_3'b_4''b_2'''b_4'''] \\
&[x_1x_2xb_4a_4'b_3''b_2'''b_3'''] \textrm{\ \ \  \ \  \ } [x_1x_2xb_4a_4'b_3''b_2'''b_4'''] \textrm{\ \ \ \  \ \  }
[x_1x_2xb_4a_4'b_4''b_2'''b_3''']
\textrm{\ \ \ \  \ \  }
[x_1x_2xb_4a_4'b_4''b_2'''b_4'''] \\ 
&[x_1x_2x'b_3a_3'b_3''b_2'''b_3'''] \textrm{  \  \ \ \ } [x_1x_2x'b_3a_3'b_3''b_2'''b_4'''] \textrm{\  \  \ \ \ }
[x_1x_2x'b_3a_3'b_4''b_2'''b_3''']
\textrm{  \  \ \ \ }
[x_1x_2x'b_3a_3'b_4''b_2'''b_4'''] \\
&[x_1x_2x'b_3a_4'b_3''b_2'''b_3'''] \textrm{  \  \ \ \ } [x_1x_2x'b_3a_4'b_3''b_2'''b_4'''] \textrm{\ \ \  \ \ }
[x_1x_2x'b_3a_4'b_4''b_2'''b_3''']
\textrm{ \ \  \ \  }
[x_1x_2x'b_3a_4'b_4''b_2'''b_4'''] \\
&[x_1x_2x'b_4a_3'b_3''b_2'''b_3'''] \textrm{ \ \  \ \  } [x_1x_2x'b_4a_3'b_3''b_2'''b_4'''] \textrm{\ \ \  \ \  }
[x_1x_2x'b_4a_3'b_4''b_2'''b_3''']
\textrm{\ \ \  \ \  }
[x_1x_2x'b_4a_3'b_4''b_2'''b_4'''] \\
&[x_1x_2x'b_4a_4'b_3''b_2'''b_3'''] \textrm{\ \ \  \ \  } [x_1x_2x'b_4a_4'b_3''b_2'''b_4'''] \textrm{\ \ \  \ \  }
[x_1x_2x'b_4a_4'b_4''b_2'''b_3''']
\textrm{\ \ \  \ \  }
[x_1x_2x'b_4a_4'b_4''b_2'''b_4''']
\end{split}
\end{equation}
Then we have that $s',t'$ each share a facet with the same the 7-simplex, from \Cref{eq:simplices ff}, and likewise $s'',t''$.
\item another 7-simplex $t$ that contains $x$, $x'$, $x''$ and no other vertices from $\mathcal{S}$. Then $t$ is another one of the simplices from \Cref{eq:simplices 2} that is replaced with two simplices $t',t'' \in \mathcal{C}_0$, and $t',s'$ will share a facet, while $t'',s''$ share a facet
\item another 7-simplex $t$ that contains  $\mathcal{S}$. Then $t$ will share a facet with one of the other new simplices in $\mathcal{C}_0$ that we will consider next
\end{itemize}
So again, even though simplices have been removed from the complex they have been replaced in exactly the correct way to ensure that no new holes are added.

The 7-simplices which share the 2-simplex $[xx'x''']$ with $\mathcal{S}$ work in the same way, except in this case it is $x_1$ which shares an edge with both $x$, $x'$ and every other vertex in the vertex set (In this case we do not have the complication that we need to consider simplices from \Cref{eq:simplices ff} because the we only need to consider simplices with the dummy vertex $x_1$).

The final case to consider is the 7-simplices which contain $\mathcal{S}$.
We have 16 such 7-simplices:
\begin{equation}
\begin{split}
&[xb_3x'a_3'x''b_3''x'''b_3'''] \textrm{\ \ \  \ \ \ } [xb_3x'a_3'x''b_3''x'''b_4'''] \textrm{\ \ \  \ \ \ }
[xb_3x'a_3'x''b_4''x'''b_3''']
\textrm{\ \ \  \ \ \ }
[xb_3x'a_3'x''b_4''x'''b_4'''] \\
&[xb_3x'a_4'x''b_3''x'''b_3'''] \textrm{\ \ \  \ \ \ } [xb_3x'a_4'x''b_3''x'''b_4'''] \textrm{\ \ \  \ \ \ }
[xb_3x'a_4'x''b_4''x'''b_3''']
\textrm{\ \ \  \ \ \ }
[xb_3x'a_4'x''b_4''x'''b_4'''] \\
&[xb_4x'a_3'x''b_3''x'''b_3'''] \textrm{\ \ \  \ \ \ } [xb_4x'a_3'x''b_3''x'''b_4'''] \textrm{\ \ \  \ \ \ }
[xb_4x'a_3'x''b_4''x'''b_3''']
\textrm{\ \ \  \ \ \ }
[xb_4x'a_3'x''b_4''x'''b_4'''] \\
&[xb_4x'a_4'x''b_3''x'''b_3'''] \textrm{\ \ \  \ \ \ } [xb_4x'a_4'x''b_3''x'''b_4'''] \textrm{\ \ \  \ \ \ }
[xb_4x'a_4'x''b_4''x'''b_3''']
\textrm{\ \ \  \ \ \ }
[xb_4x'a_4'x''b_4''x'''b_4''']
\end{split}
\end{equation}
The edges between the dummy vertices $\{x_3 \dots x_{10}\}$ and the rest of the complex are designed to ensure there are no gaps between the simplices that replace these simplices in $\mathcal{C}_0$.
We will show how this works for the simplex $[xb_3x'a_3'x''b_3''x'''b_3''']$ -- the other simplices can be analysed in a similar manner.
From \Cref{sec:4_qubit_construction} we can see that the vertex $x_9$ is connected to every vertex in this simplex.
Therefore, this 7-simplex is cut open, and forms two new 7-simplices, one which contains the 1-simplex $[xx_9]$ and one which contains the 1-simplex $[x'x_9]$.
There are eight facets of $[xb_3x'a_3'x''b_3''x'''b_3''']$ that we need to consider:
\begin{itemize}
\item $[b_3x'a_3'x''b_3''x'''b_3''']$ -- in $s'$ this facet becomes  $[b_3x_9a_3'x''b_3''x'''b_3''']$, which is shared with $s''$. In $s''$ this remains  $[b_3x'a_3'x''b_3''x'''b_3''']$, which is shared with $[b_2 b_3x'a_3'x''b_3''x'''b_3''']$
\item $[xx'a_3'x''b_3''x'''b_3''']$ -- in $s'$ this facet becomes $[xx_9a_3'x''b_3''x'''b_3''']$ which is shared with $[xa_4 x_9a_3'x''b_3''x'''b_3''']$. In $s''$ this becomes $[x_9x'a_3'x''b_3''x'''b_3''']$ which is shared with $[x_9a_4x'a_3'x''b_3''x'''b_3''']$ 
\item $[xb_3a_3'x''b_3''x'''b_3''']$ -- in $s'$ this remains $[xb_3a_3'x''b_3''x'''b_3''']$ which is shared with $[xb_3a'_2a_3'x''b_3''x'''b_3''']$. In $s''$ this becomes $[x_9b_3a_3'x''b_3''x'''b_3''']$ which is shared with $s'$.
\item $[xb_3x'x''b_3''x'''b_3''']$ -- in $s'$ this becomes $[xb_3x_9x''b_3''x'''b_3''']$ which is shared with $[xb_3x_9a'_4x''b_3''x'''b_3''']$. In $s''$ this becomes $[x_9b_3x'x''b_3''x'''b_3''']$ which is shared with $[x_9b_3x'a'_4x''b_3''x'''b_3''']$.
\item $[xb_3x'a_3'b_3''x'''b_3''']$ -- in $s'$ this becomes $[xb_3x_9a_3'b_3''x'''b_3''']$ which is shared with $[xb_3x_9a_3'x_1b_3''x'''b_3''']$. In $s''$ this becomes $[x_9b_3x'a_3'b_3''x'''b_3''']$ which is shared with $[x_9b_3x'a_3'x_1b_3''x'''b_3''']$.
\item $[xb_3x'a_3'x''x'''b_3''']$ -- in $s'$ this becomes $[xb_3x_9a_3'x''x'''b_3''']$ which is shared with $[xb_3x_9a_3'x''b_4''x'''b_3''']$. In $s''$ this becomes $[x_9b_3x'a_3'x''x'''b_3''']$ which is shared with $[x_9b_3x'a_3'x''b_4''x'''b_3''']$.
\item $[xb_3x'a_3'x''b_3''b_3''']$ -- in $s'$ this becomes $[xb_3x_9a_3'x''b_3''b_3''']$ which is shared with $[xb_3x_9a_3'x''b_3''x_2b_3''']$. In $s''$ this becomes $[x_9b_3x'a_3'x''b_3''b_3''']$ which is shared with $[x_9b_3x'a_3'x''b_3''x_2b_3''']$.
\item $[xb_3x'a_3'x''b_3''x''']$ -- in $s'$ this becomes $[xb_3x'_9a_3'x''b_3''x''']$ which is shared with $[xb_3x_9a_3'x''b_3''x'''x_8]$. In $s''$ this becomes $[x_9b_3x'a_3'x''b_3''x''']$ which is shared with $[x_9b_3x'a_3'x''b_3''x'''x_8]$.
\end{itemize}
So there are no resulting holes in the complex.

So we have shown that each of the 7-simplices in the cycle $c_0$ are still in $\mathcal{C}_0$, or have been removed and replaced with a pair of simplices which `fill in' the same point in the complex, leaving no gaps.
Next we consider the new simplices added by the dummy vertices, which are not simply replacing simplices from $c_0$.
That is, simplices which contain two or more dummy vertices.
The thickening procedure is designed precisely so that every new simplex that is added shares every facet (except those that bound the hole $\mathfrak{g}_7)$ with another simplex in the complex, leaving no gaps (see \Cref{sec:thickening}).
So, as claimed, the only hole we have cut in the copy of $\mathfrak{g}_8$ is that given by the copy of $\mathfrak{g}_7$ formed out of the vertices from $\mathcal{S}$ along with the dummy vertices (the complex $\mathcal{C}_0$ has trivial homology despite being formed by cutting a hole in $\mathfrak{g}_8$ becaues the process of cutting a hole takes a sphere, which has non-trivial homology, to a hyper-plane, which has trivial homology).

Precisely the same argument applies to cutting open the cycle for the basis state $\ket{1000}$.
So we are left with complexes which are each triangulations of a 7-dimensional hyper-planes with boundaries of the same size.
Gluing these hyper-planes together along the boundaries clearly gives a triangulation of $S^7$ as claimed.
Moreover, at every step when constructing this complex we have added edges, and every simplex induced by the edge.
So the resulting complex is 2-determined, and therefore a clique-complex.

\end{proof}

\begin{lemma}
The procedure outlined in \Cref{sec:3 qubit construction} for cutting and gluing more than three three-qubit cycles results in a clique complex $\Kcyc$ which triangulates $S^5$.
\end{lemma}
\begin{proof}
For concreteness we will consider the state $-5\ket{011}+4\ket{100}+3\ket{101}$ throughout the proof, but it trivially generalises to other three-qubit states.

First we demonstrate that applying the cutting process from \Cref{sec:3 qubit cutting} to the cycle $c_0$ corresponding to the basis state $\ket{011}$ results in a complex $\mathcal{C}_0$ which is a triangulation of $S^5$ with a copy of $\mathfrak{g}_5$ cut into it by the dummy vertices.
Is is immediate from \Cref{sec:3 qubit cutting} that the dummy vertices $\{x_1 \dots x_9\}$ along with the original vertices $x,x',x''$ form a copy of $\mathfrak{g}_5$.
So all that remains is to show that no additional holes were formed in the process of constructing this hole.

As in the proof of the previous lemma, we will do this by showing that the facet of every simplex in $\mathcal{C}_0$ is either shared with another simplex in $\mathcal{C}_0$, or is a face of the hole $\mathfrak{g}_5$.
Again, we will only need to show this for maximal simplices, which are 5-simplices in this case.

In this case the hole in $c_0$ is cut into the 2-simplex $\mathcal{S} = [xx'x'']$.
So any simplex from $c_0$ that shares no lower simplices with $\mathcal{S}$ is present in $\mathcal{C}_0$, and shares all facets with other simplices in $\mathcal{C}_0$.

We also have immediately that any simplex that shares exactly one vertex with $\mathcal{S}$ is still in the complex $\mathcal{C}_0$ and shares all facets with other simplices in the complex because the lower simplices they contain are unaffected by the cutting procedure.
Moreover, any simplex that shares only the 1-simplices $[xx'']$ or $[x'x'']$ with $\mathcal{S}$ is still in the complex $\mathcal{C}_0$ and shares all facets with other simplices in the complex for the same reason.

Next consider 5-simplices which share the 1-simplex $[xx']$ with $\mathcal{S}$.
There are eight such 5-simplices:
\begin{equation}\label{eq:3 simplices}
\begin{split}
&[xa_3x'b_3'b_2''b_3''] \textrm{\ \ \  \ \ \ } [xa_3x'b_3'b_2''b_4'']\textrm{\ \ \  \ \ \ } [xa_3x'b_4'b_2''b_3'']\textrm{\ \ \  \ \ \ } [xa_3x'b_4'b_2''b_4''] \\
& [xa_4x'b_3'b_2''b_3''] \textrm{\ \ \  \ \ \ } [xa_4x'b_3'b_2''b_4'']\textrm{\ \ \  \ \ \ } [xa_4x'b_4'b_2''b_3'']\textrm{\ \ \  \ \ \ } [xa_4x'b_4'b_2''b_4'']
\end{split}
\end{equation}
The dummy vertex $x_1$ is connected to the original vertices $\{x,a_3,a_4,x',b'_3,b'_4,b''_2,b''_3,b''_4\}$ (see \Cref{sec:3 qubit cutting}).
Therefore each simplex $s \in c_0$ from \Cref{eq:3 simplices} is replaced in $\mathcal{C}_0$ by two simplices -- one $s'$ containing the 1-simplex $[xx_1]$ and the other $s''$ containing the 1-simplex $[x'x_1]$.
The simplices $s',s''$ share the 4-simplex given by removing the vertices $x$ and $x'$ from the original vertex set, and adding the vertex $x_1$. 
In $c_0$  the other facets of each simplex $s$ in \Cref{eq:3 simplices} were shared with either:
\begin{itemize}
\item another 5-simplex, $t$, that contains one of $x,x'$. Since $t$ only contains one of $x,x'$ we have that $t \in \mathcal{C}_0$, and moreover $t$ will now share a facet with either $s$ or $s''$ (depending on whether it is $x$ or $x'$ that is contained in $t$)
\item another 5-simplex that contains $x,x'$ but not $x''$. Then $t$ is one of the other simplices from \Cref{eq:3 simplices} that has now been replaced by two simplices $t',t'' \in \mathcal{C}_0$. We will have that $t'$ shares a facet with $s'$, and $t''$ shares a facet with $s''$
\item another 5-simplex that contains $x,x',x''$ -- we will consider such 5-simplices next
\end{itemize}
So while the 5-simplices in \Cref{eq:3 simplices} are not present in $\mathcal{C}_0$, they are replaced by new 5-simplices that share every facet with a simplex that is in $\mathcal{C}_0$.

Next we consider simplices that contain $\mathcal{S}$.
There are eight such 5-simplices:
\begin{equation}\label{eq:3 simplices 2}
\begin{split}
&[xa_3x'b_3'x''b_3''] \textrm{\ \ \  \ \ \ } [xa_3x'b_3'x''b_4'']\textrm{\ \ \  \ \ \ } [xa_3x'b_4'x''b_3'']\textrm{\ \ \  \ \ \ } [xa_3x'b_4'x''b_4''] \\
& [xa_4x'b_3'x''b_3''] \textrm{\ \ \  \ \ \ } [xa_4x'b_3'x''b_4'']\textrm{\ \ \  \ \ \ } [xa_4x'b_4'x''b_3'']\textrm{\ \ \  \ \ \ } [xa_4x'b_4'x''b_4'']
\end{split}
\end{equation}
The edges between the dummy vertices $x_2 \dots x_9$ are designed to ensure there are no gaps between the simplices that replace these simplices in $\mathcal{C}_0$.
We will show how this works for the simplex $s=[xa_3x'b_3'x''b_3'']$, the other simplices can be analysed in a similar manner.
From \Cref{sec:3 qubit cutting} we can see that the simplex $x_2$ is connected to every vertex in this simplex.
Therefore this 5-simplex is cut open and forms two new 5-simplices: $s=[xa_3x_2b_3'x''b_3'']$, $s=[x_2a_3x'b_3'x''b_3'']$.
There are six facets of $s$ that we need to consider:
\begin{itemize}
\item $[a_3x'b_3'x''b_3'']$ -- in $s'$ this facet becomes $[a_3x_2b_3'x''b_3'']$, which is shared with $s'$. In $s''$ this is unchanged, and is shared with $[a_2a_3x'b_3'x''b_3'']$.
\item $[xx'b_3'x''b_3'']$ -- in $s'$ this facet becomes $[xx_2b_3'x''b_3'']$, which is shared with $[xx_3x_2b_3'x''b_3'']$. In $s''$ this facet becomes $[x_2x'b_3'x''b_3'']$, which is shared with $[x_2x_3x'b_3'x''b_3'']$.
\item $[xa_3b_3'x''b_3'']$ -- in $s'$ this facet is unchanged, and is shared with $[xa_3b_2'b_3'x''b_3'']$. In $s''$ this facet becomes $[x_2a_3b_3'x''b_3'']$ and is shared with $s'$.
\item $[xa_3x'x''b_3'']$ -- in $s'$ this facet becomes $[xa_3x_2x''b_3'']$, which is shared with $[xa_3x_2x_8x''b_3'']$. In $s''$ this facet becomes $[x_2a_3x'x''b_3'']$, and is shared with $[x_2a_3x'x_8x''b_3'']$.
\item $[xa_3x'b_3'b_3'']$ -- in $s'$ this facet becomes $[xa_3x_2b_3'b_3'']$, which is shared with $[xa_3x_2b_3'x_1b_3'']$. In $s''$ this facet becomes $[x_2a_3x'b_3'b_3'']$, and is shared with $[x_2a_3x'b_3'x_1b_3'']$.
\item $[xa_3x_2b_3'x'']$ -- in $s'$ this facet becomes $[xa_3x_2b_3'x'']$, which is shared with $[xa_3x'b_3'x''x_4]$. In $s''$ this facet becomes $[x_2a_3x'b_3'x'']$, and is shared with $[x_2a_3x'b_3'x''x_4]$.
\end{itemize}
So there are no resulting holes in the complex.
We have shown that each of the 5-simplices in the cycle $c_0$ is either still in $\mathcal{C}_0$, or have been removed and replaced with a pair of simplices which share every facet with another facet in the complex, leaving no gaps.
Now we consider the new simplices added by the dummy vertices, which are not simply replacing simplices from $c_0$.
That is, simplices which contain two or more dummy vertices.
It can be seen by inspection of \Cref{fig:3 qubit 7}, \Cref{fig:3 qubit 6} and \Cref{fig:3 qubit 8} that every cycle involving the dummy vertices can be continuously deformed to a cycle  which includes only the original vertices. 
Therefore, every facet of every new simplex must be shared with another simplex in the complex, except the simplices that bound the hole $\mathfrak{g}_5$ (the complex $\mathcal{C}_0$ has trivial homology despite beging formed by cutting a hole in $\mathfrak{g}_6$ because the process of cutting a hole takes a sphere to a hyper-plane). 

Precisely the same argument applies to cutting open five more copies of $c_0$, and to cutting open four copies of $c_1$ (the cycle corresponding to $\ket{100}$) and three copies of $c_2$ (the cycle corresponding to $\ket{101}$).
Therefore if we label the dummy vertices as follows:
\begin{itemize}
\item $c_0^{(1)}$: $x_1 \dots x_{9}$
\item $c_0^{(2)}$: $x_1$ \& $x_6 \dots x_{13}$
\item $c_0^{(3)}$: $x_1$ \& $x_{10} \dots x_{17}$
\item $c_0^{(4)}$: $x_1$ \& $x_{14} \dots x_{21}$
\item $c_0^{(5)}$: $x_1$ \& $x_{18} \dots x_{25}$
\item $c_1^{(1)}$: $x_1$ \& $x_{22} \dots x_{29}$
\item $c_1^{(2)}$: $x_1$ \&  $x_{26} \dots x_{33}$
\item $c_1^{(3)}$: $x_1$ \& $x_{30} \dots x_{37}$
\item $c_1^{(4)}$: $x_1$ \& $x_{34} \dots x_{41}$
\item $c_2^{(1)}$: $x_1$ \& $x_{38} \dots x_{45}$
\item $c_2^{(2)}$: $x_1$ \& $x_{42} \dots x_{49}$
\item $c_2^{(3)}$: $x_1$ \& $x_{46} \dots x_{49}$ \& $x_2 \dots x_5$
\end{itemize}
Then we are taking 12 copies of the hyper plane, and gluing them along common boundaries.
However, care has to be taken because this way of gluing together hyperplanes does not glue together the entire boundary (as in the case when we are just gluing two hyperplanes together), instead it only glues along parts of the boundaries -- see \Cref{fig:3 qubit 11}.
This procedure therefore leads to a triangulation of $S^5$ with a number of holes in it.
To fill in the holes we add additional dummy vertices, connected up to the existing dummy vertices -- see \Cref{fig:3 qubit 12} and \Cref{fig:3 qubit 13}.
It can be checked that this process of identifying vertices and adding additional dummy vertices closes every hole in the surface of the triangulation of $S^5$, and therefore gives a complex which is a triangulation of $S^5$ as claimed.
Moreover, at every point when constructing the complex whenever we have added an edge we have also added every simplex containing that edge, so the resulting complex is 2-determined, and hence a clique complex.

\end{proof}

\pagebreak
\bibliographystyle{alpha}
\bibliography{refs}

\end{document}